\documentclass[aps,pra,floatfix,superscriptaddress,twocolumn,10pt,bibliography]{revtex4-2}

\usepackage[page,header]{appendix}

\usepackage{amsmath,amsfonts,amssymb,amstext, amsthm}
\usepackage[sort&compress]{natbib}
\usepackage[pdftex]{graphicx}

\usepackage[dvipsnames,svgnames]{xcolor}
\usepackage{comment}

\usepackage{mathtools}
\usepackage{nccmath}
\usepackage{blkarray}
\usepackage{here}

\usepackage{array}

\usepackage{afterpage}

\usepackage{makecell}
\usepackage{tikz}
\usetikzlibrary{decorations.pathreplacing}

\usepackage{enumitem}

\usepackage[linesnumbered,ruled,vlined]{algorithm2e}
\usepackage[most]{tcolorbox}

\usepackage{enumitem}

\usepackage{thmtools}

\usepackage{url}
\usepackage[linktocpage, colorlinks]{hyperref}
\hypersetup{
	colorlinks = true,
	citecolor = blue,
	linkcolor = magenta,
	urlcolor = brown,
}

\usepackage[normalem]{ulem}

\usepackage[T1]{fontenc}
\usepackage{tgtermes}
\usepackage{newtxtext,newtxmath}

\usepackage{cleveref}

\Crefname{section}{Sec.}{Secs.}
\Crefname{subsection}{Sec.}{Secs.}
\Crefname{subsubsection}{Sec.}{Secs.}
\Crefname{figure}{Fig.}{Figs.}
\Crefname{table}{Table}{Tables}
\Crefname{algorithm}{Algorithm}{Algorithms}
\Crefname{theorem}{Theorem}{Theorems}
\Crefname{lemma}{Lemma}{Lemmas}
\Crefname{proposition}{Proposition}{Propositions}
\Crefname{corollary}{Corollary}{Corollaries}
\Crefname{appendix}{Appendix}{Appendices}

\newtheorem{theorem}{Theorem}
\newtheorem{proposition}[theorem]{Proposition}
\newtheorem{lemma}[theorem]{Lemma}
\newtheorem{corollary}[theorem]{Corollary}

\theoremstyle{definition}

\theoremstyle{definition}

\newcommand{\ket}[1]{|#1\rangle}
\newcommand{\bra}[1]{\langle#1|}
\newcommand{\ketbra}[2]{|#1\rangle\langle#2|}
\newcommand{\braket}[2]{\langle #1 \vert #2 \rangle}

\DeclareMathOperator{\tr}{Tr}

\newcommand{\rid}{{\mathrm{id}}}
\newcommand{\bR}{\mathbb{R}}
\newcommand{\bC}{\mathbb{C}}
\newcommand{\bN}{\mathbb{N}}
\newcommand{\bZ}{\mathbb{Z}}
\newcommand{\bI}{\mathbb{I}}
\newcommand{\bE}{\mathbb{E}}

\newcommand{\cA}{{\mathcal{A}}}
\newcommand{\cB}{{\mathcal{B}}}
\newcommand{\cC}{{\mathcal{C}}}
\newcommand{\cD}{{\mathcal{D}}}
\newcommand{\cE}{{\mathcal{E}}}
\newcommand{\cF}{{\mathcal{F}}}
\newcommand{\cG}{{\mathcal{G}}}
\newcommand{\cH}{{\mathcal{H}}}

\newcommand{\cJ}{{\mathcal{J}}}
\newcommand{\cK}{{\mathcal{K}}}
\newcommand{\cL}{{\mathcal{L}}}
\newcommand{\cM}{{\mathcal{M}}}
\newcommand{\cN}{{\mathcal{N}}}
\newcommand{\cO}{{\mathcal{O}}}
\newcommand{\cP}{{\mathcal{P}}}
\newcommand{\cQ}{{\mathcal{Q}}}
\newcommand{\cR}{{\mathcal{R}}}

\newcommand{\cT}{{\mathcal{T}}}
\newcommand{\cU}{{\mathcal{U}}}
\newcommand{\cV}{{\mathcal{V}}}
\newcommand{\cW}{{\mathcal{W}}}

\newcommand{\rF}{{\mathrm{F}}}

\newcommand{\rL}{{\mathrm{L}}}

\newcommand{\rR}{{\mathrm{R}}}

\newcommand{\rU}{{\mathrm{U}}}

\newcommand{\sA}{{\mathsf{A}}}
\newcommand{\sB}{{\mathsf{B}}}
\newcommand{\sC}{{\mathsf{C}}}
\newcommand{\sD}{{\mathsf{D}}}
\newcommand{\sE}{{\mathsf{E}}}
\newcommand{\sF}{{\mathsf{F}}}
\newcommand{\sG}{{\mathsf{G}}}
\newcommand{\sH}{{\mathsf{H}}}

\newcommand{\sJ}{{\mathsf{J}}}
\newcommand{\sK}{{\mathsf{K}}}
\newcommand{\sL}{{\mathsf{L}}}
\newcommand{\sM}{{\mathsf{M}}}
\newcommand{\sN}{{\mathsf{N}}}

\newcommand{\sP}{{\mathsf{P}}}
\newcommand{\sQ}{{\mathsf{Q}}}
\newcommand{\sR}{{\mathsf{R}}}
\newcommand{\sS}{{\mathsf{S}}}

\newcommand{\sX}{{\mathsf{X}}}
\newcommand{\sY}{{\mathsf{Y}}}

\newcommand{\tF}{\mathtt{F}}

\newcommand{\f}{\frac}

\newcommand{\rarr}{\rightarrow}

\newcommand{\til}{\tilde}

\newcommand{\supp}{\mathrm{supp}}
\newcommand{\ctrl}{\mathrm{ctrl}}

\newcommand{\poly}{\mathrm{poly}}
\newcommand{\sgn}{\mathrm{sgn}}
\newcommand{\im}{\mathrm{im}}

\newtcolorbox{dashedbox}[1][]{
  enhanced,
  sharp corners,
  boxrule=0.5pt,
  colback=white,
  colframe=black,
  fonttitle=\bfseries,
  title=#1,
  dash pattern=on 2pt off 2pt,
  borderline={0.5pt}{0pt}{black!60!white,dashed},
  before skip=5pt, after skip=5pt,
  boxsep=5pt
}

\begin{document}

\begin{titlepage}
\thispagestyle{empty}
\setcounter{page}{1}

\begin{flushright}
YITP-25-135
\end{flushright}

\title{Quantum algorithms for Uhlmann transformation}

\author{Takeru Utsumi} 
\email{takeru-utsumi@g.ecc.u-tokyo.ac.jp}
\affiliation{Graduate School of Arts and Sciences, The University of Tokyo, Komaba, Meguro-ku, Tokyo 153-8902, Japan}

\author{Yoshifumi Nakata} 
\email{nakata@comp.isct.ac.jp}
\affiliation{School of Computing, Department of Computer Science, Institute of Science Tokyo, 4259 Nagatsuta-cho, Midori-ku, Yokohama, Kanagawa 226-8501, Japan}
\affiliation{Yukawa Institute for Theoretical Physics, Kyoto University, Kitashirakawa Oiwake-Cho, Sakyo-Ku, Kyoto 606-8502, Japan}

\author{Qisheng Wang} 
\email{QishengWang1994@gmail.com}
\affiliation{School of Computer Science, Shanghai Jiao Tong University, Shanghai 200240, China}

\author{Ryuji Takagi} 
\email{ryujitakagi.pat@gmail.com}
\affiliation{Graduate School of Arts and Sciences, The University of Tokyo, Komaba, Meguro-ku, Tokyo 153-8902, Japan}

\date{\today}

\begin{abstract}

Uhlmann's theorem is a central result in quantum information theory, which associates the closeness of two quantum states with that of their purifications. This theorem characterizes a fundamental task: transforming one quantum state into another via local operations acting only on a subsystem. The optimal transformation for this task is called the Uhlmann transformation, which has applications in various information-processing tasks. However, its quantum circuit implementation and computational cost have remained unclear, limiting its utility. In this work, we fill this gap by proposing quantum algorithms that realize the Uhlmann transformation in query and sample access models. Notably, our algorithms in the purified access models can be implemented in polynomial time for low-rank states, exhibiting an exponential improvement over the prior approach and other naive approaches based on quantum state tomography. We also derive lower bounds on the query and sample costs of the Uhlmann transformation, thereby revealing the computational limits for this task. In the purified access models, these bounds match the costs of our algorithms up to logarithmic factors in terms of the rank, the accuracy, and the minimum non-zero singular value associated with the fidelity. In the mixed sample access model with the canonical purification, our algorithm attains the optimal scaling in the dimension and the accuracy. We apply our Uhlmann transformation algorithms to the fidelity estimation problem and substantially improve the previous best query and sample costs, where our costs attain the optimal scaling in the rank, and the sample cost also in the accuracy. We further discuss other applications to several information-theoretic tasks, including quantum communication, quantum state merging, and the algorithmic implementation of the Petz recovery map, providing a comprehensive evaluation of their computational costs.

\end{abstract}
\maketitle
\end{titlepage}

\tableofcontents

\section{Introduction}
\label{sec:introduction}

Uhlmann's theorem \cite{uhlmann1976transition} has played a crucial role in quantum information
science. The theorem states that for quantum states $\rho^\sA$ and $\sigma^\sA$ on system $\sA$ and their purifications $\ket{\rho}^{\sA\sB}$ and $\ket{\sigma}^{\sA\sB}$, it holds that 
\begin{align}
\label{inteq:57}
    &\rF(\rho^\sA, \sigma^\sA) \notag\\
    &\hspace{1pc}= \max_{U^\sB}\rF\big((\bI^\sA \otimes U^\sB) \ketbra{\rho}{\rho}^{\sA\sB}(\bI^\sA \otimes U^\sB)^\dag, \ketbra{\sigma}{\sigma}^{\sA\sB}\big),
\end{align}
where the maximization is taken over all unitaries acting on the purifying system $\sB$, and $\bI^\sA$ is the identity operator on $\sA$.
The fidelity $\rF(\rho, \sigma) = \big(\tr\big[\sqrt{\sqrt{\sigma}\rho\sqrt{\sigma}}\big]\big)^2$ is a closeness measure of quantum states $\rho$ and $\sigma$; it increases as the states become closer. 
Notably, Uhlmann's theorem characterizes the optimal local transformation between $\ket{\rho}^{\sA\sB}$ and $\ket{\sigma}^{\sA\sB}$ because, for any quantum channel $\cT^\sB$ acting on $\sB$, $\rF\big((\rid^\sA \otimes \cT^\sB)(\ketbra{\rho}{\rho}^{\sA\sB}), \ketbra{\sigma}{\sigma}^{\sA\sB}\big)$ never exceeds $\rF(\rho^\sA, \sigma^\sA)$, where $\rid^\sA$ is the identity map on $\sA$. 
This implies that a unitary achieving the maximization in Eq.~\eqref{inteq:57} is the optimal local transformation on $\sB$, which takes $\ket{\rho}^{\sA\sB}$ as close as theoretically possible to $\ket{\sigma}^{\sA\sB}$. 
This unitary $U^\sB$ is called the Uhlmann unitary, and the resulting transformation is the Uhlmann transformation. The Uhlmann unitary $U^\sB$ is not determined solely by the states $\rho^\sA$ and $\sigma^\sA$, but only once their purifications are specified.
The Uhlmann transformation has long served as a key tool in various fields, including quantum Shannon theory \cite{devetak2005distillation, Abeyesinghe2009Motherfamilytree, Datta2011apexfamilytree, Datta2013OneShotCompression, bennett2014quantum, Wilde2017ConvrersePrivate, Metger2024entropyaccumulation, mazzola2025uhlmannstheoremrelativeentropies, fang2025variexpressuhlmann}, complexity theory \cite{kitaev2000parallelamplif, arnonfriedman2023computationalentanglementtheory, metger2023stateqipspace, bostanci2023unicompuhlmann, chia2024complexitytheoryquantumpromise, Poremba2026LearningStabilizersNoise}, quantum cryptography \cite{Mayer1997unconditionallysecure, Lo1998WhyBitCommit, Yan2022BitCommit, Khurana2024commitment}, and also fundamental physics \cite{hayden2007black, harlow2013CompvsFirewall, aaronson2016moneyblackhole, Kirklin2020holographicdual, Kirklin2022islandUhlmannphase, brakerski2023blackholeradiation}.

Despite its broad applicability, however, the Uhlmann transformation has an underexplored aspect: Uhlmann's theorem does not offer an explicit way to realize the transformation in the form of quantum circuits, leaving its algorithmic properties open.
In particular, when we do not have prior knowledge of the descriptions of the states $\ket{\rho}^{\sA\sB}$ and $\ket{\sigma}^{\sA\sB}$, yet have multiple access to samples of these states or queries to the unitaries that prepare them, the following question naturally arises:
\begin{center}
\emph{How can we construct a quantum circuit implementing the Uhlmann transformation using these resources?}
\end{center}
As quantum technologies advance and practical interest grows, the importance of these algorithmic features, such as the computational cost, is increasing.

Recently, there has been progress along this line.
In Ref.~\cite{metger2023stateqipspace}, a quantum algorithm was proposed with an explicit quantum circuit for implementing the Uhlmann transformation. Their algorithm realizes the transformation with polynomial-space computation, marking substantial progress. 
However, it incurs unavoidable exponential costs other than space complexity, because, in our setting, it relies on an exponentially large number of purified state preparations and measurements, as well as an exponentially deep quantum circuit (see \Cref{sec:metger and yuen's algorithm} for details).
These high computational costs motivate us to develop more efficient algorithms.

In this paper, we address this question by proposing quantum query and sample algorithms that implement the Uhlmann transformation in the form of a quantum circuit.
Our algorithms are constructed in three commonly studied computational models: the purified query, the purified sample, and the mixed sample access models.
We then analyze their computational costs, including the query and sample complexities, the number of one- and two-qubit gates, and the number of qubits at any one time.
It is shown that our algorithms can significantly outperform the approach in Ref.~\cite{metger2023stateqipspace} and other naive approaches based on quantum state tomography.
In particular, for states $\rho^\sA$ and $\sigma^\sA$ whose smaller rank is $r_{\rm min}$, our algorithms realize the Uhlmann transformation within accuracy $\delta$ in the fidelity using $\cO(r_{\rm min}/\delta)$ queries in the purified query access model and $\cO(r_{\rm min}^2/\delta^2)$ samples in the purified sample access model. Thus, when either $\rho^\sA$ or $\sigma^\sA$ is a polynomial-rank state, our algorithms can be implemented with polynomial computational cost up to inverse-polynomial accuracy, achieving an \emph{exponential} improvement.
While previous approaches rely on a large number of state measurements, our improvement is built on leveraging powerful techniques such as the quantum singular value transformation (QSVT) \cite{gilyen2019qsvt, Gilyn2019thesis} and density matrix exponentiation \cite{lloyd2014QuantPrincCompAnaly}, enabling us to perform the Uhlmann transformation without state measurements.

In addition, we provide lower bounds on the query and sample complexities for the Uhlmann transformation in the three models, which hold for any quantum algorithm. In the purified query and sample access models, $\Omega(r_{\rm min}/\delta)$ queries and $\Omega(r_{\rm min}^2/\delta^2)$ samples are required, respectively, which shows that our algorithms in these models are optimal with respect to $r_{\rm min}$ and $\delta$. We obtain these lower bounds by hiding an unknown unitary with a small amplitude in one of the purified states, for which implementing the Uhlmann transformation amounts to learning this unitary.

In the mixed sample access model, where the purifications are not specified, we consider the transformation for the so-called canonical purification~\cite{Winter2004ExtrinsicIntrinsic}.
To this end, we first construct a canonical purification algorithm, which may be of independent interest, and then combine it with our algorithm in the purified query access model. The resulting algorithm uses a number of samples that scales linearly with the dimension $d_\sA$ of system $\sA$. Our lower bound shows that, in contrast to the purified models, this dimension dependence cannot be replaced by the rank: $\Omega(d_\sA/\delta)$ samples are required even for pure states. We obtain this bound by reducing the task to preparing the complex conjugate of an unknown pure state. Our algorithm matches this lower bound up to logarithmic factors for pure states with constant fidelity.

We also address fidelity estimation, the task of estimating the fidelity between two states $\rho^\sA$ and $\sigma^\sA$. This task is crucial for information processing, as it allows us to assess the accuracy of quantum state preparation, the reliability of quantum communication, and the outputs of quantum circuits and algorithms. We present estimation algorithms for the square root fidelity $\sqrt{\rF}(\rho^\sA, \sigma^\sA)$ based on the Uhlmann transformation and evaluate their computational costs in the three query and sample models. Our algorithms achieve a significant improvement in both query and sample complexities over the prior state-of-the-art algorithms in Refs.~\cite{gilyen2022improvedfidelity, Liu2024geometricmean, lowe2026randomdimensionreductionlearning}, with the optimal dependence on the rank up to logarithmic factors. In the mixed sample access model with collective operations, our algorithm uses $\til{\cO}(r_{\rm min}^2/\delta^2)$ samples, which is optimal in both $r_{\rm min}$ and $\delta$.

The key idea behind our approach is that, according to Uhlmann's theorem in Eq.~\eqref{inteq:57}, we reduce the problem of estimating the fidelity between generally mixed states $\rho^\sA$ and $\sigma^\sA$ to estimating the fidelity between pure states. The fidelity between these pure states can be estimated by applying the Uhlmann transformation using our algorithms and then known optimal fidelity estimation techniques for pure states~\cite{Wang2024OptimalFidelityPure, wang2024sampleoptimal}. In the mixed sample access model, our approaches require the given mixed states to be purified. To this end, we employ the \emph{random purification} recently proposed in Ref.~\cite{tang2025conjugatequerieshelp}, which converts multiple copies of a mixed state into copies of a purified state that is chosen at random but shared by all the copies. As the fidelity does not depend on the choice of purified states, this suffices for fidelity estimation. The random purification requires collective operations on the copies, so we instead use the canonical purification when they are not available. Since its proposal, the random purification has been actively explored as a subroutine in algorithms, such as those for state tomography~\cite{pelecanos2025mixedtomographypure} and mixed-state cloning~\cite{fanizza2026optimalcloningmixed, jeon2026optimalcopycloning}. Our results are among the first applications of the random purification and provide further examples in which it leads to optimal algorithms.

To highlight the wide-ranging utility of our algorithms, we further explore other applications to several information-theoretic tasks. 
Based on the decoupling approach \cite{hayden2008decoupling, dupuis2010decoupling, Szehr2013decouple2design, dupuis2014one, preskill2016quantumshannon}, the Uhlmann transformation is known as an operation that can achieve theoretical limits, such as the quantum capacity \cite{hayden2008decoupling}.
In this study, we focus on two tasks: quantum communication~\cite{lloyd1997capacity, shor2002quantum, devetak2005private, barnum1998inftrans, datta2013oneshotentassistQCcom, wakakuwa2023oneshot, khatri2024principlemodern} and quantum state merging \cite{Horodecki2005Partialquantinfo, berta2008singleshot, dupuis2014one, yamasaki2019statemergesmalldimension}.
We apply our Uhlmann transformation algorithms to accomplish these tasks in the form of quantum circuits and evaluate their computational costs, which are independent of the dimensions of the systems and polynomial when the decoding party's state has polynomial rank. In both tasks, the random purification also allows the decoding party to work with copies of a mixed state, rather than with its purified state.

We apply our algorithms to the algorithmic implementation of the Petz recovery map \cite{petz1986sufficient, petz1988sufficiency}, a map which plays an important role in various fields of quantum science \cite{barnum2002reversing, ng2010simpleapproach, chen2020entanglement, lauten2022approx, gilyen2022petzmap, biswas2023noiseadapted, nakayama2023petz}.
The Petz recovery map is defined for a quantum channel and a reference state, and is known to reverse the effect of the channel with near-optimal error \cite{barnum2002reversing, beigi2016decoding}.
While quantum circuits for implementing the Petz recovery map were previously studied \cite{gilyen2022petzmap, biswas2023noiseadapted, zhao2026petzrecoverysinglepolar}, they relied on a crucial assumption that a Stinespring dilation unitary \cite{stinespring1955} of the channel and its inverse can be implemented.
We provide an algorithm without this assumption---which approximately realizes the Petz recovery map using only the quantum channel itself and a reference state---by leveraging our Uhlmann transformation algorithm for the purified sample access model. The key observation is that the Uhlmann transformation from a maximally entangled state to a purification of the Choi--Jamio{\l}kowski state of the channel realizes a Stinespring isometry of the channel. The copies of this purification are prepared from the outputs of the channel by the random purification.

The rest of this paper is organized as follows. An overview of our main results is given in \Cref{sec:mainresults}, and a summary and outlook are in \Cref{sec:conclusion}. 
We provide thorough preliminaries in \Cref{sec:preliminaries}.
The details of the algorithms in each computational model are explained in Secs.~\ref{sec:purif query},~\ref{sec:purif sample}, and~\ref{sec:mixed sample}, respectively, and the lower bounds are proven in \Cref{sec:lower bound Uhlmann}.
Applications of our algorithm are presented in \Cref{sec:fidelity estimation} for fidelity estimation, in \Cref{sec:decoupling and uhlmann} for the combination with the decoupling approach, and in \Cref{sec:sample petz} for the algorithmic implementation of the Petz recovery map.
Technical materials are provided in Appendices~\ref{sec:altanative mixed sample uhlmann} to~\ref{sec:lower bound merging}.

\section{Overview of the main results}
\label{sec:mainresults}
We provide an overview of our main results here.
We begin with a brief introduction to the necessary notation and background for providing a clear overview.
The full preliminaries can be found in \Cref{sec:preliminaries}.

We use superscripts to indicate the systems on which linear operators and linear maps are defined.
For instance, $M^{\sA}$ denotes a linear operator on system $\sA$ and $\cT^{\sA \rarr \sB}$ denotes a linear map from $\sA$ to $\sB$, and additionally, $\cT^{\sA \rightarrow \sA}$ is simply written as $\cT^{\sA}$.
The superscript is omitted when it is clear from the context.
The identity operator $\bI$ and the identity map $\rid$ are often implicit; for instance, we denote $\bI^\sA \otimes M^\sB$ simply by $M^{\sB}$.
For quantum states, we denote a pure state on $\sA\sB$ such as $\ket{\omega}^{\sA\sB}$, and its reduced density operator on $\sA$ by $\omega^\sA$, i.e., $\omega^\sA = \tr_\sB\big[\ketbra{\omega}{\omega}^{\sA\sB}\big]$, where $\tr_\sB$ is the partial trace over the system $\sB$.

The (squared) fidelity between quantum states $\rho$ and $\sigma$ is defined as $\rF(\rho, \sigma) = \big(\tr\big[\sqrt{\sqrt{\sigma}\rho\sqrt{\sigma}}\big]\big)^2$.
We also use the square root fidelity, defined as $\sqrt{\rF}(\rho, \sigma) = \sqrt{\rF(\rho, \sigma)}$.
The fidelity is monotonic with respect to the partial trace: $\rF(\rho^\sA, \sigma^\sA) \geq \rF(\rho^{\sA\sB}, \sigma^{\sA\sB})$.
For notational convenience, we write the fidelity as $\rF(\rho, \ket{\phi})$ or $\rF(\ket{\psi}, \ket{\phi})$, when one or both of its arguments are pure states.

Our goal is to construct a quantum algorithm that \emph{approximately} realizes the unitary $U^\sB$ achieving Eq.~\eqref{inteq:57}.
In the algorithm construction, we consider the following three computational models:
\begin{enumerate}[label=\Roman*.]
    \item Purified query access model~\cite{watrous2002limitQSZK, watrous2009QSZKquantumattacks, belovs2019ClassicalDist, gilyen2019DistPropTesting, brandao2019SDPOpt, vanapeldoorn2019impleveSDP, Subramanian2021estimatealphalenni, gilyen2022improvedfidelity, Apeldoorn2023TomoPurifQuery, wang2023QAfidelityest, wang2023SampletoQuary, luo2024SuccinctTest, wang2024NewQuantEnt, Liu2024geometricmean}: let $U_\rho^{\sA\sB}$ and $U_\sigma^{\sA\sB}$ be unitaries which prepare states $\ket{\rho}^{\sA\sB}$ and $\ket{\sigma}^{\sA\sB}$, respectively, that is, $U_\rho^{\sA\sB}\ket{0}^{\sA\sB}=\ket{\rho}^{\sA\sB}$, and $U_\sigma^{\sA\sB}\ket{0}^{\sA\sB}=\ket{\sigma}^{\sA\sB}$.
    We assume that $U_\rho^{\sA\sB}$ and $U_\sigma^{\sA\sB}$ are available multiple times as unitary oracles.
    It is also assumed that we can query their inverses $(U_\rho^{\sA\sB})^\dag$ and $(U_\sigma^{\sA\sB})^\dag$.
    An algorithm in this model is required to work for any such unitaries.
    \item Purified sample access model~\cite{chen2024localtestUniInv, liu2024exponentialseparations}:
    multiple independent and identical copies of purified states $\ket{\rho}^{\sA\sB}$ and $\ket{\sigma}^{\sA\sB}$ are available.
    \item Mixed sample access model~\cite{gilyen2022improvedfidelity, wang2023SampletoQuary, wang2024TimeEfficientEntEstSamp, wang2024sampleoptimal}: multiple independent and identical copies of states $\rho^\sA$ and $\sigma^\sA$ are available.
\end{enumerate}

For each model, we evaluate the number of queries/samples for implementing the Uhlmann transformation.
Since each model listed above can simulate those listed below, the models are ordered from the one with the strongest assumptions and highest computational power (I) to the one with the weakest assumptions and lowest power (III).
The details and motivations behind these computational models are presented in \Cref{sec:three computational models}.

When we consider the Uhlmann transformation, the purified sample access model (II) appears to be the most natural setting, as the transformation is specified for a pair of purified states and cannot be determined from mixed states alone.
Nevertheless, there is still motivation to consider the mixed sample access model (III), due to its utility in various information-theoretic tasks, as discussed in our applications.
In the mixed sample access model (III), we consider the Uhlmann transformation tailored to a specific purification named the \emph{canonical purification}~\cite{Winter2004ExtrinsicIntrinsic} (its definition is given by Eq.~\eqref{eq:def of canonical purification}).
We construct a quantum algorithm that approximately implements a unitary preparing the canonical purified state and its inverse from multiple sampled states, and utilize it as part of our Uhlmann transformation algorithm in model (III).

We further introduce some technical notation used throughout this paper; see \Cref{tab:technical notation} for a summary.
These quantities are related to each other. For example, $r_\omega \leq \kappa_\omega$ holds for any state $\omega$, while $r_\omega$ does not imply any upper bound on $\kappa_\omega$.
The inequality $R \leq r_{\rm min}$ follows from the fact that the rank of a product of matrices is at most the rank of each matrix.
Moreover, since $R s_{\rm min} \leq \sqrt{\rF}(\rho^\sA, \sigma^\sA) \leq 1$, it follows that $R \leq 1/s_{\rm min}$.


\renewcommand{\arraystretch}{2.0}
\begin{table*}
\centering
\caption{A table of some technical notation.}
\begin{tabular}{wc{6em}|wl{43em}} 
    $d_\sA$, $d_\sB$ & \hspace{1mm} The dimensions of the Hilbert spaces $\cH^\sA$ and $\cH^\sB$, respectively. \\ \hline
    $s_{\rm min}$, $R$ & \hspace{1mm} The minimum non-zero singular value and the rank of $\sqrt{\sigma^\sA}\sqrt{\rho^\sA}$, respectively. \\ \hline
        $r_\rho$, $r_\sigma$& \hspace{1mm} The ranks of $\rho^\sA$ and $\sigma^\sA$, respectively. \\ \hline
    $r_{\rm min}$, $r_{\rm max}$ & \hspace{1mm} The smaller and larger of the ranks of $\rho^\sA$ and $\sigma^\sA$: $r_{\rm min} = \min\{r_\rho, r_\sigma\}$ and $r_{\rm max} = \max\{r_\rho, r_\sigma\}$. \\ \hline
    $\rho_{\rm min}$, $\sigma_{\rm min}$ & \hspace{1mm} The minimum non-zero eigenvalues of $\rho^\sA$ and $\sigma^\sA$, respectively. \\ \hline
    $\kappa_\rho$, $\kappa_\sigma$ & \hspace{1mm} The reciprocals of $\rho_{\rm min}$ and $\sigma_{\rm min}$, respectively: $\kappa_\rho = 1/\rho_{\rm min}$ and $\kappa_\sigma = 1/\sigma_{\rm min}$. \\
    \hline
    $\kappa_{\rm min}$, $\kappa_{\rm max}$ & \hspace{1mm} The smaller and larger of the reciprocals of $\rho_{\rm min}$ and $\sigma_{\rm min}$: $\kappa_{\rm min} = \min\{\kappa_\rho, \kappa_\sigma\}$ and $\kappa_{\rm max} = \max\{\kappa_\rho, \kappa_\sigma\}$.
\end{tabular}
\label{tab:technical notation}
\end{table*}
\renewcommand{\arraystretch}{1.0}


For evaluating computational costs, we use the Landau notation $\cO(\cdot)$, $\Omega(\cdot)$, and $\Theta(\cdot)$. These describe the scaling in terms of dimension $d$ and the approximation accuracy $\delta$; that is, a notation such as $\cO(d/\delta)$ is understood as $\cO(d)$ with respect to $d$ and $\cO(1/\delta)$ with respect to $\delta$.
The tilde notation, such as $\til{\cO}(\cdot)$, $\til{\Omega}(\cdot)$, and $\til{\Theta}(\cdot)$, additionally hides polylogarithmic factors in $d$ and $1/\delta$; for instance, $\til{\cO}\big(f(d, 1/\delta)\big) = \cO\big(f(d, 1/\delta){\rm polylog}(d, 1/\delta)\big)$ for a function $f$.
We also use the notation $\poly(d, 1/\delta)$ to denote a polynomial in $d$ and $1/\delta$. Throughout this paper, $\log$ is the logarithm to base $2$, and $\ln$ is the natural logarithm.

We first provide upper bounds on the computational costs for implementing the Uhlmann transformation in \Cref{sec:upper bound uhl}.
Next, in \Cref{sec:comparison with state tom}, we compare the computational costs of our algorithm with other algorithms, highlighting the advantage of our algorithm.
In \Cref{sec:overview of lower bound}, we provide lower bounds on the query and sample complexities for implementing the Uhlmann transformation in the three computational models.
In \Cref{sec:overview Optimality}, we combine these bounds and show that they are tight in the purified query and sample access models.
In \Cref{sec:overview application}, we present applications of the Uhlmann transformation algorithm.


\subsection{Upper bounds on the computational costs via Uhlmann transformation algorithms}
\label{sec:upper bound uhl}

We here provide upper bounds on the computational costs required for the Uhlmann transformation in each computational model in Secs.~\ref{sec:upper purif query main},~\ref{sec:upper purif sample main}, and~\ref{sec:upper mixed sample main}.

\subsubsection{In the purified query access model}
\label{sec:upper purif query main}

Our statement in the purified query access model is as follows.

\begin{theorem}[Uhlmann transformation algorithm in the purified query access model]
\label{infthm:uhlmenn purif query}
For any $\delta \in (0, 1)$, there exists a quantum query algorithm that realizes a quantum channel $\cT^{\sB}$ satisfying 
\begin{align}
\label{inteq:17}
     \rF(\cT^{\sB}(\ketbra{\rho}{\rho}^{\sA\sB}), \ket{\sigma}^{\sA\sB}) \geq \rF(\rho^\sA, \sigma^\sA) - \delta,
\end{align}
using $u$ queries to $U_\rho^{\sA\sB}$, $U_\sigma^{\sA\sB}$, and their inverses, where 
\begin{align}
\label{inteq:query cost}
    u = \cO\Big(\min\Big\{\f{1}{s_{\rm min}}\log{\Big(\f{1}{\delta}\Big)}, \f{R}{\delta}\Big\}\Big).
\end{align}
The quantum circuit of this algorithm consists of $\cO\big(u\log{(d_\sA d_\sB)}\big)$ one- and two-qubit gates, and uses $\cO\big(\log(d_\sA d_\sB)\big)$ qubits at once. The circuit description can be computed classically in $\poly\big(u, \log{(d_\sA d_\sB)}, \log{(1/\delta)}\big)$ time.
\end{theorem}

A more formal and technical version of \Cref{infthm:uhlmenn purif query} is provided in \Cref{sec:purif query} as \Cref{thm:Uhlmann alg purif query model}. Here, the approximation error of the transformation is measured not only by the fidelity difference for the specified input $\ket{\rho}^{\sA\sB}$ as in Eq.~\eqref{inteq:17} but also by the diamond norm, which evaluates the transformation as a quantum channel for arbitrary inputs.

Theorem~\ref{infthm:uhlmenn purif query} states that, by this query algorithm, the Uhlmann transformation is approximately realized with accuracy $\delta$.
Remarkably, when either $1/s_{\rm min}$ or $R/\delta$ is sufficiently small, for instance, a polynomial in the number of qubits, the algorithm is implementable with a polynomial number of queries.
This implies that our approach achieves an exponential improvement in query and gate costs over the previous approach in Ref.~\cite{metger2023stateqipspace}, which requires exponential costs, $\poly(d_\sA d_\sB, 1/\delta)$. We will discuss the comparison in detail in \Cref{sec:comparison with state tom}.

Our algorithm builds on the observation that the Uhlmann unitary is the polar factor of $\tr_\sA\big[\ketbra{\sigma}{\rho}^{\sA\sB}\big]$~\cite{Jozsa1994FidelityforMixedQuantumStates}. Following the QSVT-based algorithm for polar decomposition in Ref.~\cite{Quek2022Fastalgorithmquantumpolardecomposition}, we block-encode this operator and apply the QSVT with the sign function, which approximates the Uhlmann transformation as a quantum channel for arbitrary inputs and leads to the dependence on $1/s_{\rm min}$ in \Cref{infthm:uhlmenn purif query}. When the approximation error is measured for the specified input $\ket{\rho}^{\sA\sB}$, the deviation from the sign function at each singular value contributes to the error in proportion to the singular value. It then suffices to approximate the sign function only coarsely at small singular values, with each of the $R$ contributions kept below $\delta/R$. This yields the dependence $R/\delta$. As shown in \Cref{sec:overview Optimality}, both the dependences on $1/s_{\rm min}$ and $R/\delta$ turn out to be optimal.

This algorithm is a query algorithm. However, if quantum circuit descriptions of the unitaries $U_\rho^{\sA\sB}$ and $U_\sigma^{\sA\sB}$ are given, in other words, if one can implement $U_\rho^{\sA\sB}$, $U_\sigma^{\sA\sB}$, and their inverses on one's own, the circuit complexity, the total number of one- and two-qubit gates required in the algorithm, is evaluated as
\begin{align}
    \cO\Big(u \big(\cC(U_\rho) + \cC(U_\sigma) + \log{(d_\sA d_\sB)}\big)\Big),
\end{align}
where $\cC(U_\rho)$ and $\cC(U_\sigma)$ denote the circuit complexity of $U_\rho^{\sA\sB}$ and $U_\sigma^{\sA\sB}$, respectively.

One might be concerned about the necessity of knowing either the value of $s_{\rm min}$ or $R$ for implementing this algorithm.
To estimate the singular value, quantum algorithms based on the quantum phase estimation have been proposed in Refs.~\cite{chakraborty2019powerblock, Gilyn2019thesis}.
We can estimate $s_{\rm min}$ with the desired accuracy in advance.
Similarly, several algorithms for rank estimation are known~\cite{Tan2021ValEstRank, wang2024NewQuantEnt}.
In our construction of the Uhlmann transformation algorithm, we employ the QSVT with the sign function introduced in \Cref{sec:BE and QSVT intro}.
Due to a property of the QSVT with the sign function, it is sufficient to have a lower bound on $s_{\rm min}$ or an upper bound on $R$, and Eq.~\eqref{inteq:query cost} holds with $s_{\rm min}$ and $R$ replaced by these bounds.

\subsubsection{In the purified sample access model}
\label{sec:upper purif sample main}

Next, we provide the statement in the purified sample access model.

\begin{theorem}[Uhlmann transformation algorithm in the purified sample access model]
\label{infthm:uhlmenn purif sample}
For any $\delta \in (0, 1)$, there exists a quantum sample algorithm that realizes a quantum channel $\cT^{\sB}$ satisfying
\begin{align}
\label{eq:fid diff pure samp}
     \rF(\cT^{\sB}(\ketbra{\rho}{\rho}^{\sA\sB}), \ket{\sigma}^{\sA\sB}) \geq \rF(\rho^\sA, \sigma^\sA) - \delta,
\end{align}
using $w$ samples of $\ket{\rho}^{\sA\sB}$ and $\ket{\sigma}^{\sA\sB}$, where 
\begin{align}
    w = \cO\bigg(\min\Big\{\f{R}{\delta s_{\rm min}}\Big(\log{\Big(\f{1}{\delta}\Big)}\Big)^2, \f{R^2}{\delta^2}\Big\}\bigg).
\end{align}
The quantum circuit of this algorithm consists of $\cO\big(w\log (d_\sA d_\sB)\big)$ one- and two-qubit gates, and uses $\cO\big(\log{(d_\sA d_\sB)} + \log{(1/\delta)}\big)$ qubits at once. The circuit description can be computed classically in $\poly\big(w, \log{(d_\sA d_\sB)}, \log{(1/\delta)}\big)$ time.
\end{theorem}

A more formal and technical version of \Cref{infthm:uhlmenn purif sample} is \Cref{thm:algorithm Uhlmann} in \Cref{sec:purif sample}, where the approximation error is evaluated by both fidelity difference and diamond norm.

Theorem~\ref{infthm:uhlmenn purif sample} states that, rather remarkably, the sample complexity still does not scale with the dimension $d_\sA$ or $ d_\sB$---when either $R/(\delta s_{\rm min})$ or $R^2/\delta^2$ is sufficiently small, the Uhlmann transformation can be efficiently implemented using multiple copies of the purified states $\ket{\rho}^{\sA\sB}$ and $\ket{\sigma}^{\sA\sB}$.
Moreover, because this algorithm is sequential and processes the quantum states one by one, its quantum circuit uses $\cO\big(\log(d_\sA d_\sB) + \log{(1/\delta)}\big)$ qubits at any one time, including qubits to store these states.

The block-encoding of $\tr_\sA\big[\ketbra{\sigma}{\rho}^{\sA\sB}\big]$ is constructed from the reflections about $\ket{\rho}^{\sA\sB}$ and $\ket{\sigma}^{\sA\sB}$.
To do so, we rely on the density matrix exponentiation~\cite{lloyd2014QuantPrincCompAnaly, kimmel2017HamSimSampleComp}, which approximately applies $e^{it\omega}$ for an unknown state $\omega$ by using its samples.
Since each reflection requires a number of samples inversely proportional to its error, accurately implementing all the reflections is costly, yielding $\cO(u^2/\delta)$ samples with the number $u$ of queries in \Cref{infthm:uhlmenn purif query}.
The key idea for reducing the number of samples is to apply the QSVT with the sign function in stages from large to small singular values and to stop each component once it is transformed.
The many reflections required for small singular values act only on components that contribute little to the fidelity, and can thus be implemented with low accuracy.
This leads to the term $R/(\delta s_{\rm min})$, the product of the two terms $R/\delta$ and $1/s_{\rm min}$ in the query complexity.



\subsubsection{In the mixed sample access model}
\label{sec:upper mixed sample main}

In general, purifications of a mixed state are not unique. 
This implies that the Uhlmann transformation is not determined solely by a given pair of mixed states.
In other words, in the mixed sample access model, the Uhlmann transformation is ill-defined in the sense that, from samples of mixed states, it is impossible to perform the Uhlmann transformation that universally works for all possible purifications.
We thus first perform a specific purification using multiple copies of the given mixed states.
This allows us to implement the Uhlmann transformation between the specific purified states obtained by this purification.

In this work, we consider the canonical purification~\cite{Winter2004ExtrinsicIntrinsic}: for any state $\omega^\sA$, the canonical purified state $\ket{\omega_{\rm c}}^{\sA\sB}$ is defined by 
\begin{equation}
\label{eq:def of canonical purification}
    \ket{\omega_{\rm c}}^{\sA\sB} = \Big(\sqrt{\omega^\sA} \otimes \bI^{\sB}\Big)\ket{\Gamma}^{\sA\sB},
\end{equation}
where $\ket{\Gamma}^{\sA\sB} = \sum_{i=0}^{d_\sA-1}\ket{i}^{\sA}\ket{i}^{\sB}$ is the unnormalized maximally entangled state in the computational basis $\ket{i}$.
It is natural to assume that $d_\sA = d_\sB$ here.
The state $\ket{\omega_{\rm c}}^{\sA\sB}$ is indeed one of the purified states of $\omega^\sA$ as $\tr_\sB[\ketbra{\omega_{\rm c}}{\omega_{\rm c}}^{\sA\sB}] = \omega^\sA$.
A quantum algorithm for the canonical purification from multiple identical copies of $\omega^\sA$ is provided in \Cref{sec:appendix B}.

Let $\ket{\rho_{\rm c}}^{\sA\sB}$ and $\ket{\sigma_{\rm c}}^{\sA\sB}$ be the canonical purified states of $\rho^\sA$ and $\sigma^\sA$, respectively.
Our statement for the mixed sample access model is as follows.

\begin{theorem}[Uhlmann transformation algorithm in the mixed sample access model]
\label{infthm:uhlmenn mixed sample}
For any $\delta \in (0, 1)$, there exists a quantum sample algorithm that realizes a quantum channel $\cT^\sB$ satisfying 
\begin{align} 
    \rF\big(\cT^\sB(\ketbra{\rho_{\rm c}}{\rho_{\rm c}}^{\sA\sB}), \ket{\sigma_{\rm c}}^{\sA\sB}\big)  
    \geq \rF(\rho^\sA, \sigma^\sA) - \delta, 
\end{align}
using $\zeta$ samples of $\rho^\sA$ and $\sigma^\sA$, where 
\begin{align}
    \zeta 
    = \til{\cO}\bigg(\f{d_\sA}{\delta}\min\bigg\{\f{1}{s_{\rm min}^2}, \f{R^2}{\delta^2}\bigg\}\min\bigg\{\kappa_\rho^2 + \kappa_\sigma^2, \f{r_\rho^2 + r_\sigma^2}{\delta^4}\bigg\}\bigg).
\end{align}
The quantum circuit of this algorithm consists of $\cO(\zeta\log{d_\sA})$ one- and two-qubit gates, and uses $\cO(\log{d_\sA})$ qubits at once. The circuit description can be computed classically in $\poly\big(\zeta, \log{d_\sA}, \log{(1/\delta)}\big)$ time.
\end{theorem}

A formal and technical version of \Cref{infthm:uhlmenn mixed sample} is given by \Cref{thm:Uhlmann mixed state sample} in \Cref{sec:mixed sample}, where the approximation error is evaluated by both fidelity difference and diamond norm.

Due to the symmetry of $\ket{\Gamma}^{\sA\sB}$, the canonical purification has the property that the reduced states on each system are related by the transpose. That is, for the canonical purified state $\ket{\rho_{\rm c}}^{\sA\sB}$ of a state $\rho$ on $\sA$, the reduced states are $\rho_{\rm c}^\sA = \rho$ and $\rho_{\rm c}^\sB = \rho^\top$.
For $\rho_{\rm c}^\sB = \rho^\top$ and $\sigma_{\rm c}^\sB = \sigma^\top$, the Uhlmann unitary $U^\sA$ satisfies $\rF(\rho^\top, \sigma^\top) = \rF(U^\sA\ket{\rho_{\rm c}}^{\sA\sB}, \ket{\sigma_{\rm c}}^{\sA\sB})$, where the states with transpose on the left-hand side are defined on $\sB$.
Since the transpose does not change the fidelity, we obtain $\rF(\rho^\sA, \sigma^\sA) = \rF(U^\sA\ket{\rho_{\rm c}}^{\sA\sB}, \ket{\sigma_{\rm c}}^{\sA\sB})$.
This shows that, when we consider the canonical purification, we can apply a transformation by $U^\sA$ to the original system $\sA$ to achieve the fidelity of the states on the original system. 
This is unlike the original Uhlmann transformation, which is applied to the purifying system $\sB$.
We call the transformation on $\sA$ a \emph{variant} of the Uhlmann transformation and propose an algorithm for implementing it in the mixed sample access model.

Note that such a unitary $U^\sA$ exists due to the property of the canonical purification.
For general purifications, a unitary on the original system $\sA$ does not necessarily achieve the fidelity $\rF(\rho^\sA, \sigma^\sA)$.

\begin{theorem}[Algorithm for a variant of the Uhlmann transformation in the mixed sample access model]
\label{thm:variant Uhl overview}
For any $\delta \in (0, 1)$, there exists a quantum sample algorithm that realizes a quantum channel $\cT^\sA$ satisfying 
\begin{align} 
    \rF\big(\cT^\sA(\ketbra{\rho_{\rm c}}{\rho_{\rm c}}^{\sA\sB}), \ket{\sigma_{\rm c}}^{\sA\sB}\big)  
    \geq \rF(\rho^\sA, \sigma^\sA) - \delta, 
\end{align}
using $\zeta$ samples of $\rho^\sA$ and $\sigma^\sA$, where 
\begin{align}
    \zeta 
    = \til{\cO}\bigg(\min\Big\{\f{1}{\delta s_{\rm min}^2}, \f{R^2}{\delta^3}\Big\}\min\Big\{\kappa_\rho^2 + \kappa_\sigma^2, \f{R^2}{\delta^4}\Big\}\bigg).
\end{align}
The quantum circuit of this algorithm consists of $\cO(\zeta\log{d_\sA})$ one- and two-qubit gates, and uses $\cO(\log{d_\sA})$ qubits at once. The circuit description can be computed classically in $\poly\big(\zeta, \log{d_\sA}, \log{(1/\delta)}\big)$ time.
\end{theorem}

A formal and technical version of \Cref{thm:variant Uhl overview} is given by \Cref{thm:variant of uhlmann mixed sample} in \Cref{sec:altanative mixed sample uhlmann}, where the approximation error is evaluated using both fidelity difference and the diamond norm.

Compared with the algorithm in \Cref{infthm:uhlmenn mixed sample}, this algorithm in \Cref{thm:variant Uhl overview} can be implemented with substantially fewer samples of the mixed states.
This is because this algorithm does not rely on the quantum algorithm for the canonical purification, which requires a large number of samples. 
Instead, it directly block-encodes the square root of the given mixed states into a large unitary through a combination of the QSVT and the density matrix exponentiation.
Further details on this variant of the Uhlmann transformation can be found in \Cref{sec:altanative mixed sample uhlmann}.

\renewcommand{\arraystretch}{2.5}
\begin{table*}
\centering
\caption{Comparison of the query/sample complexity for implementing the Uhlmann transformation with the fidelity difference at most $\delta$.
For the notation, see \Cref{tab:technical notation}.
For our algorithm in the mixed sample access model, we show the bound in terms of $\kappa_\rho$ and $\kappa_\sigma$ for simplicity. The results of the state tomography-based approach are derived using the state tomography algorithms in Refs.~\cite{ODonnell2016EffTomo, Haah2017samploptTomo, Apeldoorn2023TomoPurifQuery}. In the purified query access model and the purified sample access model, the query and sample complexities of our algorithms are independent of the dimensions $d_\sA$ and $d_\sB$, which indicates exponential advantages. These complexities are tight up to logarithmic factors in the worst case for $R \geq 3$; see \Cref{cor:tight Uhlmann} for the precise statement.}
\begin{tabular}{wc{8em}|wc{15em}|wc{16em}|wc{15em}}
                & Our algorithm  & State tomography-based & Prior work \\ \hline
    Purified query & $\til{\Theta}\Big(\min\Big\{\f{1}{s_{\rm min}}, \f{R}{\delta}\Big\}\Big)$ & $\til{\cO}\Big(d_\sA d_\sB\min\Big\{\f{1}{\delta}, \f{1}{\sqrt{\delta s_{\rm min}}}\Big\}\Big)$ & $\til{\cO}\Big(\f{d_\sA^{13} d_\sB^{9}}{\delta^2}\min\Big\{\f{1}{s_{\rm min}^3}, \f{R^3}{\delta^3}\Big\}\Big)$~\cite{metger2023stateqipspace} \\
    Purified sample & $\til{\Theta}\Big(\min\Big\{\f{R}{\delta s_{\rm min}}, \f{R^2}{\delta^2}\Big\}\Big)$ & $\til{\cO}\Big(d_\sA d_\sB\min\Big\{\f{1}{\delta^2}, \f{1}{\delta s_{\rm min}}\Big\}\Big)$ & N/A \\
    Mixed sample & $\til{\cO}\Big(\f{d_\sA(\kappa_\rho^2 + \kappa_\sigma^2)}{\delta}\min\Big\{\f{1}{s_{\rm min}^2}, \f{R^2}{\delta^2}\Big\}\Big)$ & $\til{\cO}\Big(d_\sA(r_\rho + r_\sigma)\min\Big\{\f{1}{\delta^2}, \f{1}{\delta s_{\rm min}}\Big\}\Big)$ & N/A \\
\end{tabular}
\label{tab:compare Uhlmann and tomography}
\end{table*}
\renewcommand{\arraystretch}{1.0}


\subsection{Comparison with approaches based on state measurements}
\label{sec:comparison with state tom}

The most naive approach to implementing the Uhlmann transformation would be to use \emph{quantum state tomography}~\cite{ODonnell2016EffTomo, Haah2017samploptTomo, Guta2020FastTomo, Apeldoorn2023TomoPurifQuery, Chen2023WhenAdapTomo, hu2024sampleoptimalmemoryefficient}.
Using quantum state tomography, information about the states is obtained as classical data.
Then, the Uhlmann unitary in Eq.~\eqref{inteq:57} is directly computed.
By a brute-force decomposition of the computed unitary into one- and two-qubit gates~\cite{nielsen2010quantum}, a quantum circuit for the Uhlmann transformation can be obtained.
However, this approach clearly requires exponential query/sample complexity, as well as exponential classical time and space complexities in the number of bits, since it needs to handle matrices of exponential size in classical computation.
In addition, the number of gates arising from the decomposition is, in general, exponentially large.

Recently, an explicit quantum algorithm implementing the Uhlmann transformation was proposed in Ref.~\cite{metger2023stateqipspace}. This algorithm resolves the issue of exponential space use by sequentially encoding classically computed matrix elements into small unitaries.
This approach enables the implementation of the Uhlmann transformation in polynomial space, while exponential time in the number of bits in classical computation and exponential query/sample complexity, $\poly(d_\sA d_\sB, 1/\delta)$, cannot be avoided, as, in our setting, the matrix elements are estimated by an exponential number of measurements.

Our algorithms do not require such a measurement on exponentially many copies of the states, and thus offer advantages in the query/sample complexity in the purified query and sample access models.
Moreover, our algorithms only use a polynomial number of qubits at any one time.
Since our algorithms do not use classical data obtained by measurements on states, they also avoid the exponential time and space complexities in classical computation that the naive tomography-based approach requires.

To compare our algorithms with the above approaches, we summarize the query/sample complexity of these approaches in \Cref{tab:compare Uhlmann and tomography}, where the approximation error is evaluated by the fidelity difference. In the purified query access model, our Uhlmann transformation algorithm clearly has an advantage over the other approaches. Since $\min\{1/s_{\rm min}, R/\delta\} \leq \sqrt{R/(\delta s_{\rm min})}$ and $R \leq \min\{d_\sA, d_\sB\}$, our query complexity is at most that of the state tomography-based approach, and is exponentially smaller when $R$ is small. When $s_{\rm min}$ is not small, our algorithm also makes the dependence on $\delta$ logarithmic. In the purified sample access model, our algorithm also remarkably outperforms the other approaches.

In the mixed sample access model, neither our algorithm nor the state tomography-based approach outperforms the other in general.
For states whose $s_{\rm min}$, $\rho_{\rm min}$, and $\sigma_{\rm min}$ are bounded below by constants, both approaches use $\til{\cO}(d_\sA/\delta)$ samples, which is optimal for these states, as shown in \Cref{sec:overview Optimality}.
Even in this case, our algorithm has a significant advantage in the space costs required for classical computation and in the number of qubits used at any one time, and can have an advantage in the time complexity, as it processes the samples sequentially without classical data obtained by state measurements.
For the variant of the Uhlmann transformation, the number of samples in \Cref{thm:variant Uhl overview} is independent of $d_\sA$.

In \Cref{sec:Comparing with a naive tomography-based}, we provide detailed derivations of the computational costs by the naive state tomography-based approach, employing the most efficient known algorithms for quantum state tomography: $\tilde{\mathcal{O}}(rd/\epsilon)$ queries~\cite{Apeldoorn2023TomoPurifQuery} or $\tilde{\mathcal{O}}(rd/\epsilon^2)$ samples~\cite{ODonnell2016EffTomo, Haah2017samploptTomo} for a rank-$r$ state on a $d$-dimensional system within trace distance error $\epsilon$, and $\til{\cO}(rd/\epsilon)$ samples~\cite{Haah2017samploptTomo} with fidelity at least $1 - \epsilon$.
Here, the time and space complexities in classical computation required by this approach are tolerated. 
In addition, in the sample access model, it is allowed to perform collective measurements on multiple copies, possibly even an exponential number of them.

In \Cref{sec:metger and yuen's algorithm}, we provide an overview and performance evaluation of the approach proposed in Ref.~\cite{metger2023stateqipspace}, adapted to our setting.
We consider only the purified query access model, since this approach uses the unitaries $U_\rho^{\sA\sB}$ and $U_\sigma^{\sA\sB}$ and their inverses in the circuit, and it is not immediately clear whether it can be applied to the sample access models.

\subsection{Lower bounds on the query and sample complexities for the Uhlmann transformation}
\label{sec:overview of lower bound}

Our explicit algorithms provide upper bounds on the computational complexity for the Uhlmann transformation.
Here, we ask the converse question---what is the limit that can never be surpassed by any potential algorithm? To this end, we provide lower bounds on the query and sample complexities for the Uhlmann transformation in the three computational models. These lower bounds hold for any quantum algorithm, with no restriction on its computation, including collective operations on the samples. Our results are summarized in the following theorem, which we prove in \Cref{sec:lower bound Uhlmann}.

\begin{restatable}[Lower bounds for the Uhlmann transformation]{theorem}{thmloweruhlmann}
\label{thm:lower Uhlmann}
Let $\delta \in (0, 1/32]$, $s \in (0, 1/8]$, and $3 \leq r \leq \min\{d_\sA, d_\sB\}$, and consider a quantum algorithm that, for every pair of purified states $\ket{\rho}^{\sA\sB}$ and $\ket{\sigma}^{\sA\sB}$ such that $\rho^\sA$ and $\sigma^\sA$ satisfy $r_\rho, r_\sigma \leq r$ and $s_{\rm min} \geq s$, realizes a quantum channel $\cT^\sB$ satisfying
\begin{align}
    \rF\big(\cT^\sB(\ketbra{\rho}{\rho}^{\sA\sB}), \ket{\sigma}^{\sA\sB}\big) \geq \rF(\rho^\sA, \sigma^\sA) - \delta.
\end{align}
Depending on the access model, such an algorithm requires
\begin{itemize}
    \item $\Omega\big(\min\{r/\delta, 1/s\}\big)$ queries to $U_\rho^{\sA\sB}$, $U_\sigma^{\sA\sB}$, and their inverses in the purified query access model,
    \item $\Omega\big(\min\{r^2/\delta^2, r/(\delta s), 1/(\delta s^2)\}\big)$ samples of $\ket{\rho}^{\sA\sB}$ and $\ket{\sigma}^{\sA\sB}$ in the purified sample access model, and
    \item $\Omega\big(d_\sA/\delta + 1/(\delta s)\big)$ samples of $\rho^\sA$ and $\sigma^\sA$ in the mixed sample access model with $\ket{\rho}^{\sA\sB}$ and $\ket{\sigma}^{\sA\sB}$ replaced by the canonical purified states $\ket{\rho_{\rm c}}^{\sA\sB}$ and $\ket{\sigma_{\rm c}}^{\sA\sB}$, if $\delta \leq 2^{-13}s$.
\end{itemize}
\end{restatable}

In each model, we construct a family of pairs of states for which the condition on the fidelity forces the algorithm to recover unknown information contained in the queries or samples. Our idea in the purified sample access model is to hide a Haar random unitary in a component of $\ket{\rho}^{\sA\sB}$ with a small amplitude. The Uhlmann transformation has to undo this unitary, and the task thus reduces to learning an unknown unitary~\cite{Bisio2010Optimalquantumlearningunitary} from samples in which it appears only with this amplitude. We choose the dimension of the unitary as large and the amplitude as small as allowed by $r$, $\delta$, and $s$. Moreover, we obtain the bound in the purified query access model from that in the purified sample access model, by converting any query algorithm into a sample algorithm with the density matrix exponentiation. This bound in fact remains valid even when complex conjugates, transposes, and controlled versions of $U_\rho^{\sA\sB}$ and $U_\sigma^{\sA\sB}$ are available. As a related result, Ref.~\cite{cochran2026blackholeradiationdecoding} obtained a lower bound linear in rank in the purified query access model, while our bound also captures the dependence on $\delta$ and $s$.

In the mixed sample access model, the key is that the canonical purified state of a pure state $\ket{\psi}$ is $\ket{\psi}\ket{\psi^*}$. The algorithm thus has to prepare the complex conjugate $\ket{\psi^*}$, or equivalently the transpose~\cite{Buscemi2003OptimalRealizationTranspositionMaps, brzic2026optimalphysicalapproximationspurestate}, of an unknown pure state from its samples. Since complex conjugation is not a physical operation, this essentially requires learning the state, which leads to the term $\Omega(d_\sA/\delta)$. Adding a known component orthogonal to this state, we further require the algorithm to reproduce a phase in the unknown state relative to the known component, while in the samples this phase appears only relative to a small component of amplitude about $s$. As in the no-programming theorem~\cite{Nielsen1997ProgrammableQuantumGateArrays}, the algorithm can realize different operations for different phases only to the extent that the samples distinguish them, which leads to the term $\Omega\big(1/(\delta s)\big)$.

Taking $r = \min\{d_\sA, d_\sB\}$ and $s \leq \delta/r$, we find that implementing the Uhlmann transformation for general states requires $\Omega\big(\min\{d_\sA, d_\sB\}/\delta\big)$ queries and $\Omega\big(\min\{d_\sA, d_\sB\}^2/\delta^2\big)$ samples in the purified query and sample access models, respectively, and $\Omega(d_\sA/\delta)$ samples in the mixed sample access model, all of which are exponential in the number of qubits. While these lower bounds hold without any complexity-theoretic assumption, known results suggest that the transformation remains hard in general even when the circuit descriptions of $U_\rho^{\sA\sB}$ and $U_\sigma^{\sA\sB}$ are given. This is because an efficient algorithm for the Uhlmann transformation yields one for fidelity estimation (see \Cref{sec:overview of fidelity est}). Given circuits that prepare purifications of two states, testing their closeness is $\mathsf{QSZK}$-complete~\cite{watrous2002limitQSZK, watrous2009QSZKquantumattacks}, and a polynomial-time quantum algorithm for fidelity estimation would imply $\mathsf{BQP} = \mathsf{QSZK}$~\cite{Rethinasamy2023estimatingdist}. Thus, unless $\mathsf{BQP} = \mathsf{QSZK}$, no polynomial-time quantum algorithm implements the Uhlmann transformation in general, which is consistent with prior suggestions from the context of the information recovery problem~\cite{harlow2013CompvsFirewall, aaronson2016moneyblackhole, brakerski2023blackholeradiation}.


\subsection{Tightness of the bounds for the Uhlmann transformation}
\label{sec:overview Optimality}

Combining the upper bounds in \Cref{infthm:uhlmenn purif query,infthm:uhlmenn purif sample} with the lower bounds in \Cref{thm:lower Uhlmann}, we characterize the query and sample complexities of the Uhlmann transformation in the purified query and sample access models up to logarithmic factors.
The characterization is given in terms of the rank $R$ of $\sqrt{\sigma^\sA}\sqrt{\rho^\sA}$ and the minimum non-zero singular value $s_{\rm min}$.

\begin{corollary}[Query and sample complexities of the Uhlmann transformation]
\label{cor:tight Uhlmann}
Let $\delta \in (0, 1/32]$, $s \in (0, 1/8]$, and $3 \leq r \leq \min\{d_\sA, d_\sB\}$, and consider quantum algorithms that, for every pair of purified states $\ket{\rho}^{\sA\sB}$ and $\ket{\sigma}^{\sA\sB}$ such that $R \leq r$ and $s_{\rm min} \geq s$, realize a quantum channel $\cT^\sB$ satisfying
\begin{align}
    \rF\big(\cT^\sB(\ketbra{\rho}{\rho}^{\sA\sB}), \ket{\sigma}^{\sA\sB}\big) \geq \rF(\rho^\sA, \sigma^\sA) - \delta.
\end{align}
The minimum numbers of queries to $U_\rho^{\sA\sB}$, $U_\sigma^{\sA\sB}$, and their inverses in the purified query access model and of samples of $\ket{\rho}^{\sA\sB}$ and $\ket{\sigma}^{\sA\sB}$ in the purified sample access model required for such algorithms are
\begin{align}
    \til{\Theta}\Big(\min\Big\{\f{r}{\delta}, \f{1}{s}\Big\}\Big)
    \ \ \text{and} \ \ \
    \til{\Theta}\bigg(\min\Big\{\f{r^2}{\delta^2}, \f{r}{\delta s}, \f{1}{\delta s^2}\Big\}\bigg),
\end{align}
respectively.
\end{corollary}

\Cref{cor:tight Uhlmann} shows that, in the purified query and sample access models, the rank and the minimum non-zero singular value together determine the complexities of the Uhlmann transformation, and our algorithms are optimal for all these parameters up to logarithmic factors.
A bound in terms of only one of these parameters cannot be tight; for instance, the query complexity is $\til{\Theta}(1/s)$ for $r \geq \delta/s$ and $\til{\Theta}(r/\delta)$ otherwise.
In these complexities, the terms $r/\delta$ and $r^2/\delta^2$ come from the rank, the terms $1/s$ and $1/(\delta s^2)$ come from the minimum non-zero singular value, and the term $r/(\delta s)$ involves both.
For example, when $s$ is bounded below by a constant, the query complexity is independent of $\delta$ up to logarithmic factors, and the sample complexity is $\til{\Theta}(1/\delta)$ rather than $\til{\Theta}(1/\delta^2)$.

We obtain \Cref{cor:tight Uhlmann} by writing the upper bounds in \Cref{infthm:uhlmenn purif query,infthm:uhlmenn purif sample} in terms of $r$ and $s$.
For states with $R \leq r$ and $s_{\rm min} \geq s$, we have $R \leq 1/s_{\rm min} \leq 1/s$.
Thus, our algorithms in \Cref{infthm:uhlmenn purif query,infthm:uhlmenn purif sample}, run with $s$ as a lower bound on $s_{\rm min}$ and $\min\{r, 1/s\}$ as an upper bound on $R$, use at most
\begin{align}
    \til{\cO}\Big(\min\Big\{\f{r}{\delta}, \f{1}{s}\Big\}\Big)
    \ \ \text{and} \ \ \
    \til{\cO}\Big(\min\Big\{\f{r^2}{\delta^2}, \f{r}{\delta s}, \f{1}{\delta s^2}\Big\}\Big)
\end{align}
queries and samples, respectively, where the term $R/(\delta s_{\rm min})$ in \Cref{infthm:uhlmenn purif sample} is bounded by both $r/(\delta s)$ and $1/(\delta s^2)$.
These coincide with the lower bounds in the first and second statements of \Cref{thm:lower Uhlmann} up to logarithmic factors, which apply here since states of rank at most $r$ satisfy $R \leq r$.

In the mixed sample access model with the canonical purification, the upper and lower bounds coincide in restricted cases.
For pure states with fidelity $1/2$, for which $R = 1$, $\rho_{\rm min} = \sigma_{\rm min} = 1$, and $s_{\rm min} = 1/\sqrt{2}$, \Cref{infthm:uhlmenn mixed sample} uses $\til{\cO}(d_\sA/\delta)$ samples, which matches the term $\Omega(d_\sA/\delta)$ in the third statement of \Cref{thm:lower Uhlmann}.
Thus, the linear dependence on $d_\sA$ of \Cref{infthm:uhlmenn mixed sample}, which uses at most $\til{\cO}\big(d_\sA R^2(r_\rho^2 + r_\sigma^2)/\delta^7\big)$ samples for any states, is necessary, and unlike in the purified models, it cannot be replaced by the rank or the minimum non-zero eigenvalues of the states.
The term $\Omega\big(1/(\delta s)\big)$ in the third statement is also tight with respect to both $\delta$ and $s$ for fixed $d_\sA$, $r_\rho$, and $r_\sigma$, since the state tomography-based approach in \Cref{sec:Comparing with a naive tomography-based} uses $\til{\cO}\big(d_\sA(r_\rho + r_\sigma)/(\delta s_{\rm min})\big)$ samples when $\delta \leq s_{\rm min}$.


\subsection{Applications}
\label{sec:overview application}

To demonstrate the usefulness of the Uhlmann transformation algorithm, we present three concrete applications: estimation of the square root fidelity, the decoupling approach~\cite{hayden2008decoupling, dupuis2010decoupling, Szehr2013decouple2design, dupuis2014one, preskill2016quantumshannon}, and algorithmic implementation of the Petz recovery map~\cite{petz1986sufficient, petz1988sufficiency}.
Comprehensive discussions for each application are provided in Secs.~\ref{sec:fidelity estimation}, \ref{sec:decoupling and uhlmann}, and \ref{sec:sample petz}, respectively.


\subsubsection{Square root fidelity estimation}
\label{sec:overview of fidelity est}

The Uhlmann transformation algorithm is applicable to estimating the square root fidelity between states $\rho^\sA$ and $\sigma^\sA$.
The key idea of our approach is to first implement the transformation, and then run an optimal algorithm that estimates the absolute value of the inner product between two pure states~\cite{Wang2024OptimalFidelityPure, wang2024sampleoptimal}.
Due to Uhlmann's theorem, the output coincides with the square root fidelity $\sqrt{\rF}(\rho^\sA, \sigma^\sA)$.
This approach provides algorithms for estimating the fidelity based on the Uhlmann transformation, in the purified query, purified sample, and mixed sample access models.

In the sample access models, our discussion so far has focused only on algorithms that use quantum states sequentially, one by one, \emph{non-collectively}.
Nevertheless, it is also possible to consider algorithms that manipulate multiple quantum states \emph{collectively}.
Collective algorithms potentially lead to a notable reduction in the number of samples, while they require sufficient space to store all states simultaneously. In contrast, non-collective algorithms may require more samples, while they allow us to reuse the qubits that hold the sampled states, thereby reducing the space cost.

Our results on the computational cost of the square root fidelity estimation in all of these settings are provided in the following theorem. 
An in-depth discussion is provided in~\Cref{sec:fidelity estimation}.


\begin{table*}[t]
\centering
\caption{A comparison of the query and sample complexities for square root fidelity estimation within error $\delta$ between our algorithms and the prior best algorithms. For notation, see \Cref{tab:technical notation}. The left side summarizes results focusing on the scaling in $\kappa_\rho$ and $\kappa_\sigma$, the reciprocals of the minimum non-zero eigenvalues of $\rho^\sA$ and $\sigma^\sA$, respectively, while the right side focuses on their ranks.
The mark $\ddag$ indicates the results obtained by using Eq.~\eqref{inteq:135} under the assumption that either $\supp[\rho^\sA] \subset \supp[\sigma^\sA]$ or $\supp[\sigma^\sA] \subset \supp[\rho^\sA]$, and the mark $\S$ indicates the results obtained from those in terms of the ranks by using the facts that $r_{\rm min} \leq \kappa_{\rm min}$ and $r_{\rm max} \leq \kappa_{\rm max}$.
To our knowledge, no previous work has directly addressed fidelity estimation in the purified sample access model, and thus the mark $\P$ indicates the results from the mixed sample access model, relying on the fact that the purified sample access model is computationally stronger than the mixed sample access model. The mark $\|$ indicates that the algorithms leading to the results perform collective operations over multiple samples of the states, and $\til{\Theta}$ indicates that the bound is tight~\cite{wang2026unifiedcomplexityframeworkquantum, lee2026fidelityestimationknownquantum}.}
\renewcommand{\arraystretch}{2.0}
\setlength{\tabcolsep}{4pt}
\begin{tabular}{wc{8em}|wc{12em}wc{12em}|wc{12em}wc{8em}}
& \multicolumn{2}{c|}{In terms of $\kappa$-scaling} & \multicolumn{2}{c}{In terms of $r$-scaling} \\
\hline
& This work & Previous work & This work & Previous work \\
\hline
\rule{0pt}{6ex}
    Purified query
    & $^\ddag\til{\cO}\big(\sqrt{\kappa_\rho\kappa_\sigma}/\delta\big)$
    & $^\ddag\til{\cO}\big(\kappa_{\rm min}^2\kappa_\rho\kappa_\sigma/\delta\big)$~\cite{Liu2024geometricmean}
    & $\til{\cO}\big(r_{\rm min}/\delta^2\big)$
    & $\til{\cO}\big(r_{\rm min}^{5/2}/\delta^5\big)$~\cite{gilyen2022improvedfidelity}
    \\[5ex]
    Purified sample
    & \makecell[c]{$^{\|, \S}\til{\cO}\big(\kappa_{\rm min}^2/\delta^2\big)$ \\[1.2ex] $^\ddag\til{\cO}\big(\kappa_\rho\kappa_\sigma/\delta^2\big)$}
    & \makecell[c]{$^{\P, \S}\til{\cO}\big(\kappa_{\rm min}^{11/2}/\delta^{12}\big)$~\cite{gilyen2022improvedfidelity} \\[1.2ex] $^{\P, \ddag}\til{\cO}\big(\kappa_{\rm min}^5\kappa_\rho^2\kappa_\sigma^2/\delta^3\big)$~\cite{Liu2024geometricmean} \\[1.2ex] $^{\P, \|, \S}\cO\big(\kappa_{\rm max}^2/\delta^2\big)$~\cite{lowe2026randomdimensionreductionlearning}}
    & \makecell[c]{$^\|\til{\cO}\big(r_{\rm min}^2/\delta^2\big)$ \\[1.2ex] $\til{\cO}\big(r_{\rm min}^2/\delta^4\big)$}
    & \makecell[c]{$^\P\til{\cO}\big(r_{\rm min}^{11/2}/\delta^{12}\big)$~\cite{gilyen2022improvedfidelity} \\[1.2ex] $^{\P, \|}\cO\big(r_{\rm max}^2/\delta^2\big)$~\cite{lowe2026randomdimensionreductionlearning}}
    \\[8ex]
    Mixed sample
    & \makecell[c]{$^{\|, \S}\til{\cO}\big(\kappa_{\rm min}^2/\delta^2\big)$ \\[1.2ex] $^\ddag\til{\cO}\big(d_\sA\kappa_\rho\kappa_\sigma(\kappa_\rho^2 + \kappa_\sigma^2)/\delta^2\big)$}
    & \makecell[c]{$^\S\til{\cO}\big(\kappa_{\rm min}^{11/2}/\delta^{12}\big)$~\cite{gilyen2022improvedfidelity} \\[1.2ex] $^\ddag\til{\cO}\big(\kappa_{\rm min}^5\kappa_\rho^2\kappa_\sigma^2/\delta^3\big)$~\cite{Liu2024geometricmean} \\[1.2ex] $^{\|, \S}\cO\big(\kappa_{\rm max}^2/\delta^2\big)$~\cite{lowe2026randomdimensionreductionlearning}}
    & \makecell[c]{$^\|\til{\Theta}\big(r_{\rm min}^2/\delta^2\big)$ \\[1.2ex] $\til{\cO}\big(d_\sA r_{\rm min}^2(r_\rho^2 + r_\sigma^2)/\delta^8\big)$}
    & \makecell[c]{$\til{\cO}\big(r_{\rm min}^{11/2}/\delta^{12}\big)$~\cite{gilyen2022improvedfidelity} \\[1.2ex] $^\|\cO\big(r_{\rm max}^2/\delta^2\big)$~\cite{lowe2026randomdimensionreductionlearning}}
    \rule[-5.5ex]{0pt}{0pt}
    \\
\end{tabular}
\label{tab:fidelity}
\end{table*}
\renewcommand{\arraystretch}{1.0}


\begin{theorem}[Square root fidelity estimation based on the Uhlmann transformation]
\label{thm:informal fidelity result}
In each of the three query/sample access models, for any $\delta \in (0, 1)$, there exists a quantum query/sample algorithm that outputs $\sqrt{\tilde{\rF}}$ satisfying $\big|\sqrt{\tilde{\rF}} - \sqrt{\rF}(\rho^\sA, \sigma^\sA)\big| \leq \delta$ with probability at least $2/3$. 
The computational costs of these algorithms in each model are as follows.
\begin{itemize}
\item The purified query access model: 
\begin{align}
    q = \til{\cO}\Big(\f{1}{\delta}{\rm min}\Big\{\f{1}{s_{\rm min}}, \f{R}{\delta}\Big\}\Big)
\end{align}
queries to $U_\rho^{\sA\sB}$, $U_\sigma^{\sA\sB}$, and their inverses. 
The quantum circuit consists of $\til{\cO}(q)$ one- and two-qubit gates, and uses $\cO(\log{(d_\sA d_\sB)} + \log{(1/\delta}))$ qubits at once. The circuit description can be computed classically in $\poly\big(q, \log{(d_\sA d_\sB)}, \log{(1/\delta)}\big)$ time.
\item The purified sample access model: 
\begin{itemize}
\item With collective operations: 
\begin{align}
    n = \til{\cO}\big(r_{\rm min}^2/\delta^2\big)
\end{align}
samples of $\ket{\rho}^{\sA\sB}$ and $\ket{\sigma}^{\sA\sB}$. 
The quantum circuit consists of $\til{\cO}(n^4)$ one- and two-qubit gates, and uses $\til{\cO}(n^2)$ qubits at once. The circuit description can be computed classically in $\poly\big(n, \log{d_\sA}, \log{(1/\delta)}\big)$ time.
\item With non-collective operations: 
\begin{align}
    n = \til{\cO}\bigg(\f{1}{\delta^2} \min\Big\{\f{1}{s_{\rm min}^2}, \f{R^2}{\delta^2}\Big\}\bigg)
\end{align}
samples of $\ket{\rho}^{\sA\sB}$ and $\ket{\sigma}^{\sA\sB}$. 
The quantum circuit consists of $\til{\cO}(n)$ one- and two-qubit gates, and uses $\cO(\log{(d_\sA d_\sB)} + \log{(1/\delta)})$ qubits at once. The circuit description can be computed classically in $\poly\big(n, \log{(d_\sA d_\sB)}, \log{(1/\delta)}\big)$ time.
\end{itemize}
\item The mixed sample access model: 
\begin{itemize}
\item With collective operations: 
\begin{align}
    n = \til{\cO}\bigg(\f{1}{\delta^2} \min\Big\{\f{1}{s_{\rm min}^2}, \f{R^2}{\delta^2}, r_{\rm min}^2\Big\}\bigg)
\end{align}
samples of $\rho^\sA$ and $\sigma^\sA$.
The quantum circuit consists of $\til{\cO}(n^4)$ one- and two-qubit gates, and uses $\til{\cO}(n^2)$ qubits at once. The circuit description can be computed classically in $\poly\big(n, \log{d_\sA}, \log{(1/\delta)}\big)$ time.
\item With non-collective operations: 
\begin{align}
    n = \til{\cO}\bigg(\f{d_\sA}{\delta^2}\min\bigg\{\f{1}{s_{\rm min}^2}, \f{R^2}{\delta^2}\bigg\}\min\bigg\{\kappa_\rho^2 + \kappa_\sigma^2, \f{r_\rho^2 + r_\sigma^2}{\delta^4}\bigg\}\bigg)
\end{align}
samples of $\rho^\sA$ and $\sigma^\sA$. 
The quantum circuit consists of $\til{\cO}(n)$ one- and two-qubit gates, and uses $\cO(\log{d_\sA} + \log{(1/\delta)})$ qubits at once. The circuit description can be computed classically  in $\poly\big(n, \log{d_\sA}, \log{(1/\delta)}\big)$ time.
\end{itemize}
\end{itemize}

\end{theorem}

We summarize our results on the query and sample complexities in \Cref{tab:fidelity}, along with the most efficient results in prior works to our knowledge.
In comparison, we should note that $R \leq r_{\rm min} \leq \kappa_{\rm min}$. Moreover, as derived in \Cref{sec:appendix product of singvalue}, if either $\supp[\rho^\sA] \subset \supp[\sigma^\sA]$ or $\supp[\sigma^\sA] \subset \supp[\rho^\sA]$, we have
\begin{align}
\label{inteq:135}
    1/s_{\rm min} 
    &\leq 1/\sqrt{\rho_{\rm min}\sigma_{\rm min}} = \sqrt{\kappa_\rho\kappa_\sigma}, 
\end{align}
where $\supp[A]$ is the support of a matrix $A$, i.e., the orthogonal complement of the kernel of $A$.
By applying Eq.~\eqref{inteq:135} to the computational costs in \Cref{thm:informal fidelity result} for clarity, the results in \Cref{tab:fidelity} (marked by $\ddag$) are obtained.

In the purified query access model, our result achieves a significant improvement over the previous best-known results in Refs.~\cite{gilyen2022improvedfidelity, Liu2024geometricmean}.
The result in Ref.~\cite{Liu2024geometricmean} can be derived under the assumption that $\supp[\rho^\sA] \subset \supp[\sigma^\sA]$ or $\supp[\sigma^\sA] \subset \supp[\rho^\sA]$.
For simple comparison, we suppose $\kappa$ such that $\kappa_\rho, \kappa_\sigma \leq \kappa$.
While the result in Ref.~\cite{Liu2024geometricmean} scales as $\til{\cO}(\kappa^4/\delta)$, our result scales linearly as $\til{\cO}(\kappa/\delta)$, improving the quartic dependence on $\kappa$ to a linear one.
Furthermore, the result in Ref.~\cite{gilyen2022improvedfidelity} scales as $\til{\cO}\big(r_{\rm min}^{2.5}/\delta^5\big)$, whereas ours scales as $\til{\cO}(r_{\rm min}/\delta^2)$, highlighting two improvements: linear scaling in $r_{\rm min}$ and a substantial reduction in the dependence on $\delta$.
We also demonstrate significant improvements in the purified sample access model, although this may be due to the model not having been thoroughly explored.

In the mixed sample access model, our results $\til{\cO}(\kappa_{\rm min}^2/\delta^2)$ and $\til{\Theta}(r_{\rm min}^2/\delta^2)$ with collective operations (marked by $\|$ in \Cref{tab:fidelity}) show notable improvement on the sample complexity over the previous results $\til{\cO}(\kappa_{\rm min}^{5.5}/\delta^{12})$ and $\til{\cO}(r_{\rm min}^{5.5}/\delta^{12})$~\cite{gilyen2022improvedfidelity} in the dependence on both these parameters and $\delta$.
Compared with $\cO(\kappa_{\rm max}^2/\delta^2)$ and $\cO(r_{\rm max}^2/\delta^2)$~\cite{lowe2026randomdimensionreductionlearning}, our results depend only on the smaller of the two parameters, which leads to a substantial reduction when one of the states has a much smaller rank than the other.
Moreover, our result $\til{\Theta}(r_{\rm min}^2/\delta^2)$ achieves, up to logarithmic factors, the number of samples $\cO(r_{\rm min}^2/\delta^2)$ known for the case where one of the states is classically specified~\cite{lee2026fidelityestimationknownquantum}, while both states are given only as samples in our setting.

We achieve these results by employing a recently developed technique known as the \emph{random purification}~\cite{tang2025conjugatequerieshelp}. Our approach requires purified states, while only copies of the mixed states $\rho^\sA$ and $\sigma^\sA$ are given in the mixed sample access model. The random purification collectively converts copies of each of these states into copies of a single purified state, in which a unitary on the purifying system is chosen at random. As the fidelity does not depend on the choice of purified state, this randomness does not affect the estimation. For the tight sample complexity $\til{\Theta}(r_{\rm min}^2/\delta^2)$, we estimate the square root fidelity directly from these copies through the block-encoding used in our Uhlmann transformation algorithm, without implementing the transformation (see \Cref{sec:fidelity estimation collective}). We here replace the density matrix exponentiation with its virtual version~\cite{wada2025statetohamiltonianconversioncopies}, which incurs no approximation error for pure states when only expectation values are needed. By applying the random purification and the virtual density matrix exponentiation to the fidelity estimation, we highlight their remarkable power.

It is also noteworthy that, focusing on the $\delta$-scaling, our results, $\til{\cO}(1/(\delta s_{\rm min}))$ queries and $\til{\cO}(1/(\delta^2 s_{\rm min}^2))$ samples in the purified query and mixed sample access models, are nearly optimal. This is because the lower bounds on the query and sample complexities for square root fidelity estimation are known as $\Omega(1/\delta)$ and $\Omega(1/\delta^2)$~\cite{Liu2024geometricmean}, respectively.
In terms of the rank, our results $\til{\cO}(r_{\rm min}/\delta^2)$ and $\til{\Theta}(r_{\rm min}^2/\delta^2)$ in these models also achieve the optimal dependence on $r_{\rm min}$ up to logarithmic factors, given the lower bounds $\til{\Omega}(r_{\rm min}/\delta)$~\cite{wang2026unifiedcomplexityframeworkquantum} and $\til{\Omega}(r_{\rm min}^2/\delta^2)$~\cite{wang2026unifiedcomplexityframeworkquantum, lee2026fidelityestimationknownquantum}, and the number of samples is optimal also in $\delta$.
The optimality with respect to the $s_{\rm min}$-scaling is still open.

Without collective operations, the number of samples of our algorithm in the mixed sample access model depends on the dimension $d_\sA$, while it still improves on the previous result in Ref.~\cite{Liu2024geometricmean}, for example, when $\kappa_\rho = \kappa_\sigma \geq d_\sA$.
This dependence on $d_\sA$ comes from the canonical purification. Although the Uhlmann transformation inherently requires purified states, no algorithm is currently known that realizes purification solely through sequential operations, with a number of samples that scales with the rank rather than the dimension.
If such a purification algorithm were developed, the Uhlmann transformation would become more tractable. This could improve the sample complexity for fidelity estimation in the mixed sample access model, even without collective operations.

We comment on the number of qubits in the cases with and without collective operations.
When we do not use collective operations, the quantum circuits use $\cO\big(\log{(d_\sA d_\sB)} + \log{(1/\delta)}\big)$ qubits at any one time. The term $\cO(\log(d_\sA d_\sB))$ comes from the systems holding the quantum states, while $\cO(\log{(1/\delta)})$ comes from auxiliary qubits, such as those for the quantum phase estimation in the algorithm in Refs.~\cite{Wang2024OptimalFidelityPure, wang2024sampleoptimal}.
On the other hand, when we employ collective operations, $\til{\cO}(n^2)$ qubits are used simultaneously, where $n$ is the number of sampled states. 
This is due to the Schur transform applied in the random purification~\cite{tang2025conjugatequerieshelp} and follows from the results in Ref.~\cite{burchardt2025highdimensionalschur}.
As Theorem~\ref{thm:informal fidelity result} shows, collective operations can dramatically reduce the sample cost $n$ in the sample access models; however, the number of qubits required at once increases with $n$.
Thus, a non-collective approach may be more suitable when the number of simultaneously available qubits is limited.

\subsubsection{Information-theoretic tasks with the decoupling approach}
\label{sec:overview of decoup}

To achieve the theoretical limits of various information-theoretic tasks, the decoupling approach~\cite{hayden2008decoupling, dupuis2010decoupling, Szehr2013decouple2design, dupuis2014one, preskill2016quantumshannon}, which relies on the existence of the Uhlmann transformation, is widely employed (see e.g., Refs.~\cite{hayden2008decoupling, Abeyesinghe2009Motherfamilytree, Datta2011apexfamilytree}).
Although the existence of the Uhlmann transformation is guaranteed, its implementation methods have long been of interest in the context of the decoupling approach.

Within the decoupling approach, we apply the proposed Uhlmann transformation algorithms to information-theoretic tasks and evaluate the computational costs. 
We particularly focus on two tasks: \emph{quantum communication}~\cite{lloyd1997capacity, shor2002quantum, devetak2005private, barnum1998inftrans, datta2013oneshotentassistQCcom, wakakuwa2023oneshot, khatri2024principlemodern} and \emph{quantum state merging}~\cite{Horodecki2005Partialquantinfo, berta2008singleshot, dupuis2014one, yamasaki2019statemergesmalldimension}.
Remarkably, the number of channel uses and samples required for the decoding scales with the rank of the decoding party's state, rather than with the dimensions of the systems.
As a result, when this rank is polynomial in the number of qubits, the decoding can be implemented with computational cost polynomial in the number of qubits, even though the dimensions of the systems are exponential. 

The same scaling with the rank holds even in the setting corresponding to the mixed sample access model, where the decoding party has access only to copies of a mixed state. In the decoding, both the input and target states of the Uhlmann transformation are prepared from this mixed state by operations that do not act on its purifying system, which makes the transformation independent of the choice of purified state. The random purification using collective operations thus allows the decoding party to work with these copies.

For quantum state merging, we also show that, in the standard protocol based on the decoupling approach, the number of samples in our decoding operation is optimal up to logarithmic factors in the worst case among all operations of the decoding party achieving the same guarantee on the fidelity (see \Cref{sec:lower bound merging}).
This is shown by constructing a family of initial states containing unknown unitaries, for which the merging requires undoing these unitaries, and by reducing the task to learning an unknown unitary from samples, as in the lower bounds in \Cref{sec:overview of lower bound}.

Comprehensive discussions, including the settings, are provided in \Cref{sec:decoupling and uhlmann}, and the main results can be found in \Cref{thm:Uhlmann decoder ent trans} for quantum communication and in \Cref{thm:Uhlmann state merging} for quantum state merging.

\subsubsection{Algorithmic implementation of the Petz recovery map}
\label{sec:overview of petz}

The Petz recovery map~\cite{petz1986sufficient, petz1988sufficiency} is one of the few powerful tools for reversing the effect of a noisy quantum channel, achieving nearly optimal performance and having an explicit form.
Quantum algorithms for the Petz recovery map have been studied, for example, in Refs.~\cite{gilyen2022petzmap, biswas2023noiseadapted, zhao2026petzrecoverysinglepolar}, and have demonstrated significant progress toward its implementation.
These algorithms, however, require access to a Stinespring unitary of the noisy quantum channel and its inverse, which can be challenging in practice.
We introduce a method for implementing the Petz recovery map directly from uses of the noisy quantum channel, without relying on its Stinespring unitary.

The key idea is to realize a Stinespring isometry of the channel $\cF^{\sA\rarr\sB}$ by the Uhlmann transformation algorithm for the purified sample access model. Let $\tau^{\sR\sB} = \cF^{\sA\rarr\sB}(\Phi^{\sR\sA})$ be the \emph{Choi--Jamio\l{}kowski state}~\cite{jamiolkowski1972linear, choi1975completely} of $\cF^{\sA\rarr\sB}$, where $\ket{\Phi}^{\sR\sA}$ is a maximally entangled state. Any purified state of $\tau^{\sR\sB}$ is written as $\ket{\tau}^{\sR\sB\sE} = V_\cF^{\sA\rarr\sB\sE}\ket{\Phi}^{\sR\sA}$ with a Stinespring isometry $V_\cF^{\sA\rarr\sB\sE}$ of $\cF^{\sA\rarr\sB}$. Since $\tau^\sR$ is the completely mixed state, the Uhlmann partial isometry from $\ket{\Phi}^{\sR\sA}\ket{0}^\sG$ to $\ket{\tau}^{\sR\sB\sE}$ coincides with $V_\cF^{\sA\rarr\sB\sE}$ when the input on the auxiliary system $\sG$ is $\ket{0}^\sG$.

To run the algorithm, we prepare the copies of a purified state of $\tau^{\sR\sB}$ from the outputs of $\cF^{\sA\rarr\sB}$ by the random purification. Because all the copies share the same randomly chosen purified state, this construction simulates sequential uses of a block-encoding of a randomly chosen Stinespring isometry and its inverse.
Following Ref.~\cite{zhao2026petzrecoverysinglepolar}, we obtain the Petz recovery map by applying the QSVT with the sign function to an operator constructed from the Stinespring isometry and the reference state, which requires such sequential uses. The randomness of the isometry does not matter, as the Petz recovery map does not depend on the choice of the Stinespring isometry. Related efficient conversions of channel uses into uses of a randomly chosen Stinespring isometry were also given in Refs.~\cite{girardi2026randomstinespringsuperchannelconverting, yoshida2026randomdilationsuperchannel}. As a result, we implement the Petz recovery map for any reference state and any quantum channel, and evaluate the number of uses of the channel and samples of the reference state. A formal statement (\Cref{thm:petz implementation}) and a detailed discussion are provided in \Cref{sec:sample petz}.


\section{Summary and outlook}
\label{sec:conclusion}

This work provides quantum query and sample algorithms for implementing the Uhlmann transformation in the form of quantum circuits, building on the QSVT and the density matrix exponentiation. We evaluate the query and sample complexities, the number of one- and two-qubit gates, and the number of qubits at any one time, for our algorithms in the purified query, purified sample, and mixed sample access models. In the purified query and sample access models, our algorithms demonstrate significant improvements in these computational costs over the naive tomography-based approach and the one in Ref.~\cite{metger2023stateqipspace}. We also consider a variant of the Uhlmann transformation for the canonical purification and propose a quantum sample algorithm in the mixed sample access model, which can be implemented with relatively few samples. Moreover, we derive lower bounds on the query and sample complexities required for any quantum algorithm to perform the Uhlmann transformation in the three models. In the purified query and sample access models, these bounds coincide with the costs of our algorithms up to logarithmic factors in terms of the rank, the minimum non-zero singular value $s_{\rm min}$, and the accuracy, and thus characterize the query and sample complexities of the Uhlmann transformation. In the mixed sample access model with the canonical purification, the dependence on the dimension is unavoidable, and our algorithm attains the optimal scaling in the dimension and the accuracy.

With our Uhlmann transformation algorithms, we investigate square root fidelity estimation.
In particular, we show that our algorithms achieve better upper bounds on query and sample complexities than the previous best algorithms \cite{gilyen2022improvedfidelity, Liu2024geometricmean, lowe2026randomdimensionreductionlearning}, with the optimal dependence on the rank.
In the mixed sample access model with collective operations, the number of samples is also optimal in the accuracy, which we achieve by estimating the square root fidelity directly from the block-encoding used in our Uhlmann transformation algorithm.
Combined with the decoupling approach, we apply our Uhlmann transformation algorithms to two tasks: quantum communication and quantum state merging.
For these tasks, we provide a comprehensive analysis of the computational costs, demonstrating the broad applicability of the algorithms. In the mixed sample access model, the random purification plays an important role in accomplishing these tasks with copies of a mixed state.
We also apply our Uhlmann transformation algorithm to an algorithmic implementation of the Petz recovery map.
While quantum algorithms for the Petz recovery map have been proposed using a Stinespring unitary and its inverse~\cite{gilyen2022petzmap, biswas2023noiseadapted,zhao2026petzrecoverysinglepolar}, we develop an algorithm that relies on the quantum channel defining the Petz recovery map, rather than its Stinespring unitary.
These results, hence, contribute to the realization of such widely recognized and useful protocols.

An important future direction is to explore how our results could be related to quantum complexity theory and quantum cryptography.
Recently, the Uhlmann transformation has been studied from a complexity-theoretic perspective, e.g., in Refs.~\cite{metger2023stateqipspace, bostanci2023unicompuhlmann, chia2024complexitytheoryquantumpromise, legall2026slightlyimprovedQSZK}. 
Specifically, Ref.~\cite{metger2023stateqipspace} has provided valuable insights into complexity classes related to space and interactive proofs via the transformation, while complexity classes related to time are still not well understood.
Our Uhlmann transformation algorithm can be implemented with polynomial circuit complexity, provided that the purified state preparation unitaries are efficiently implementable and that either $1/s_{\rm min}$ or $R/\delta$ scales polynomially with the number of qubits.
This fact may offer a novel viewpoint on quantum complexity theory and quantum cryptography.

Another potential avenue is to investigate the gap between the purified and mixed sample access models.
The Uhlmann transformation is ill-defined in the mixed sample access model because it inherently requires purification. To address this, we employed a quantum algorithm that realizes the canonical purification; however, this results in exponential costs, which our lower bound shows to be unavoidable. If a more efficient algorithm for another deterministic purification is developed, this exponential factor would be removed, making the Uhlmann transformation more tractable even in the mixed sample access model. As a result, the exponential gap between the purified and mixed sample access models could be bridged. On the other hand, if purification is easily implementable, it may break certain cryptographic assumptions, potentially leading to a ``no-go'' result for purification.
Although several results are known for specific cases~\cite{Chen2021ValiationalTraceFidelity, Ezzell2023MixedCompil, liu2024exponentialseparations, chen2024localtestUniInv, Sala2024SpontaneousPurif, weinstein2024efficientdetectionstrongtoweakspontaneous}, and efficient universal purification is impossible in the worst case~\cite{liu2025universalpurification}, a comprehensive understanding---especially whether it can be done at a cost on the order of the rank---has not yet been achieved.

\section{Preliminaries}
\label{sec:preliminaries}

We here provide technical preliminaries.
We introduce extra notation in \Cref{sec:notation}. In \Cref{sec:uhl intro}, we overview known facts about the Uhlmann transformation.
We present the three computational frameworks for quantum algorithm design in \Cref{sec:three computational models}.
In \Cref{sec:previous key subroutine}, we give a quick review of key algorithmic primitives, including block-encoding, the QSVT, and the density matrix exponentiation, which form the building blocks of our algorithms.

\subsection{Extra notations}
\label{sec:notation}

We denote a Hilbert space by $\cH$, and the space associated with a system $\sA$ is denoted by $\cH^\sA$. Hilbert spaces such as $\cH^{\sA'}$, $\cH^{\hat{\sA}}$, or $\cH^{\sA_1}$ are isomorphic to $\cH^\sA$, meaning that they have the same dimension and can be equipped with the same orthonormal basis as $\cH^\sA$.
This applies not only to the system $\sA$, but also to any systems, such as $\cH^{\sB'}$ and $\cH^{\hat{\sC}}$. The dimension of $\cH$ is denoted by $d$ and, for instance, the dimension of $\cH^\sA$ is denoted by $d_\sA$.

For an operator $M$, the kernel of $M$, i.e., the subspace spanned by vectors mapped to zero by $M$, is denoted by $\ker[M]$. The support of $M$ is denoted by $\supp[M]$, which is defined as the orthogonal complement of $\ker[M]$.
The image of $M$, i.e., the subspace spanned by all vectors of the form $M\ket{v}$, is denoted by $\im[M]$.
For any operator $M$, we denote the complex conjugate and the transpose in a given orthonormal basis by $M^*$ and $M^\top$, respectively, and denote the Hermitian conjugate by $M^\dag$.

We omit the symbol of the tensor product between vectors and denote it as $\ket{\rho}\otimes\ket{\sigma} = \ket{\rho}\ket{\sigma}$.
We denote by $\ket{\Phi}$ the maximally entangled state defined in the orthonormal computational basis. For instance, the maximally entangled state between $\sA$ and $\hat{\sA}$ is $\ket{\Phi}^{\sA\hat{\sA}} = \f{1}{\sqrt{d_\sA}}\sum_{i=0}^{d_\sA-1}\ket{i}^\sA\ket{i}^{\hat{\sA}}$, where $\{\ket{i}\}_{i=0}^{d_\sA-1}$ is the computational basis in $\sA$ and $\hat{\sA}$, respectively. 
We also denote an unnormalized maximally entangled state by $\ket{\Gamma}$, e.g., $\ket{\Gamma}^{\sA\hat{\sA}} = \sum_{i=0}^{d_\sA-1}\ket{i}^\sA\ket{i}^{\hat{\sA}}$ for $\sA$ and $\hat{\sA}$. Note that $\sum_{i=0}^{d-1}\ket{i}\ket{i} = \sum_{i=0}^{d-1}\ket{e_i}\ket{e_i^*}$ holds for any orthonormal basis $\{\ket{e_i}\}_{i=0}^{d-1}$, where the complex conjugate is taken with respect to the computational basis $\{\ket{i}\}_{i=0}^{d-1}$. The completely mixed state is denoted by $\pi$, such as $\pi^\sA = \bI^\sA/d_\sA$ for a system $\sA$.

For an isometry $U$, the corresponding isometry channel $\cU$ is defined by $\cU(\cdot) = U(\cdot)U^\dag$, and its inverse is defined by $\cU^\dag = U^\dag(\cdot)U$.
For any state $\psi^\sA$, let $\cP_{\psi}^{\bC\rarr\sA}$ be the state-preparation channel defined by $\cP_{\psi}^{\bC\rarr\sA} = \psi^\sA$.
We denote the trace map by $\tr$, and denote the partial trace map over a subsystem, for example, $\sA$, by $\tr_\sA$.
We denote the swap operator by $F$, such as $F^{\sA\hat{\sA}}$, which swaps the systems $\sA$ and $\hat{\sA}$.

For a matrix $M$, the Schatten-$p$ norm is defined by $\| M \|_p = (\tr[|M|^p])^{1/p}$ for $p\in[1, \infty]$, where $|M| = \sqrt{M^\dag M}$. We particularly use the trace norm, the Hilbert--Schmidt norm, and the operator norm, corresponding to $p=1, 2, \infty$, respectively.
The Schatten $p$-norm is invariant under isometries~\cite{wilde2013QItheory, bhatia2013matrix, watrous2018TheoryQI}: $\|M\|_p = \|UMV^\dag\|_p$ for isometries $U$ and $V$. It is also non-increasing in $p$: $\|M\|_q \leq \|M\|_p$ for $1 \leq p \leq q \leq \infty$.
For the Schatten $p$-norm, H\"{o}lder's inequality~\cite{Rogers1888extension, Holder1889inequality} holds: for any matrices $M_1$ and $M_2$,
\begin{align}
\label{eq:Holder ineq}
    \|M_1M_2\|_r \leq \|M_1\|_p\|M_2\|_q,
\end{align}
where $1/r = 1/p + 1/q$.
The trace norm has a contraction property against the partial trace~\cite{rastegin2012NormPartialTrace, watrous2018TheoryQI} such that for any matrix $M^{\sA\sB}$,
\begin{equation}
\label{eq:contraction trace norm}
    \big\|M^{\sA\sB}\big\|_1 \geq \big\|\tr_{\sB}\big[M^{\sA\sB}\big]\big\|_1.
\end{equation}
Trace norm also has a variational characterization~\cite{nielsen2010quantum, wilde2013QItheory, watrous2018TheoryQI}: for any quantum states $\rho$ and $\sigma$, it holds that
\begin{align}
\label{eq:trace distance variational}
    \f{1}{2}\|\rho - \sigma\|_1 = \max_{M; 0 \leq M \leq \bI}\tr[M(\rho - \sigma)].
\end{align}

The diamond norm~\cite{Kitaev1997QCAlgoEC, watrous2004notessuperoperatornorm, watrous2018TheoryQI} for a linear map $\cT^{\sA\rarr\sB}$ is defined by
\begin{align}
    \big\|\cT^{\sA\rarr\sB}\big\|_\diamond = \max_{\sR}\max_{M^{\sR\sA}; \|M^{\sR\sA}\|_1 \leq 1}\big\|\rid^\sR \otimes \cT^{\sA\rarr\sB}(M^{\sR\sA})\big\|_1, 
\end{align}
where the maximization is taken over all reference systems $\sR$ of arbitrary size and all operators $M^{\sR\sA}$ with trace norm at most one.
For the difference of two quantum channels, it suffices to take the maximization over all pure states on $\sR\sA$, where $\sR$ can be chosen such that $d_\sR = d_\sA$.
The diamond norm is subadditive for the composition of quantum channels~\cite{watrous2018TheoryQI}: for any quantum channels $\cT_1^{\sA\rarr\sB}$, $\cT_2^{\sB\rarr\sC}$, $\cE_1^{\sA\rarr\sB}$, and $\cE_2^{\sB\rarr\sC}$, it holds that
\begin{align}
    &\big\|\cT_2^{\sB\rarr\sC} \circ \cT_1^{\sA\rarr\sB} - \cE_2^{\sB\rarr\sC} \circ \cE_1^{\sA\rarr\sB}\big\|_\diamond  \notag\\
    &\hspace{2pc}\leq \big\|\cT_2^{\sB\rarr\sC} - \cE_2^{\sB\rarr\sC}\big\|_\diamond + \big\|\cT_1^{\sA\rarr\sB} - \cE_1^{\sA\rarr\sB}\big\|_\diamond.
\end{align}

For isometry channels $\til{\cV}$ and $\cV$, the diamond norm and the operator norm are related as~\cite{Kitaev1997QCAlgoEC} 
\begin{align}
\label{eq:rel of operator and diamond norm}
    \f{1}{2}\big\|\til{\cV} - \cV\big\|_\diamond \leq \big\|\til{V} - V\big\|_\infty.
\end{align}
The diamond norm and the fidelity satisfy the following relation: for any quantum channels $\cT_1$ and $\cT_2$, any state $\rho$, and any pure state $\ket{\phi}$,
\begin{align}
\label{eq:fidelity and diamond}
    \big|\rF\big(\cT_1(\rho), \ket{\phi}\big) - \rF\big(\cT_2(\rho), \ket{\phi}\big)\big| \leq \f{1}{2}\big\|\cT_1 - \cT_2\big\|_\diamond,
\end{align}
which follows from the variational characterization of the trace norm and the definition of the diamond norm.

For a vector $v \in \bC^d$, the Euclidean norm is defined by $\|v\| =\sqrt{v^\dag v}$. For pure states $\ket{\psi}$ and $\ket{\varphi}$, the trace norm and the Euclidean norm satisfy the inequalities:
\begin{align}
\label{eq:relation of Euclidean and trace norm}
    \f{1}{2}\big\|\ketbra{\psi}{\psi} - \ketbra{\varphi}{\varphi}\big\|_1 
    &\leq \min_\theta\big\|\ket{\psi} - e^{i\theta}\ket{\varphi}\big\| \\
    &\leq \f{1}{\sqrt{2}}\big\|\ketbra{\psi}{\psi} - \ketbra{\varphi}{\varphi}\big\|_1.
\end{align}


\subsection{Uhlmann transformation}
\label{sec:uhl intro}

We first recall Uhlmann's theorem~\cite{uhlmann1976transition}.
\begin{theorem}[Uhlmann's theorem~\cite{uhlmann1976transition}]
\label{thm:Uhlmann theorem}
    Let $\rho^\sA$ and $\sigma^\sA$ be quantum states. The fidelity $\rF(\rho^\sA, \sigma^\sA)$ can be rephrased using their purified states $\ket{\rho}^{\sA\sB}$ and $\ket{\sigma}^{\sA\sB}$ as
\begin{align}
    \rF(\rho^\sA, \sigma^\sA) 
    &= \max_{U^\sB}\big| \bra{\sigma}^{\sA\sB}(\bI^\sA \otimes U^{\sB})\ket{\rho}^{\sA\sB} \big|^2 \\
    &= \max_{U^\sB}\rF(U^\sB \ket{\rho}^{\sA\sB}, \ket{\sigma}^{\sA\sB}),
\end{align}
where the maximization is taken over all unitaries on $\sB$.
\end{theorem}
This theorem states that for a given state $\ket{\rho}^{\sA\sB}$, there exists a unitary $U^\sB$ that realizes the transformation from $\ket{\rho}^{\sA\sB}$ to $\ket{\sigma}^{\sA\sB}$, with the minimum error achievable by local operations on $\sB$.
This optimal transformation and the corresponding unitary $U^\sB$ are called the \emph{Uhlmann transformation} and the \emph{Uhlmann unitary}, respectively.

While we focus here on transformations between pure states, the Uhlmann transformation actually characterizes the optimal local transformation between mixed states, by considering an environment. See \Cref{sec:opt local and Uhlmann}.

We should remark on several points.
First, since we can always expand the purifying systems of $\rho^\sA$ and $\sigma^\sA$ by adding additional qubits, we can assume that their purifying systems are of the same size, which we denote by $\sB$.
Second, the Uhlmann unitary is not unique; for instance, the unitary $U^\sB$ can act arbitrarily on the kernel of $\rho^\sB$.
Hence, the partial isometry $V^\sB$, which is the part of $U^\sB$ uniquely determined up to a global phase, is truly crucial.
We refer to this partial isometry $V^\sB$ as a \emph{Uhlmann partial isometry}.

An explicit form of the Uhlmann partial isometry $V^\sB$, which consists of $\ket{\rho}^{\sA\sB}$ and $\ket{\sigma}^{\sA\sB}$, is implied in Ref.~\cite{Jozsa1994FidelityforMixedQuantumStates}, and recently its uniqueness and robustness have been studied in Ref.~\cite{bostanci2025localtransformationsbipart}.
To present this, we define the sign function over real numbers as  
\begin{equation}
\sgn(x) = 
\begin{cases}
-1 & x<0, \\
0 & x=0, \\
1 & x>0.
\end{cases}
\end{equation}
Then, for a matrix $A$, $\sgn^{(\rm SV)}(A)$ is defined as the sign function acting on the singular values of $A$; that is, when $A$ has the singular value decomposition $A = \sum_j s_j \ketbra{\eta_j}{\xi_j}$, we have
\begin{align}
\label{inteq:158}
    \sgn^{(\rm SV)}(A) = \sum_j \sgn(s_j) \ketbra{\eta_j}{\xi_j}.
\end{align}

\begin{proposition}[Explicit form of the Uhlmann partial isometry~{\cite[Lemma 6]{Jozsa1994FidelityforMixedQuantumStates}}]
\label{prop:explicit form of Uhl}
An explicit form of the Uhlmann partial isometry is given by 
\begin{align}
\label{eq:explicit form Uhl partial iso}
    V^\sB = \sgn^{(\rm SV)}\big(\tr_\sA\big[\ketbra{\sigma}{\rho}^{\sA\sB}\big]\big).
\end{align}
\end{proposition}

One can readily check that the singular values of $\tr_\sA\big[\ketbra{\sigma}{\rho}^{\sA\sB}\big]$ and $\sqrt{\sigma^\sA}\sqrt{\rho^\sA}$ are identical.
Their ranks are also equal, as they have the same number of non-zero singular values.
In \Cref{sec:analyze partial iso}, we provide a derivation of \Cref{prop:explicit form of Uhl} for completeness.

Note that, since $V^\sB$ in Eq.~\eqref{eq:explicit form Uhl partial iso} is a partial isometry, it cannot be directly implemented as a deterministic physical operation. However, with the assistance of an ancillary system, the Uhlmann transformation can be approximately realized as a quantum channel.
That is, we aim for a transformation: $\ket{\rho}^{\sA\sB}\ket{\phi}^\sF \mapsto V^\sB\ket{\rho}^{\sA\sB}\ket{\phi'}^\sF$, with auxiliary states $\ket{\phi}^\sF$ and $\ket{\phi'}^\sF$, by operating on the system $\sB\sF$.


\subsection{Computational models for algorithm design}
\label{sec:three computational models}

We construct quantum algorithms that realize the Uhlmann transformation in the following three standard computational models. 
Here, we explain these computational models for a general state $\omega$; in the case of the Uhlmann transformation, $\omega$ is taken to be $\rho$ and $\sigma$.
\begin{enumerate}[label=\Roman*.]
    \item Purified query access model~\cite{watrous2002limitQSZK, watrous2009QSZKquantumattacks, belovs2019ClassicalDist, gilyen2019DistPropTesting, brandao2019SDPOpt, vanapeldoorn2019impleveSDP, Subramanian2021estimatealphalenni, gilyen2022improvedfidelity, Apeldoorn2023TomoPurifQuery, wang2023QAfidelityest, wang2023SampletoQuary, luo2024SuccinctTest, wang2024NewQuantEnt, Liu2024geometricmean}: let $U_\omega^{\sA\sB}$ be a unitary that prepares its purified state $\ket{\omega}^{\sA\sB}$, i.e., $U_\omega^{\sA\sB}\ket{0}^{\sA\sB}=\ket{\omega}^{\sA\sB}$. We assume that $U_\omega^{\sA\sB}$ is available multiple times as a unitary oracle.
    It is also assumed that we can utilize its inverse $(U_\omega^{\sA\sB})^\dag$. An algorithm in this model is required to work for any such unitary.
    \item Purified sample access model~\cite{chen2024localtestUniInv, liu2024exponentialseparations}:
    multiple independent and identical copies of purified states $\ket{\omega}^{\sA\sB}$ are available.
    \item Mixed sample access model~\cite{gilyen2022improvedfidelity, wang2023SampletoQuary, wang2024TimeEfficientEntEstSamp, wang2024sampleoptimal}: multiple independent and identical copies of states $\omega^\sA$ are available.
\end{enumerate}

The purified query access model (Model~I) has been actively investigated in connection with quantum algorithms for processing classical-data distributions in oracle settings~\cite{belovs2019ClassicalDist, gilyen2019DistPropTesting, luo2024SuccinctTest}.
It is also commonly employed in quantum computational complexity~\cite{watrous2002limitQSZK, watrous2009QSZKquantumattacks}, quantum query algorithms~\cite{gilyen2022improvedfidelity, Apeldoorn2023TomoPurifQuery}, and semidefinite programming~\cite{brandao2019SDPOpt, vanapeldoorn2019impleveSDP}. 
The purified sample access model (Model~II) was recently considered in Refs.~\cite{chen2024localtestUniInv, liu2024exponentialseparations}.
Model~I and Model~II are suitable when we can manipulate the entire system, for example, in simulating quantum systems on a quantum computer.
Model~II further allows the relevant states to be prepared by other parties, even when one is unable to generate them. 
The mixed sample access model (Model~III) is now the standard sample input model for quantum property testing (see, e.g., Ref.~\cite{MdW16surveyquantproptest}), employed in many fundamental tasks such as quantum state tomography~\cite{ODonnell2016EffTomo, Haah2017samploptTomo, Guta2020FastTomo, Chen2023WhenAdapTomo, hu2024sampleoptimalmemoryefficient}.

These computational models are hierarchically related in terms of their computational power.
It is clear that Model~III can be simulated by Model~II, and Model~II can be simulated by Model~I. 
In other words, let $\mathrm{Q}$, $\mathrm{S}_{\mathrm{p}}$, and $\mathrm{S}_{\mathrm{m}}$ be the query/sample complexity of Models I, II, and III, respectively. 
Then, it immediately holds that $\mathrm{Q} \leq \mathrm{S}_{\mathrm{p}} \leq \mathrm{S}_{\mathrm{m}}$ for any computational tasks. 
However, their quantitative relation is generally unclear. 
A known simple separation between $\mathrm{Q}$ and $\mathrm{S}_{\mathrm{p}}$ is by, for example, the estimation of pure-state closeness: for the problem of estimating the trace distance between two pure states to within additive error $\delta$, it is known that $\mathrm{Q} = \Theta(1/\delta)$~\cite{Wang2024OptimalFidelityPure,Fang2025OptimalFidelityToPure} whereas $\mathrm{S}_{\mathrm{p}} = \Omega(1/\delta^2)$ (implied by Ref.~\cite{ALL22} using quantum state discrimination). 
On the other hand, it was shown in Ref.~\cite{chen2024localtestUniInv} that, for testing unitarily invariant properties, $\mathrm{S_p} = \mathrm{S_m}$ holds, that is, Model~II has no advantage over Model~III.
In general cases, however, whether there is a strict separation between $\mathrm{S}_{\mathrm{p}}$ and $\mathrm{S}_{\mathrm{m}}$ is not yet known.

Our results provide some insights into this open problem.
For the Uhlmann transformation for the canonical purification, our lower bound shows that $\mathrm{S}_{\mathrm{m}} = \Omega(d_\sA/\delta)$ even for pure states, whereas our algorithm in Model~II uses a number of samples independent of $d_\sA$ when given copies of the canonical purified states.
Although this gap relies on the purification given in Model~II being the canonical one, it suggests that the Uhlmann transformation is a viable candidate that could show the separation.


\subsection{Building blocks of quantum algorithms}
\label{sec:previous key subroutine}

The Uhlmann transformation algorithm can be constructed by combining several quantum algorithms as subroutines. 
Here, we briefly review these primitives: the block-encoding and the QSVT in \Cref{sec:BE and QSVT intro}, the density matrix exponentiation in \Cref{sec:DME intro}, and the random purification in \Cref{sec:random purification intro}.

\subsubsection{Block-encoding and the quantum singular value transformation}
\label{sec:BE and QSVT intro}

The notion of block-encoding was introduced in Refs.~\cite{Low2019HamiltonianQubitize, gilyen2019qsvt} to describe that certain operators of interest $A$ are encoded as a block of a larger unitary $U$.
The block is specified by an $a$-qubit auxiliary state $\ket{0^a}$ as 
\begin{align}
        U
     = \hspace{1mm}
\begin{blockarray}{ccc}
& \bra{0^a} &  & \vspace{1mm}\\
\begin{block}{c(cc)}
  \ket{0^a} \hspace*{1mm} & A/\alpha & \hspace{0mm} * \hspace{3mm}\\
   \hspace*{1mm} & * & \hspace{0mm} * \hspace{3mm}\\
\end{block}
\end{blockarray}
\ \  \longleftrightarrow \ \  A = \alpha(\bra{0^a} \otimes \bI)U(\ket{0^a} \otimes \bI).
\end{align}
Here, $\alpha$ is the normalization constant such that $\alpha \geq \|A\|_\infty$, which ensures that $U$ is a unitary with appropriate choices of the remaining blocks in $U$. 
More generally, a unitary $U$ is called an $(\alpha, a, \epsilon)$-block-encoding (unitary) of $A$, when $\big\|A - \alpha(\bra{0^a} \otimes \bI)U(\ket{0^a} \otimes \bI)\big\|_\infty \leq \epsilon$. If $(\alpha, a, \epsilon) = (1, a, 0)$ and $a$ is clear from the context, we simply refer to $U$ as an (exact) block-encoding (unitary) of $A$ and omit $a$ from $\ket{0^a}$.

Block-encoding is a powerful tool for implementing flexible operations.
In particular, one can construct a unitary $\tilde{U}$ that block-encodes a certain degree odd or even polynomial $P$ of a matrix.  
The construction is called the \emph{quantum singular value transformation (QSVT)}~\cite{gilyen2019qsvt, Gilyn2019thesis}:
\begin{align}
        U
     = \!\hspace{1mm}
\begin{blockarray}{ccc}
& \bra{0^a} &  & \vspace{1mm}\\
\begin{block}{c(cc)}
  \ket{0^a} \hspace*{1mm} & A/\alpha & \hspace{0mm} * \hspace{3mm}\\
   \hspace*{1mm} & * & \hspace{0mm} * \hspace{3mm}\\
\end{block}
\end{blockarray}
\ \ \overset{\rm QSVT}{\longrightarrow}  \
        \til{U}
     = \!\hspace{1mm}
\begin{blockarray}{ccc}
& \bra{0^{\til{a}}} &  & \vspace{1mm}\\
\begin{block}{c(cc)}
  \ket{0^{\til{a}}} \hspace*{1mm} & P^{(\rm SV)}(A/\alpha) & \hspace{0mm} * \hspace{3mm}\\
   \hspace*{1mm} & * & \hspace{0mm} * \hspace{3mm}\\
\end{block}
\end{blockarray}\hspace{1mm},
\end{align}
where $P^{(\rm SV)}(A/\alpha)$ is defined by 
\begin{align}
\label{eq:P poly def}
    P^{(\rm SV)}(A/\alpha) =
    \begin{cases}
    \sum_j P(s_j/\alpha)\ketbra{\eta_j}{\xi_j}, \ \  \text{if $P$ is odd,} \\
    \sum_j P(s_j/\alpha)\ketbra{\xi_j}{\xi_j}, \ \ \text{if $P$ is even,}
    \end{cases}
\end{align}
and $A = \sum_j s_j \ketbra{\eta_j}{\xi_j}$ is the singular value decomposition of $A$. The sum includes the zero singular values so that $\{\ket{\xi_j}\}_j$ is an orthonormal basis of the input space.
The formal statement of the QSVT, specialized to real even/odd polynomials, is as follows, where the polynomial $P$ must satisfy some conditions.

\begin{theorem}[Quantum singular value transformation with a real polynomial~{\cite[Corollary 11]{gilyen2019qsvt}}; see also Ref.~{\cite[Corollary 2.3.8]{Gilyn2019thesis}}]
\label{prop:general QSVT}
Let $P$ be a real even/odd polynomial of degree $l$ satisfying $|P(x)| \leq 1$ for $x \in [-1,1]$, and let $U$ be an $(\alpha, a, 0)$-block-encoding unitary of $A$.
Then, a quantum circuit for implementing a $(1, a+1, 0)$-block-encoding unitary $\tilde{U}$ of $P^{(\rm SV)}(A/\alpha)$ can be constructed using $(l+1)/2$ calls to $U$ and $(l-1)/2$ calls to $U^\dag$ if $P$ is odd, and $l/2$ calls to $U$ and $l/2$ calls to $U^\dag$ if $P$ is even, along with $\cO(al)$ other one- and two-qubit gates.
\end{theorem}

Note that it is straightforward to construct $\tilde{U}^\dag$ by slightly modifying the construction of $\tilde{U}$~\cite{gilyen2019qsvt, Gilyn2019thesis, martyn2021grand}.

To implement the QSVT, the phase factors in its quantum circuit, which correspond to a target polynomial $P$, need to be computed in a classical pre-processing step, and algorithms computing them in $\poly(l)$ classical time for a degree-$l$ polynomial are known~\cite{Gilyn2019thesis, Haah2019product, Chao2020FindingAF, dong2021findingphase, Lin2022LectureNotes, Mizuta2023vzw}. Meanwhile, as pointed out in Refs.~\cite{metger2023stateqipspace, gall2024Spaceboundedqsvt}, the space cost of this pre-processing step has not been well controlled and can become prohibitively large, especially when the degree of the polynomial grows exponentially. These works have also shown that, at the cost of worse error scaling or increased query cost, the QSVT for a wide class of functions can be implemented with space-efficient classical pre-processing, even when the degree of the polynomial is exponential. In this study, we do not consider the space complexity of the pre-processing step in detail; if needed, one may follow the approach in these works.

By choosing a polynomial $P$ appropriately, we can implement various useful quantum algorithms within the QSVT framework, such as the amplitude amplification, matrix inversion, and Hamiltonian simulation.
The following proposition is a result of selecting an appropriate real odd polynomial $P_\sgn$ that closely approximates the sign function~\cite{Low2017HamiltonianSB, gilyen2019qsvt, Gilyn2019thesis, martyn2021grand}. For future reference, we present it in an algorithmic form.

\begin{proposition}[QSVT with the sign function~{\cite[Theorem 1]{gilyen2019qsvt}}; see also Ref.~{\cite[Theorem 3.2.3]{Gilyn2019thesis}}]
\label{prop:FPAA general}
    Let $\delta \in (0, 1/2)$ and $\beta \in (0, 1]$, and let $U$ be a $(1, a, 0)$-block-encoding unitary of a matrix $M$. Then, there exists a quantum algorithm $\mathtt{QSVTSIGN}(U, U^\dag; \delta, \beta)$ that outputs a unitary $\til{U}$, which is a $(1, a+1, 0)$-block-encoding of a degree-$u$ real odd polynomial $P_\sgn^{(\rm SV)}(M)$ such that $\big|P_\sgn(x) - \sgn(x)\big| \leq \delta$ for $|x| \in [\beta, 1]$ and $|P_\sgn(x)| \leq 1$ for $|x| \in [0, 1]$, where $u$ is the minimum odd integer satisfying $u \geq \big\lceil\f{8e}{\beta}\log{(2/\delta)}\big\rceil$.
    This algorithm uses $U$ and $U^\dag$ a total of $u$ times. The circuit description of $\til{U}$ can be computed classically in $\poly\big(u, a, \log{(1/\delta)}\big)$ time.
\end{proposition}

This proposition is based on the fact that, for given $\delta$ and $\beta$, there exists a degree-$u$ polynomial $P_{\rm sgn}$ that well approximates the sign function.
More precisely, for any odd degree $u \geq \gamma_{\delta, \beta}$, where
\begin{align}
    &\gamma_{\delta, \beta} = 2 \,\bigg\lceil\max\bigg\{
    \f{e}{2\beta} \sqrt{W\Big(\f{32}{\pi\delta^2}\Big) W\Big(\f{2048}{e^2\pi\delta^2}\Big)}, \notag\\
    &\hspace{7pc}\sqrt{2} W\Big(\f{8}{\beta\delta}\sqrt{\f{2}{\pi}W\Big(\f{32}{\pi\delta^2}\Big)}\Big)\bigg\}\bigg\rceil + 1,
\end{align}
there is a polynomial that approximates $\sgn(x)$ within error $\delta$ over $|x| \in [\beta, 1]$~\cite{martyn2023efficient}.
Here, $W(x)$ is the Lambert $W$ function~\cite{Corless1996onthelambert}.
From the fact that $W(x) \leq \log x$ holds for $x > e$~\cite{Hassani2005Approximationlambert, Hoorfar2008inequlambert}, we have $\gamma_{\delta, \beta} \leq \big\lceil\frac{8e}{\beta}\log(2/\delta)\big\rceil$ for $\delta \in (0,1/2)$ and $\beta \in (0,1]$. 
Thus, it is sufficient to set the degree $u$ of the polynomial to the minimum odd integer satisfying $u \geq \big\lceil\frac{8e}{\beta}\log(2/\delta)\big\rceil$, where $u$ is the number of uses of the unitaries $U$ and $U^\dag$.

In \Cref{prop:FPAA general}, the error of $P_\sgn$ is bounded uniformly over $|x| \in [\beta, 1]$, which results in the factor $\log{(1/\delta)}$ in the number of uses of $U$ and $U^\dag$.
The following proposition instead bounds the error multiplied by $|x|$, i.e., $|x|\big|P_\sgn(x) - \sgn(x)\big| \leq \beta$, by a polynomial of degree $\cO(1/\beta)$.
Although the error can be large for small $|x|$, this avoids the factor $\log{(1/\beta)}$ needed to obtain the same guarantee from \Cref{prop:FPAA general}.

\begin{proposition}[QSVT with the Fej\'{e}r mean of the sign function]
\label{prop:QSVT sign Fejer}
    Let $\beta \in (0, 1]$, and let $U$ be a $(1, a, 0)$-block-encoding unitary of a matrix $M$. Then, there exists a quantum algorithm $\mathtt{QSVTSIGNFEJER}(U, U^\dag; \beta)$ that outputs a unitary $\til{U}$, which is a $(1, a+1, 0)$-block-encoding of a degree-$u$ real odd polynomial $P_\sgn^{(\rm SV)}(M)$ such that $|x|\big|P_\sgn(x) - \sgn(x)\big| \leq \beta$ and $|P_\sgn(x)| \leq 1$ for $|x| \in [0, 1]$, where $u$ is the minimum odd integer satisfying $u \geq 2\pi/\beta$.
    This algorithm uses $U$ and $U^\dag$ a total of $u$ times. The circuit description of $\til{U}$ can be computed classically in $\poly(u, a)$ time.
\end{proposition}

The polynomial $P_\sgn$ in \Cref{prop:QSVT sign Fejer} is obtained from the Fej\'{e}r mean, i.e., the arithmetic mean of the partial sums, of the Fourier series of the square wave~\cite{Fejer1903UntersuchungenFourierscheReihen, Katznelson2004IntroductionHarmonicAnalysis}.
Unlike the truncated series, this mean is bounded by one in absolute value since the Fej\'{e}r kernel is non-negative. A similar idea of averaging the partial sums has been used for the QSVT in Ref.~\cite{gall2024Spaceboundedqsvt}, which employs the de la Vall\'{e}e Poussin means of Chebyshev expansions.

\begin{proof}[Proof of \Cref{prop:QSVT sign Fejer}]
We construct a polynomial $P_\sgn$ from the square wave function $f(\theta) = \sgn(\cos\theta)$ for $\theta \in \bR$, whose Fourier series is given by
\begin{align}
    f(\theta) = \f{4}{\pi}\sum_{l=0}^{\infty}\f{(-1)^l}{2l+1}\cos\big((2l+1)\theta\big).
\end{align}
The Fej\'{e}r mean of order $u$ of this series is
\begin{align}
\label{eq:Fejer mean}
    \sigma_u(\theta) 
    &= \f{4}{\pi}\sum_{l=0}^{(u-1)/2}\Big(1 - \f{2l+1}{u+1}\Big)\f{(-1)^l}{2l+1}\cos\big((2l+1)\theta\big) \\
    &= \f{1}{2\pi}\int_{-\pi}^{\pi}K_u(t)f(\theta - t)dt,
\end{align}
where $K_u$ is the Fej\'{e}r kernel given by
\begin{align}
    K_u(t) = \f{1}{u+1}\bigg(\f{\sin\big((u+1)t/2\big)}{\sin(t/2)}\bigg)^2.
\end{align}
Since $\cos(k\theta) = T_k(\cos\theta)$ for the degree-$k$ Chebyshev polynomial $T_k$ of the first kind, which is odd for odd $k$, the polynomial
\begin{align}
\label{eq:Fejer polynomial}
    P_\sgn(x) = \f{4}{\pi}\sum_{l=0}^{(u-1)/2}\Big(1 - \f{2l+1}{u+1}\Big)\f{(-1)^l}{2l+1}T_{2l+1}(x)
\end{align}
satisfies $P_\sgn(\cos\theta) = \sigma_u(\theta)$ and is a real odd polynomial of degree $u$ with a non-zero coefficient of $T_u$.
From $K_u(t) \geq 0$, $\f{1}{2\pi}\int_{-\pi}^{\pi}K_u(t)dt = 1$, and $|f(\theta)| \leq 1$, we have $|P_\sgn(x)| \leq 1$ for $x \in [-1, 1]$.

Let $x \in [-1, 1]$ with $x \neq 0$, and let $\theta = \arccos x \in [0, \pi]$ and $d = |\theta - \pi/2|$.
The discontinuities of $f$ are at $\pi/2 + k\pi$ for $k \in \bZ$, and the nearest one to $\theta \in [0, \pi]$ is $\pi/2$, which is at distance $d$; hence, $f(\theta - t) = f(\theta) = \sgn(x)$ for $|t| < d$.
Using $\sin(t/2) \geq t/\pi$ for $t \in [0, \pi]$, we also have $K_u(t) \leq \pi^2/\big((u+1)t^2\big)$ for $t \in (0, \pi]$.
Thus, we obtain
\begin{align}
    \big|P_\sgn(x) - \sgn(x)\big|
    &= \Big|\f{1}{2\pi}\int_{-\pi}^{\pi}K_u(t)\big(f(\theta - t) - f(\theta)\big)dt\Big| \\
    &\leq \f{2}{\pi}\int_{d}^{\pi}K_u(t)dt \\
    &\leq \f{2\pi}{(u+1)d},
\end{align}
where we used $\f{1}{2\pi}\int_{-\pi}^{\pi}K_u(t)dt = 1$ in the first line, and $|f(\theta - t) - f(\theta)| \leq 2$ and $K_u(-t) = K_u(t)$ in the second line.
Together with $d \geq |\sin(\theta - \pi/2)| = |x|$, this implies that
\begin{align}
    |x|\big|P_\sgn(x) - \sgn(x)\big| \leq \f{2\pi}{u+1} < \beta,
\end{align}
where we used $u \geq 2\pi/\beta$.
This inequality also holds at $x = 0$, where $P_\sgn(0) = \sgn(0) = 0$.

Applying \Cref{prop:general QSVT} to $P_\sgn$, we obtain a $(1, a+1, 0)$-block-encoding $\til{U}$ of $P_\sgn^{(\rm SV)}(M)$, using $U$ and $U^\dag$ a total of $u$ times~\cite{gilyen2019qsvt, Gilyn2019thesis}.
The coefficients of $P_\sgn$ in the basis of the Chebyshev polynomials are explicitly given by Eq.~\eqref{eq:Fejer polynomial}, and the phase factors of the QSVT can be computed classically in $\poly(u)$ time~\cite{Gilyn2019thesis, Haah2019product, Chao2020FindingAF, dong2021findingphase, Lin2022LectureNotes, Mizuta2023vzw}.

\end{proof}

\subsubsection{Density matrix exponentiation}
\label{sec:DME intro}

Given identical copies of a quantum state $\rho$, we can approximately implement a unitary $e^{it\rho}$ using a technique known as the \emph{density matrix exponentiation}, which was introduced in Ref.~\cite{lloyd2014QuantPrincCompAnaly} and developed in subsequent works~\cite{marvian2016universalquantumemulator, kimmel2017HamSimSampleComp, wei2023DMEhermitianpreserv, go2024DMEsamplebased}.

\begin{theorem}[Density matrix exponentiation~{\cite[Theorem 2]{go2024DMEsamplebased}}]
\label{prevthm:cntrl-DME}
Let $\rho$ be a quantum state in a $d$-dimensional Hilbert space.
For $\delta > 0$ and $t \geq \delta/4$, there exists a quantum algorithm that realizes a quantum channel $\cG$ such that
$\f{1}{2}\|\cG - \cU\|_\diamond \leq \delta$,
where $\cU(\cdot) = e^{it\rho}(\cdot)e^{-it\rho}$.
The algorithm uses $m = \lceil4t^2/\delta\rceil$ copies of $\rho$, and $\cO(m\log{d})$ one- and two-qubit gates.
\end{theorem}

The key idea of the density matrix exponentiation is as follows.
If, for a small parameter $\Delta t$, we apply the partial swap unitary $e^{i\Delta t F^{\sA\hat{\sA}}}$, we see that
\begin{align}
    &\tr_{\hat{\sA}}\big[e^{i\Delta t F^{\sA\hat{\sA}}} (\psi^{\sR\sA} \otimes \rho^{\hat{\sA}}) e^{-i\Delta t F^{\sA\hat{\sA}}}\big] \notag\\
    &\hspace{1pc}= \psi^{\sR\sA} + i\Delta t [\rho^\sA, \psi^{\sR\sA}] + \cO\big((\Delta t)^2\big) \\ 
    \label{inteq:39}
    &\hspace{1pc}= e^{i\Delta t\rho^\sA} \psi^{\sR\sA} e^{-i\Delta t \rho^\sA} + \cO\big((\Delta t)^2\big),
\end{align}
where the last line follows from the Hadamard lemma, also known as the Baker-Campbell-Hausdorff formula~\cite{campbell1896BCHlemma, Baker1901BCHlemma, hausdorff1906symbolische, dynkin1947calculation}.
By repeating this procedure $m$ times, we end up implementing the unitary $e^{i\rho^\sA t}$ up to the error $\cO\big(m(\Delta t)^2\big)$, where $m \Delta t = t$. 
Thus, to achieve the approximation of $e^{it\rho^\sA}$ within an error of $\delta$, it is sufficient to set $m = \Omega(t^2/\delta)$, which represents the number of copies of $\rho^\sA$.
A quantum channel that approximates $\cU^\dag(\cdot) = e^{-it\rho}(\cdot)e^{it\rho}$ can also be constructed by simply modifying this procedure to replace $e^{i\Delta t F}$ with $e^{-i\Delta t F}$.

Intuitively, the first-order relation $\tr_{\hat{\sA}}\big[F^{\sA\hat{\sA}}(\psi^{\sR\sA}\otimes\rho^{\hat{\sA}})\big] = \rho^\sA\psi^{\sR\sA}$ implies that by repeatedly preparing $\rho^{\hat{\sA}}$ and applying $e^{i\Delta t F^{\sA\hat{\sA}}}$ and $\tr_{\hat{\sA}}$, the higher-order deviations can be suppressed below $\delta$, which enables us to approximately simulate the unitary $e^{it\rho^\sA}$.
For detailed discussions of higher-order terms, refer to Refs.~\cite{wei2023DMEhermitianpreserv, go2024DMEsamplebased}.

For a general mixed state, a block-encoding of the state can be obtained from its samples by combining the density matrix exponentiation and the QSVT~\cite{gilyen2022improvedfidelity, wang2023SampletoQuary}.
For a pure state $\ket{\psi}$, the density matrix exponentiation alone suffices, since $\ketbra{\psi}{\psi}$ is a projector and thus $e^{i\pi\ketbra{\psi}{\psi}} = \bI - 2\ketbra{\psi}{\psi}$.

\begin{lemma}[Block-encoding of a pure state from samples]
\label{lem:BE pure state}
Let $\delta \in (0, 1)$, and let $\ket{\psi}^\sA$ be a pure state in a $d_\sA$-dimensional Hilbert space.
There exists a quantum algorithm $\mathtt{BlockEncPureState}(\ket{\psi}^\sA; \delta)$ that outputs a quantum channel $\cG^{\sC\sA}$ such that $\f{1}{2}\|\cG^{\sC\sA} - \cW^{\sC\sA}\|_\diamond \leq \delta$, where $W^{\sC\sA}$ is a $(1, 1, 0)$-block-encoding unitary of $\ketbra{\psi}{\psi}^\sA$ given by
\begin{align}
\label{eq:W pure state}
    W^{\sC\sA} = X^\sC H^\sC e^{i\pi\ketbra{1}{1}^\sC\otimes\ketbra{\psi}{\psi}^\sA} H^\sC,
\end{align}
and $X^\sC$ and $H^\sC$ are the one-qubit Pauli-$X$ and Hadamard gates, respectively.
The algorithm uses $m = \lceil4\pi^2/\delta\rceil$ samples of $\ket{\psi}^\sA$ and $\cO(m\log{d_\sA})$ one- and two-qubit gates.
Since $(W^{\sC\sA})^\dag = W^{\sC\sA}$, $\cG^{\sC\sA}$ also satisfies $\f{1}{2}\|\cG^{\sC\sA} - (\cW^{\sC\sA})^\dag\|_\diamond \leq \delta$.
\end{lemma}

\begin{proof}[Proof of \Cref{lem:BE pure state}]
Since $\ketbra{1}{1}^\sC\otimes\ketbra{\psi}{\psi}^\sA$ is a projector, we have
\begin{align}
    &e^{i\pi\ketbra{1}{1}^\sC\otimes\ketbra{\psi}{\psi}^\sA} \notag\\
    &\hspace{1pc}= \ketbra{0}{0}^\sC\otimes\bI^\sA + \ketbra{1}{1}^\sC\otimes\big(\bI^\sA - 2\ketbra{\psi}{\psi}^\sA\big).
\end{align}
Conjugating it by $H^\sC$ yields a unitary whose diagonal and off-diagonal blocks are $\bI^\sA - \ketbra{\psi}{\psi}^\sA$ and $\ketbra{\psi}{\psi}^\sA$, respectively.
Applying $X^\sC$ exchanges these blocks, and we get 
\begin{align}
\label{eq:BE pure state}
    W^{\sC\sA}
    = \hspace{1mm}
\begin{blockarray}{ccc}
& \bra{0}^\sC & \bra{1}^\sC \vspace{1mm}\\
\begin{block}{c(cc)}
  \ket{0}^\sC \hspace*{1mm} & \ketbra{\psi}{\psi}^\sA & \bI^\sA - \ketbra{\psi}{\psi}^\sA \\
  \ket{1}^\sC \hspace*{1mm} & \bI^\sA - \ketbra{\psi}{\psi}^\sA & \ketbra{\psi}{\psi}^\sA \\
\end{block}
\end{blockarray}\hspace{1mm},
\end{align}
which is a Hermitian $(1, 1, 0)$-block-encoding of $\ketbra{\psi}{\psi}^\sA$.

To prepare $e^{i\pi\ketbra{1}{1}^\sC\otimes\ketbra{\psi}{\psi}^\sA}$, we apply \Cref{prevthm:cntrl-DME} with $t = \pi$ to the pure state $\ket{1}^\sC\ket{\psi}^\sA$.
The resulting channel approximates $e^{i\pi\ketbra{1}{1}^\sC\otimes\ketbra{\psi}{\psi}^\sA}$ within error $\delta$ in the diamond norm, using $m = \lceil4\pi^2/\delta\rceil$ samples of $\ket{\psi}^\sA$ and $\cO(m\log{d_\sA})$ one- and two-qubit gates.
Since the diamond norm is invariant under the composition with the unitaries $H^\sC$ and $X^\sC H^\sC$, the channel $\cG^{\sC\sA}$ obtained by this composition satisfies $\f{1}{2}\|\cG^{\sC\sA} - \cW^{\sC\sA}\|_\diamond \leq \delta$.

\end{proof}

The density matrix exponentiation in \Cref{prevthm:cntrl-DME} implements $e^{it\rho}$ only approximately, and $\Omega(t^2/\delta)$ copies of $\rho$ are necessary for any quantum channel to approximate it within error $\delta$~\cite{kimmel2017HamSimSampleComp, go2024DMEsamplebased}.
When only the expectation value of a measurement outcome is needed, this error can be removed for a pure state by the \emph{virtual density matrix exponentiation}~\cite{wada2025statetohamiltonianconversioncopies}, which uses a quantum channel followed by classical post-processing.

\begin{theorem}[Virtual density matrix exponentiation for pure states~{\cite[Theorem~S5]{wada2025statetohamiltonianconversioncopies}}]
\label{thm:virtual DME}
Let $\ket{\psi}^\sA$ be a pure state in a $d_\sA$-dimensional Hilbert space.
For $t \in [-2\pi, 2\pi]\setminus\{0\}$ and an even integer $m \geq 2\max\{1, |t|\}$, there exists a quantum algorithm $\mathtt{VirtualDME}(\ket{\psi}^\sA; t, m)$ that outputs a quantum channel $\cG^{\sK\sA}$ with a one-qubit system $\sK$ such that the map $\hat{\cG}^\sA(\cdot) = \tr_\sK\big[X^\sK\cG^{\sK\sA}(\ketbra{+}{+}^\sK\otimes\cdot)\big]$ satisfies $\alpha\hat{\cG}^\sA = \cU^\sA$. Here, $\cU^\sA(\cdot) = e^{it\ketbra{\psi}{\psi}^\sA}(\cdot)e^{-it\ketbra{\psi}{\psi}^\sA}$, and $\alpha$ is a real number determined by $t$ and $m$ that satisfies $1 < \alpha \leq e^{4t^2/m}$.
The algorithm uses $m$ copies of $\ket{\psi}^\sA$, and $\cO(m\log{d_\sA})$ one- and two-qubit gates.
\end{theorem}

\Cref{thm:virtual DME} allows us to estimate the expectation of the real-valued output of a quantum circuit containing unitaries of the form $e^{it\ketbra{\psi}{\psi}}$ without the error of the density matrix exponentiation in \Cref{prevthm:cntrl-DME}. Indeed, if the output takes values in $[-1, 1]$ and the circuit contains $q$ such unitaries, replacing each of them with $\alpha\hat{\cG}^\sA$ does not change this expectation, since $\alpha\hat{\cG}^\sA = \cU^\sA$. Although $\hat{\cG}^\sA$ is not a quantum channel, it is realized by applying $\cG^{\sK\sA}$ to $\ket{+}^\sK$ and the system on which the unitary acts, measuring $\sK$ in the basis $\{\ket{+}, \ket{-}\}$, and multiplying the output by the outcome $\pm 1$. The resulting estimate, which is the output multiplied by the $q$ outcomes and the $q$ values of $\alpha$, has absolute value at most $e$ if the channel for each unitary uses $m \geq 4t^2q$ copies of the pure state $\ket{\psi}^\sA$.


\subsubsection{Random purification}
\label{sec:random purification intro}

Given identical copies of a mixed state $\omega^\sA$, we can approximately prepare the same number of copies of a purified state of $\omega^\sA$, using a technique known as the \emph{random purification}, which was introduced in Ref.~\cite{tang2025conjugatequerieshelp} as an algorithmic strengthening of the results in Refs.~\cite{Soleimanifar2022testingmatrixproduct, chen2024localtestUniInv}.
Below, $\ket{\omega_0}^{\sA\sB}$ is a fixed purified state of $\omega^\sA$ with purifying system $\sB$, and $\bE_U[\cdot]$ denotes the expectation over a Haar random unitary $U^\sB$ on $\sB$.

\begin{theorem}[Random purification~{\cite[Lemma~2.11]{tang2025conjugatequerieshelp}}]
\label{prevthm:random purification}
Let $\delta \in (0, 1)$ and $n \geq 1$, and let $\cQ_{\rm ideal}$ be the quantum channel satisfying
\begin{align}
\label{eq:random purification ideal}
    \cQ_{\rm ideal}\big((\omega^\sA)^{\otimes n}\big) = \bE_U\big[\big(U^\sB\ketbra{\omega_0}{\omega_0}^{\sA\sB}(U^\sB)^\dag\big)^{\otimes n}\big],
\end{align}
for every state $\omega^\sA$ whose rank is at most $d_\sB$.
Then, there exists a quantum algorithm $\mathtt{RandomPurification}(\omega^\sA; \delta, n)$ that outputs a quantum channel $\cQ$ such that $\f{1}{2}\|\cQ - \cQ_{\rm ideal}\|_\diamond \leq \delta$, using $n$ samples of $\omega^\sA$.
The quantum circuit of this algorithm consists of $\til{\cO}(n^4)$ one- and two-qubit gates, and uses $\til{\cO}(n^2)$ qubits at any one time. The circuit description can be computed classically in $\poly\big(n, \log{d_\sA}, \log{(1/\delta)}\big)$ time.
\end{theorem}

Every purified state of $\omega^\sA$ on $\sA\sB$ is written as $U^\sB\ket{\omega_0}^{\sA\sB}$ with a unitary $U^\sB$ on the purifying system. The term \emph{random} refers to drawing this unitary according to the Haar measure, under which the right-hand side of Eq.~\eqref{eq:random purification ideal} is a state uniformly averaged over the choice of purification and does not depend on $\ket{\omega_0}^{\sA\sB}$.
Remarkably, although the state is averaged, the algorithm prepares $n$ purified states from $n$ samples of $\omega^\sA$ without reducing the number of samples, and, importantly, the outputs share the same unitary $U^\sB$. The primary component of this subroutine is the Schur transform~\cite{harrow2005applicationSchurtrans, Bacon2006efficientcircuitSchur, Krovi2019efficienthigh, Nguyen2024theorynuralnet, burchardt2025highdimensionalschur}, which determines the number of gates and qubits in \Cref{prevthm:random purification}~\cite{burchardt2025highdimensionalschur} and, unlike the algorithms in Secs.~\ref{sec:BE and QSVT intro} and~\ref{sec:DME intro}, requires the $n$ samples to be handled collectively.

The purifying system $\sB$ can be taken to be of dimension $d_\sB = d_\sA$, for which the condition on the rank in \Cref{prevthm:random purification} is satisfied by every state $\omega^\sA$.
When an upper bound on the rank of $\omega^\sA$ is known in advance, a smaller purifying system suffices.
We also remark that no quantum channel purifies copies of multiple different states with a common unitary on the purifying systems~\cite{lowe2026randomdimensionreductionlearning}.
The Uhlmann transformation involves two different states, and hence \Cref{prevthm:random purification} may seem of limited use for it.
Still, it is useful when the task does not depend on the choice of purification or when the two states are constructed from copies of a single state, as we will see in Secs.~\ref{sec:fidelity estimation collective}, \ref{sec:decoupling and uhlmann}, and~\ref{sec:sample petz}.

We consider a quantum query algorithm that makes queries to a purified state preparation unitary of $\omega^\sA$. Using \Cref{prevthm:random purification}, such an algorithm can be simulated by a quantum sample algorithm that consumes samples of $\omega^\sA$.

\begin{theorem}[Simulating queries to state preparation unitary given copies of the state~{\cite[Theorem 1.5]{tang2025conjugatequerieshelp}}]
\label{thm:sample to query lift}
Let $\delta, \eta \in (0, 1)$ and let $U_\omega^{\sA\sB}$ be a unitary such that $U_\omega^{\sA\sB}\ket{0}^{\sA\sB} = \ket{\omega}^{\sA\sB}$. Suppose that there is a quantum query algorithm that makes $q$ queries to $U_\omega^{\sA\sB}$ and $(U_\omega^{\sA\sB})^\dag$, and produces an output with probability $p^{\rm qry}$.
If the query algorithm satisfies $p^{\rm qry} \geq 1 - \eta$ for any choice of the purifying system $\sB$ and any $U_\omega^{\sA\sB}$, then there exists a quantum sample algorithm that simulates the output of the query algorithm with probability $p^{\rm samp} \geq 1 - \eta - \delta$, using $n = \cO(q^2 / \delta)$ samples of the state $\omega^\sA$.
The quantum circuit of the sample algorithm consists of $\til{\cO}(n^4)$ one- and two-qubit gates, and $\til{\cO}(n^2)$ qubits at any one time, in addition to those of the query algorithm.
\end{theorem}

At the heart of this theorem are the following three primitives.
The first primitive is $\mathtt{RandomPurification}$ in \Cref{prevthm:random purification}, which maps $(\omega^\sA)^{\otimes n}$ to a state uniformly averaged over the choice of purification.
This primitive determines the number of additional gates and qubits in \Cref{thm:sample to query lift}.

The second primitive, which appeared earlier in Refs.~\cite{kretschmer2021pseudorandomnessclassical, Goldin2025translatingcommonhaar}, relates to a state with a phase: $\ket{\omega_\theta}^{\sA\sB\sK} = \f{1}{\sqrt{2}}(e^{i\theta}\ket{0}^{\sA\sB}\ket{0}^{\sK} + \ket{\omega}^{\sA\sB}\ket{1}^{\sK})$, where $\sK$ is a one-qubit auxiliary system. 
This primitive maps $(\ketbra{\omega}{\omega}^{\sA\sB})^{\otimes n}$ to $\bE_\theta[(\ketbra{\omega_\theta}{\omega_\theta}^{\sA\sB\sK})^{\otimes n}]$, where the expectation $\bE_\theta[\cdot]$ is taken over $\theta$ chosen uniformly at random from $[0, 2\pi)$. 
The quantum circuit implementing this primitive consists of $\cO(n^2\log{(d_\sA d_\sB)})$ one- and two-qubit gates and uses $\cO(n\log{(d_\sA d_\sB)})$ qubits at any one time \cite{tang2025conjugatequerieshelp}.

The last primitive is the density matrix exponentiation in \Cref{prevthm:cntrl-DME} with $t=\pi$, which maps $(\ketbra{\omega_\theta}{\omega_\theta}^{\sA\sB\sK})^{\otimes m}$ to $e^{i\pi \ketbra{\omega_\theta}{\omega_\theta}^{\sA\sB\sK}}$ within error $\delta'$ for each use, where $m = \cO(1/\delta')$. 
The unitary $e^{i\pi \ketbra{\omega_\theta}{\omega_\theta}^{\sA\sB\sK}}$ is a purified state preparation unitary because $e^{i\pi \ketbra{\omega_\theta}{\omega_\theta}^{\sA\sB\sK}} \ket{0}^{\sA\sB\sK} = -e^{-i\theta} \ket{\omega}^{\sA\sB} \ket{1}^\sK$ and the state $-e^{-i\theta} \ket{\omega}^{\sA\sB} \ket{1}^\sK$ is one of the purified states of $\omega^\sA$ with purifying system $\sB\sK$ of dimension at least twice the rank of $\omega^\sA$.
To apply the unitary for $q$ times with total error at most $\delta/2$, we set $\delta' = \delta/(2q)$ and use $n = qm = \cO(q^2/\delta)$ copies of $\ket{\omega_\theta}^{\sA\sB\sK}$, while $\mathtt{RandomPurification}$ is run with error $\delta/2$.
Naturally, the inverse of the purified state preparation unitary is obtained in a similar manner.

Now, consider a query algorithm that outputs a certain result with probability $p^{\rm qry}$, using $q$ queries to a unitary preparing a purified state of $\omega^\sA$.
We suppose that this query algorithm is independent of the choice of purification, i.e., it satisfies $p^{\rm qry} \geq 1 - \eta$ for any given purified state preparation unitary.
Without loss of generality, we then consider a sufficiently large purifying system, at least twice the rank of $\omega^\sA$.
By replacing all $q$ queries in the query algorithm with the three primitives above, we obtain a sample algorithm.
Using $n = \cO(q^2/\delta)$ copies of the state $\omega^\sA$, the sample algorithm can simulate the original query algorithm with probability $p^{\rm samp}$, such that $\big|p^{\rm samp} - \bE_\theta\bE_U [p^{\rm qry}]\big| \leq \delta$.
Since $p^{\rm qry} \geq 1 - \eta$ for any $U^\sB$ and $\theta$, it holds that $p^{\rm samp} \geq \bE_\theta \bE_U [p^{\rm qry}] - \delta \ge 1 - \eta - \delta$.
Hence, \Cref{thm:sample to query lift} follows.

\section{Uhlmann transformation algorithm in the purified query access model}
\label{sec:purif query}

We consider the Uhlmann transformation algorithm in the purified query access model.
The formal statement of \Cref{infthm:uhlmenn purif query} is as follows, where $\sD$ is a $\big(\log{(d_\sA d_\sB)} + 1\big)$-qubit system. 
A quantum algorithm $\mathtt{UhlmannPurifiedQuery}$ is provided in \Cref{alg:Uhl purif query}, in which the symbol ``$\gets$'' denotes assignment.
Recall that $s_{\rm min}$ and $R$ are the minimum non-zero singular value and the rank of $\sqrt{\sigma^\sA}\sqrt{\rho^\sA}$, respectively, and $U_\rho^{\hat{\sA}\hat{\sB}}$ and $U_\sigma^{\hat{\sA}\hat{\sB}}$ are unitaries that prepare states $\ket{\rho}^{\hat{\sA}\hat{\sB}}$ and $\ket{\sigma}^{\hat{\sA}\hat{\sB}}$, respectively.

\begin{theorem}[Uhlmann transformation algorithm in the purified query access model]
\label{thm:Uhlmann alg purif query model}
Let $\delta \in (0, 1)$ and $\chi \in \{\diamond, \tF\}$. Then, the quantum query algorithm $\mathtt{UhlmannPurifiedQuery}$ given by \Cref{alg:Uhl purif query} satisfies the following. 

A unitary $\til{W}_\diamond^{\sB\sD}$ given by 
\begin{align}
    \til{W}_\diamond^{\sB\sD} = \mathtt{UhlmannPurifiedQuery}(U_\rho^{\hat{\sA}\hat{\sB}}, U_\sigma^{\hat{\sA}\hat{\sB}}; \delta, \diamond), 
\end{align}
satisfies that
\begin{equation}
\label{eq:error purif access}
    \f{1}{2}\big\|\til{\cW}_\diamond^{\sB\sD} \circ\cP_{\ket{0}}^{\bC\rarr\sD} - \cU_{\rm ideal}^{\sB\sD}\circ\cP_{\ket{0}}^{\bC\rarr\sD}\big\|_\diamond
    \leq \delta,
\end{equation}
where $U_{\rm ideal}^{\sB\sD}$ is an exact block-encoding unitary of the Uhlmann partial isometry $V^\sB$, and $\cP_{\ket{0}}^{\bC\rarr\sD}$ is the channel preparing the state $\ketbra{0}{0}^\sD$ on system $\sD$.
The algorithm uses $u_\diamond = \cO\big(\f{1}{s_{\rm min}}\log{\big(\f{1}{\delta}\big)}\big)$ queries to $U_\rho^{\hat{\sA}\hat{\sB}}$, $U_\sigma^{\hat{\sA}\hat{\sB}}$, and their inverses, where $s_{\rm min}$ is the minimum non-zero singular value of $\sqrt{\sigma^{\sA}}\sqrt{\rho^{\sA}}$.

Moreover, for a unitary $\til{W}_{\tF}^{\sB\sD}$ given by 
\begin{align}
    \til{W}_{\tF}^{\sB\sD} = \mathtt{UhlmannPurifiedQuery}(U_\rho^{\hat{\sA}\hat{\sB}}, U_\sigma^{\hat{\sA}\hat{\sB}}; \delta, \tF),
\end{align}
let $\cT^\sB = \tr_{\sD}\circ\til{\cW}_{\tF}^{\sB\sD}\circ\cP_{\ket{0}}^{\bC\rarr\sD}$. This channel satisfies that
\begin{align}
        \rF\big(\cT^\sB(\ketbra{\rho}{\rho}^{\sA\sB}), \ket{\sigma}^{\sA\sB}\big) \geq \rF(\rho^\sA, \sigma^\sA) - \delta.
\end{align}
The algorithm uses $u_{\tF} = \cO\Big(\min\big\{\f{1}{s_{\rm min}}\log{\big(\f{1}{\delta}\big)}, \f{R}{\delta}\big\}\Big)$ queries to $U_\rho^{\hat{\sA}\hat{\sB}}$, $U_\sigma^{\hat{\sA}\hat{\sB}}$, and their inverses, where $R$ is the rank of $\sqrt{\sigma^\sA}\sqrt{\rho^\sA}$.

In both cases, the quantum circuit for implementing the algorithm consists of $\cO\big(u_\chi \log{(d_\sA d_\sB)}\big)$ one- and two-qubit gates, and uses $\cO\big(\log{(d_\sA d_\sB)}\big)$ qubits simultaneously. The circuit description can be computed classically in $\poly\big(u_\chi, \log{(d_\sA d_\sB)}, \log{(1/\delta)}\big)$ time.
\end{theorem}


\begin{algorithm}[ht]
\caption{Uhlmann transformation algorithm in the purified query access model \\ \parbox{\linewidth}{\centering $\mathtt{UhlmannPurifiedQuery}(U_\rho^{\hat{\sA}\hat{\sB}}, U_\sigma^{\hat{\sA}\hat{\sB}}; \delta, \chi)$ \\ (In Theorem~\ref{thm:Uhlmann alg purif query model})}}
\label{alg:Uhl purif query}
\SetKwInput{KwInput}{Input}
\SetKwInput{KwOutput}{Output}
\SetKwInput{KwParameters}{Parameters}

\SetAlgoNoEnd
\SetAlgoNoLine
\KwInput{Unitary oracles $U_\rho^{\hat{\sA}\hat{\sB}}$, $U_\sigma^{\hat{\sA}\hat{\sB}}$, and their inverses.}
\KwParameters{$\delta \in (0, 1)$ and $\chi \in \{\diamond, \tF\}$.}
\KwOutput{Unitary $\til{W}^{\sB\sD}$.}
\SetAlgoLined

Set $W^{\sB\hat{\sA}\hat{\sB}} \gets (U_\rho^{\hat{\sA}\hat{\sB}})^\dag(\bI^{\hat{\sA}} \otimes F^{\sB\hat{\sB}})U_\sigma^{\hat{\sA}\hat{\sB}}$, where $F^{\sB\hat{\sB}}$ is the swap operator between $\sB$ and $\hat{\sB}$. \\
\If{$\chi = \tF$ and $\f{4\pi R}{\delta} < \f{8e}{s_{\rm min}}\log{(18/\delta^2)}$}{
  Set $\delta_1 \gets \delta/2$. \\
  Set $\til{W}^{\sB\sD} \gets \mathtt{QSVTSIGNFEJER}(W^{\sB\hat{\sA}\hat{\sB}}, (W^{\sB\hat{\sA}\hat{\sB}})^\dag; \delta_1/R)$ (Proposition~\ref{prop:QSVT sign Fejer}). \\
}
\Else{
  Set $\delta_1 \gets (\delta/3)^2$. \\
  Set $\til{W}^{\sB\sD} \gets \mathtt{QSVTSIGN}(W^{\sB\hat{\sA}\hat{\sB}}, (W^{\sB\hat{\sA}\hat{\sB}})^\dag; \delta_1, s_{\rm min})$ (Proposition~\ref{prop:FPAA general}). \\
}
Return $\til{W}^{\sB\sD}$.
\end{algorithm}


As we show in the following subsection, the unitary
$W^{\sB\hat{\sA}\hat{\sB}}$ in Algorithm~\ref{alg:Uhl purif query} is a block-encoding of the operator
$M^{\sB} = \tr_{\sA'}\big[\ketbra{\sigma}{\rho}^{\sA'\sB}\big]$.
Applying $\mathtt{QSVTSIGN}$, the QSVT with the sign function described in \Cref{prop:FPAA general}, with input $W^{\sB\hat{\sA}\hat{\sB}}$ and $(W^{\sB\hat{\sA}\hat{\sB}})^\dag$ yields a unitary $\til{W}^{\sB\sD}$ that is a block-encoding of $P_\sgn^{(\rm SV)}(M^\sB)$. In the case of $\chi = \tF$, we use $\mathtt{QSVTSIGNFEJER}$ in \Cref{prop:QSVT sign Fejer} instead of $\mathtt{QSVTSIGN}$ if it uses fewer queries. If the polynomial $P_\sgn$ sufficiently approximates the sign function, then by \Cref{prop:explicit form of Uhl}, $\til{W}^{\sB\sD}$ is close to a block-encoding unitary of the Uhlmann partial isometry $V^\sB = \sgn^{(\rm SV)}(M^\sB)$.

This result connects quantum algorithms for polar decomposition with Uhlmann's theorem, which makes these algorithms applicable to various tasks involving the Uhlmann transformation. The underlying observation is that the Uhlmann partial isometry $V^\sB$ is the polar partial isometry of $M^\sB$. A QSVT-based algorithm that implements the polar partial isometry using queries to unitaries was provided in Ref.~\cite{Quek2022Fastalgorithmquantumpolardecomposition}, improving on the algorithm in Ref.~\cite{lloyd2020polardecomposition}. As an application, Ref.~\cite{Quek2022Fastalgorithmquantumpolardecomposition} considered the quantum Procrustes problem, which can be regarded as the Uhlmann transformation for particular purified states. Our algorithm follows the approach in Ref.~\cite{Quek2022Fastalgorithmquantumpolardecomposition}, and we analyze its approximation error in detail both in the diamond norm and in the fidelity difference, the latter yielding the optimal number of queries up to logarithmic factors (see \Cref{sec:overview Optimality}). This connection allows the Uhlmann transformation to be implemented for low-rank states at a cost independent of the dimension, substantially improving on the exponential cost of the previous Uhlmann transformation algorithm~\cite{metger2023stateqipspace}. Our applications to fidelity estimation, the decoupling approach, and the Petz recovery map also build on this connection.

We emphasize the importance of clarifying which measure is used to evaluate the degree of approximation to the Uhlmann transformation.
In Theorem~\ref{thm:Uhlmann alg purif query model}, the first statement evaluates the difference between the quantum \emph{channels} implemented by our algorithm and the ideal Uhlmann transformation using the diamond norm, accommodating the setting where input states to the Uhlmann transformation can be arbitrary.
On the other hand, the second statement evaluates the accuracy in terms of the output \emph{state} when the input to the Uhlmann transformation is $\ket{\rho}^{\sA\sB}$---where we compare the fidelity between the actual output and $\ket{\sigma}^{\sA\sB}$ with its optimal value $\rF(\rho^\sA, \sigma^\sA)$.
While evaluating the approximation error via fidelity difference may suffice, it is indeed more convenient to use the diamond norm in cases where the Uhlmann transformation algorithm is used as a subroutine (e.g., in \Cref{sec:sample petz}).

One might think that post-selection on $\ket{0}^\sD$ is required to extract the Uhlmann partial isometry $V^\sB$ encoded in the top-left block of $U_{\rm ideal}^{\sB\sD}$.
However, this is not necessary for our purpose: since $\bra{0}^\sD U_{\rm ideal}^{\sB\sD}\ket{0}^\sD = V^\sB$, the monotonicity of the fidelity under the partial trace over $\sD$ yields $\rF\big( \tr_\sD \circ\, \cU_{\rm ideal}^{\sB\sD} \circ \cP_{\ket{0}}^{\bC\rarr\sD}(\ketbra{\rho}{\rho}^{\sA\sB}),\, \ket{\sigma}^{\sA\sB} \big) \geq \big| \bra{\sigma}^{\sA\sB} V^\sB \ket{\rho}^{\sA\sB} \big|^2 = \rF(\rho^\sA, \sigma^\sA)$.
We note that $U_{\rm ideal}^{\sB\sD}\ket{\rho}^{\sA\sB}\ket{0}^\sD = V^\sB\ket{\rho}^{\sA\sB}\ket{0}^\sD$ holds if and only if $\supp[\rho^\sA] \cap \ker[\sigma^\sA] = \{0\}$, in which case the domain of $V^\sB$ covers $\supp[\rho^\sB]$.

If we are able to implement $U_\rho^{\hat{\sA}\hat{\sB}}$, $U_\sigma^{\hat{\sA}\hat{\sB}}$, and their inverses by ourselves, the circuit complexity, i.e., the total number of one- and two-qubit gates, for the Uhlmann transformation can be quantified. 
This applies to both cases where the error is evaluated using the diamond norm and where it is evaluated using the fidelity difference.
Since $W^{\sB\hat{\sA}\hat{\sB}}$ is given by $(U_\rho^{\hat{\sA}\hat{\sB}})^\dag(\bI^{\hat{\sA}} \otimes F^{\sB\hat{\sB}})U_\sigma^{\hat{\sA}\hat{\sB}}$ (see line~1 of \Cref{alg:Uhl purif query}), its implementation requires the unitaries $(U_\rho^{\hat{\sA}\hat{\sB}})^\dag$ and $U_\sigma^{\hat{\sA}\hat{\sB}}$, along with the swap operation $F^{\sB\hat{\sB}}$, which can be realized using $\log d_\sB$ two-qubit swap gates.
Furthermore, the number of one- and two-qubit gates used in $\mathtt{QSVTSIGN}$ and $\mathtt{QSVTSIGNFEJER}$ in \Cref{alg:Uhl purif query} is $\cO\big(u_\chi\log{(d_\sA d_\sB)}\big)$, where $\chi \in \{\diamond, \tF\}$, because $W^{\sB\hat{\sA}\hat{\sB}}$ is a $(1, \log{(d_\sA d_\sB)}, 0)$-block-encoding unitary (see also Theorem~\ref{prop:general QSVT}). 
Therefore, the total number of one- and two-qubit gates for the overall algorithm is 
\begin{align}
    &\cO\Big(u_\chi \big(\cC(U_\rho) + \cC(U_\sigma) + \log{d_\sB}\big) + u_\chi\log{(d_\sA d_\sB)} \Big) \notag\\ 
    &= \cO\big(u_\chi\big(\cC(U_\rho) + \cC(U_\sigma) + \log{(d_\sA d_\sB)}\big)\big),
\end{align}
where $\cC(U_\rho)$ and $\cC(U_\sigma)$ are the number of one- and two-qubit gates required to implement $U_\rho^{\hat{\sA}\hat{\sB}}$ and $U_\sigma^{\hat{\sA}\hat{\sB}}$, respectively.

We first prove Theorem~\ref{thm:Uhlmann alg purif query model} in \Cref{sec:proof purif query Uhl}, and then provide a proof of a technical lemma used to show the theorem in \Cref{sec:proof of lem FPAA diamond}.


\subsection{Proof of Theorem~\ref{thm:Uhlmann alg purif query model}}
\label{sec:proof purif query Uhl}

The unitary $W^{\sB\hat{\sA}\hat{\sB}}$ in Algorithm~\ref{alg:Uhl purif query} is a block-encoding of $\tr_{\sA'}\big[\ketbra{\sigma}{\rho}^{\sA'\sB}\big]$, which was shown in Ref.~\cite[Lemma~1]{Quek2022Fastalgorithmquantumpolardecomposition}. Intuitively, the transformation cannot uncompute $\ket{\rho}^{\sA\sB}$ by applying $(U_\rho^{\sA\sB})^\dag$, since it acts only on $\sB$. Instead, the swap operator $F^{\sB\hat{\sB}}$ moves the input on $\sB$ to $\hat{\sB}$, where the purifying system of $\ket{\rho}^{\hat{\sA}\hat{\sB}}$ is located, and moves the purifying system of $\ket{\sigma}^{\hat{\sA}\hat{\sB}}$ to $\sB$. Then, $\bra{0}^{\hat{\sA}\hat{\sB}}(U_\rho^{\hat{\sA}\hat{\sB}})^\dag = \bra{\rho}^{\hat{\sA}\hat{\sB}}$ takes the overlap of the resulting state on $\hat{\sA}\hat{\sB}$ with $\ket{\rho}^{\hat{\sA}\hat{\sB}}$, where the inner product on $\hat{\sA}$ realizes the partial trace over $\sA'$. We verify this block-encoding below for completeness.

Let $U_\rho^{\hat{\sA}\hat{\sB}}$ and $U_\sigma^{\hat{\sA}\hat{\sB}}$ be unitaries such that $U_\rho^{\hat{\sA}\hat{\sB}}\ket{0}^{\hat{\sA}\hat{\sB}} = \ket{\rho}^{\hat{\sA}\hat{\sB}}$, and $U_\sigma^{\hat{\sA}\hat{\sB}}\ket{0}^{\hat{\sA}\hat{\sB}} = \ket{\sigma}^{\hat{\sA}\hat{\sB}}$.
We then define a unitary $W^{\sB\hat{\sA}\hat{\sB}}$ by $W^{\sB\hat{\sA}\hat{\sB}} = (U_\rho^{\hat{\sA}\hat{\sB}})^\dag(\bI^{\hat{\sA}} \otimes F^{\sB\hat{\sB}})U_\sigma^{\hat{\sA}\hat{\sB}}$.
The unitary $W^{\sB\hat{\sA}\hat{\sB}}$ is a $(1, \log{(d_\sA d_\sB)}, 0)$-block-encoding of $\tr_{\sA'}\big[\ketbra{\sigma}{\rho}^{\sA'\sB}\big]$:
\begin{align}
\label{eq:block W fidelity}
        W^{\sB\hat{\sA}\hat{\sB}}
        = \hspace{1mm}
\begin{blockarray}{ccc}
& \bra{0}^{\hat{\sA}\hat{\sB}}&  & \vspace{1mm}\\
\begin{block}{c(cc)}
  \ket{0}^{\hat{\sA}\hat{\sB}} \hspace*{1mm} & \tr_{\sA'}\big[\ketbra{\sigma}{\rho}^{\sA'\sB}\big] & \hspace{0mm} * \hspace{3mm}\\
   \hspace*{1mm} & * & \hspace{0mm} * \hspace{3mm}\\
\end{block}
\end{blockarray}\hspace{1mm}.
\end{align}
To see this, we use the Schmidt decomposition of $\ket{\rho}^{\sA\sB}$ and $\ket{\sigma}^{\sA\sB}$:
\begin{align}
      \ket{\rho}^{\sA\sB} = \sum_{k=1}^{r_\rho}\sqrt{p_k}\ket{e_k}^\sA\ket{f_k}^\sB, \ \text{and} \ \ 
      \ket{\sigma}^{\sA\sB} = \sum_{l=1}^{r_\sigma}\sqrt{q_l}\ket{g_l}^\sA\ket{h_l}^\sB.
\end{align}
We then compute the elements of $W^{\sB\hat{\sA}\hat{\sB}}$ as
\begin{align}
    &\bra{i}^\sB\bra{0}^{\hat{\sA}\hat{\sB}}W^{\sB\hat{\sA}\hat{\sB}}\ket{0}^{\hat{\sA}\hat{\sB}}\ket{j}^{\sB} \notag\\
    &=  \bra{i}^\sB\bra{\rho}^{\hat{\sA}\hat{\sB}}(\bI^{\hat{\sA}} \otimes F^{\sB\hat{\sB}})\ket{\sigma}^{\hat{\sA}\hat{\sB}}\ket{j}^{\sB} \\
    &= \sum_{k, l} \sqrt{p_k}\sqrt{q_l}\braket{e_k}{g_l}^{\hat{\sA}} \bra{i}^\sB\bra{f_k}^{\hat{\sB}} F^{\sB\hat{\sB}} \ket{h_l}^{\hat{\sB}}\ket{j}^\sB \\
    &= \sum_{k, l} \sqrt{p_k}\sqrt{q_l}\braket{e_k}{g_l}\braket{i}{h_l}\braket{f_k}{j} \\
    &= \bra{i}^\sB \Big(\sum_{k, l} \sqrt{p_k}\sqrt{q_l}\braket{e_k}{g_l}\ketbra{h_l}{f_k}^\sB \Big) \ket{j}^\sB \\
    \label{inteq:171}
    &= \bra{i}^\sB \big(\tr_{\sA'}\big[\ketbra{\sigma}{\rho}^{\sA'\sB}\big]\big) \ket{j}^\sB.
\end{align}
Thus, $\bra{0}^{\hat{\sA}\hat{\sB}}W^{\sB\hat{\sA}\hat{\sB}}\ket{0}^{\hat{\sA}\hat{\sB}} = \tr_{\sA'}\big[\ketbra{\sigma}{\rho}^{\sA'\sB}\big]$.

Due to \Cref{prop:explicit form of Uhl}, it suffices to apply the sign function on the top-left block of $W^{\sB\hat{\sA}\hat{\sB}}$.
We implement $\mathtt{QSVTSIGN}$ in \Cref{prop:FPAA general} or $\mathtt{QSVTSIGNFEJER}$ in \Cref{prop:QSVT sign Fejer} to approximate the sign function by a polynomial $P_\sgn$.
We denote by $\til{W}^{\sB\sD}$ a unitary which is a block-encoding of $P_\sgn^{(\rm SV)}\big(\tr_{\sA'}\big[\ketbra{\sigma}{\rho}^{\sA'\sB}\big]\big)$:
\begin{align}
\label{eq:block tildeW fidelity}
    \til{W}^{\sB\sD}
    = \hspace{1mm}
\begin{blockarray}{ccc}
& \bra{0}^{\sD}&  & \vspace{1mm}\\
\begin{block}{c(cc)}
  \ket{0}^{\sD} \hspace*{1mm} & P_\sgn^{(\rm SV)}\big(\tr_{\sA'}\big[\ketbra{\sigma}{\rho}^{\sA'\sB}\big]\big) & \hspace{0mm} * \hspace{3mm}\\
   \hspace*{1mm} & * & \hspace{0mm} * \hspace{3mm}\\
\end{block}
\end{blockarray}\hspace{1mm}, 
\end{align}
where $\sD$ is a system including $\log{(d_\sA d_\sB)} + 1$ qubits.
We denote by $u$ the number of queries to $W^{\sB\hat{\sA}\hat{\sB}}$ for implementing $\til{W}^{\sB\sD}$. The quantum circuit consists of $\cO(u \log(d_\sA d_\sB))$ one- and two-qubit gates, and its circuit description can be computed classically in $\poly\big(u, \log{(d_\sA d_\sB)}, \log{(1/\delta)}\big)$ time. 
Since the algorithm is conducted sequentially, $\cO(\log(d_\sA d_\sB))$ qubits are needed simultaneously.

The discussion so far applies to both cases, in which the performance of our algorithm is measured in terms of the diamond norm and the fidelity difference.
In the following, we discuss each case separately.
In Secs.~\ref{sec:eval diamond in purif query} and~\ref{sec:uhlfidquery}, we evaluate the approximation error using the diamond norm and the fidelity difference, respectively.
These results demonstrate the complete statement of Theorem~\ref{thm:Uhlmann alg purif query model}.


\subsubsection{Evaluation in the diamond norm}
\label{sec:eval diamond in purif query}

To evaluate the error in the diamond norm, we use the following lemma.
For its application in the following sections, we state it in a slightly more general form than the form used in this section. We provide a proof of this lemma in \Cref{sec:proof of lem FPAA diamond}.

\begin{lemma}
\label{lem:FPAA diamond}
    Let $\delta \in (0, 1)$, $M^\sB$ be a matrix, and $\ket{\varphi}^\sE$ and $\ket{\phi}^\sE$ be two pure states. 
    Suppose that $\til{U}^{\sB\sD\sE}$ is a block-encoding unitary of $\ketbra{\varphi}{\phi}^\sE\otimes P_\sgn^{(\rm SV)}(M^\sB)$ that satisfies
\begin{align}
\label{inteq:24}
    \big\|P_\sgn^{(\rm SV)}(M^\sB) - \sgn^{(\rm SV)}(M^\sB)\big\|_\infty \leq \delta.
\end{align}
Then, there exists a unitary $\til{U}_{\rm ideal}^{\sB\sD\sE}$ which is an exact block-encoding of $\ketbra{\varphi}{\phi}^\sE \otimes \sgn^{(\rm SV)}(M^\sB)$ and satisfies that 
\begin{align}
\label{eq:distance of L and L_ideal}
    \f{1}{2}\big\|\til{\cU}^{\sB\sD\sE} \circ \cP_{\ket{0}\ket{\phi}}^{\bC\rarr\sD\sE} - \til{\cU}_{\rm ideal}^{\sB\sD\sE} \circ \cP_{\ket{0}\ket{\phi}}^{\bC\rarr\sD\sE}\big\|_\diamond \leq 3\sqrt{\delta},
\end{align}
where $\cP_{\ket{0}\ket{\phi}}^{\bC\rarr\sD\sE}$ is a state-preparation channel such that $\cP_{\ket{0}\ket{\phi}}^{\bC\rarr\sD\sE} = \ketbra{0}{0}^{\sD}\otimes\ketbra{\phi}{\phi}^\sE$.
\end{lemma}

When we use $\mathtt{QSVTSIGN}$ in \Cref{prop:FPAA general}, we set $\beta$ to the minimum non-zero singular value of $\tr_{\sA'}\big[\ketbra{\sigma}{\rho}^{\sA'\sB}\big]$, denoted by $s_{\rm min}$.
Then, we obtain a block-encoding unitary 
\begin{align}
    \til{W}^{\sB\sD} = \mathtt{QSVTSIGN}(W^{\sB\hat{\sA}\hat{\sB}}, (W^{\sB\hat{\sA}\hat{\sB}})^\dag; \delta_1, \beta_\diamond)    
\end{align}
of a polynomial $P_\sgn^{(\rm SV)}\big(\tr_{\sA'}\big[\ketbra{\sigma}{\rho}^{\sA'\sB}\big]\big)$, which satisfies  
\begin{align}
\label{inteq:65}
    &\big\|P_\sgn^{(\rm SV)}\big(\tr_{\sA'}\big[\ketbra{\sigma}{\rho}^{\sA'\sB}\big]\big) - \sgn^{(\rm SV)}\big(\tr_{\sA'}\big[\ketbra{\sigma}{\rho}^{\sA'\sB}\big]\big)\big\|_\infty \notag\\
    &\leq \delta_1,
\end{align}
where $\beta_\diamond = s_{\rm min}$ and $\delta_1 \in (0, 1/2)$.
The number of iterations of the unitary $W^{\sB\hat{\sA}\hat{\sB}}$ is given by an integer $u_\diamond$ satisfying $u_\diamond \geq \big\lceil\f{8e}{\beta_\diamond}\log{(2/\delta_1)}\big\rceil = \big\lceil\f{8e}{s_{\rm min}}\log{(2/\delta_1)}\big\rceil$.
By direct calculation, one can check that $s_{\rm min}$ corresponds to the minimum non-zero singular value of $\sqrt{\sigma^\sA}\sqrt{\rho^\sA}$ (see also \Cref{sec:analyze partial iso}).

From Eq.~\eqref{inteq:65} and \Cref{lem:FPAA diamond} with $d_\sE = 1$, we have that
\begin{align}
    \f{1}{2}\big\|\til{\cW}^{\sB\sD}\circ\cP_{\ket{0}}^{\bC\rarr\sD} - \cU_{\rm ideal}^{\sB\sD}\circ \cP_{\ket{0}}^{\bC\rarr\sD}\big\|_\diamond \leq 3\sqrt{\delta_1},
\end{align}
where $U_{\rm ideal}^{\sB\sD}$ is an exact block-encoding of $\sgn^{(\rm SV)}\big(\tr_{\sA'}\big[\ketbra{\sigma}{\rho}^{\sA'\sB}\big]\big)$, which is nothing but the Uhlmann partial isometry $V^\sB$ (see \Cref{prop:explicit form of Uhl}).
As the number of queries to $W^{\sB\hat{\sA}\hat{\sB}}$ and $(W^{\sB\hat{\sA}\hat{\sB}})^\dag$ is given by $u_\diamond = \cO\big(\log{(1/\delta_1)}/s_{\rm min}\big)$, the numbers of queries to $U_\rho^{\hat{\sA}\hat{\sB}}$, $(U_\rho^{\hat{\sA}\hat{\sB}})^\dag$, $U_\sigma^{\hat{\sA}\hat{\sB}}$, and $(U_\sigma^{\hat{\sA}\hat{\sB}})^\dag$ are the same.
By rescaling $\delta_1$ as $\delta_1 = (\delta/3)^2$, we obtain the result for the number of queries.


\subsubsection{Evaluation in the fidelity difference}
\label{sec:uhlfidquery}

When we evaluate the error between the ideal Uhlmann transformation and the transformation $\cT^\sB$ realized by our algorithm using the fidelity difference,
\begin{align}
    \rF(\rho^\sA, \sigma^\sA) - \rF(\cT^\sB(\ketbra{\rho}{\rho}^{\sA\sB}), \ket{\sigma}^{\sA\sB}),
\end{align}
there is a case where the transformation can be realized with fewer queries than $\cO\big(\log{(1/\delta)/s_{\rm min}}\big)$ derived in \Cref{sec:eval diamond in purif query}.

Let $\til{W}^{\sB\sD}$ be the unitary given by \Cref{alg:Uhl purif query} with $\chi = \tF$, and $M^\sB = \tr_{\sA'}\big[\ketbra{\sigma}{\rho}^{\sA'\sB}\big]$.
The square root fidelity between $\til{W}^{\sB\sD}\ket{\rho}^{\sA\sB}\ket{0}^{\sD}$ and $\ket{\sigma}^{\sA\sB}\ket{0}^{\sD}$ is given by 
\begin{align}
\label{inteq:14}
    \sqrt{\rF'} &=\sqrt{\rF}\big(\til{W}^{\sB\sD}\ket{\rho}^{\sA\sB}\ket{0}^{\sD}, \ket{\sigma}^{\sA\sB}\ket{0}^{\sD}\big) \\
    &= \big|\bra{\sigma}^{\sA\sB}\bra{0}^{\sD}\til{W}^{\sB\sD}\ket{0}^{\sD}\ket{\rho}^{\sA\sB}\big| \\
    &= \big|\bra{\sigma}^{\sA\sB}P_\sgn^{(\rm SV)}(M^\sB)\ket{\rho}^{\sA\sB}\big|,
\end{align}
where $P_\sgn$ is the real odd polynomial in \Cref{prop:FPAA general} or \Cref{prop:QSVT sign Fejer}, which satisfies $|P_\sgn(x)| \leq 1$ for $|x| \in [0, 1]$.
We then evaluate the difference between $\sqrt{\rF'}$ and $\sqrt{\rF} = \sqrt{\rF}(\rho^\sA, \sigma^\sA) = \sqrt{\rF}(V^\sB\ket{\rho}^{\sA\sB}, \ket{\sigma}^{\sA\sB})$ as
\begin{align}
    \label{inteq:104}
    \big|\sqrt{\rF} - \sqrt{\rF'}\big| 
    &= \big||\bra{\sigma}^{\sA\sB}V^\sB\ket{\rho}^{\sA\sB}| - |\bra{\sigma}^{\sA\sB}P_\sgn^{(\rm SV)}(M^\sB)\ket{\rho}^{\sA\sB}| \big| \\
    &\leq \big| \bra{\sigma}^{\sA\sB}V^\sB\ket{\rho}^{\sA\sB} - \bra{\sigma}^{\sA\sB}P_\sgn^{(\rm SV)}(M^\sB)\ket{\rho}^{\sA\sB} \big| \\
    &= \Big|\tr\big[\big(V^\sB - P_\sgn^{(\rm SV)}(M^\sB)\big)\tr_{\sA}\big[\ketbra{\rho}{\sigma}^{\sA\sB}\big]\big]\Big| \\
    &= \Big|\tr\big[\big(\sgn^{(\rm SV)}(M^\sB) - P_\sgn^{(\rm SV)}(M^\sB)\big)(M^{\sB})^\dag\big]\Big| \\
    &= \Big|\sum_k \big(1 - P_\sgn(s_k)\big)s_k\Big| \\
    \label{inteq:8}
    &\leq \sum_k \big|1 - P_\sgn(s_k)\big|s_k,
\end{align}
where $\{s_k\}_k$ are the singular values of $M^\sB$.
We should note that they correspond to the singular values of $\sqrt{\sigma^\sA}\sqrt{\rho^\sA}$ and that $\sum_k s_k = \big\|\sqrt{\sigma^\sA}\sqrt{\rho^\sA}\big\|_1 = \sqrt{\rF} \leq 1$.

We evaluate Eq.~\eqref{inteq:8} for each of the two polynomials used in \Cref{alg:Uhl purif query}.
When we use $\mathtt{QSVTSIGN}$ with $\beta = s_{\rm min}$ and $\delta_1 = (\delta/3)^2$, all the non-zero singular values $s_k$ satisfy $s_k \geq \beta$, and hence \Cref{prop:FPAA general} implies that
\begin{align}
    \sum_k \big|1 - P_\sgn(s_k)\big|s_k \leq \delta_1\sum_k s_k \leq \Big(\f{\delta}{3}\Big)^2 \leq \f{\delta}{2}.
\end{align}
When we use $\mathtt{QSVTSIGNFEJER}$ with $\beta = \delta_1/R$ and $\delta_1 = \delta/2$, \Cref{prop:QSVT sign Fejer} implies that $\big|1 - P_\sgn(s_k)\big|s_k \leq \beta$ holds for each $k$, and hence
\begin{align}
    \sum_k \big|1 - P_\sgn(s_k)\big|s_k \leq R\beta = \f{\delta}{2},
\end{align}
where we used the fact that the number of non-zero singular values $s_k$ is $R$, the rank of $\sqrt{\sigma^\sA}\sqrt{\rho^\sA}$.
In both cases, we obtain that
\begin{align}
\label{inteq:13}
    \big|\sqrt{\rF} - \sqrt{\rF'}\big| \leq \f{\delta}{2}.
\end{align}
Using the inequality that $|x-y|\leq 2 |\sqrt{x} - \sqrt{y}|$ for $0 \leq x, y \leq 1$, we have that
\begin{align}
    \big|\rF\big(\til{W}^{\sB\sD}\ket{\rho}^{\sA\sB}\ket{0}^{\sD}, \ket{\sigma}^{\sA\sB}\ket{0}^{\sD}\big)
    - \rF(\rho^\sA, \sigma^\sA)\big|\leq \delta.
\end{align}

By denoting $\tr_{\sD}\circ\til{\cW}^{\sB\sD}\circ\cP_{\ket{0}}^{\bC\rarr\sD}$ by $\cT^\sB$, we obtain that
\begin{align}
    \rF(\rho^\sA, \sigma^\sA) - \delta 
    &\leq \rF\big(\til{W}^{\sB\sD}\ket{\rho}^{\sA\sB}\ket{0}^{\sD}, \ket{\sigma}^{\sA\sB}\ket{0}^{\sD}\big) \\
    &\leq \rF\big(\cT^\sB(\ketbra{\rho}{\rho}^{\sA\sB}), \ket{\sigma}^{\sA\sB}\big),
\end{align}
where we used the monotonicity of the fidelity under the partial trace. 
Since \Cref{alg:Uhl purif query} uses the one of $\mathtt{QSVTSIGN}$ and $\mathtt{QSVTSIGNFEJER}$ that requires fewer queries, the number of queries to $U_\rho^{\hat{\sA}\hat{\sB}}$, $U_\sigma^{\hat{\sA}\hat{\sB}}$, and their inverses is evaluated as
\begin{align}
    u_\tF &= \cO\Big(\min\Big\{\f{1}{s_{\rm min}}\log{\Big(\f{18}{\delta^2}\Big)}, \f{R}{\delta}\Big\}\Big) \\
    &= \cO\Big(\min\Big\{\f{1}{s_{\rm min}}\log{\Big(\f{1}{\delta}\Big)}, \f{R}{\delta}\Big\}\Big).
\end{align}
This completes the proof of Theorem~\ref{thm:Uhlmann alg purif query model}.


\subsection{Proof of \texorpdfstring{\Cref{lem:FPAA diamond}}{Lemma~\ref{lem:FPAA diamond}}}
\label{sec:proof of lem FPAA diamond}

We now give the proof of \Cref{lem:FPAA diamond}.
\begin{proof}[Proof of \Cref{lem:FPAA diamond}]

We use the (right) polar decomposition; for any matrix $A$, there exists a partial isometry $V$ with $\supp[V] \supseteq \im[|A|]$ such that $A$ can be expressed as $A = V|A|$, where $|A| = \sqrt{A^\dag A}$.
Since $\til{U}^{\sB\sD\sE}$ is unitary, we have $(\til{U}^{\sB\sD\sE})^\dag \til{U}^{\sB\sD\sE} = \bI^{\sB\sD\sE}$.
Thus, by applying the polar decomposition to the bottom-left block of $\til{U}^{\sB\sD\sE}$, its first column block can be represented as
\begin{align}
\til{U}^{\sB\sD\sE} =
\begin{blockarray}{ccc}
& \bra{0}^{\sD}&  \vspace{1mm}\\
\begin{block}{c(cc)} 
  \ket{0}^{\sD}\hspace*{1mm} &  \ketbra{\varphi}{\phi}^\sE\otimes P_\sgn^{(\rm SV)}(M^\sB) & \hspace*{2mm} * \hspace*{2mm}\\ 
  &  V^{\sB\sE \rarr \sB\sD\sE}\sqrt{\bI^{\sB\sE} - S^{\sB\sE}} &\hspace*{2mm} * \hspace*{2mm}\\
\end{block}
\end{blockarray}\hspace{1mm},
\end{align}
where $S^{\sB\sE} = \ketbra{\phi}{\phi}^\sE \otimes P_\sgn^{(\rm SV)}((M^\sB)^\dag)P_\sgn^{(\rm SV)}(M^\sB)$, and $V^{\sB\sE \rarr \sB\sD\sE}$ is a partial isometry such that $\supp\big[V^{\sB\sE\rarr\sB\sD\sE}\big] \supseteq \im\big[\bI^{\sB\sE} - S^{\sB\sE}\big]$.
We define a unitary $\til{U}_{\rm ideal}^{\sB\sD\sE}$, which is an exact block-encoding of $\ketbra{\varphi}{\phi}^\sE \otimes \sgn^{(\rm SV)}(M^\sB)$, by
\begin{align}
\til{U}_{\rm ideal}^{\sB\sD\sE} = 
\begin{blockarray}{ccc}
& \bra{0}^{\sD}&  \vspace{1mm}\\
\begin{block}{c(cc)} 
  \ket{0}^{\sD}\hspace*{1mm} &  \ketbra{\varphi}{\phi}^\sE\otimes\sgn^{(\rm SV)}(M^\sB) & \hspace*{2mm} * \hspace*{2mm}\\
  &  V^{\sB\sE \rarr \sB\sD\sE}\Pi_\bot^{\sB\sE} &\hspace*{2mm} * \hspace*{2mm}\\
\end{block}
\end{blockarray}\hspace{1mm},
\end{align}
where $\Pi_\bot^{\sB\sE}$ is the projector onto $\ker[\ketbra{\phi}{\phi}^\sE\otimes M^\sB]$, given by
\begin{align}
    \Pi_\bot^{\sB\sE} = \bI^{\sB\sE} - \ketbra{\phi}{\phi}^\sE\otimes \sgn^{(\rm SV)}((M^\sB)^\dag) \sgn^{(\rm SV)}(M^\sB).
\end{align}

For any state $\ket{\psi}^{\sR\sB}$ with any size of reference system $\sR$, we consider two states such that
\begin{align}
    \ket{\Psi}^{\sR\sB\sD\sE} 
    &= \til{U}^{\sB\sD\sE}\ket{0}^{\sD}\ket{\phi}^\sE\ket{\psi}^{\sR\sB} \\
    &=\ket{0}^{\sD}\ket{\varphi}^\sE\otimes P_\sgn^{(\rm SV)}(M^\sB)\ket{\psi}^{\sR\sB} \notag\\
    &\hspace{3pc}+ V^{\sB\sE \rarr \sB\sD\sE}\sqrt{\bI^{\sB\sE} - S^{\sB\sE}}\ket{\phi}^\sE\ket{\psi}^{\sR\sB},
\end{align}
and
\begin{align}
    \ket{\Psi'}^{\sR\sB\sD\sE} 
    &= \til{U}_{\rm ideal}^{\sB\sD\sE}\ket{0}^{\sD}\ket{\phi}^\sE\ket{\psi}^{\sR\sB} \\
    &= \ket{0}^{\sD}\ket{\varphi}^\sE \otimes \sgn^{(\rm SV)}(M^\sB)\ket{\psi}^{\sR\sB} \notag\\
    &\hspace{2pc}+ V^{\sB\sE \rarr \sB\sD\sE}\Pi_\bot^{\sB\sE}\ket{\phi}^\sE\ket{\psi}^{\sR\sB}.
\end{align}
The distance between these states is bounded as
\begin{align}
\label{inteq:5}
    &\big\|\ket{\Psi}^{\sR\sB\sD\sE} - \ket{\Psi'}^{\sR\sB\sD\sE}\big\| \notag\\
    &\leq \big\|P_\sgn^{(\rm SV)}(M^\sB) - \sgn^{(\rm SV)}(M^\sB)\big\|_\infty
    + \Big\|\sqrt{\bI^{\sB\sE} - S^{\sB\sE}} - \Pi_\bot^{\sB\sE}\Big\|_\infty.
\end{align}
The first term in the right-hand side of Eq.~\eqref{inteq:5} is, by assumption Eq.~\eqref{inteq:24}, bounded above by $\delta \in (0, 1)$. 
Regarding the second term, we can see that
\begin{align}
    &\sqrt{\bI^{\sB\sE} - S^{\sB\sE}} \notag\\
    &= \sqrt{\bI^{\sB\sE} - \ketbra{\phi}{\phi}^\sE\otimes P_\sgn^{(\rm SV)}((M^\sB)^\dag)P_\sgn^{(\rm SV)}(M^\sB)} \\
    &= \sqrt{\sum_k\Big(1-\big(P_\sgn(s_k)\big)^2\Big)\ketbra{\phi}{\phi}^\sE\otimes\ketbra{\xi_k}{\xi_k}^\sB + \Pi_\bot^{\sB\sE}} \\
    \label{inteq:4}
    &= \sum_k\sqrt{1-\big(P_\sgn(s_k)\big)^2}\ketbra{\phi}{\phi}^\sE\otimes\ketbra{\xi_k}{\xi_k}^{\sB} + \Pi_\bot^{\sB\sE},
\end{align}
where the singular value decomposition of $M^\sB$ is expressed by $\sum_k s_k\ketbra{\eta_k}{\xi_k}^\sB$ with the non-zero singular values $s_k$.
Thus, the second term on the right-hand side of Eq.~\eqref{inteq:5} is bounded as 
\begin{align}
    \Big\|\sqrt{\bI^{\sB\sE} -S^{\sB\sE}} - \Pi_\bot^{\sB\sE}\Big\|_\infty 
    &=\max_k\Big|\sqrt{1-\big(P_\sgn(s_k)\big)^2}\Big| \\
    &\leq\max_k\Big|\sqrt{2\big(1-P_\sgn(s_k)\big)}\Big| \\
    &\leq \sqrt{2\delta}.
\end{align}
In the first inequality, we used $1-x^2 \leq 2(1-x)$ for $|x| \leq 1$, and in the last inequality, we used the assumption in Eq.~\eqref{inteq:24}.

Hence, we can bound Eq.~\eqref{inteq:5} as
\begin{equation}
\label{inteq:20}
    \big\|\ket{\Psi}^{\sR\sB\sD\sE} - \ket{\Psi'}^{\sR\sB\sD\sE}\big\| \leq \delta + \sqrt{2\delta}.
\end{equation}
Using the fact that $\f{1}{2}\|\ketbra{\Psi}{\Psi}^{\sR\sB\sD\sE} - \ketbra{\Psi'}{\Psi'}^{\sR\sB\sD\sE}\|_1 \leq \|\ket{\Psi}^{\sR\sB\sD\sE} - \ket{\Psi'}^{\sR\sB\sD\sE}\|$ (see Eq.~\eqref{eq:relation of Euclidean and trace norm}) and that Eq.~\eqref{inteq:20} holds for any state $\ket{\psi}^{\sR\sB}$ with any system $\sR$, we can conclude that
\begin{align}
    \f{1}{2}\big\|\til{\cU}^{\sB\sD\sE} \circ \cP_{\ket{0}\ket{\phi}}^{\bC\rarr\sD\sE} - \til{\cU}_{\rm ideal}^{\sB\sD\sE} \circ \cP_{\ket{0}\ket{\phi}}^{\bC\rarr\sD\sE}\big\|_\diamond 
    &\leq \delta + \sqrt{2\delta} \\
    \label{inteq:15}
    &\leq 3\sqrt{\delta},
\end{align}
where $\delta \in (0, 1)$ and $\cP_{\ket{0}\ket{\phi}}^{\bC\rarr\sD\sE} = \ketbra{0}{0}^\sD \otimes \ketbra{\phi}{\phi}^\sE$.

\end{proof}

\section{Uhlmann transformation algorithm in the purified sample access model}
\label{sec:purif sample}

We provide an in-depth analysis of the Uhlmann transformation algorithm in the purified sample access model. The main challenge here is to prepare a unitary that encodes $\tr_{\sA'}\big[\ketbra{\sigma}{\rho}^{\sA'\sB}\big]$ directly from given states $\ket{\rho}^{\hat{\sA}\hat{\sB}}$ and $\ket{\sigma}^{\hat{\sA}\hat{\sB}}$, rather than from unitaries $U_\rho^{\hat{\sA}\hat{\sB}}$ and $U_\sigma^{\hat{\sA}\hat{\sB}}$ that generate them.
In the purified query access model, the top-left block of the block-encoding unitary $(U_\rho^{\hat{\sA}\hat{\sB}})^\dag(\bI^{\hat{\sA}}\otimes F^{\sB\hat{\sB}})U_\sigma^{\hat{\sA}\hat{\sB}}$ is $\bra{\rho}^{\hat{\sA}\hat{\sB}}(\bI^{\hat{\sA}}\otimes F^{\sB\hat{\sB}})\ket{\sigma}^{\hat{\sA}\hat{\sB}} = \tr_{\sA'}\big[\ketbra{\sigma}{\rho}^{\sA'\sB}\big]$, which depends on these unitaries only through the states $\ket{\rho}^{\hat{\sA}\hat{\sB}}$ and $\ket{\sigma}^{\hat{\sA}\hat{\sB}}$.
This implies that we can replace $(U_\rho^{\hat{\sA}\hat{\sB}})^\dag$ and $U_\sigma^{\hat{\sA}\hat{\sB}}$ with $\ketbra{\rho}{\rho}^{\hat{\sA}\hat{\sB}}$ and $\ketbra{\sigma}{\sigma}^{\hat{\sA}\hat{\sB}}$, respectively, both of which can be block-encoded from samples using the density matrix exponentiation. Then we have
\begin{align}
\label{eq:projector product}
    \ketbra{\rho}{\rho}^{\hat{\sA}\hat{\sB}}(\bI^{\hat{\sA}}\otimes F^{\sB\hat{\sB}})\ketbra{\sigma}{\sigma}^{\hat{\sA}\hat{\sB}}
    = \ketbra{\rho}{\sigma}^{\hat{\sA}\hat{\sB}}\otimes\tr_{\sA'}\big[\ketbra{\sigma}{\rho}^{\sA'\sB}\big].
\end{align}
The additional factor $\ketbra{\rho}{\sigma}^{\hat{\sA}\hat{\sB}}$ arises because the reference state $\ket{0}^{\hat{\sA}\hat{\sB}}$ is replaced by $\ket{\sigma}^{\hat{\sA}\hat{\sB}}$ and $\ket{\rho}^{\hat{\sA}\hat{\sB}}$ on the input and output sides. Since this factor has a single non-zero singular value equal to one, the right-hand side has the same non-zero singular values as $\tr_{\sA'}\big[\ketbra{\sigma}{\rho}^{\sA'\sB}\big]$.

The following is a formal statement of \Cref{infthm:uhlmenn purif sample}.
Here, $\sH$ is the auxiliary system, which consists of three qubits for $\chi = \diamond$ and $\cO\big(\log{(R/\delta)}\big)$ qubits for $\chi = \tF$, and $\cP_{\ket{0}\ket{\sigma}}^{\bC\rarr\sH\hat{\sA}\hat{\sB}}$ is a state-preparation channel given by $\cP_{\ket{0}\ket{\sigma}}^{\bC\rarr\sH\hat{\sA}\hat{\sB}} = \ketbra{0}{0}^\sH \otimes \ketbra{\sigma}{\sigma}^{\hat{\sA}\hat{\sB}}$.

\begin{theorem}[Uhlmann transformation algorithm in the purified sample access model]
\label{thm:algorithm Uhlmann}
Let $\delta \in (0, 1)$ and $\chi \in \{\diamond, \tF\}$. The quantum sample algorithm $\mathtt{UhlmannPurifiedSample}$ given by \Cref{alg:Uhl purif sample} satisfies the following.

A quantum channel $\cJ_\diamond^{\sB\sH\hat{\sA}\hat{\sB}}$ given by 
\begin{align}
    \cJ_\diamond^{\sB\sH\hat{\sA}\hat{\sB}} = \mathtt{UhlmannPurifiedSample}(\ket{\rho}^{\hat{\sA}\hat{\sB}}, \ket{\sigma}^{\hat{\sA}\hat{\sB}}; \delta, \diamond),
\end{align}
satisfies that
\begin{equation}
    \f{1}{2}\big\|\cJ_\diamond^{\sB\sH\hat{\sA}\hat{\sB}} \circ\cP_{\ket{0}\ket{\sigma}}^{\bC\rarr\sH\hat{\sA}\hat{\sB}} - \cU_{\rm ideal}^{\sB\sH\hat{\sA}\hat{\sB}}\circ\cP_{\ket{0}\ket{\sigma}}^{\bC\rarr\sH\hat{\sA}\hat{\sB}}\big\|_\diamond \leq \delta,
\end{equation}
where $U_{\rm ideal}^{\sB\sH\hat{\sA}\hat{\sB}}$ is an exact block-encoding unitary of $\ketbra{\rho}{\sigma}^{\hat{\sA}\hat{\sB}} \otimes V^\sB$, and $V^\sB$ is the Uhlmann partial isometry.
The algorithm uses $w_\diamond = \cO\Big(\f{1}{\delta s_{\rm min}^2}\big(\log{\big(\f{1}{\delta}\big)}\big)^2\Big)$ samples of $\ket{\rho}^{\hat{\sA}\hat{\sB}}$ and $\ket{\sigma}^{\hat{\sA}\hat{\sB}}$.

Moreover, for a quantum channel $\cJ_{\tF}^{\sB\sH\hat{\sA}\hat{\sB}}$ given by 
\begin{align}
    \cJ_{\tF}^{\sB\sH\hat{\sA}\hat{\sB}} = \mathtt{UhlmannPurifiedSample}(\ket{\rho}^{\hat{\sA}\hat{\sB}}, \ket{\sigma}^{\hat{\sA}\hat{\sB}}; \delta, \tF),
\end{align}
let $\cT^\sB = \tr_{\sH\hat{\sA}\hat{\sB}}\circ\cJ_{\tF}^{\sB\sH\hat{\sA}\hat{\sB}}\circ\cP_{\ket{0}\ket{\sigma}}^{\bC\rarr\sH\hat{\sA}\hat{\sB}}$. This channel satisfies that
\begin{align}
        \rF\big(\cT^\sB(\ketbra{\rho}{\rho}^{\sA\sB}), \ket{\sigma}^{\sA\sB}\big) \geq \rF(\rho^\sA, \sigma^\sA) - \delta.
\end{align}
The algorithm uses $w_{\tF} = \cO\Big(\min\big\{\f{R}{\delta s_{\rm min}}\big(\log{\big(\f{1}{\delta}\big)}\big)^2, \f{R^2}{\delta^2}\big\}\Big)$ samples of $\ket{\rho}^{\hat{\sA}\hat{\sB}}$ and $\ket{\sigma}^{\hat{\sA}\hat{\sB}}$.

In both cases, the quantum circuit for implementing the algorithm consists of $\cO\big(w_\chi \log{(d_\sA d_\sB)}\big)$ one- and two-qubit gates, and uses $\cO\big(\log{(d_\sA d_\sB)}\big)$ qubits simultaneously for $\chi = \diamond$, and $\cO\big(\log{(d_\sA d_\sB)} + \log{(R/\delta)}\big)$ qubits for $\chi = \tF$. The circuit description can be computed classically in $\poly\big(w_\chi, \log{(d_\sA d_\sB)}, \log{(1/\delta)}\big)$ time.

\end{theorem}

The additional factor $\ketbra{\rho}{\sigma}^{\hat{\sA}\hat{\sB}}$ appearing in the top-left block of $U_{\rm ideal}^{\sB\sH\hat{\sA}\hat{\sB}}$ in Theorem~\ref{thm:algorithm Uhlmann} does not cause any issues. To see this, we fix the input state to $\ket{\rho}^{\sA\sB}$ and observe that 
\begin{align}
    &|\bra{\sigma}^{\sA\sB}\bra{0}^\sH\bra{\rho}^{\hat{\sA}\hat{\sB}}U_{\rm ideal}^{\sB\sH\hat{\sA}\hat{\sB}}\ket{\rho}^{\sA\sB}\ket{0}^\sH\ket{\sigma}^{\hat{\sA}\hat{\sB}}|^2 \notag\\
    &\hspace{1pc}= |\bra{\sigma}^{\sA\sB}\bra{0}^\sH \bra{\rho}^{\hat{\sA}\hat{\sB}}V^\sB\ket{\rho}^{\sA\sB}\ket{0}^\sH\ket{\rho}^{\hat{\sA}\hat{\sB}}|^2 \\
    &\hspace{1pc}=  \rF(\rho^\sA, \sigma^\sA),
\end{align} 
where $V^\sB$ is the Uhlmann partial isometry.
Thus, when we focus only on the system $\sA\sB$, the algorithm approximately implements the Uhlmann transformation on $\sB$, and the states $\ket{0}^\sH\ket{\sigma}^{\hat{\sA}\hat{\sB}}$ merely serve as auxiliary states.
The state $\ket{\sigma}^{\hat{\sA}\hat{\sB}}$ can be prepared in the purified sample access model.

Similarly to Theorem~\ref{thm:Uhlmann alg purif query model}, the first part of Theorem~\ref{thm:algorithm Uhlmann} measures the approximation error to the Uhlmann transformation using the diamond norm, while the latter part evaluates it in terms of fidelity difference.
We should also recall that the bottom-left block of $U_{\rm ideal}^{\sB\sH\hat{\sA}\hat{\sB}}$ does not degrade the fidelity: by its monotonicity under the partial trace, discarding $\sH\hat{\sA}\hat{\sB}$ can only increase it.


\begin{algorithm}[ht]
\caption{Uhlmann transformation algorithm in the purified sample access model \\ \parbox{\linewidth}{\centering $\mathtt{UhlmannPurifiedSample}(\ket{\rho}^{\hat{\sA}\hat{\sB}}, \ket{\sigma}^{\hat{\sA}\hat{\sB}}; \delta, \chi)$ \\ (In Theorem~\ref{thm:algorithm Uhlmann})}}
\label{alg:Uhl purif sample}
\SetKwInput{KwInput}{Input}
\SetKwInput{KwOutput}{Output}
\SetKwInput{KwParameters}{Parameters}
\SetKwComment{Comment}{$\triangleright$\ }{}
\SetCommentSty{textnormal}

\SetAlgoNoEnd
\SetAlgoNoLine
\KwInput{Two pure quantum states $\ket{\rho}^{\hat{\sA}\hat{\sB}}$ and $\ket{\sigma}^{\hat{\sA}\hat{\sB}}$.}
\KwParameters{$\delta \in (0, 1)$ and $\chi \in \{\diamond, \tF\}$.}
\KwOutput{Quantum channel $\cJ^{\sB\sH\hat{\sA}\hat{\sB}}$.}
\SetAlgoLined

\If{$\chi = \tF$}{
  Set $\cJ^{\sB\sH\hat{\sA}\hat{\sB}} \gets \mathtt{QSVTSIGNSTAGES}(\ket{\rho}^{\hat{\sA}\hat{\sB}}, \ket{\sigma}^{\hat{\sA}\hat{\sB}}, \bI^{\hat{\sA}}\otimes F^{\sB\hat{\sB}}; \delta/(8R), s_{\rm min})$ (Proposition~\ref{prop:QSVT sign stages}). \\
  Return $\cJ^{\sB\sH\hat{\sA}\hat{\sB}}$. \\
}
Set $\delta_1 \gets (\delta/6)^2$ and $\beta \gets s_{\rm min}$. \\
Set $u$ to be the minimum odd integer satisfying $u \geq \big\lceil\f{8e}{\beta}\log{(2/\delta_1)}\big\rceil$. \\
Set $\delta_2 \gets \delta/(2u)$. \\
Set $\cG_\rho^{\sC_1\hat{\sA}\hat{\sB}} \gets \mathtt{BlockEncPureState}(\ket{\rho}^{\hat{\sA}\hat{\sB}}; \delta_2/2)$ and $\cG_\sigma^{\sC_2\hat{\sA}\hat{\sB}} \gets \mathtt{BlockEncPureState}(\ket{\sigma}^{\hat{\sA}\hat{\sB}}; \delta_2/2)$ (Lemma~\ref{lem:BE pure state}). \\
Set $\cG^{\sC\hat{\sA}\hat{\sB}\sB} \gets \cG_\rho^{\sC_1\hat{\sA}\hat{\sB}}\circ\cF^{\sB\hat{\sB}}\circ\cG_\sigma^{\sC_2\hat{\sA}\hat{\sB}}$ and $\cG_{\rm inv}^{\sC\hat{\sA}\hat{\sB}\sB} \gets \cG_\sigma^{\sC_2\hat{\sA}\hat{\sB}}\circ\cF^{\sB\hat{\sB}}\circ\cG_\rho^{\sC_1\hat{\sA}\hat{\sB}}$, where $\sC = \sC_1\sC_2$ and $\cF^{\sB\hat{\sB}}(\cdot) = F^{\sB\hat{\sB}}(\cdot)F^{\sB\hat{\sB}}$ is the swap channel. \\
Set $\cJ^{\sB\sH\hat{\sA}\hat{\sB}} \gets \mathtt{QSVTSIGN}(\cG^{\sC\hat{\sA}\hat{\sB}\sB}, \cG_{\rm inv}^{\sC\hat{\sA}\hat{\sB}\sB}; \delta_1, \beta)$ (Proposition~\ref{prop:FPAA general}). \\
Return $\cJ^{\sB\sH\hat{\sA}\hat{\sB}}$.
\end{algorithm}


At a high level, for $\chi = \diamond$, \Cref{alg:Uhl purif sample} proceeds in two steps:
\begin{align}
\ket{\rho}^{\hat{\sA}\hat{\sB}}, \ket{\sigma}^{\hat{\sA}\hat{\sB}}
&\hspace{0.5pc}\overset{\makebox[6.5pc]{$\scriptstyle 1. \, \mathtt{BlockEncPureState}$}}{\longrightarrow}\hspace{0.5pc}
\begin{pmatrix}
    \hspace*{1.5mm} \ketbra{\rho}{\sigma}^{\hat{\sA}\hat{\sB}} \otimes M^\sB & \hspace*{0mm} * \hspace*{3mm} \\
    \hspace*{2mm} * & \hspace{0mm} * \hspace*{3mm}
\end{pmatrix} \\
&\hspace{0.5pc}\overset{\makebox[6.5pc]{$\scriptstyle 2. \, \mathtt{QSVTSIGN}$}}{\longrightarrow}\hspace{0.5pc}
\begin{pmatrix}
    \hspace*{1.5mm} \ketbra{\rho}{\sigma}^{\hat{\sA}\hat{\sB}} \otimes V^\sB & \hspace*{0mm} * \hspace*{3mm} \\
    \hspace*{2mm} * & \hspace{0mm} * \hspace*{3mm}
\end{pmatrix},
\end{align}
where $M^\sB = \tr_{\sA'}\big[\ketbra{\sigma}{\rho}^{\sA'\sB}\big]$.
In step~1, each use of the block-encoding unitary of $\ketbra{\rho}{\sigma}^{\hat{\sA}\hat{\sB}} \otimes M^\sB$ is approximated within error $\delta_2$, using $\cO(1/\delta_2)$ samples of $\ket{\rho}^{\hat{\sA}\hat{\sB}}$ and $\ket{\sigma}^{\hat{\sA}\hat{\sB}}$ by the density matrix exponentiation in \Cref{lem:BE pure state}.
In step~2, this unitary is used $u = \cO(\log(1/\delta_1)/\beta)$ times by \Cref{prop:FPAA general}, where $\delta_1$ is the accuracy of the polynomial approximating the sign function for the singular values at least $\beta$.
Thus, with $\delta_1 = (\delta/6)^2$ and $\delta_2 = \delta/(2u)$, the total number of samples of $\ket{\rho}^{\hat{\sA}\hat{\sB}}$ and $\ket{\sigma}^{\hat{\sA}\hat{\sB}}$ for \Cref{alg:Uhl purif sample} is $\cO(u^2/\delta) = \cO\big((\log(1/\delta))^2/(\beta^2\delta)\big)$. The parameters $\delta_1$, $\delta_2$, and $\beta$ are chosen so that the overall approximation error remains within $\delta$.

Note that, in \Cref{alg:Uhl purif sample}, $\mathtt{QSVTSIGN}$ takes a quantum channel $\cG$ (and $\cG_{\rm inv}$) as an input, instead of the block-encoding unitary $U_2$ (and $U_2^\dag$) in step~1.
This means that in the quantum circuit for $\mathtt{QSVTSIGN}$, every access to $U_2$ is replaced by an access to $\cG$.  As a result, the output of $\mathtt{QSVTSIGN}$ is not a unitary $U_1$ whose top-left block approximates the sign function, but rather a quantum channel $\cJ$.
We need to evaluate the distance between the actual output channel $\cJ$ and the desired unitary channel $\cU_1$. The error in approximating $\cU_2$ by $\cG$ accumulates proportionally to the number of repetitions $u$ in $\mathtt{QSVTSIGN}$, resulting in the error in approximating $\cU_1$ by $\cJ$.
To keep the final error sufficiently small, the initial error of $\cG$ in approximating $\cU_2$ is pre-scaled by $1/u$. For this reason, $u$ appears in the denominator of the assignment of $\delta_2$ in \Cref{alg:Uhl purif sample}.

This accumulation of the error becomes the bottleneck when we evaluate the error using the fidelity difference. If we use $\mathtt{QSVTSIGN}$ with $\beta = \Theta(\delta/R)$, each of the $u = \cO\big((R/\delta)\log{(1/\delta)}\big)$ uses of $\cG$ requires an error of $\cO(\delta/u)$, which results in $\cO\big((R^2/\delta^3)(\log{(1/\delta)})^2\big)$ samples in total.
This is larger than the number of samples $\cO(R^2/\delta^2)$ in Theorem~\ref{thm:algorithm Uhlmann} by a factor of $(\log{(1/\delta)})^2/\delta$.
With $\beta = s_{\rm min}$, $\mathtt{QSVTSIGN}$ uses $\cO\big((\log{(1/\delta)})^2/(\delta s_{\rm min}^2)\big)$ samples, which is larger than $\cO\big((R/(\delta s_{\rm min}))(\log{(1/\delta)})^2\big)$ in Theorem~\ref{thm:algorithm Uhlmann} by a factor of $1/(Rs_{\rm min})$, where $Rs_{\rm min} \leq 1$ since the sum of the $R$ non-zero singular values of $M^\sB$ is $\sqrt{\rF}(\rho^\sA, \sigma^\sA) \leq 1$.
To remove these factors, in the case of $\chi = \tF$, \Cref{alg:Uhl purif sample} uses a quantum sample algorithm $\mathtt{QSVTSIGNSTAGES}$ described in the following proposition. For its application in the following sections, we state it in a slightly more general form than the form used in this section.

\begin{restatable}[Staged QSVT with the sign function from samples]{proposition}{propsignstages}
\label{prop:QSVT sign stages}
Let $\beta \in (0, 1]$ and $s \in [0, 1]$, and let $\ket{\varphi}^\sE$ and $\ket{\phi}^\sE$ be pure states. Suppose that a unitary $U^{\sB\sE}$ and a matrix $M^\sB$ satisfy $M^\sB = \bra{\varphi}^\sE U^{\sB\sE}\ket{\phi}^\sE$, and that all the non-zero singular values of $M^\sB$ are at least $s$. Then, there exists a quantum sample algorithm $\mathtt{QSVTSIGNSTAGES}(\ket{\varphi}^\sE, \ket{\phi}^\sE, U^{\sB\sE}; \beta, s)$ that outputs a quantum channel $\cJ^{\sB\sH\sE}$, where $\sH$ is a system of $\cO\big(\log{(1/\beta)}\big)$ qubits. The channel $\cJ^{\sB\sH\sE}$ has a Kraus operator $K^{\sB\sH\sE}$ satisfying
\begin{align}
\label{eq:QSVT sign stages}
    s_k\big|1 - \bra{\eta_k}^\sB\bra{0}^\sH\bra{\varphi}^\sE K^{\sB\sH\sE}\ket{0}^\sH\ket{\phi}^\sE\ket{\xi_k}^\sB\big| \leq 4\beta,
\end{align}
for any singular value decomposition $M^\sB = \sum_k s_k\ketbra{\eta_k}{\xi_k}^\sB$ and all $k$.
The algorithm uses 
\begin{align}
    n= \cO\bigg(\f{1}{\beta\max\{\beta, s\}}\Big(\log{\Big(\f{2\max\{\beta, s\}}{\beta}\Big)}\Big)^2\bigg)
\end{align}
samples of $\ket{\varphi}^\sE$ and $\ket{\phi}^\sE$, and uses the controlled versions of $U^{\sB\sE}$ and $(U^{\sB\sE})^\dag$ a total of $\cO(\beta n)$ times.
In addition to these uses, the quantum circuit for this algorithm consists of $\cO(\log{d_\sE})$ one- and two-qubit gates for each sample, and uses $\cO\big(\log{(d_\sB d_\sE)} + \log{(1/\beta)}\big)$ qubits simultaneously. The circuit description can be computed classically in $\poly\big(1/\beta, \log{d_\sE}\big)$ time.
\end{restatable}

The idea of $\mathtt{QSVTSIGNSTAGES}$ is that the sign function at a large singular value can be approximated by a polynomial of low degree, and a high degree is required only for small singular values.
We thus apply the QSVT in stages with thresholds $1, 1/2, 1/4, \ldots$, where each stage coherently tests whether the singular value of a component is above the threshold, applies the sign function to the component if it is, and passes it to the next stage otherwise. A similar staged structure is used in the quantum linear system algorithms based on variable-time amplitude amplification~\cite{Ambainis2012VariableTimeAmplitudeAmplification, Childs2017QuantumLinearSystemsExponentiallyImproved}, where the test in each stage is performed by phase estimation.

The components passed to a later stage have smaller singular values and contribute less to the left-hand side of Eq.~\eqref{eq:QSVT sign stages}. A later stage can thus be implemented with a larger error, which compensates for its larger number of uses of $U^{\sB\sE}$ by requiring fewer samples per use.
The stages end at a threshold of about $\max\{\beta, s/2\}$, where the remaining components have singular values of order at most $\beta$ and contribute at most $\cO(\beta)$, as all the non-zero singular values of $M^\sB$ are at least $s$.
The stages are determined only by $\beta$ and $s$, and the algorithm does not require knowing the singular values of $M^\sB$.
A lower bound $s$ larger than $\beta$ reduces the number of stages, while for $s \leq \beta$, including $s = 0$ when no lower bound is known, the algorithm uses $\cO(1/\beta^2)$ samples and $\cO(1/\beta)$ uses of the controlled versions of $U^{\sB\sE}$ and $(U^{\sB\sE})^\dag$.

In \Cref{alg:Uhl purif sample}, we use $\mathtt{QSVTSIGNSTAGES}$ with $\ket{\varphi}^\sE \gets \ket{\rho}^{\hat{\sA}\hat{\sB}}$, $\ket{\phi}^\sE \gets \ket{\sigma}^{\hat{\sA}\hat{\sB}}$, $U^{\sB\sE} \gets \bI^{\hat{\sA}}\otimes F^{\sB\hat{\sB}}$, $\beta \gets \delta/(8R)$, and $s \gets s_{\rm min}$, for which $M^\sB = \tr_{\sA'}\big[\ketbra{\sigma}{\rho}^{\sA'\sB}\big]$, whose non-zero singular values are those of $\sqrt{\sigma^\sA}\sqrt{\rho^\sA}$. As shown in \Cref{sec:uhlfidpurifsamp}, the number of samples is then $\cO\big(\min\big\{\f{R}{\delta s_{\rm min}}\big(\log{\big(\f{1}{\delta}\big)}\big)^2, \f{R^2}{\delta^2}\big\}\big)$. We provide a proof of \Cref{prop:QSVT sign stages} in \Cref{sec:proof of prop QSVT sign stages}.


\subsection{Proof of Theorem~\ref{thm:algorithm Uhlmann}}
\label{sec:proofUhlmann}

We first consider the case of $\chi = \diamond$ and explain steps 1 and 2 in order, while for step 2, the details are similar to those in the previous section, \Cref{sec:proof purif query Uhl}.

In step 1, let $W_\rho^{\sC_1\hat{\sA}\hat{\sB}}$ and $W_\sigma^{\sC_2\hat{\sA}\hat{\sB}}$ be the $(1, 1, 0)$-block-encoding unitaries of $\ketbra{\rho}{\rho}^{\hat{\sA}\hat{\sB}}$ and $\ketbra{\sigma}{\sigma}^{\hat{\sA}\hat{\sB}}$, respectively, given in \Cref{lem:BE pure state}, where each of $\sC_1$ and $\sC_2$ is a one-qubit system.
We define a unitary $U_2^{\sC\hat{\sA}\hat{\sB}\sB}$ with $\sC = \sC_1\sC_2$ by
\begin{align}
\label{eq:def of U2}
    U_2^{\sC\hat{\sA}\hat{\sB}\sB} = W_\rho^{\sC_1\hat{\sA}\hat{\sB}}(\bI^{\hat{\sA}}\otimes F^{\sB\hat{\sB}})W_\sigma^{\sC_2\hat{\sA}\hat{\sB}}.
\end{align}
Since the swap operator $F^{\sB\hat{\sB}}$ does not act on $\sC$, the top-left block of $U_2^{\sC\hat{\sA}\hat{\sB}\sB}$ is given by the left-hand side of Eq.~\eqref{eq:projector product}, and hence
\begin{align}
    U_2^{\sC\hat{\sA}\hat{\sB}\sB}
    &= \hspace{1mm}
\begin{blockarray}{ccc}
& \bra{0^2}^\sC &  \vspace{1mm}\\
\begin{block}{c(cc)}
  \ket{0^2}^\sC \hspace*{1mm} & \ketbra{\rho}{\sigma}^{\hat{\sA}\hat{\sB}}\otimes M^\sB & \hspace{0mm} * \hspace{3mm}\\
  & * & \hspace{0mm} * \hspace{3mm}\\
\end{block}
\end{blockarray}\hspace{1mm},
\end{align}
where $M^\sB = \tr_{\sA'}\big[\ketbra{\sigma}{\rho}^{\sA'\sB}\big]$.
That is, $U_2^{\sC\hat{\sA}\hat{\sB}\sB}$ is a $(1, 2, 0)$-block-encoding unitary of $\ketbra{\rho}{\sigma}^{\hat{\sA}\hat{\sB}}\otimes M^\sB$.
Since $W_\rho^{\sC_1\hat{\sA}\hat{\sB}}$, $W_\sigma^{\sC_2\hat{\sA}\hat{\sB}}$, and $F^{\sB\hat{\sB}}$ are Hermitian, its inverse is given by $(U_2^{\sC\hat{\sA}\hat{\sB}\sB})^\dag = W_\sigma^{\sC_2\hat{\sA}\hat{\sB}}(\bI^{\hat{\sA}}\otimes F^{\sB\hat{\sB}})W_\rho^{\sC_1\hat{\sA}\hat{\sB}}$.

Let $\cG_\rho^{\sC_1\hat{\sA}\hat{\sB}}$ and $\cG_\sigma^{\sC_2\hat{\sA}\hat{\sB}}$ be the quantum channels given by $\mathtt{BlockEncPureState}(\ket{\rho}^{\hat{\sA}\hat{\sB}}; \delta_2/2)$ and $\mathtt{BlockEncPureState}(\ket{\sigma}^{\hat{\sA}\hat{\sB}}; \delta_2/2)$ in \Cref{lem:BE pure state}, which approximate $\cW_\rho^{\sC_1\hat{\sA}\hat{\sB}}$ and $\cW_\sigma^{\sC_2\hat{\sA}\hat{\sB}}$, respectively, within error $\delta_2/2$.
Each of them uses $m$ samples of $\ket{\rho}^{\hat{\sA}\hat{\sB}}$ or $\ket{\sigma}^{\hat{\sA}\hat{\sB}}$, where
\begin{equation}
\label{eq:m}
m = \bigg\lceil\f{8\pi^2}{\delta_2}\bigg\rceil.
\end{equation}
From the subadditivity of the diamond norm, the quantum channel $\cG^{\sC\hat{\sA}\hat{\sB}\sB} = \cG_\rho^{\sC_1\hat{\sA}\hat{\sB}}\circ\cF^{\sB\hat{\sB}}\circ\cG_\sigma^{\sC_2\hat{\sA}\hat{\sB}}$ satisfies
\begin{equation}
\label{eq:matrix exponentiation of M}
    \f{1}{2}\|\cG^{\sC\hat{\sA}\hat{\sB}\sB} - \cU_2^{\sC\hat{\sA}\hat{\sB}\sB}\|_\diamond \leq \delta_2,
\end{equation}
where $\cF^{\sB\hat{\sB}}(\cdot) = F^{\sB\hat{\sB}}(\cdot)F^{\sB\hat{\sB}}$.
Since the channels in \Cref{lem:BE pure state} also approximate the inverses of the block-encoding unitaries, the quantum channel $\cG_{\rm inv}^{\sC\hat{\sA}\hat{\sB}\sB} = \cG_\sigma^{\sC_2\hat{\sA}\hat{\sB}}\circ\cF^{\sB\hat{\sB}}\circ\cG_\rho^{\sC_1\hat{\sA}\hat{\sB}}$ satisfies
\begin{align}
\label{inteq:92}
    \f{1}{2}\big\|\cG_{\rm inv}^{\sC\hat{\sA}\hat{\sB}\sB} - (\cU_2^{\sC\hat{\sA}\hat{\sB}\sB})^\dag\big\|_\diamond \leq \delta_2.
\end{align}

Since $M^\sB = \tr_{\sA'}\big[\ketbra{\sigma}{\rho}^{\sA'\sB}\big]$, we obtain $\ketbra{\rho}{\sigma}^{\hat{\sA}\hat{\sB}} \otimes V^\sB$ by lifting up all the non-zero singular values of $\ketbra{\rho}{\sigma}^{\hat{\sA}\hat{\sB}}\otimes M^\sB$ to unity, where $V^\sB$ is the Uhlmann partial isometry (see also \Cref{prop:explicit form of Uhl}).
To achieve that, we use the QSVT with the sign function in step 2.

Using $\mathtt{QSVTSIGN}$ in \Cref{prop:FPAA general} with input $\cG^{\sC\hat{\sA}\hat{\sB}\sB}$ and $\cG_{\rm inv}^{\sC\hat{\sA}\hat{\sB}\sB}$, we can realize a quantum channel $\cJ^{\sB\sH\hat{\sA}\hat{\sB}}$ that satisfies
\begin{equation}
\label{inteq:16}
    \f{1}{2}\big\|\cJ^{\sB\sH\hat{\sA}\hat{\sB}} - \cU_1^{\sB\sH\hat{\sA}\hat{\sB}}\big\|_\diamond \leq u\delta_2,
\end{equation}
where $\sH$ is a three-qubit system consisting of $\sC$ and the additional qubit used in $\mathtt{QSVTSIGN}$, and $u$ is the minimum odd integer such that 
\begin{align} 
\label{eq:u} 
    u \geq \bigg\lceil \f{8 e}{\beta}\log{\Big(\f{2}{\delta_1}\Big)} \bigg\rceil.
\end{align}
The unitary $U_1^{\sB\sH\hat{\sA}\hat{\sB}}$ is given by
\begin{align}
\label{inteq:61}
    U_1^{\sB\sH\hat{\sA}\hat{\sB}} 
    = \hspace{1mm}
\begin{blockarray}{ccc}
& \bra{0^3}^\sH&  \vspace{1mm}\\
\begin{block}{c(cc)}
  \ket{0^3}^\sH \hspace*{1mm} & \ketbra{\rho}{\sigma}^{\hat{\sA}\hat{\sB}}\otimes P_\sgn^{(\rm SV)}\big(M^\sB\big) & \hspace{0mm} * \hspace{3mm}\\
  & * & \hspace{0mm} * \hspace{3mm}\\
\end{block}
\end{blockarray}\hspace{1mm},
\end{align}
where $P_\sgn$ satisfies $|P_\sgn(x) - \sgn(x)| \leq \delta_1$ for $x \in [\beta, 1]$.
This follows from replacing $U_2^{\sC\hat{\sA}\hat{\sB}\sB}$ and $(U_2^{\sC\hat{\sA}\hat{\sB}\sB})^\dag$, used a total of $u$ times in the quantum circuit for $\mathtt{QSVTSIGN}$, with the channels $\cG^{\sC\hat{\sA}\hat{\sB}\sB}$ and $\cG_{\rm inv}^{\sC\hat{\sA}\hat{\sB}\sB}$, respectively, which approximate these unitaries as Eqs.~\eqref{eq:matrix exponentiation of M} and~\eqref{inteq:92}.
The overall error is evaluated using the subadditivity of the diamond norm.
Note that
\begin{align}
    P_\sgn^{(\rm SV)}\big(\ketbra{\rho}{\sigma}^{\hat{\sA}\hat{\sB}} \otimes M^\sB\big) 
    = \ketbra{\rho}{\sigma}^{\hat{\sA}\hat{\sB}} \otimes P_\sgn^{(\rm SV)}\big(M^\sB\big).
\end{align}

The discussion so far concerns the case of $\chi = \diamond$.
In \Cref{sec:FPAA part}, we complete the proof for this case by evaluating the approximation error in the diamond norm.
In \Cref{sec:uhlfidpurifsamp}, we consider the case of $\chi = \tF$, where \Cref{alg:Uhl purif sample} uses $\mathtt{QSVTSIGNSTAGES}$ instead of $\mathtt{QSVTSIGN}$, and the approximation error is evaluated by the fidelity difference, which leads to fewer samples.
These results give the full statement of Theorem~\ref{thm:algorithm Uhlmann}.


\subsubsection{Evaluation in the diamond norm}
\label{sec:FPAA part}

We evaluate a distance between $U_1^{\sB\sH\hat{\sA}\hat{\sB}}$ and a unitary $U_{\rm ideal}^{\sB\sH\hat{\sA}\hat{\sB}}$ which is an exact block-encoding of $\ketbra{\rho}{\sigma}^{\hat{\sA}\hat{\sB}}\otimes\sgn^{(\rm SV)}\big(M^\sB\big)$:
\begin{align}
    U_{\rm ideal}^{\sB\sH\hat{\sA}\hat{\sB}} 
    = \hspace{1mm}
\begin{blockarray}{ccc}
& \bra{0^3}^\sH&  \vspace{1mm}\\
\begin{block}{c(cc)}
  \ket{0^3}^\sH \hspace*{1mm} & \ketbra{\rho}{\sigma}^{\hat{\sA}\hat{\sB}}\otimes \sgn^{(\rm SV)}\big(M^\sB\big) & \hspace{0mm} * \hspace{3mm}\\
  & * & \hspace{0mm} * \hspace{3mm}\\
\end{block}
\end{blockarray}\hspace{1mm}.
\end{align}
We note that $\sgn^{(\rm SV)}\big(M^\sB\big)$ is the Uhlmann partial isometry $V^\sB$ (see \Cref{prop:explicit form of Uhl}), and that the singular values of $M^\sB$ are those of $\sqrt{\sigma^\sA}\sqrt{\rho^\sA}$.
Thus, by setting $\beta$ as $\beta = \beta_\diamond = s_{\rm min}$, where $s_{\rm min}$ is the minimum non-zero singular value of $\sqrt{\sigma^\sA}\sqrt{\rho^\sA}$, from \Cref{prop:FPAA general}, we have that
\begin{align}
\label{inteq:89}
    \big\|P_\sgn^{(\rm SV)}\big(M^\sB\big) - \sgn^{(\rm SV)}\big(M^\sB\big)\big\|_\infty \leq \delta_1.
\end{align}
From \Cref{lem:FPAA diamond}, when Eq.~\eqref{inteq:89} holds, we obtain that
\begin{align}
\label{inteq:25}
    \f{1}{2}\big\|\cU_1^{\sB\sH\hat{\sA}\hat{\sB}}\circ\cP_{\ket{0}\ket{\sigma}}^{\bC\rarr\sH\hat{\sA}\hat{\sB}} - \cU_{\rm ideal}^{\sB\sH\hat{\sA}\hat{\sB}}\circ\cP_{\ket{0}\ket{\sigma}}^{\bC\rarr\sH\hat{\sA}\hat{\sB}}\big\|_\diamond \leq 3\sqrt{\delta_1},
\end{align}
where $\cP_{\ket{0}\ket{\sigma}}^{\bC\rarr\sH\hat{\sA}\hat{\sB}} = \ketbra{0^3}{0^3}^\sH \otimes \ketbra{\sigma}{\sigma}^{\hat{\sA}\hat{\sB}}$.
Note that, in the purified sample access model, we can easily implement the channel $\cP_{\ket{\sigma}}^{\bC\rarr\hat{\sA}\hat{\sB}}$, since it only involves assigning one of the given copies of $\ket{\sigma}^{\hat{\sA}\hat{\sB}}$.

Finally, we evaluate the total error of the whole algorithm using the triangle inequality. 
We can evaluate the error as
\begin{align}
    &\big\|\cJ^{\sB\sH\hat{\sA}\hat{\sB}}\circ\cP_{\ket{0}\ket{\sigma}}^{\bC\rarr\sH\hat{\sA}\hat{\sB}} - \cU_{\rm ideal}^{\sB\sH\hat{\sA}\hat{\sB}}\circ\cP_{\ket{0}\ket{\sigma}}^{\bC\rarr\sH\hat{\sA}\hat{\sB}}\big\|_\diamond  \notag\\
    &\leq \big\|\cJ^{\sB\sH\hat{\sA}\hat{\sB}}\circ\cP_{\ket{0}\ket{\sigma}}^{\bC\rarr\sH\hat{\sA}\hat{\sB}} - \cU_1^{\sB\sH\hat{\sA}\hat{\sB}}\circ\cP_{\ket{0}\ket{\sigma}}^{\bC\rarr\sH\hat{\sA}\hat{\sB}}\big\|_\diamond     \notag\\
    &\hspace{2pc}+ \big\|\cU_1^{\sB\sH\hat{\sA}\hat{\sB}}\circ\cP_{\ket{0}\ket{\sigma}}^{\bC\rarr\sH\hat{\sA}\hat{\sB}} - \cU_{\rm ideal}^{\sB\sH\hat{\sA}\hat{\sB}}\circ\cP_{\ket{0}\ket{\sigma}}^{\bC\rarr\sH\hat{\sA}\hat{\sB}}\big\|_\diamond \\
    &\leq \big\|\cJ^{\sB\sH\hat{\sA}\hat{\sB}} - \cU_1^{\sB\sH\hat{\sA}\hat{\sB}}\big\|_\diamond \big\|\cP_{\ket{0}\ket{\sigma}}^{\bC\rarr\sH\hat{\sA}\hat{\sB}}\big\|_\diamond    \notag\\
    &\hspace{2pc}+ \big\|\cU_1^{\sB\sH\hat{\sA}\hat{\sB}}\circ\cP_{\ket{0}\ket{\sigma}}^{\bC\rarr\sH\hat{\sA}\hat{\sB}} - \cU_{\rm ideal}^{\sB\sH\hat{\sA}\hat{\sB}}\circ\cP_{\ket{0}\ket{\sigma}}^{\bC\rarr\sH\hat{\sA}\hat{\sB}}\big\|_\diamond.
\end{align}
Thus, as $\big\|\cP_{\ket{0}\ket{\sigma}}^{\bC\rarr\sH\hat{\sA}\hat{\sB}}\big\|_\diamond = 1$, and from Eqs.~\eqref{inteq:16} and~\eqref{inteq:25}, we obtain that, for $\delta_2 \in (0, 1)$ and $\delta_1 \in (0, 1/2)$,
\begin{align}
\label{inteq:103}
    \f{1}{2} \big\|\cJ^{\sB\sH\hat{\sA}\hat{\sB}}\circ\cP_{\ket{0}\ket{\sigma}}^{\bC\rarr\sH\hat{\sA}\hat{\sB}} - \cU_{\rm ideal}^{\sB\sH\hat{\sA}\hat{\sB}}\circ\cP_{\ket{0}\ket{\sigma}}^{\bC\rarr\sH\hat{\sA}\hat{\sB}}\big\|_\diamond
    \leq 3\sqrt{\delta_1} + u\delta_2.
\end{align}
By rescaling these parameters, such as $\delta_1 = (\delta/6)^2$ and $\delta_2 = \delta/(2u)$, the overall error is bounded by $\delta$.

The number of samples $w_\diamond = um$ of $\ket{\rho}^{\hat{\sA}\hat{\sB}}$ and $\ket{\sigma}^{\hat{\sA}\hat{\sB}}$ is given by 
\begin{align}
    w_\diamond 
    = \cO\Big(\f{1}{\delta_2 \beta_\diamond}\log{\Big(\f{1}{\delta_1}\Big)}\Big)
    = \cO\Big(\f{1}{\delta s_{\rm min}^2}\Big(\log{\Big(\f{1}{\delta}\Big)}\Big)^2\Big),
\end{align}
where $u$ is the minimum odd integer satisfying Eq.~\eqref{eq:u} and $m$ is given by Eq.~\eqref{eq:m}, while we set $\beta = \beta_\diamond = s_{\rm min}$.

By \Cref{lem:BE pure state}, each of $\cG_\rho^{\sC_1\hat{\sA}\hat{\sB}}$ and $\cG_\sigma^{\sC_2\hat{\sA}\hat{\sB}}$ involves $\cO\big(m\log{(d_\sA d_\sB)}\big)$ one- and two-qubit gates, and the swap channel $\cF^{\sB\hat{\sB}}$ involves $\cO(\log{d_\sB})$ two-qubit gates.
Since $\cG^{\sC\hat{\sA}\hat{\sB}\sB}$ and $\cG_{\rm inv}^{\sC\hat{\sA}\hat{\sB}\sB}$ are used $u$ times in total, and the remaining part of $\mathtt{QSVTSIGN}$ involves $\cO(u)$ one- and two-qubit gates, the quantum circuit of the algorithm involves $\cO\big(w_\diamond\log{(d_\sA d_\sB)}\big)$ one- and two-qubit gates.
Note that, although it might be expected that $\cO\big(w_\diamond\log{(d_\sA d_\sB)}\big)$ auxiliary qubits are required in the circuit, in fact, $\cO\big(\log{(d_\sA d_\sB)}\big)$ qubits are used at any one time, since this algorithm is carried out sequentially.
These qubits include those for storing the samples of $\ket{\rho}^{\hat{\sA}\hat{\sB}}$ and $\ket{\sigma}^{\hat{\sA}\hat{\sB}}$, and a constant number of additional qubits.
Considering that this algorithm uses the QSVT and the density matrix exponentiation, its circuit description can be computed classically in $\poly\big(w_\diamond, \log(d_\sA d_\sB), \log(1/\delta)\big)$ time.


\subsubsection{Evaluation in the fidelity difference}
\label{sec:uhlfidpurifsamp}

When we evaluate the error using the fidelity difference, 
\begin{align}
    \rF(\rho^\sA, \sigma^\sA) - \rF(\cT^\sB(\ketbra{\rho}{\rho}^{\sA\sB}), \ket{\sigma}^{\sA\sB}),
\end{align}
the sample complexity of the Uhlmann transformation could be smaller than $\cO\Big(\big(\log{(1/\delta)}\big)^2/(\delta s_{\rm min}^2)\Big)$ derived in \Cref{sec:FPAA part}.
In this case, \Cref{alg:Uhl purif sample} outputs $\cJ^{\sB\sH\hat{\sA}\hat{\sB}} = \mathtt{QSVTSIGNSTAGES}(\ket{\rho}^{\hat{\sA}\hat{\sB}}, \ket{\sigma}^{\hat{\sA}\hat{\sB}}, \bI^{\hat{\sA}}\otimes F^{\sB\hat{\sB}}; \beta_\tF, s_{\rm min})$ with $\beta_\tF = \delta/(8R)$.
From Eq.~\eqref{eq:projector product}, we have that $(\bra{\rho}^{\hat{\sA}\hat{\sB}}\otimes\bI^\sB)(\bI^{\hat{\sA}}\otimes F^{\sB\hat{\sB}})(\ket{\sigma}^{\hat{\sA}\hat{\sB}}\otimes\bI^\sB) = M^\sB$ with $M^\sB = \tr_{\sA'}\big[\ketbra{\sigma}{\rho}^{\sA'\sB}\big]$, whose non-zero singular values are at least $s_{\rm min}$.
Let $M^\sB = \sum_k s_k\ketbra{\eta_k}{\xi_k}^\sB$ be a singular value decomposition, and let $K^{\sB\sH\hat{\sA}\hat{\sB}}$ be the Kraus operator of $\cJ^{\sB\sH\hat{\sA}\hat{\sB}}$ given in \Cref{prop:QSVT sign stages}.
We define $\rF'$ by
\begin{align}
    \rF' = \rF\big(\cJ^{\sB\sH\hat{\sA}\hat{\sB}}\circ\cP_{\ket{0}\ket{\sigma}}^{\bC\rarr\sH\hat{\sA}\hat{\sB}}(\ketbra{\rho}{\rho}^{\sA\sB}), \ket{\sigma}^{\sA\sB}\ket{0}^\sH\ket{\rho}^{\hat{\sA}\hat{\sB}}\big).
\end{align}
Since $\rF'$ is the sum of the non-negative contributions from all the Kraus operators of $\cJ^{\sB\sH\hat{\sA}\hat{\sB}}$, it holds that
\begin{align}
    &\sqrt{\rF'} \notag\\
    &\geq \big|\bra{\sigma}^{\sA\sB}\bra{0}^\sH\bra{\rho}^{\hat{\sA}\hat{\sB}}K^{\sB\sH\hat{\sA}\hat{\sB}}\ket{\rho}^{\sA\sB}\ket{0}^\sH\ket{\sigma}^{\hat{\sA}\hat{\sB}}\big| \\
    &= \Big|\sum_k s_k\bra{\eta_k}^\sB\bra{0}^\sH\bra{\rho}^{\hat{\sA}\hat{\sB}}K^{\sB\sH\hat{\sA}\hat{\sB}}\ket{0}^\sH\ket{\sigma}^{\hat{\sA}\hat{\sB}}\ket{\xi_k}^\sB\Big| \\
    &\geq \sum_k s_k
    - \sum_k s_k\Big|1 - \bra{\eta_k}^\sB\bra{0}^\sH\bra{\rho}^{\hat{\sA}\hat{\sB}}K^{\sB\sH\hat{\sA}\hat{\sB}}\ket{0}^\sH\ket{\sigma}^{\hat{\sA}\hat{\sB}}\ket{\xi_k}^\sB\Big| \\
    \label{inteq:stages root}
    &\geq \sqrt{\rF}(\rho^\sA, \sigma^\sA) - 4R\beta_\tF.
\end{align}
In the equality, we used $\tr_\sA\big[\ketbra{\rho}{\sigma}^{\sA\sB}\big] = (M^\sB)^\dag = \sum_k s_k\ketbra{\xi_k}{\eta_k}^\sB$.
In the last inequality, we used $\sum_k s_k = \sqrt{\rF}(\rho^\sA, \sigma^\sA)$, Eq.~\eqref{eq:QSVT sign stages}, and the fact that the number of non-zero singular values $s_k$ is $R$.
Since $\beta_\tF = \delta/(8R)$, Eq.~\eqref{inteq:stages root} implies $\sqrt{\rF'} \geq \sqrt{\rF}(\rho^\sA, \sigma^\sA) - \delta/2$, and hence $\rF' \geq \rF(\rho^\sA, \sigma^\sA) - \delta$, where we used $\sqrt{\rF}(\rho^\sA, \sigma^\sA) \leq 1$ when $\sqrt{\rF}(\rho^\sA, \sigma^\sA) \geq \delta/2$ and $\rF(\rho^\sA, \sigma^\sA) < \delta$ otherwise.
Defining $\cT^\sB$ by $\cT^\sB = \tr_{\sH\hat{\sA}\hat{\sB}}\circ\cJ^{\sB\sH\hat{\sA}\hat{\sB}}\circ\cP_{\ket{0}\ket{\sigma}}^{\bC\rarr\sH\hat{\sA}\hat{\sB}}$ and using the monotonicity of the fidelity under the partial trace, we obtain that
\begin{align}
    \rF\big(\cT^\sB(\ketbra{\rho}{\rho}^{\sA\sB}), \ket{\sigma}^{\sA\sB}\big)
    \geq \rF(\rho^\sA, \sigma^\sA) - \delta.
\end{align}

By \Cref{prop:QSVT sign stages}, the number of samples $w_\tF$ of the states $\ket{\rho}^{\hat{\sA}\hat{\sB}}$ and $\ket{\sigma}^{\hat{\sA}\hat{\sB}}$ is given by
\begin{align}
    w_\tF = \cO\bigg(\f{1}{\beta_\tF\max\{\beta_\tF, s_{\rm min}\}}\Big(\log{\Big(\f{2\max\{\beta_\tF, s_{\rm min}\}}{\beta_\tF}\Big)}\Big)^2\bigg).
\end{align}
As $(\log{(2x)})^2/x \leq 3$ for $x \geq 1$, taking $x = \max\{\beta_\tF, s_{\rm min}\}/\beta_\tF$ yields $w_\tF = \cO(1/\beta_\tF^2) = \cO(R^2/\delta^2)$.
Moreover, since $Rs_{\rm min} \leq \sum_k s_k \leq 1$, we have $\max\{\beta_\tF, s_{\rm min}\}/\beta_\tF = \max\{1, 8Rs_{\rm min}/\delta\} \leq 8/\delta$, which, together with $\max\{\beta_\tF, s_{\rm min}\} \geq s_{\rm min}$, implies $w_\tF = \cO\big((R/(\delta s_{\rm min}))(\log{(1/\delta)})^2\big)$.
Thus, we obtain
\begin{align}
    w_\tF = \cO\bigg(\min\Big\{\f{R}{\delta s_{\rm min}}\Big(\log{\Big(\f{1}{\delta}\Big)}\Big)^2, \f{R^2}{\delta^2}\Big\}\bigg).
\end{align}
The circuit of the algorithm consists of
$\cO\big(w_\tF \log{(d_\sA d_\sB)}\big)$ one- and two-qubit gates, and uses $\cO\big(\log{(d_\sA d_\sB)} + \log{(R/\delta)}\big)$ qubits at any one time, where the second term comes from the system $\sH$ in \Cref{prop:QSVT sign stages}.
The circuit description can be computed classically in $\poly\big(w_\tF, \log{(d_\sA d_\sB)}, \log{(1/\delta)}\big)$ time.
This completes the proof of Theorem~\ref{thm:algorithm Uhlmann}.

\section{Uhlmann transformation algorithm in the mixed sample access model}
\label{sec:mixed sample}

We provide a quantum sample algorithm for the Uhlmann transformation in the mixed sample access model.
It is natural to assume that $d_\sA = d_\sB$ in this model. The key idea for construction is to use a property of the maximally entangled state $\ket{\Phi}^{\sA\sB} = \f{1}{\sqrt{d_\sA}}\sum_{i=0}^{d_\sA-1}\ket{i}^\sA\ket{i}^{\sB}$: for any state $\omega^\sA$,
\begin{equation}
\Big(\sqrt{\omega^\sA} \otimes \bI^{\sB}\Big)\ket{\Phi}^{\sA\sB}
= \f{1}{\sqrt{d_\sA}}\ket{\omega_{\rm c}}^{\sA\sB},
\end{equation}
where $\ket{\omega_{\rm c}}^{\sA\sB}$ is the canonical purified state of $\omega^\sA$~\cite{Winter2004ExtrinsicIntrinsic}. 
Note that $\tr_\sB[\ketbra{\omega_{\rm c}}{\omega_{\rm c}}^{\sA\sB}] = \omega^\sA$.
Our approach is to first construct, from multiple copies of the given states $\rho^{\hat{\sA}}$ and $\sigma^{\hat{\sA}}$, quantum channels approximating unitaries that prepare the canonical purified states $\ket{\rho_{\rm c}}^{\hat{\sA}\hat{\sB}}$ and $\ket{\sigma_{\rm c}}^{\hat{\sA}\hat{\sB}}$, together with their inverses, and then to use them as the oracles in the Uhlmann transformation algorithm for the purified query access model in \Cref{alg:Uhl purif query}.
This is to find the Uhlmann transformation from $\ket{\rho_{\rm c}}^{\sA\sB}$ to $\ket{\sigma_{\rm c}}^{\sA\sB}$.

The statement is as follows, where $\sD$ is a $\big(\log{(d_\sA d_\sB)} + 7\big)$-qubit system.
Remember that $s_{\rm min}$ and $R$ are the minimum non-zero singular value and the rank of $\sqrt{\sigma^\sA}\sqrt{\rho^\sA}$, respectively. Also, $r_\rho$ and $r_\sigma$ are the ranks of $\rho^\sA$ and $\sigma^\sA$, and $\rho_{\rm min}$ and $\sigma_{\rm min}$ are their minimum non-zero eigenvalues, respectively.

\begin{theorem}[Uhlmann transformation algorithm in the mixed sample access model]
\label{thm:Uhlmann mixed state sample}
Let $\delta \in (0, 1)$ and $\chi \in \{\diamond, \tF\}$. 
Then, the quantum sample algorithm $\mathtt{UhlmannMixedSample}$ given by \Cref{alg:Uhl mixed sample} satisfies the following.

The quantum channel $\cJ_\diamond^{\sB\sD} = \mathtt{UhlmannMixedSample}$ $(\rho^{\hat{\sA}}, \sigma^{\hat{\sA}}; \delta, \diamond)$ satisfies
\begin{equation}
    \f{1}{2}\big\|\cJ_\diamond^{\sB\sD}\circ\cP_{\ket{0}}^{\bC\rarr\sD} - \cU_{\rm ideal}^{\sB\sD}\circ\cP_{\ket{0}}^{\bC\rarr\sD}\big\|_\diamond \leq \delta,
\end{equation}
where $U_{\rm ideal}^{\sB\sD}$ is an exact block-encoding unitary of the Uhlmann partial isometry $V^\sB$.
The algorithm uses
\begin{align}
    \zeta_\diamond = \til{\cO}\bigg(\f{d_\sA}{\delta s_{\rm min}^2} \min\bigg\{\f{1}{\rho_{\rm min}^2} + \f{1}{\sigma_{\rm min}^2}, \f{r_\rho^2 + r_\sigma^2}{\delta^4 s_{\rm min}^4}\bigg\}\bigg)
\end{align}
samples of $\rho^{\hat{\sA}}$ and $\sigma^{\hat{\sA}}$.
Moreover, the quantum channel $\cJ_\tF^{\sB\sD} = \mathtt{UhlmannMixedSample}$$(\rho^{\hat{\sA}}, \sigma^{\hat{\sA}}; \delta, \tF)$ satisfies that
\begin{align}
    \rF\big(\cT^\sB(\ketbra{\rho_{\rm c}}{\rho_{\rm c}}^{\sA\sB}), \ket{\sigma_{\rm c}}^{\sA\sB}\big) \geq \rF(\rho^\sA, \sigma^\sA) - \delta,
\end{align}
where $\cT^{\sB}= \tr_{\sD}\circ\cJ_\tF^{\sB\sD}\circ\cP_{\ket{0}}^{\bC\rarr\sD}$.
The algorithm uses 
\begin{align}
    \zeta_{\tF} = \til{\cO}\bigg(\f{d_\sA}{\delta}\min\bigg\{\f{1}{s_{\rm min}^2}, \f{R^2}{\delta^2}\bigg\}\min\bigg\{\f{1}{\rho_{\rm min}^2} + \f{1}{\sigma_{\rm min}^2}, \f{r_\rho^2 + r_\sigma^2}{\delta^4}\bigg\}\bigg)
\end{align}
samples of $\rho^{\hat{\sA}}$ and $\sigma^{\hat{\sA}}$.

In both cases, the quantum circuit for implementing the algorithm consists of $\cO\big(\zeta_\chi \log{d_\sA}\big)$ one- and two-qubit gates, and uses $\cO\big(\log{d_\sA}\big)$ qubits simultaneously. The circuit description can be computed classically in $\poly\big(\zeta_\chi, \log{d_\sA}, \log{(1/\delta)}\big)$ time.
\end{theorem}

\begin{algorithm}[ht]
\caption{Uhlmann transformation algorithm in the mixed sample access model \\ \parbox{\linewidth}{\centering $\mathtt{UhlmannMixedSample}(\rho^{\hat{\sA}}, \sigma^{\hat{\sA}}; \delta, \chi)$ \\ (In Theorem~\ref{thm:Uhlmann mixed state sample})}}
\label{alg:Uhl mixed sample}
\SetKwInput{KwInput}{Input}
\SetKwInput{KwOutput}{Output}
\SetKwInput{KwParameters}{Parameters}
\SetKwComment{Comment}{$\triangleright$\ }{}
\SetCommentSty{textnormal}
\SetKwProg{Fn}{Subroutine}{}{end}
\SetKwFunction{UPS}{UhlmannPurifiedSample}

\SetAlgoNoEnd
\SetAlgoNoLine
\KwInput{Two quantum states $\rho^{\hat{\sA}}$ and $\sigma^{\hat{\sA}}$.}
\KwParameters{$\delta \in (0, 1)$ and $\chi \in \{\diamond, \tF\}$.}
\KwOutput{Quantum channel $\cJ^{\sB\sD}$.}
\SetAlgoLined

\If{$\chi = \diamond$}{
  Set $\delta_1 \gets \delta/(6u)$ and $\epsilon_1 \gets \delta/(6u)$, where $u$ is the number of uses of $\cG^{\sB\hat{\sA}\hat{\sB}\sG}$ and $\cG_{\rm inv}^{\sB\hat{\sA}\hat{\sB}\sG}$ in $\mathtt{QSVTSIGN}$ below. \\
}
\Else{
  Set $\delta_1 \gets \delta/(4u)$ and $\epsilon_1 \gets \delta/64$, where $u$ is the number of uses of $\cG^{\sB\hat{\sA}\hat{\sB}\sG}$ and $\cG_{\rm inv}^{\sB\hat{\sA}\hat{\sB}\sG}$ in $\mathtt{QSVTSIGNFEJER}$ or $\mathtt{QSVTSIGN}$ below. \\
}
Set $\cL_\rho^{\hat{\sA}\hat{\sB}\sG}$ and $\cL_{\rho, {\rm inv}}^{\hat{\sA}\hat{\sB}\sG}$ to be the quantum channels given by $\mathtt{CanonicalPurification}(\rho^{\hat{\sA}}; \delta_1, \epsilon_1)$, and $\cL_\sigma^{\hat{\sA}\hat{\sB}\sG}$ and $\cL_{\sigma, {\rm inv}}^{\hat{\sA}\hat{\sB}\sG}$ to be those given by $\mathtt{CanonicalPurification}(\sigma^{\hat{\sA}}; \delta_1, \epsilon_1)$ (Algorithm~\ref{alg:canonical purification}). \\
Set $\cG^{\sB\hat{\sA}\hat{\sB}\sG} \gets \cL_{\rho, {\rm inv}}^{\hat{\sA}\hat{\sB}\sG}\circ\cF^{\sB\hat{\sB}}\circ\cL_\sigma^{\hat{\sA}\hat{\sB}\sG}$ and $\cG_{\rm inv}^{\sB\hat{\sA}\hat{\sB}\sG} \gets \cL_{\sigma, {\rm inv}}^{\hat{\sA}\hat{\sB}\sG}\circ\cF^{\sB\hat{\sB}}\circ\cL_\rho^{\hat{\sA}\hat{\sB}\sG}$, where $\cF^{\sB\hat{\sB}}(\cdot) = F^{\sB\hat{\sB}}(\cdot)F^{\sB\hat{\sB}}$ is the swap channel. \\
\If{$\chi = \tF$ and $\mathtt{QSVTSIGNFEJER}$ below uses $\cG^{\sB\hat{\sA}\hat{\sB}\sG}$ and $\cG_{\rm inv}^{\sB\hat{\sA}\hat{\sB}\sG}$ fewer times than $\mathtt{QSVTSIGN}$ below}{
  Set $\cJ^{\sB\sD} \gets \mathtt{QSVTSIGNFEJER}(\cG^{\sB\hat{\sA}\hat{\sB}\sG}, \cG_{\rm inv}^{\sB\hat{\sA}\hat{\sB}\sG}; \delta/(8R))$ (Proposition~\ref{prop:QSVT sign Fejer}). \\
}
\Else{
  Set $\delta_3 \gets (\delta/9)^2$ and $\beta \gets s_{\rm min}$ if $\chi = \diamond$, and $\delta_3 \gets \delta/8$ and $\beta \gets s_{\rm min}/2$ if $\chi = \tF$. \\
  Set $\cJ^{\sB\sD} \gets \mathtt{QSVTSIGN}(\cG^{\sB\hat{\sA}\hat{\sB}\sG}, \cG_{\rm inv}^{\sB\hat{\sA}\hat{\sB}\sG}; \delta_3, \beta)$ (Proposition~\ref{prop:FPAA general}). \\
}
Return $\cJ^{\sB\sD}$. 
\end{algorithm}

In \Cref{alg:Uhl mixed sample}, $\mathtt{CanonicalPurification}$ is a quantum sample algorithm that outputs quantum channels approximating a unitary $U_\omega^{\hat{\sA}\hat{\sB}\sG}$ and its inverse for $\omega \in \{\rho, \sigma\}$, where $U_\omega^{\hat{\sA}\hat{\sB}\sG}\ket{0}^{\hat{\sA}\hat{\sB}\sG}$ approximates $\ket{\omega_{\rm c}}^{\hat{\sA}\hat{\sB}}\ket{0}^\sG$. We will introduce and explain this algorithm in detail in the next subsection, \Cref{sec:appendix B}.
The rest of \Cref{alg:Uhl mixed sample} is the same as \Cref{alg:Uhl purif query} with $U_\rho^{\hat{\sA}\hat{\sB}\sG}$ and $U_\sigma^{\hat{\sA}\hat{\sB}\sG}$ as the oracles, except that these unitaries and their inverses are replaced with the channels given by $\mathtt{CanonicalPurification}$.
In the case of $\chi = \tF$, \Cref{alg:Uhl mixed sample} uses either $\mathtt{QSVTSIGNFEJER}$ or $\mathtt{QSVTSIGN}$, whichever uses $\cG^{\sB\hat{\sA}\hat{\sB}\sG}$ and $\cG_{\rm inv}^{\sB\hat{\sA}\hat{\sB}\sG}$ fewer times.

Since $U_\rho^{\hat{\sA}\hat{\sB}\sG}$ and $U_\sigma^{\hat{\sA}\hat{\sB}\sG}$ prepare the canonical purified states only approximately, the singular values of the matrix block-encoded by $(U_\rho^{\hat{\sA}\hat{\sB}\sG})^\dag(\bI^{\hat{\sA}\sG}\otimes F^{\sB\hat{\sB}})U_\sigma^{\hat{\sA}\hat{\sB}\sG}$ differ from those of $\sqrt{\sigma^\sA}\sqrt{\rho^\sA}$.
In the case of $\chi = \tF$, we evaluate the effect of this difference on the fidelity without requiring the singular values to remain above $s_{\rm min}$, which is why $\mathtt{QSVTSIGN}$ is used with $\beta = s_{\rm min}/2$.
In this evaluation, the error $\epsilon_1$ of $U_\rho^{\hat{\sA}\hat{\sB}\sG}$ and $U_\sigma^{\hat{\sA}\hat{\sB}\sG}$ does not accumulate over their $u$ uses, unlike the error $\delta_1$ of their implementation from samples, and thus $\epsilon_1$ is chosen independently of $u$.
We provide the proof of Theorem~\ref{thm:Uhlmann mixed state sample} in \Cref{sec:proof eval sample mixed}.


\subsection{Canonical purification algorithm}
\label{sec:appendix B}

We provide a quantum sample algorithm for the canonical purification $\mathtt{CanonicalPurification}$, which outputs quantum channels approximating a unitary that prepares the canonical purified state $\ket{\omega_{\rm c}}^{\sA\sB}$ and its inverse, from multiple copies of $\omega^\sA$. 
The details of $\mathtt{CanonicalPurification}$ are provided in \Cref{alg:canonical purification}, where $\sF$ and $\sG$ are five- and six-qubit systems, respectively.
A quantum sample algorithm $\mathtt{BlockEncSqrtState}$ is described in \Cref{prop:block enc sqrt state} below.
Regarding the performance of $\mathtt{CanonicalPurification}$, the following theorem holds.

\begin{theorem}[Quantum algorithm for the canonical purification]
\label{prop:canonical purification}
Let $\delta, \epsilon \in (0, 1)$, and let $\omega^\sA$ be a state whose rank and minimum non-zero eigenvalue are $r_\omega$ and $\omega_{\rm min}$, respectively. The quantum sample algorithm $\mathtt{CanonicalPurification}(\omega^\sA; \delta, \epsilon)$ given by \Cref{alg:canonical purification} outputs quantum channels $\cL^{\sA\sB\sG}$ and $\cL_{\rm inv}^{\sA\sB\sG}$ such that
\begin{align}
\label{eq:canonical purification channel}
    \f{1}{2}\big\|\cL^{\sA\sB\sG} - \cU_\omega^{\sA\sB\sG}\big\|_\diamond \leq \delta, \ \ \text{and} \ \ \
    \f{1}{2}\big\|\cL_{\rm inv}^{\sA\sB\sG} - (\cU_\omega^{\sA\sB\sG})^\dag\big\|_\diamond \leq \delta,
\end{align}
where $U_\omega^{\sA\sB\sG}$ is a unitary satisfying $\big\|U_\omega^{\sA\sB\sG} - U_{\omega, {\rm ideal}}^{\sA\sB\sG}\big\|_\infty \leq \epsilon$ for a unitary $U_{\omega, {\rm ideal}}^{\sA\sB\sG}$ such that $U_{\omega, {\rm ideal}}^{\sA\sB\sG}\ket{0}^{\sA\sB\sG} = \ket{\omega_{\rm c}}^{\sA\sB}\ket{0}^\sG$.
The algorithm uses $f = \til{\cO}\Big(\f{d_\sA}{\delta}\min\Big\{\f{1}{\omega_{\rm min}^2}, \f{r_\omega^2}{\epsilon^4}\Big\}\Big)$ samples of $\omega^\sA$ to implement each of these channels, and the quantum circuit for each channel consists of $\cO(f\log d_\sA)$ one- and two-qubit gates and uses $\cO(\log{d_\sA})$ qubits simultaneously.
In particular, the state $\til{\omega}_{\rm c}^{\sA\sB} = \tr_\sG\circ\cL^{\sA\sB\sG}\circ\cP_{\ket{0}}^{\bC\rarr\sG}(\ketbra{0}{0}^{\sA\sB})$ satisfies 
\begin{align}
    \f{1}{2}\|\til{\omega}_{\rm c}^{\sA\sB} - \ketbra{\omega_{\rm c}}{\omega_{\rm c}}^{\sA\sB}\|_1 \leq \delta + \epsilon.
\end{align}
\end{theorem}

\begin{algorithm}[ht]
\caption{Canonical purification algorithm \\ \parbox{\linewidth}{\centering $\mathtt{CanonicalPurification}(\omega^\sA; \delta, \epsilon)$ \\ (In Theorem~\ref{prop:canonical purification})}}
\label{alg:canonical purification}
\SetKwInput{KwInput}{Input}
\SetKwInput{KwOutput}{Output}
\SetKwInput{KwParameters}{Parameters}
\SetKwComment{Comment}{$\triangleright$\ }{}
\SetCommentSty{textnormal}

\SetAlgoNoEnd
\SetAlgoNoLine
\KwInput{Quantum state $\omega^\sA$.}
\KwParameters{$\delta, \epsilon \in (0, 1)$.}
\KwOutput{Quantum channels $\cL^{\sA\sB\sG}$ and $\cL_{\rm inv}^{\sA\sB\sG}$.}
\SetAlgoLined

Set $\delta_2 \gets \epsilon^2/8$ and $\beta_\diamond \gets 1/(4\sqrt{2d_\sA})$. \\
Set $u$ to be the minimum odd integer satisfying $u \geq \big\lceil\f{8e}{\beta_\diamond}\log{(2/\delta_2)}\big\rceil$. \\
Set $\delta_1 \gets \delta/u$ and $\epsilon_1 \gets \epsilon/(4\sqrt{2r_\omega})$. \\

Set $\cK^{\sA\sF} \gets \mathtt{BlockEncSqrtState}(\omega^\sA; \delta_1, \epsilon_1)$ 
and $\cK_{\rm inv}^{\sA\sF} \gets \mathtt{BlockEncSqrtState}^\dag(\omega^\sA; \delta_1, \epsilon_1)$ (Proposition~\ref{prop:block enc sqrt state}). \\
Set $\cG^{\sA\sB\sF} \gets \cK^{\sA\sF}\circ\cU_{\Phi}^{\sA\sB}$ and $\cG_{\rm inv}^{\sA\sB\sF} \gets (\cU_{\Phi}^{\sA\sB})^\dag\circ\cK_{\rm inv}^{\sA\sF}$, where $U_{\Phi}^{\sA\sB}$ is a unitary preparing the maximally entangled state: $U_{\Phi}^{\sA\sB}\ket{0}^{\sA\sB} = \ket{\Phi}^{\sA\sB}$. \\
Set $\cL^{\sA\sB\sG} \gets \mathtt{QSVTSIGN}(\cG^{\sA\sB\sF}, \cG_{\rm inv}^{\sA\sB\sF}; \delta_2, \beta_\diamond)$ (Proposition~\ref{prop:FPAA general}). \\
Set $\cL_{\rm inv}^{\sA\sB\sG}$ to be the quantum channel obtained by applying the adjoint circuit of the QSVT above, with $\cG^{\sA\sB\sF}$ and $\cG_{\rm inv}^{\sA\sB\sF}$ interchanged. \\
Return $\cL^{\sA\sB\sG}$ and $\cL_{\rm inv}^{\sA\sB\sG}$.
\end{algorithm}

Our approach is to apply a block-encoding of $\sqrt{\omega^\sA}$, approximately implemented by $\mathtt{BlockEncSqrtState}$, to the state $\ket{0}^{\sF}\ket{\Phi}^{\sA\sB}$, and to amplify the resulting coefficient $1/(2\sqrt{2d_\sA})$ of $\ket{\omega_{\rm c}}^{\sA\sB}$ to unity by the QSVT with the sign function.
This amplification uses the block-encoding $\cO\big(\sqrt{d_\sA}\log(1/\epsilon)\big)$ times, and the error of its implementation from samples, which determines $\delta$ in Eq.~\eqref{eq:canonical purification channel}, accumulates over these uses. This is the origin of the factor $d_\sA$ in Theorem~\ref{prop:canonical purification}.
The error $\epsilon$ originates from the polynomial approximations in the QSVT. In particular, the error of the polynomial approximating the square root function changes only the vector to be amplified and vanishes on the kernel of $\omega^\sA$, which leads to the dependence on $r_\omega$ rather than $d_\sA$. An upper bound on $r_\omega$, such as $d_\sA$, can also be used in \Cref{alg:canonical purification}.

\begin{proposition}[Block-encoding of the square root of a quantum state from samples]
\label{prop:block enc sqrt state}
Let $\delta, \epsilon \in (0, 1)$, and let $\omega$ be a state in a $d$-dimensional Hilbert space, whose minimum non-zero eigenvalue is $\omega_{\rm min}$. 
Then, there is a quantum sample algorithm $\mathtt{BlockEncSqrtState}(\omega; \delta, \epsilon)$ that outputs a channel $\cK$ such that $\f{1}{2}\|\cK - \cU\|_\diamond \leq \delta$, using $h = \til{\cO}\big(\f{1}{\delta}\min\big\{1/\omega_{\rm min}^2, 1/\epsilon^4\big\}\big)$ samples of $\omega$,
where $U$ is a unitary satisfying $\|U - U_{\rm ideal}\|_\infty \leq \epsilon$ for a $(1, 5, 0)$-block-encoding unitary $U_{\rm ideal}$ of $\sqrt{\omega}/\big(2\sqrt{2}\big)$.
The unitary $U$ is a $(1, 5, 0)$-block-encoding of $P_{\rm sqrt}(\omega/2)$ for a real odd polynomial $P_{\rm sqrt}$ such that $\big\|P_{\rm sqrt}(\omega/2) - \sqrt{\omega}/\big(2\sqrt{2}\big)\big\|_\infty \leq \epsilon/2$.
The quantum circuit for the algorithm consists of $\cO(h\log{d})$ one- and two-qubit gates, and uses $\cO(\log{d})$ qubits simultaneously.

Moreover, there exists a quantum sample algorithm $\mathtt{BlockEncSqrtState}^\dag(\omega; \delta, \epsilon)$ that outputs a channel $\cK_{\rm inv}$ such that $\f{1}{2}\|\cK_{\rm inv} - \cU^\dag\|_\diamond \leq \delta$, using the same number of samples, gates, and qubits.
\end{proposition}

We first prove \Cref{prop:block enc sqrt state}.    
At a high level, this is achieved by the following two steps:
\begin{align}
\label{eq:flow of trans purif}
\omega^\sA \overset{1. \, \text{Theorem~\ref{thm:BEquantum state form sample}}}{\longrightarrow}
\begin{pmatrix}
    \hspace{1mm} \omega^\sA & \hspace*{1.5mm} * \hspace*{2mm} \\
    * & \hspace{1.5mm} * \hspace*{2mm}
\end{pmatrix}
\overset{2. \, \text{Proposition~\ref{prop:QSVT sqrt}}}{\longrightarrow}
\begin{pmatrix}
    \sqrt{\omega^\sA} & \hspace*{1.5mm} * \hspace*{2mm} \\
    * & \hspace{1.5mm} * \hspace*{2mm}
\end{pmatrix}.
\end{align}

In step~1, we use quantum algorithms that block-encode the density matrix of a quantum state \cite{gilyen2022improvedfidelity, wang2024fastQAtracedist, wang2023SampletoQuary, wang2024TimeEfficientEntEstSamp}, which are obtained by combining the density matrix exponentiation and the QSVT. 

\begin{theorem}[Block-encoding of a quantum state from samples~{\cite[Lemma 2.21]{wang2023SampletoQuary}}]
\label{thm:BEquantum state form sample}
Let $\delta \in (0, 1)$, and let $\omega$ be a state in a $d$-dimensional Hilbert
space. 
Then, there are quantum sample algorithms that output channels $\cE$ and $\cE_{\rm inv}$, respectively, which satisfy $\f{1}{2}\|\cE - \cU\|_\diamond \leq \delta$ and $\f{1}{2}\|\cE_{\rm inv} - \cU^\dag\|_\diamond \leq \delta$, where $U$ is a $(1, 4, 0)$-block-encoding unitary of $\omega/2$.
Both algorithms use $y$ samples of $\omega$, where $y = \cO\big(\f{1}{\delta}\big(\log{\big(\f{1}{\delta}\big)}\big)^2\big)$.
The quantum circuits for these algorithms consist of $\cO(y\log{d})$ one- and two-qubit gates, and use $\cO(\log{d})$ qubits simultaneously.
\end{theorem}

In step~2, we use the QSVT to approximate the square root function. 
Specifically, we take the target function for a polynomial approximation to be $\sqrt{x}/2$. This is a special case of the QSVT for positive power functions; see, e.g., Refs.~\cite{Gilyn2019thesis, gilyen2022improvedfidelity, wang2024NewQuantEnt} for details.

\begin{proposition}[QSVT with the square root function~{\cite[Lemma 17 and Corollary 18]{gilyen2022improvedfidelity}},~{\cite[Lemma II.6 and Lemma II.13]{wang2024NewQuantEnt}}]
\label{prop:QSVT sqrt}
Let $\delta \in (0, 1/2)$, and let $U$ be a $(1, a, 0)$-block-encoding unitary of a positive semidefinite matrix $M$ whose minimum non-zero eigenvalue is $m_{\rm min}$. Then, there are quantum query algorithms that output unitaries $\til{U}$ and $\til{U}^\dag$, respectively, where $\til{U}$ is a $(1, a+1, 0)$-block-encoding of $P_{\rm sqrt}(M)$ for a real odd polynomial $P_{\rm sqrt}$.
The polynomial $P_{\rm sqrt}$ satisfies $\big\|P_{\rm sqrt}(M) - \sqrt{M}/2\big\|_\infty \leq \delta$, and $|P_{\rm sqrt}(x)| \leq 1$ for $x \in [-1, 1]$. 
Both algorithms use $U$ and $U^\dag$ a total of $l = \cO\big(\min\{1/m_{\rm min}, 1/\delta^2\}\log{(1/\delta)}\big)$ times.
\end{proposition}

This proposition follows from the existence of odd polynomials $P_{\rm sqrt}^{(1)}$ and $P_{\rm sqrt}^{(2)}$ with degrees $\cO\big(\f{1}{\alpha}\log{(1/\delta)}\big)$ and $\cO\big(\f{1}{\delta^2}\log{(1/\delta)}\big)$, respectively, where $\alpha \in (0, 1/2)$.
These polynomials satisfy that 
\begin{align}
    &\big|P_{\rm sqrt}^{(1)}(x)\big|, \big|P_{\rm sqrt}^{(2)}(x)\big| \leq 1 
    \ \ \textrm{for} \ \ x \in [-1, 1], \\
    &\Big|P_{\rm sqrt}^{(1)}(x) - \sqrt{x}/2\Big| \leq \delta
    \ \ \textrm{for} \ \ x \in \big[\alpha/2, 1\big], \\
    &\Big|P_{\rm sqrt}^{(2)}(x) - \sqrt{x}/2\Big| \leq \delta
    \ \ \textrm{for} \ \ x \in [0, 1].
\end{align}
For our purpose, it suffices to set $\alpha$ so that $\alpha/2 = \min\{m_{\rm min}, 1/8\}$.
The polynomial $P_{\rm sqrt}$ in \Cref{prop:QSVT sqrt} is the one of $P_{\rm sqrt}^{(1)}$ and $P_{\rm sqrt}^{(2)}$ with the smaller degree, which gives the number $l$ of uses of $U$ and $U^\dag$.

\begin{proof}[Proof of \Cref{prop:block enc sqrt state}]

Let $\delta_1 \in (0, 1)$, and $\sE$ be a four-qubit system.
From Theorem~\ref{thm:BEquantum state form sample}, we can obtain quantum channels $\cE^{\sA\sE}$ and $\cE_{\rm inv}^{\sA\sE}$, which satisfy that
\begin{align}
    \f{1}{2}\big\|\cE^{\sA\sE} - \cU_1^{\sA\sE}\big\|_\diamond \leq \delta_1, \ \ \text{and} \ \ \ 
    \f{1}{2}\big\|\cE_{\rm inv}^{\sA\sE} - (\cU_1^{\sA\sE})^\dag\big\|_\diamond \leq \delta_1,
\end{align}
where $U_1^{\sA\sE}$ is a block-encoding unitary of $\omega/2$, using $y = \cO\Big(\big(\log{(1/\delta_1)}\big)^2/\delta_1\Big)$ samples of $\omega^\sA$.

Let $\delta_2 \in (0, 1/2)$, $\omega_{\rm min}$ be the minimum non-zero eigenvalue of $\omega^\sA$, and $\sF$ be a five-qubit system. 
We evaluate the accumulated error that arises from using quantum channels $\cE^{\sA\sE}$ and $\cE_{\rm inv}^{\sA\sE}$ instead of a unitary $U_1^{\sA\sE}$ and $(U_1^{\sA\sE})^\dag$ in the QSVT with the square root function.
From \Cref{prop:QSVT sqrt}, by using $\cE^{\sA\sE}$ and $\cE_{\rm inv}^{\sA\sE}$ for a total of $l = \cO\big(\min\{1/\omega_{\rm min}, 1/\delta_2^2\}\log{(1/\delta_2)}\big)$ times, we can obtain quantum channels $\cK^{\sA\sF}$ and $\cK_{\rm inv}^{\sA\sF}$ such that 
\begin{align}
\label{inteq:93}
    \f{1}{2}\big\|\cK^{\sA\sF} - \cU_2^{\sA\sF}\big\|_\diamond \leq l\delta_1, \ \  \text{and} \ \ \
    \f{1}{2}\big\|\cK_{\rm inv}^{\sA\sF} - (\cU_2^{\sA\sF})^\dag\big\|_\diamond \leq l\delta_1,
\end{align}
where we used the subadditivity of the diamond norm to bound the overall error.
The unitary $U_2^{\sA\sF}$ is a block-encoding unitary of a polynomial $P_{\rm sqrt}(\omega^\sA/2)$, which satisfies that 
\begin{align}
    \big\|P_{\rm sqrt}(\omega^\sA/2) - \sqrt{\omega^\sA}/\big(2\sqrt{2}\big)\big\|_\infty \leq \delta_2.
\end{align}

To evaluate the distance between $U_2^{\sA\sF}$ and $U_{\rm ideal}^{\sA\sF}$, where $U_{\rm ideal}^{\sA\sF}$ is an exact block-encoding unitary of $\sqrt{\omega^\sA}/\big(2\sqrt{2}\big)$, we use the following lemma~\cite{Gilyn2019thesis, gilyen2022improvedfidelity, wang2023SampletoQuary}.

\begin{lemma}[Robustness of block-encoding unitaries~{\cite[Corollary 12]{gilyen2022improvedfidelity}}]
\label{prevthm:robustness BE unitary}
Suppose that $U$ is a $d \times d$ unitary that is a block-encoding of an $m' \times m$ matrix $A$ and $d \geq 4(m'+m)$. Then, for any $m' \times m$ matrix $\til{A}$ satisfying $\|A-\til{A}\|_\infty + \Big\|\f{A+\til{A}}{2}\Big\|_\infty^2 \leq 1$, there exists a $d \times d$ unitary $\til{U}$ that is a block-encoding of $\til{A}$ satisfying
\begin{align}
    \big\|U - \til{U}\big\|_\infty &\leq \sqrt{\f{2}{1-\big\|\f{A+\til{A}}{2}\big\|_\infty^2}}\big\|A-\til{A}\big\|_\infty.
\end{align}  
\end{lemma}

Observe that, for $\delta_2 \in (0, 1/2)$,
\begin{align}
    &\big\|P_{\rm sqrt}(\omega^\sA/2) - \sqrt{\omega^\sA}/\big(2\sqrt{2}\big)\big\|_\infty \notag\\
    &\hspace{3pc}+ \Big\|\f{1}{2}\Big(P_{\rm sqrt}(\omega^\sA/2) + \sqrt{\omega^\sA}/\big(2\sqrt{2}\big)\Big)\Big\|_\infty^2 \notag\\ 
    &\leq \delta_2 + \f{1}{4}\big(2\big\|\sqrt{\omega^\sA}/\big(2\sqrt{2}\big)\big\|_\infty + \delta_2\big)^2 \\
    &\leq \delta_2 + \f{1}{4}\big(1/\sqrt{2} + \delta_2\big)^2 \\
    &< 1/2 + \f{1}{4}\big(1/\sqrt{2} + 1/2\big)^2 \\
    &< 1.
\end{align}
Thus, applying \Cref{prevthm:robustness BE unitary}, it holds that 
\begin{equation}
\label{inteq:94}
    \big\|U_2^{\sA\sF} - U_{\rm ideal}^{\sA\sF}\big\|_\infty \leq 2 \delta_2,
\end{equation}
where we used $1-\Big\|\f{1}{2}\Big(P_{\rm sqrt}(\omega^\sA/2) + \sqrt{\omega^\sA}/\big(2\sqrt{2}\big)\Big)\Big\|_\infty^2 \geq 1/2$.
Rescaling the parameters as $\delta_1 = \delta/l$ and $\delta_2 = \epsilon/2$, we obtain the statement of \Cref{prop:block enc sqrt state} with $U = U_2^{\sA\sF}$ from Eqs.~\eqref{inteq:93} and~\eqref{inteq:94}.

The number of samples of $\omega^\sA$ required so far is evaluated as
\begin{align}
    h 
    &= yl \\ 
    &= \cO\Big(\f{1}{\delta_1}\Big(\log{\Big(\f{1}{\delta_1}\Big)}\Big)^2\Big) \cO\Big(\min\Big\{\f{1}{\omega_{\rm min}}, \f{1}{\delta_2^2}\Big\}\log{\Big(\f{1}{\delta_2}\Big)}\Big) \\
    &=\cO\bigg(\f{1}{\delta}\min\Big\{\f{1}{\omega_{\rm min}^2}, \f{1}{\epsilon^4}\Big\} \Big(\log{\Big(\f{1}{\epsilon}\Big)}\Big)^2 \notag\\
    &\hspace{4pc}\Big(\log{\Big(\f{1}{\delta}\min\Big\{\f{1}{\omega_{\rm min}}, \f{1}{\epsilon^2}\Big\}\log\Big(\f{1}{\epsilon}\Big)}\Big)\Big)^2\bigg) \\
    &=\til{\cO}\Big(\f{1}{\delta}\min\Big\{\f{1}{\omega_{\rm min}^2}, \f{1}{\epsilon^4}\Big\}\Big).
\end{align}
The number of one- and two-qubit gates in the quantum circuit of this algorithm is given by $\cO\big(h\log{d}\big)$.
Since this algorithm is conducted sequentially, its quantum circuit uses $\cO(\log{d})$ qubits at any one time.

\end{proof}

We now provide a proof of Theorem~\ref{prop:canonical purification}.

\begin{proof}[Proof of Theorem~\ref{prop:canonical purification}]

We denote by $\cK^{\sA\sF}$ and $\cK_{\rm inv}^{\sA\sF}$ the channels given by $\cK^{\sA\sF} = \mathtt{BlockEncSqrtState}(\omega^\sA; \delta_1, \epsilon_1)$ and $\cK_{\rm inv}^{\sA\sF} = \mathtt{BlockEncSqrtState}^\dag(\omega^\sA; \delta_1, \epsilon_1)$ for $\delta_1, \epsilon_1 \in (0, 1)$, respectively.
Proposition~\ref{prop:block enc sqrt state} ensures that
\begin{align}
\label{inteq:canonical K}
    \f{1}{2}\big\|\cK^{\sA\sF} - \cU_2^{\sA\sF}\big\|_\diamond \leq \delta_1, \ \ \text{and} \ \ \ \f{1}{2}\big\|\cK_{\rm inv}^{\sA\sF} - (\cU_2^{\sA\sF})^\dag\big\|_\diamond \leq \delta_1,
\end{align}
when we use $h = \til{\cO}\big(\f{1}{\delta_1}\min\big\{\f{1}{\omega_{\rm min}^2}, \f{1}{\epsilon_1^4}\big\}\big)$ samples of $\omega^\sA$, where $U_2^{\sA\sF}$ is a $(1, 5, 0)$-block-encoding unitary of $P_{\rm sqrt}(\omega^\sA/2)$ for a real odd polynomial $P_{\rm sqrt}$ satisfying
\begin{align}
\label{inteq:canonical sqrt poly}
    \big\|P_{\rm sqrt}(\omega^\sA/2) - \sqrt{\omega^\sA}/\big(2\sqrt{2}\big)\big\|_\infty \leq \epsilon_1/2.
\end{align}

Let $U_{\Phi}^{\sA\sB}$ be a unitary that prepares the maximally entangled state $\ket{\Phi}^{\sA\sB}$: $U_{\Phi}^{\sA\sB}\ket{0}^{\sA\sB} = \ket{\Phi}^{\sA\sB}$.
We define 
\begin{align}
    \ket{v}^{\sA\sB} = P_{\rm sqrt}(\omega^\sA/2)\ket{\Phi}^{\sA\sB},
\end{align}
which satisfies $\bra{0}^{\sF}U_2^{\sA\sF}U_{\Phi}^{\sA\sB}\ket{0}^{\sF}\ket{0}^{\sA\sB} = \ket{v}^{\sA\sB}$.
That is, the unitary $U_2^{\sA\sF}U_{\Phi}^{\sA\sB}$ is a block-encoding of $\ket{v}^{\sA\sB}$:
\begin{align}
\label{inteq:95}
    U_2^{\sA\sF}U_{\Phi}^{\sA\sB}
    = \hspace{1mm}
\begin{blockarray}{ccc}
& \bra{0}^\sF\bra{0}^{\sA\sB}&  & \vspace{1mm}\\
\begin{block}{c(cc)}
  \ket{0}^\sF \hspace*{1mm} & \ket{v}^{\sA\sB} & \hspace{0mm} * \hspace{3mm}\\
   \hspace*{1mm} & * & \hspace{0mm} * \hspace{3mm}\\
\end{block}
\end{blockarray}\hspace{1mm}.
\end{align}

We first show that $\ket{v}^{\sA\sB}$ is close to $\ket{\omega_{\rm c}}^{\sA\sB}/(2\sqrt{2d_\sA})$.
Since $P_{\rm sqrt}(0) = 0$, the operator $P_{\rm sqrt}(\omega^\sA/2) - \sqrt{\omega^\sA}/(2\sqrt{2})$ vanishes on $\ker[\omega^\sA]$, and its rank is at most $r_\omega$.
Its Hilbert--Schmidt norm is thus at most $\sqrt{r_\omega}$ times its operator norm, which is at most $\epsilon_1/2$ by Eq.~\eqref{inteq:canonical sqrt poly}.
Using $\sqrt{\omega^\sA}\ket{\Phi}^{\sA\sB} = \ket{\omega_{\rm c}}^{\sA\sB}/\sqrt{d_\sA}$ and $\|X^\sA\ket{\Phi}^{\sA\sB}\| = \|X^\sA\|_2/\sqrt{d_\sA}$ for any operator $X^\sA$, we obtain
\begin{align}
\label{inteq:canonical v error}
    \Big\|\ket{v}^{\sA\sB} - \f{1}{2\sqrt{2d_\sA}}\ket{\omega_{\rm c}}^{\sA\sB}\Big\|
    &= \f{1}{\sqrt{d_\sA}}\Big\|P_{\rm sqrt}(\omega^\sA/2) - \f{\sqrt{\omega^\sA}}{2\sqrt{2}}\Big\|_2 \notag\\
    &\leq \f{\sqrt{2r_\omega}\epsilon_1}{2\sqrt{2d_\sA}}.
\end{align}
By the QSVT with the sign function in \Cref{prop:FPAA general}, we amplify the norm of $\ket{v}^{\sA\sB}$ to unity.
In the following, we assume $\sqrt{2r_\omega}\epsilon_1 \leq 1/2$, which holds for $\epsilon_1$ in \Cref{alg:canonical purification}.
Then, Eq.~\eqref{inteq:canonical v error} implies $\|\ket{v}^{\sA\sB}\| \geq 1/(4\sqrt{2d_\sA})$, and we set $\beta$ to $\beta_\diamond = 1/(4\sqrt{2d_\sA})$.
Note that, on the right-hand side of Eq.~\eqref{inteq:95}, the states $\ket{0}^\sF$ and $\ket{0}^\sF\ket{0}^{\sA\sB}$, which specify the encoded block, do not have the same dimensions. Even in this case, the QSVT can still be performed without any issue~\cite{gilyen2019qsvt, Gilyn2019thesis, martyn2021grand}.

For $\delta_2 \in (0, 1/2)$, let $U_\omega^{\sA\sB\sG}$ be the unitary given by $\mathtt{QSVTSIGN}$ in \Cref{prop:FPAA general} with $\delta_2$ and $\beta_\diamond$, applied to $U_2^{\sA\sF}U_{\Phi}^{\sA\sB}$.
This QSVT uses $U_2^{\sA\sF}U_{\Phi}^{\sA\sB}$ and its inverse a total of $u$ times, where $u$ is the minimum odd integer such that $u \geq \big\lceil\f{8e}{\beta_\diamond}\log(2/\delta_2)\big\rceil = \big\lceil32e\sqrt{2d_\sA}\log(2/\delta_2)\big\rceil$.
The block in Eq.~\eqref{inteq:95} has the single non-zero singular value $\|\ket{v}^{\sA\sB}\| \geq \beta_\diamond$ with the left singular vector $\ket{\hat{v}}^{\sA\sB} = \ket{v}^{\sA\sB}/\|\ket{v}^{\sA\sB}\|$.
Thus, from \Cref{prop:FPAA general}, the unitary $U_\omega^{\sA\sB\sG}$ satisfies
\begin{align}
    \bra{0}^\sG U_\omega^{\sA\sB\sG}\ket{0}^{\sA\sB\sG} 
    = P_\sgn\big(\|\ket{v}^{\sA\sB}\|\big)\ket{\hat{v}}^{\sA\sB},
\end{align}
where $P_\sgn$ is a real odd polynomial such that $\big|P_\sgn\big(\|\ket{v}^{\sA\sB}\|\big) - 1\big| \leq \delta_2$.
Since $U_\omega^{\sA\sB\sG}\ket{0}^{\sA\sB\sG}$ is a unit vector, it holds that
\begin{align}
    \big\|U_\omega^{\sA\sB\sG}\ket{0}^{\sA\sB\sG} - \ket{\hat{v}}^{\sA\sB}\ket{0}^\sG\big\|^2 
    = 2 - 2P_\sgn\big(\|\ket{v}^{\sA\sB}\|\big) 
    \leq 2\delta_2.
\end{align}
Moreover, using $\big\|\ket{x}/\|\ket{x}\| - \ket{y}/\|\ket{y}\|\big\| \leq 2\|\ket{x} - \ket{y}\|/\|\ket{y}\|$ for non-zero vectors $\ket{x}$ and $\ket{y}$, Eq.~\eqref{inteq:canonical v error} implies $\big\|\ket{\hat{v}}^{\sA\sB} - \ket{\omega_{\rm c}}^{\sA\sB}\big\| \leq 2\sqrt{2r_\omega}\epsilon_1$.
Hence, we obtain
\begin{align}
\label{inteq:canonical state error}
    \big\|U_\omega^{\sA\sB\sG}\ket{0}^{\sA\sB\sG} - \ket{\omega_{\rm c}}^{\sA\sB}\ket{0}^\sG\big\| 
    \leq \sqrt{2\delta_2} + 2\sqrt{2r_\omega}\epsilon_1.
\end{align}

We next construct a unitary that exactly prepares $\ket{\omega_{\rm c}}^{\sA\sB}\ket{0}^\sG$ and is close to $U_\omega^{\sA\sB\sG}$.
For unit vectors $\ket{x}$ and $\ket{y}$ such that $c = \braket{x}{y}$ is real and $c \in (-1, 1)$, the unitary
\begin{align}
    &G = c\big(\ketbra{x}{x} + \ketbra{w}{w}\big) + \sqrt{1-c^2}\big(\ketbra{w}{x} - \ketbra{x}{w}\big) \notag\\
    &\hspace{5pc}+ \bI - \ketbra{x}{x} - \ketbra{w}{w},
\end{align}
where $\ket{w} = (\ket{y} - c\ket{x})/\sqrt{1-c^2}$, satisfies $G\ket{x} = \ket{y}$.
Since the eigenvalues of $G$ other than one are $c \pm i\sqrt{1-c^2}$, we have $\|G - \bI\|_\infty = \sqrt{2 - 2c} = \|\ket{x} - \ket{y}\|$.
We apply this to $\ket{x} = U_\omega^{\sA\sB\sG}\ket{0}^{\sA\sB\sG}$ and $\ket{y} = \ket{\omega_{\rm c}}^{\sA\sB}\ket{0}^\sG$, for which $c = P_\sgn\big(\|\ket{v}^{\sA\sB}\|\big)\braket{\hat{v}}{\omega_{\rm c}}$ is real since $\braket{v}{\omega_{\rm c}} = \tr\big[P_{\rm sqrt}(\omega^\sA/2)\sqrt{\omega^\sA}\big]/\sqrt{d_\sA}$, and define $U_{\omega, {\rm ideal}}^{\sA\sB\sG} = GU_\omega^{\sA\sB\sG}$, where we take $G = \bI$ when $c = 1$.
Then, $U_{\omega, {\rm ideal}}^{\sA\sB\sG}\ket{0}^{\sA\sB\sG} = \ket{\omega_{\rm c}}^{\sA\sB}\ket{0}^\sG$, and from Eq.~\eqref{inteq:canonical state error}, we obtain that
\begin{align}
    \big\|U_\omega^{\sA\sB\sG} - U_{\omega, {\rm ideal}}^{\sA\sB\sG}\big\|_\infty 
    &= \big\|U_\omega^{\sA\sB\sG}\ket{0}^{\sA\sB\sG} - \ket{\omega_{\rm c}}^{\sA\sB}\ket{0}^\sG\big\| \\
    \label{inteq:canonical total bias}
    &\leq \sqrt{2\delta_2} + 2\sqrt{2r_\omega}\epsilon_1.
\end{align}

We evaluate the error originated by using the quantum channels $\cG^{\sA\sB\sF} = \cK^{\sA\sF}\circ\cU_{\Phi}^{\sA\sB}$ and $\cG_{\rm inv}^{\sA\sB\sF} = (\cU_{\Phi}^{\sA\sB})^\dag \circ \cK_{\rm inv}^{\sA\sF}$, instead of $U_2^{\sA\sF}U_{\Phi}^{\sA\sB}$ and $(U_2^{\sA\sF}U_{\Phi}^{\sA\sB})^\dag$, in the input of $\mathtt{QSVTSIGN}$.
From Eq.~\eqref{inteq:canonical K} and the subadditivity of the diamond norm, the quantum channel $\cL^{\sA\sB\sG} = \mathtt{QSVTSIGN}(\cG^{\sA\sB\sF}, \cG_{\rm inv}^{\sA\sB\sF}; \delta_2, \beta_\diamond)$ and the quantum channel $\cL_{\rm inv}^{\sA\sB\sG}$ obtained by applying its adjoint circuit with $\cG^{\sA\sB\sF}$ and $\cG_{\rm inv}^{\sA\sB\sF}$ interchanged satisfy
\begin{align}
\label{inteq:canonical L}
    &\f{1}{2}\big\|\cL^{\sA\sB\sG} - \cU_\omega^{\sA\sB\sG}\big\|_\diamond \leq u\delta_1, \\
\label{inteq:canonical L inv}
    &\f{1}{2}\big\|\cL_{\rm inv}^{\sA\sB\sG} - (\cU_\omega^{\sA\sB\sG})^\dag\big\|_\diamond \leq u\delta_1.
\end{align}
Rescaling the parameters as $\delta_1 = \delta/u$, $\epsilon_1 = \epsilon/(4\sqrt{2r_\omega})$, and $\delta_2 = \epsilon^2/8$, we obtain Eq.~\eqref{eq:canonical purification channel} and $\|U_\omega^{\sA\sB\sG} - U_{\omega, {\rm ideal}}^{\sA\sB\sG}\|_\infty \leq \epsilon$ from Eqs.~\eqref{inteq:canonical total bias}, \eqref{inteq:canonical L}, and~\eqref{inteq:canonical L inv}.
Moreover, noting that $\tr_\sG\circ\cU_{\omega, {\rm ideal}}^{\sA\sB\sG} \circ\cP_{\ket{0}}^{\bC\rarr\sG}(\ketbra{0}{0}^{\sA\sB}) = \ketbra{\omega_{\rm c}}{\omega_{\rm c}}^{\sA\sB}$, we obtain that
\begin{align}
    &\f{1}{2}\big\|\til{\omega}_{\rm c}^{\sA\sB} - \ketbra{\omega_{\rm c}}{\omega_{\rm c}}^{\sA\sB}\big\|_1 \notag\\
    &\leq \f{1}{2}\big\|\cL^{\sA\sB\sG} - \cU_\omega^{\sA\sB\sG}\big\|_\diamond + \big\|U_\omega^{\sA\sB\sG} - U_{\omega, {\rm ideal}}^{\sA\sB\sG}\big\|_\infty \\
    &\leq \delta + \epsilon,
\end{align}
where we used the definition of the diamond norm, the contraction property of the trace norm under the partial trace in Eq.~\eqref{eq:contraction trace norm}, and Eq.~\eqref{eq:rel of operator and diamond norm}.

The number of samples of $\omega^\sA$ required for each of $\cL^{\sA\sB\sG}$ and $\cL_{\rm inv}^{\sA\sB\sG}$ is evaluated as
\begin{align}
    f 
    &= hu \\ 
    &= \til{\cO}\Big(\f{1}{\delta_1}\min\Big\{\f{1}{\omega_{\rm min}^2}, \f{1}{\epsilon_1^4}\Big\}\Big)\cO\Big(\sqrt{d_\sA}\log{\Big(\f{1}{\epsilon}\Big)}\Big) \\
    &=\til{\cO}\Big(\f{d_\sA}{\delta}\min\Big\{\f{1}{\omega_{\rm min}^2}, \f{r_\omega^2}{\epsilon^4}\Big\}\Big).
\end{align}
The number of one- and two-qubit gates in the quantum circuit of this algorithm is evaluated as $\cO\big(f\log{d_\sA}\big)$.
Since this algorithm is conducted sequentially, its quantum circuit uses $\cO(\log{d_\sA})$ qubits at any one time.

\end{proof}


\subsection{Proof of Theorem~\ref{thm:Uhlmann mixed state sample}}
\label{sec:proof eval sample mixed}

We first describe the quantum channels used in both cases of $\chi = \diamond$ and $\chi = \tF$.
Using $\mathtt{CanonicalPurification}$ with $\delta_1, \epsilon_1 \in (0, 1)$, we obtain quantum channels $\cL_\rho^{\hat{\sA}\hat{\sB}\sG}$, $\cL_{\rho, {\rm inv}}^{\hat{\sA}\hat{\sB}\sG}$, $\cL_\sigma^{\hat{\sA}\hat{\sB}\sG}$, and $\cL_{\sigma, {\rm inv}}^{\hat{\sA}\hat{\sB}\sG}$ such that
\begin{align}
\label{inteq:41}
    &\f{1}{2}\big\|\cL_\rho^{\hat{\sA}\hat{\sB}\sG} - \cU_\rho^{\hat{\sA}\hat{\sB}\sG}\big\|_\diamond \leq \delta_1, \ \text{and} \ \ \f{1}{2}\big\|\cL_{\rho, {\rm inv}}^{\hat{\sA}\hat{\sB}\sG} - (\cU_\rho^{\hat{\sA}\hat{\sB}\sG})^\dag\big\|_\diamond \leq \delta_1, \\
\label{inteq:41 sigma}
    &\f{1}{2}\big\|\cL_\sigma^{\hat{\sA}\hat{\sB}\sG} - \cU_\sigma^{\hat{\sA}\hat{\sB}\sG}\big\|_\diamond \leq \delta_1, \ \text{and} \ \ \f{1}{2}\big\|\cL_{\sigma, {\rm inv}}^{\hat{\sA}\hat{\sB}\sG} - (\cU_\sigma^{\hat{\sA}\hat{\sB}\sG})^\dag\big\|_\diamond \leq \delta_1,
\end{align}
where $U_\rho^{\hat{\sA}\hat{\sB}\sG}$ and $U_\sigma^{\hat{\sA}\hat{\sB}\sG}$ are unitaries satisfying
\begin{align}
\label{inteq:canonical bias}
    \big\|U_\rho^{\hat{\sA}\hat{\sB}\sG} - U_{\rho, {\rm ideal}}^{\hat{\sA}\hat{\sB}\sG}\big\|_\infty \leq \epsilon_1, \ \ \text{and} \ \ \ \big\|U_\sigma^{\hat{\sA}\hat{\sB}\sG} - U_{\sigma, {\rm ideal}}^{\hat{\sA}\hat{\sB}\sG}\big\|_\infty \leq \epsilon_1,
\end{align}
for unitaries $U_{\rho, {\rm ideal}}^{\hat{\sA}\hat{\sB}\sG}$ and $U_{\sigma, {\rm ideal}}^{\hat{\sA}\hat{\sB}\sG}$ such that $U_{\rho, {\rm ideal}}^{\hat{\sA}\hat{\sB}\sG}\ket{0}^{\hat{\sA}\hat{\sB}\sG} = \ket{\rho_{\rm c}}^{\hat{\sA}\hat{\sB}}\ket{0}^\sG$ and $U_{\sigma, {\rm ideal}}^{\hat{\sA}\hat{\sB}\sG}\ket{0}^{\hat{\sA}\hat{\sB}\sG} = \ket{\sigma_{\rm c}}^{\hat{\sA}\hat{\sB}}\ket{0}^\sG$.
A pair of these channels, one for $\rho^{\hat{\sA}}$ and the other for $\sigma^{\hat{\sA}}$, uses $f$ samples of $\rho^{\hat{\sA}}$ and $\sigma^{\hat{\sA}}$, where 
\begin{align}
\label{inteq:43}
    f = \til{\cO}\Big(\f{d_\sA}{\delta_1}\min\Big\{\f{1}{\rho_{\rm min}^2} + \f{1}{\sigma_{\rm min}^2}, \f{r_\rho^2 + r_\sigma^2}{\epsilon_1^4}\Big\}\Big),
\end{align}
where $r_\rho$ and $r_\sigma$ are the ranks of $\rho^\sA$ and $\sigma^\sA$, and $\rho_{\rm min}$ and $\sigma_{\rm min}$ are their minimum non-zero eigenvalues, respectively. Here, we used that
\begin{align}
    &\min\{1/\rho_{\rm min}^2, r_\rho^2/\epsilon_1^4\} + \min\{1/\sigma_{\rm min}^2, r_\sigma^2/\epsilon_1^4\} \notag\\  
    &\hspace{1pc}\leq \min\{1/\rho_{\rm min}^2 + 1/\sigma_{\rm min}^2, (r_\rho^2 + r_\sigma^2)/\epsilon_1^4\}.
\end{align}
Let $\ket{\til{\rho}_{\rm c}}^{\hat{\sA}\hat{\sB}\sG} = U_\rho^{\hat{\sA}\hat{\sB}\sG}\ket{0}^{\hat{\sA}\hat{\sB}\sG}$ and $\ket{\til{\sigma}_{\rm c}}^{\hat{\sA}\hat{\sB}\sG} = U_\sigma^{\hat{\sA}\hat{\sB}\sG}\ket{0}^{\hat{\sA}\hat{\sB}\sG}$. From Eq.~\eqref{inteq:canonical bias}, they satisfy
\begin{align}
\label{inteq:tilde states}
    &\big\|\ket{\til{\rho}_{\rm c}}^{\hat{\sA}\hat{\sB}\sG} - \ket{\rho_{\rm c}}^{\hat{\sA}\hat{\sB}}\ket{0}^\sG\big\| \leq \epsilon_1, \\
\label{inteq:tilde states sigma}
    &\big\|\ket{\til{\sigma}_{\rm c}}^{\hat{\sA}\hat{\sB}\sG} - \ket{\sigma_{\rm c}}^{\hat{\sA}\hat{\sB}}\ket{0}^\sG\big\| \leq \epsilon_1.
\end{align}

We consider the unitaries
\begin{align}
    &W^{\sB\hat{\sA}\hat{\sB}\sG} = (U_\rho^{\hat{\sA}\hat{\sB}\sG})^\dag(\bI^{\hat{\sA}\sG}\otimes F^{\sB\hat{\sB}})U_\sigma^{\hat{\sA}\hat{\sB}\sG}, \\
    &W_{\rm ideal}^{\sB\hat{\sA}\hat{\sB}\sG} = (U_{\rho, {\rm ideal}}^{\hat{\sA}\hat{\sB}\sG})^\dag(\bI^{\hat{\sA}\sG}\otimes F^{\sB\hat{\sB}})U_{\sigma, {\rm ideal}}^{\hat{\sA}\hat{\sB}\sG},
\end{align}
where $W^{\sB\hat{\sA}\hat{\sB}\sG}$ is the unitary in \Cref{alg:Uhl purif query} with $U_\rho^{\hat{\sA}\hat{\sB}\sG}$ and $U_\sigma^{\hat{\sA}\hat{\sB}\sG}$ as the oracles.
The same calculation as in Eq.~\eqref{inteq:171}, with $\hat{\sA}\sG$ in place of $\hat{\sA}$, shows that $W^{\sB\hat{\sA}\hat{\sB}\sG}$ and $W_{\rm ideal}^{\sB\hat{\sA}\hat{\sB}\sG}$ are block-encodings of
\begin{align}
\label{inteq:tilde M}
    &\til{M}^\sB = \tr_{\sA'\sG'}\big[\ketbra{\til{\sigma}_{\rm c}}{\til{\rho}_{\rm c}}^{\sA'\sB\sG'}\big], \\
\label{inteq:M ideal}
    &M^\sB = \tr_{\sA'}\big[\ketbra{\sigma_{\rm c}}{\rho_{\rm c}}^{\sA'\sB}\big],
\end{align}
respectively, where the system $\sG'$ does not appear in $M^\sB$ since it is in $\ket{0}^{\sG'}$ for both $U_{\rho, {\rm ideal}}^{\hat{\sA}\hat{\sB}\sG}\ket{0}^{\hat{\sA}\hat{\sB}\sG}$ and $U_{\sigma, {\rm ideal}}^{\hat{\sA}\hat{\sB}\sG}\ket{0}^{\hat{\sA}\hat{\sB}\sG}$.
The non-zero singular values of $M^\sB$ are those of $\sqrt{\sigma^\sA}\sqrt{\rho^\sA}$ (see \Cref{sec:analyze partial iso}).
From Eq.~\eqref{inteq:canonical bias}, we have
\begin{align}
\label{inteq:W bias}
    \big\|W^{\sB\hat{\sA}\hat{\sB}\sG} - W_{\rm ideal}^{\sB\hat{\sA}\hat{\sB}\sG}\big\|_\infty \leq 2\epsilon_1.
\end{align}
Since $F^{\sB\hat{\sB}}$ is Hermitian, $(W^{\sB\hat{\sA}\hat{\sB}\sG})^\dag = (U_\sigma^{\hat{\sA}\hat{\sB}\sG})^\dag(\bI^{\hat{\sA}\sG}\otimes F^{\sB\hat{\sB}})U_\rho^{\hat{\sA}\hat{\sB}\sG}$.
Thus, from Eqs.~\eqref{inteq:41} and~\eqref{inteq:41 sigma} and the subadditivity of the diamond norm, the quantum channels $\cG^{\sB\hat{\sA}\hat{\sB}\sG} = \cL_{\rho, {\rm inv}}^{\hat{\sA}\hat{\sB}\sG}\circ\cF^{\sB\hat{\sB}}\circ\cL_\sigma^{\hat{\sA}\hat{\sB}\sG}$ and $\cG_{\rm inv}^{\sB\hat{\sA}\hat{\sB}\sG} = \cL_{\sigma, {\rm inv}}^{\hat{\sA}\hat{\sB}\sG}\circ\cF^{\sB\hat{\sB}}\circ\cL_\rho^{\hat{\sA}\hat{\sB}\sG}$ satisfy
\begin{align}
\label{inteq:79}
    &\f{1}{2}\big\|\cG^{\sB\hat{\sA}\hat{\sB}\sG} - \cW^{\sB\hat{\sA}\hat{\sB}\sG}\big\|_\diamond \leq 2\delta_1, \\
\label{inteq:79 inv}
    &\f{1}{2}\big\|\cG_{\rm inv}^{\sB\hat{\sA}\hat{\sB}\sG} - (\cW^{\sB\hat{\sA}\hat{\sB}\sG})^\dag\big\|_\diamond \leq 2\delta_1,
\end{align}
where $\cF^{\sB\hat{\sB}}(\cdot) = F^{\sB\hat{\sB}}(\cdot)F^{\sB\hat{\sB}}$ is the swap channel.
Each of $\cG^{\sB\hat{\sA}\hat{\sB}\sG}$ and $\cG_{\rm inv}^{\sB\hat{\sA}\hat{\sB}\sG}$ uses $f$ samples of $\rho^{\hat{\sA}}$ and $\sigma^{\hat{\sA}}$ in Eq.~\eqref{inteq:43}.

The rest is almost the same as the discussion in the purified query access model.
Let $\cJ^{\sB\sD}$ be the quantum channel given by $\mathtt{QSVTSIGN}(\cG^{\sB\hat{\sA}\hat{\sB}\sG}, \cG_{\rm inv}^{\sB\hat{\sA}\hat{\sB}\sG}; \delta_3, \beta)$ or $\mathtt{QSVTSIGNFEJER}(\cG^{\sB\hat{\sA}\hat{\sB}\sG}, \cG_{\rm inv}^{\sB\hat{\sA}\hat{\sB}\sG}; \beta)$ as in \Cref{alg:Uhl mixed sample}, which uses $\cG^{\sB\hat{\sA}\hat{\sB}\sG}$ and $\cG_{\rm inv}^{\sB\hat{\sA}\hat{\sB}\sG}$ a total of $u$ times.
Here, $\sD$ consists of the systems $\hat{\sA}\hat{\sB}\sG$ and one qubit added by the QSVT.
We denote by $\til{W}^{\sB\sD}$ and $\til{W}_{\rm ideal}^{\sB\sD}$ the unitaries obtained by the same QSVT with $W^{\sB\hat{\sA}\hat{\sB}\sG}$ and $W_{\rm ideal}^{\sB\hat{\sA}\hat{\sB}\sG}$, respectively, instead of $\cG^{\sB\hat{\sA}\hat{\sB}\sG}$, and their inverses instead of $\cG_{\rm inv}^{\sB\hat{\sA}\hat{\sB}\sG}$.
They are block-encodings of $P_\sgn^{(\rm SV)}(\til{M}^\sB)$ and $P_\sgn^{(\rm SV)}(M^\sB)$, respectively, where $P_\sgn$ is the real odd polynomial in \Cref{prop:FPAA general} or \Cref{prop:QSVT sign Fejer}.
From Eqs.~\eqref{inteq:79}, \eqref{inteq:79 inv}, and~\eqref{inteq:W bias}, we see that
\begin{align}
\label{inteq:78}
    \f{1}{2}\big\|\cJ^{\sB\sD} - \til{\cW}^{\sB\sD}\big\|_\diamond \leq 2u\delta_1, \ \text{and} \ \ 
    \big\|\til{W}^{\sB\sD} - \til{W}_{\rm ideal}^{\sB\sD}\big\|_\infty \leq 2u\epsilon_1,
\end{align}
where we used the subadditivity of the diamond norm and that $\|U_1U_2 - V_1V_2\|_\infty \leq \|U_1 - V_1\|_\infty + \|U_2 - V_2\|_\infty$ for unitaries $U_1$, $U_2$, $V_1$, and $V_2$.
The channel $\cJ^{\sB\sD}$ uses $uf$ samples of $\rho^{\hat{\sA}}$ and $\sigma^{\hat{\sA}}$ in total.
The above discussion applies to both cases where the error is measured in the diamond norm and the fidelity difference.


\subsubsection{Evaluation in the diamond norm}
\label{sec:mixed sample diamond}

When we evaluate the performance of this algorithm in the diamond norm, we use $\mathtt{QSVTSIGN}$ with $\beta = s_{\rm min}$.
Since all the non-zero singular values of $M^\sB$ are at least $s_{\rm min}$, as in Eq.~\eqref{inteq:65}, the polynomial $P_\sgn$ satisfies $\big\|P_\sgn^{(\rm SV)}(M^\sB) - \sgn^{(\rm SV)}(M^\sB)\big\|_\infty \leq \delta_3$ for $\delta_3 \in (0, 1/2)$.
We then apply \Cref{lem:FPAA diamond} with $d_\sE = 1$ to $\til{W}_{\rm ideal}^{\sB\sD}$ and obtain that
\begin{align}
\label{inteq:149}
    \f{1}{2}\big\|\til{\cW}_{\rm ideal}^{\sB\sD}\circ\cP_{\ket{0}}^{\bC\rarr\sD} - \cU_{\rm ideal}^{\sB\sD}\circ\cP_{\ket{0}}^{\bC\rarr\sD}\big\|_\diamond \leq 3\sqrt{\delta_3},
\end{align}
where $U_{\rm ideal}^{\sB\sD}$ is an exact block-encoding unitary of $\sgn^{(\rm SV)}(M^\sB)$, which is the Uhlmann partial isometry $V^\sB$ from $\ket{\rho_{\rm c}}^{\sA\sB}$ to $\ket{\sigma_{\rm c}}^{\sA\sB}$ (see \Cref{prop:explicit form of Uhl}).
Then, we have that 
\begin{align}
    &\f{1}{2}\big\|\cJ^{\sB\sD}\circ\cP_{\ket{0}}^{\bC\rarr\sD} - \cU_{\rm ideal}^{\sB\sD}\circ\cP_{\ket{0}}^{\bC\rarr\sD}\big\|_\diamond \notag\\
    &\leq \f{1}{2}\big\|\cJ^{\sB\sD} - \til{\cW}^{\sB\sD}\big\|_\diamond + \f{1}{2}\big\|\til{\cW}^{\sB\sD} - \til{\cW}_{\rm ideal}^{\sB\sD}\big\|_\diamond \notag\\
    &\hspace{3pc}+ \f{1}{2}\big\|\til{\cW}_{\rm ideal}^{\sB\sD}\circ\cP_{\ket{0}}^{\bC\rarr\sD} - \cU_{\rm ideal}^{\sB\sD}\circ\cP_{\ket{0}}^{\bC\rarr\sD}\big\|_\diamond \\
    \label{inteq:131}
    &\leq 2u\delta_1 + 2u\epsilon_1 + 3\sqrt{\delta_3},
\end{align}
where $u=\cO\big(\log{(1/\delta_3)}/s_{\rm min}\big)$. 
We used Eqs.~\eqref{inteq:78} and~\eqref{inteq:149}, and the relation between the diamond norm and the operator norm in Eq.~\eqref{eq:rel of operator and diamond norm} in the last line.

By rescaling these parameters, such as $\delta_3 = (\delta/9)^2$ and $\delta_1 = \epsilon_1 = \delta/(6u)$ as in \Cref{alg:Uhl mixed sample}, the quantum channel $\cJ_\diamond^{\sB\sD}$ in Theorem~\ref{thm:Uhlmann mixed state sample} satisfies
\begin{align}
    &\f{1}{2}\big\|\cJ_\diamond^{\sB\sD}\circ\cP_{\ket{0}}^{\bC\rarr\sD} - \cU_{\rm ideal}^{\sB\sD}\circ\cP_{\ket{0}}^{\bC\rarr\sD}\big\|_\diamond \leq \delta.
\end{align}

The number of samples of $\rho^{\hat{\sA}}$ and $\sigma^{\hat{\sA}}$ is given by
\begin{align}
\label{inteq:45}
    \zeta_\diamond 
    &= uf \\
    &= \cO\Big(\f{1}{s_{\rm min}}\log{(1/\delta_3)}\Big)\til{\cO}\Big(\f{d_\sA}{\delta_1}\min\Big\{\f{1}{\rho_{\rm min}^2} + \f{1}{\sigma_{\rm min}^2}, \f{r_\rho^2 + r_\sigma^2}{\epsilon_1^4}\Big\}\Big)    \\
    \label{inteq:44}
    &= \til{\cO}\Big(\f{d_\sA}{\delta s_{\rm min}^2} \min\Big\{\f{1}{\rho_{\rm min}^2} + \f{1}{\sigma_{\rm min}^2}, \f{r_\rho^2 + r_\sigma^2}{\delta^4 s_{\rm min}^4}\Big\}\Big).
\end{align}

Each use of $\cG^{\sB\hat{\sA}\hat{\sB}\sG}$ or $\cG_{\rm inv}^{\sB\hat{\sA}\hat{\sB}\sG}$ consists of $\cO(f\log{d_\sA})$ one- and two-qubit gates by Theorem~\ref{prop:canonical purification}, and the QSVT uses $\cO\big(u\log{(d_\sA d_\sB)}\big)$ gates in addition.
Hence, the total number of one- and two-qubit gates is $\cO\big((uf + u)\log{d_\sA}\big) = \cO\big(\zeta_\diamond\log{d_\sA}\big)$.
Note that $\cO\big(\log{(d_\sA d_\sB)}\big) = \cO(\log{d_\sA})$, as $d_\sA = d_\sB$.
The quantum circuit uses $\cO(\log{d_\sA})$ qubits at any one time since this algorithm is performed sequentially. The circuit description can be computed classically in $\poly\big(\zeta_\diamond, \log d_\sA, \log(1/\delta)\big)$ time, as this algorithm is based on the QSVT and the density matrix exponentiation.


\subsubsection{Evaluation in the fidelity difference}
\label{sec:uhlfidsamp}

Evaluation of the performance of the algorithm in the fidelity difference is almost identical to that in \Cref{sec:uhlfidquery}.
The difference is that $\til{W}^{\sB\sD}$ is a block-encoding of $P_\sgn^{(\rm SV)}(\til{M}^\sB)$ rather than $P_\sgn^{(\rm SV)}(M^\sB)$, where the singular values of $\til{M}^\sB$ may differ from those of $\sqrt{\sigma^\sA}\sqrt{\rho^\sA}$.
We evaluate the effect of this difference by the following lemma.

\begin{lemma}[Robustness of the overlap against errors in purified states]
\label{lem:overlap purified error}
For pure states $\ket{\phi}^{\sE\sB}$, $\ket{\psi}^{\sE\sB}$, $\ket{\til{\phi}}^{\sE\sB}$, and $\ket{\til{\psi}}^{\sE\sB}$ such that $\|\ket{\til{\phi}}^{\sE\sB} - \ket{\phi}^{\sE\sB}\| \leq \epsilon$ and $\|\ket{\til{\psi}}^{\sE\sB} - \ket{\psi}^{\sE\sB}\| \leq \epsilon$ for $\epsilon \in [0, 1)$, define $M^\sB = \tr_\sE\big[\ketbra{\psi}{\phi}^{\sE\sB}\big]$ with rank $R$ and $\til{M}^\sB = \tr_\sE\big[\ketbra{\til{\psi}}{\til{\phi}}^{\sE\sB}\big]$. Let $P$ be a real odd polynomial with $|P(x)| \leq 1$ for $|x| \in [0, 1]$, and suppose that one of the following holds:
\begin{enumerate}
\item $|x|\big|P(x) - \sgn(x)\big| \leq \beta$ for $|x| \in [0, 1]$ with $\beta \in (0, 1]$, in which case we set $\eta = R\beta$;
\item all the non-zero singular values of $M^\sB$ are at least $s \in (0, 1]$, and $\big|P(x) - \sgn(x)\big| \leq \delta$ for $|x| \in [s/2, 1]$ with $\delta \in (0, 1)$, in which case we set $\eta = \delta$.
\end{enumerate}
Then, for any unitary $U^{\sB\sD}$ that is a block-encoding of $P^{(\rm SV)}(\til{M}^\sB)$, it holds that
\begin{align}
\label{inteq:166}
    &\Big|\big|\bra{\til{\psi}}^{\sE\sB}\bra{0}^\sD U^{\sB\sD}\ket{\til{\phi}}^{\sE\sB}\ket{0}^\sD\big| - \|M^\sB\|_1\Big| \leq \eta + 6\epsilon, \\
\label{inteq:167}
    &\Big|\big|\bra{\psi}^{\sE\sB}\bra{0}^\sD U^{\sB\sD}\ket{\phi}^{\sE\sB}\ket{0}^\sD\big| - \|M^\sB\|_1\Big| \leq \eta + 8\epsilon.
\end{align}
\end{lemma}

This lemma follows from a perturbation bound on the singular values in the trace norm, which does not require the singular values of $\til{M}^\sB$ to remain above $s$.

\begin{proof}[Proof of \Cref{lem:overlap purified error}]
We first show that the singular values of $\til{M}^\sB$ are close to those of $M^\sB$.
For any vectors $\ket{x}^{\sE\sB}$ and $\ket{y}^{\sE\sB}$, the variational characterization of the trace norm gives
\begin{align}
    \big\|\tr_\sE\big[\ketbra{y}{x}^{\sE\sB}\big]\big\|_1 
    &= \max_{V^\sB}\Big|\tr\Big[V^\sB\tr_\sE\big[\ketbra{y}{x}^{\sE\sB}\big]\Big]\Big| \\
    \label{inteq:partial trace norm}
    &\leq \big\|\ket{x}^{\sE\sB}\big\|\big\|\ket{y}^{\sE\sB}\big\|,
\end{align}
where the maximum is taken over all unitaries $V^\sB$.
Then, we have
\begin{align}
    \big\|\til{M}^\sB - M^\sB\big\|_1
    &= \Big\|\tr_\sE\Big[\big(\ket{\til{\psi}}^{\sE\sB} - \ket{\psi}^{\sE\sB}\big)\bra{\til{\phi}}^{\sE\sB}\Big] \notag\\
    &\hspace{2pc}+ \tr_\sE\Big[\ket{\psi}^{\sE\sB}\big(\bra{\til{\phi}}^{\sE\sB} - \bra{\phi}^{\sE\sB}\big)\Big]\Big\|_1 \\
    &\leq \big\|\ket{\til{\psi}}^{\sE\sB} - \ket{\psi}^{\sE\sB}\big\| + \big\|\ket{\til{\phi}}^{\sE\sB} - \ket{\phi}^{\sE\sB}\big\| \\
    &\leq 2\epsilon,
\end{align}
where we used the triangle inequality and Eq.~\eqref{inteq:partial trace norm} in the second line. Equation~\eqref{inteq:partial trace norm} also implies that $\|\til{M}^\sB\|_1 \leq 1$.
Let $\{s_k\}_k$ and $\{\til{s}_k\}_k$ be the singular values of $M^\sB$ and $\til{M}^\sB$, respectively, both arranged in non-increasing order and including zeros.
Mirsky's inequality~\cite{mirsky1960symmetric,bhatia2013matrix} for the trace norm gives
\begin{align}
\label{inteq:mirsky}
    \sum_k\big|\til{s}_k - s_k\big| \leq \big\|\til{M}^\sB - M^\sB\big\|_1 \leq 2\epsilon.
\end{align}

Set 
\begin{align}
    &z = \bra{\psi}^{\sE\sB}\bra{0}^\sD U^{\sB\sD}\ket{\phi}^{\sE\sB}\ket{0}^\sD, \\
    &\til{z} = \bra{\til{\psi}}^{\sE\sB}\bra{0}^\sD U^{\sB\sD}\ket{\til{\phi}}^{\sE\sB}\ket{0}^\sD. 
\end{align}
Since $U^{\sB\sD}$ is a block-encoding of $P^{(\rm SV)}(\til{M}^\sB)$ and $(\til{M}^\sB)^\dag = \tr_\sE\big[\ketbra{\til{\phi}}{\til{\psi}}^{\sE\sB}\big]$, we have
\begin{align}
    \til{z} 
    &= \bra{\til{\psi}}^{\sE\sB}P^{(\rm SV)}(\til{M}^\sB)\ket{\til{\phi}}^{\sE\sB} \\
    &= \tr\big[P^{(\rm SV)}(\til{M}^\sB)(\til{M}^\sB)^\dag\big] \\
    &= \sum_k P(\til{s}_k)\til{s}_k.
\end{align}
Since $\|M^\sB\|_1 = \sum_k s_k$, it holds that
\begin{align}
    \big||\til{z}| - \|M^\sB\|_1\big| 
    &\leq \Big|\til{z} - \sum_k \til{s}_k\Big| + \Big|\sum_k \til{s}_k - \sum_k s_k\Big| \\
    \label{inteq:overlap split}
    &\leq \sum_k \big|1 - P(\til{s}_k)\big|\til{s}_k + 2\epsilon,
\end{align}
where we used Eq.~\eqref{inteq:mirsky}.
Note that $\big|1 - P(\til{s}_k)\big| \leq 2$, since $|P(x)| \leq 1$ for $|x| \in [0, 1]$ and $\til{s}_k \leq \|\til{M}^\sB\|_1 \leq 1$.
In the case of item~1, $\big|1 - P(\til{s}_k)\big|\til{s}_k \leq \beta$ holds for each $k$. For $k > R$, we have $s_k = 0$, and thus $\big|1 - P(\til{s}_k)\big|\til{s}_k \leq 2\til{s}_k = 2|\til{s}_k - s_k|$. Hence, from Eq.~\eqref{inteq:mirsky},
\begin{align}
    \sum_k \big|1 - P(\til{s}_k)\big|\til{s}_k 
    \leq R\beta + 2\sum_{k > R}\big|\til{s}_k - s_k\big| 
    \leq R\beta + 4\epsilon.
\end{align}
In the case of item~2, $\big|1 - P(\til{s}_k)\big| \leq \delta$ holds for $k$ such that $\til{s}_k \geq s/2$. For $k$ such that $\til{s}_k < s/2$, we have $\til{s}_k \leq |\til{s}_k - s_k|$, since either $s_k \geq s$ or $s_k = 0$. Hence, from Eq.~\eqref{inteq:mirsky},
\begin{align}
    \sum_k \big|1 - P(\til{s}_k)\big|\til{s}_k 
    &\leq \delta\sum_{k: \til{s}_k \geq s/2}\til{s}_k + 2\sum_{k: \til{s}_k < s/2}\big|\til{s}_k - s_k\big| \\
    &\leq \delta + 4\epsilon,
\end{align}
where we used $\sum_k \til{s}_k = \|\til{M}^\sB\|_1 \leq 1$.
Combining them with Eq.~\eqref{inteq:overlap split}, we obtain the bounds for $\til{z}$.
Finally, since $U^{\sB\sD}$ is unitary, $|z - \til{z}| \leq \|\ket{\psi}^{\sE\sB} - \ket{\til{\psi}}^{\sE\sB}\| + \|\ket{\phi}^{\sE\sB} - \ket{\til{\phi}}^{\sE\sB}\| \leq 2\epsilon$, which gives the bounds for $z$.

\end{proof}

We apply \Cref{lem:overlap purified error} with $\sE \gets \sA\sG'$, $\ket{\phi}^{\sE\sB} \gets \ket{\rho_{\rm c}}^{\sA\sB}\ket{0}^{\sG'}$, $\ket{\psi}^{\sE\sB} \gets \ket{\sigma_{\rm c}}^{\sA\sB}\ket{0}^{\sG'}$, $\ket{\til{\phi}}^{\sE\sB} \gets \ket{\til{\rho}_{\rm c}}^{\sA\sB\sG'}$, $\ket{\til{\psi}}^{\sE\sB} \gets \ket{\til{\sigma}_{\rm c}}^{\sA\sB\sG'}$, and $U^{\sB\sD} \gets \til{W}^{\sB\sD}$, for which $M^\sB$ and $\til{M}^\sB$ are given by Eqs.~\eqref{inteq:M ideal} and~\eqref{inteq:tilde M}, respectively. Here, $\|M^\sB\|_1 = \sqrt{\rF}(\rho^\sA, \sigma^\sA)$, and the condition on $\epsilon = \epsilon_1$ is satisfied by Eqs.~\eqref{inteq:tilde states} and~\eqref{inteq:tilde states sigma}.
When we use $\mathtt{QSVTSIGNFEJER}$ with $\beta = \delta/(8R)$ or $\mathtt{QSVTSIGN}$ with $\delta_3 = \delta/8$ and $\beta = s_{\rm min}/2$, the conditions of item~1 or item~2 are satisfied by \Cref{prop:QSVT sign Fejer} or \Cref{prop:FPAA general}, respectively.
In both cases, we have $\eta = \delta/8$. From Eq.~\eqref{inteq:167}, the square root fidelity
\begin{align}
    \sqrt{\rF'} 
    &= \sqrt{\rF}\big(\til{W}^{\sB\sD}\ket{\rho_{\rm c}}^{\sA\sB}\ket{0}^{\sD}, \ket{\sigma_{\rm c}}^{\sA\sB}\ket{0}^{\sD}\big) \\
    &= \big|\bra{\sigma_{\rm c}}^{\sA\sB}\bra{0}^{\sD}\til{W}^{\sB\sD}\ket{\rho_{\rm c}}^{\sA\sB}\ket{0}^{\sD}\big|
\end{align}
satisfies
\begin{align}
\label{inteq:77}
    \big|\sqrt{\rF}(\rho^\sA, \sigma^\sA) - \sqrt{\rF'}\big| \leq \f{\delta}{8} + 8\epsilon_1 = \f{\delta}{4},
\end{align}
where we used $\epsilon_1 = \delta/64$ in \Cref{alg:Uhl mixed sample}.

Let $\cT^{\sB}= \tr_{\sD}\circ\cJ^{\sB\sD}\circ\cP_{\ket{0}}^{\bC\rarr\sD}$. We obtain that
\begin{align}
    &\rF\big(\cT^{\sB}(\ketbra{\rho_{\rm c}}{\rho_{\rm c}}^{\sA\sB}), \ket{\sigma_{\rm c}}^{\sA\sB}\big) \notag\\
    &\geq \rF\big(\cJ^{\sB\sD}\circ\cP_{\ket{0}}^{\bC\rarr\sD}(\ketbra{\rho_{\rm c}}{\rho_{\rm c}}^{\sA\sB}), \ket{\sigma_{\rm c}}^{\sA\sB}\ket{0}^\sD\big) \\
    &\geq \rF' - \f{1}{2}\big\|\cJ^{\sB\sD} - \til{\cW}^{\sB\sD}\big\|_\diamond \\
    &\geq \rF(\rho^\sA, \sigma^\sA) - 2\big|\sqrt{\rF}(\rho^\sA, \sigma^\sA) - \sqrt{\rF'}\big| - 2u\delta_1 \\
    &\geq \rF(\rho^\sA, \sigma^\sA) - \delta.
\end{align}
In the first inequality, we used the monotonicity of the fidelity under the partial trace. In the second inequality, we used the relation between the diamond norm and the fidelity in Eq.~\eqref{eq:fidelity and diamond}. In the third inequality, we used $|x-y|\leq 2 |\sqrt{x} - \sqrt{y}|$ for $x, y \in [0, 1]$ and Eq.~\eqref{inteq:78}. We also used Eq.~\eqref{inteq:77} and $\delta_1 = \delta/(4u)$ in \Cref{alg:Uhl mixed sample} in the last inequality.

The number of uses of $\cG^{\sB\hat{\sA}\hat{\sB}\sG}$ and $\cG_{\rm inv}^{\sB\hat{\sA}\hat{\sB}\sG}$ is $u = \cO\big(\log{(1/\delta)}/s_{\rm min}\big)$ for $\mathtt{QSVTSIGN}$ and $u = \cO(R/\delta)$ for $\mathtt{QSVTSIGNFEJER}$.
Since \Cref{alg:Uhl mixed sample} uses the one of them that requires fewer uses, the number of samples $\zeta_\tF$ of $\rho^{\hat{\sA}}$ and $\sigma^{\hat{\sA}}$ is given by
\begin{align}
    \zeta_\tF 
    &= uf \\
    &= \cO(u)\til{\cO}\Big(\f{d_\sA u}{\delta}\min\Big\{\f{1}{\rho_{\rm min}^2} + \f{1}{\sigma_{\rm min}^2}, \f{r_\rho^2 + r_\sigma^2}{\delta^4}\Big\}\Big) \\
    &= \til{\cO}\Big(\f{d_\sA}{\delta}\min\Big\{\f{1}{s_{\rm min}^2}, \f{R^2}{\delta^2}\Big\}\min\Big\{\f{1}{\rho_{\rm min}^2} + \f{1}{\sigma_{\rm min}^2}, \f{r_\rho^2 + r_\sigma^2}{\delta^4}\Big\}\Big),
\end{align}
where we used Eq.~\eqref{inteq:43} with $\delta_1 = \delta/(4u)$ and $\epsilon_1 = \delta/64$.
As in \Cref{sec:mixed sample diamond}, the quantum circuit consists of $\cO\big(\zeta_\tF \log{d_\sA}\big)$ one- and two-qubit gates, and uses $\cO\big(\log{d_\sA}\big)$ qubits at any one time. The circuit description can be computed classically in $\poly\big(\zeta_\tF, \log d_\sA, \log(1/\delta)\big)$ time.
This completes the proof of \Cref{thm:Uhlmann mixed state sample}.

\section{Lower bounds for the Uhlmann transformation}
\label{sec:lower bound Uhlmann}

We prove \Cref{thm:lower Uhlmann}, which provides lower bounds on the query and sample complexities for the Uhlmann transformation in the three computational models.
The lower bounds hold for any quantum algorithm, with no restriction on its computation, including collective operations on the samples. We restate \Cref{thm:lower Uhlmann} here.

\thmloweruhlmann*

These lower bounds hold in the worst case over the states, and in the purified query and sample access models, our algorithms attain them up to logarithmic factors (see \Cref{cor:tight Uhlmann}).
Since $Rs_{\rm min} \leq \sqrt{\rF}(\rho^\sA, \sigma^\sA) \leq 1$, the rank $R$ of $\sqrt{\sigma^\sA}\sqrt{\rho^\sA}$ is at most $\min\{r, 1/s\}$ for states of rank at most $r$ with $s_{\rm min} \geq s$, and replacing $R$ and $s_{\rm min}$ with $\min\{r, 1/s\}$ and $s$ in \Cref{infthm:uhlmenn purif query,infthm:uhlmenn purif sample} yields upper bounds matching the first and second statements.
In the mixed sample access model, the pure states with fidelity $1/2$ used for the term $\Omega(d_\sA/\delta)$ have $\rho_{\rm min} = \sigma_{\rm min} = 1$ and $s_{\rm min} = 1/\sqrt{2}$, for which \Cref{infthm:uhlmenn mixed sample} uses $\til{\cO}(d_\sA/\delta)$ samples.
Hence, \Cref{infthm:uhlmenn mixed sample} is optimal in $d_\sA$ and $\delta$ for these states, and its linear dependence on $d_\sA$ cannot be replaced by the rank or the minimum non-zero eigenvalues of the states.

Our proofs follow a common approach. In each model, we construct a family of pairs of states for which the required fidelity forces the algorithm to recover unknown information contained in the queries or samples, and then bound the average fidelity of this recovery over the family. In \Cref{sec:lower bound purified}, we prove the bounds in the purified query and sample access models, based on a lower bound for recovering an unknown unitary from samples in which it appears with a small amplitude. In \Cref{sec:lower bound mixed}, we prove the bound in the mixed sample access model, where the corresponding tasks are the complex conjugation of an unknown pure state and the approximation of an unknown isometry.


\subsection{In the purified query and sample access models}
\label{sec:lower bound purified}

We first show in \Cref{sec:unknown unitary recovery} a lemma that bounds how well an unknown unitary can be recovered from samples in which it appears with a small amplitude, and then use it in \Cref{sec:lower bound purified proof} to prove the first and second statements of \Cref{thm:lower Uhlmann}.

\subsubsection{Recovery of an unknown unitary}
\label{sec:unknown unitary recovery}

Let $d \geq 2$, and let $\sA$ and $\sB$ be $(d+1)$-dimensional systems with the computational basis $\{\ket{j}\}_{j=0}^d$.
We denote by $\ket{\Phi}^{\sA\sB} = \f{1}{\sqrt{d}}\sum_{j=1}^d\ket{j}^\sA\ket{j}^\sB$ the maximally entangled state on the $d$-dimensional subspace spanned by $\{\ket{j}\}_{j=1}^d$.
Let $U^\sB$ be a unitary on $\sB$ that acts as a unitary $U$ on this subspace and trivially on $\ket{0}^\sB$, and let $\ket{\Phi_U}^{\sA\sB} = U^\sB\ket{\Phi}^{\sA\sB}$.
For $p \in (0, 1]$, we define
\begin{align}
\label{inteq:lower weighted state}
    \ket{\rho_{U,p}}^{\sA\sB} = \sqrt{1-p}\ket{0}^\sA\ket{0}^\sB + \sqrt{p}\ket{\Phi_U}^{\sA\sB}.
\end{align}
In this state, the unitary $U$ appears with amplitude $\sqrt{p}$.
The following lemma bounds how well $\ket{\Phi}^{\sA\sB}$ can be recovered from $\ket{\Phi_U}^{\sA\sB}$ by acting on $\sB$ without knowing $U$, when samples of $\ket{\rho_{U,p}}^{\sA\sB}$ are available.

\begin{lemma}
\label{lem:weighted unitary recovery}
Let $d \geq 2$, $p \in (0, 1]$, and $n \geq 0$.
Suppose that, given $n$ samples of $\ket{\rho_{U,p}}^{\sA'\sB'}$, a quantum algorithm realizes a quantum channel $\cT_U^\sB$, which depends on $U$ only through the samples.
Then, we have 
\begin{align}
\label{inteq:lower weighted recovery}
    \bE_U\Big[\rF\big(\cT_U^\sB(\ketbra{\Phi_U}{\Phi_U}^{\sA\sB}), \ket{\Phi}^{\sA\sB}\big)\Big] \leq 1 - \f{d^2-1}{4np + 2d^2},
\end{align}
where $\bE_U[\cdot]$ denotes the expectation over a Haar random unitary $U$ on the $d$-dimensional subspace.
\end{lemma}

To prove this lemma, we first reduce the problem to the estimation of $U$, and then bound the average fidelity of the estimation by a Cram\'er--Rao-type argument on the unitary group. Here, the Fisher information about $U$ obtained from the samples is limited by the amplitude $\sqrt{p}$.

\begin{proof}[Proof of \Cref{lem:weighted unitary recovery}]

We reduce the task to quantum learning of an unknown unitary~\cite{Bisio2010Optimalquantumlearningunitary}, in which a unitary is applied to a known state, and the resulting state is later used to apply the unitary or its inverse to a new input. In our case, the samples $\ket{\rho_{U,p}}^{\otimes n} = (U^{\sB'})^{\otimes n}\ket{\rho_{\bI,p}}^{\otimes n}$ are obtained by applying $U^{\sB'}$ to $n$ copies of the known state $\ket{\rho_{\bI,p}}$, and the new input is $\ket{\Phi_U}$, to which $(U^\sB)^\dag$ is to be applied. The results of Ref.~\cite{Bisio2010Optimalquantumlearningunitary} concern channels on the space where the unitary acts. To use them, we note that we may assume the output of $\cT_U^\sB$ to be supported on the subspace spanned by $\{\ket{j}\}_{j=1}^d$. Indeed, after $\cT_U^\sB$, we can measure the output with $\{\ketbra{0}{0}^\sB, \bI^\sB - \ketbra{0}{0}^\sB\}$. When we obtain the first outcome, we replace the resulting state with a fixed state in the subspace; when we obtain the second outcome, we keep the resulting state, which is already in the subspace. Since $\ket{\Phi}^{\sA\sB}$ is supported on the subspace, this $U$-independent operation does not decrease the fidelity.

As shown in Ref.~\cite[Theorem~1 and Lemma~2]{Bisio2010Optimalquantumlearningunitary}, when the action of the unitary commutes with the reduced state of the known state, the average fidelity is maximized by estimating the unitary with a measurement and then applying the corresponding operation. This also holds for the inverse and for different representations on the known state and on the new input~\cite[Secs.~III~B and~IV]{Bisio2010Optimalquantumlearningunitary}. In our case, the reduced state $\big((1-p)\ketbra{0}{0} + (p/d)\sum_{j=1}^d\ketbra{j}{j}\big)^{\otimes n}$ of the known state $\ket{\rho_{\bI,p}}^{\otimes n}$ on the $n$ copies of $\sB'$ commutes with $(U^{\sB'})^{\otimes n}$. Thus, for some POVM $\{E_V\}_{V\in\rU(d)}$ satisfying $\int\mathrm{d}V E_V = \bI$ on the span of the states $\ket{\rho_{U,p}}^{\otimes n}$ for all $U \in \rU(d)$, we have
\begin{align}
    &\bE_U\Big[\rF\big(\cT_U^\sB(\ketbra{\Phi_U}{\Phi_U}^{\sA\sB}), \ket{\Phi}^{\sA\sB}\big)\Big] \notag\\
    \label{inteq:lower reduction}
    &\leq \bE_U\Big[\int\mathrm{d}V\bra{\rho_{U,p}}^{\otimes n}E_V\ket{\rho_{U,p}}^{\otimes n}\rF\big((V^\sB)^\dag\ket{\Phi_U}^{\sA\sB}, \ket{\Phi}^{\sA\sB}\big)\Big],
\end{align}
where $\rU(d)$ is the unitary group on the $d$-dimensional subspace, and $\mathrm{d}V$ is the normalized Haar measure on it. The left-hand side is the average fidelity of $\cT_U^\sB$, and the right-hand side is that of the procedure that measures the samples with $\{E_V\}_{V\in\rU(d)}$ and then applies $(V^\sB)^\dag$.

We bound the average fidelity of this procedure, i.e., the right-hand side of Eq.~\eqref{inteq:lower reduction}, by using the derivatives of the fidelity and of the measurement probabilities with respect to $U$. Since the fidelity is high only when $V$ is close to $U$, a high average fidelity requires the measurement probabilities to vary rapidly with $U$.
Let $\ket{\psi_U} = \ket{\rho_{U,p}}^{\otimes n}$, and define
\begin{align}
    &\Pr(V|U) = \bra{\psi_U}E_V\ket{\psi_U}, \\
    &f_V(U) = \rF\big((V^\sB)^\dag\ket{\Phi_U}^{\sA\sB}, \ket{\Phi}^{\sA\sB}\big) = \f{\big|\tr[UV^\dag]\big|^2}{d^2}, \\
    &\bar{f} = \bE_U\Big[\int\mathrm{d}V\Pr(V|U)f_V(U)\Big].
\end{align}

We first compute the derivatives of the fidelity $f_V(U)$. 
Let $\{H_a\}_{a=1}^{d^2-1}$ be traceless Hermitian operators on the $d$-dimensional subspace, orthonormal with respect to the Hilbert--Schmidt inner product.
For a function of $U$, such as $f_V(U)$, we write $\partial_a f_V(U) = \f{\mathrm{d}}{\mathrm{d}t}f_V(e^{itH_a}U)\big|_{t=0}$.
Since $\{H_a\}_a$ and $\bI/\sqrt{d}$ together form an orthonormal basis of the operators on the $d$-dimensional subspace, we have
\begin{align}
    &\sum_a H_a^2 = \f{d^2-1}{d}\bI, \\
    &\sum_a\big|\tr[H_aUV^\dag]\big|^2 = d - \f{\big|\tr[UV^\dag]\big|^2}{d}.
\end{align}
Using these relations, $\partial_a\tr[UV^\dag] = i\tr[H_aUV^\dag]$, and $\partial_a^2\tr[UV^\dag] = -\tr[H_a^2UV^\dag]$, we have
\begin{align}
    &\sum_a\partial_a^2\big|\tr[UV^\dag]\big|^2 \notag\\
    &= \sum_a\Big(2\mathrm{Re}\big(\tr[UV^\dag]^*\partial_a^2\tr[UV^\dag]\big) + 2\big|\partial_a\tr[UV^\dag]\big|^2\Big) \\
    &= -\f{2(d^2-1)}{d}\big|\tr[UV^\dag]\big|^2 + 2\Big(d - \f{\big|\tr[UV^\dag]\big|^2}{d}\Big) \\
    &= 2d - 2d\big|\tr[UV^\dag]\big|^2,
\end{align}
where $\mathrm{Re}(x)$ denotes the real part of $x$.
Since $\partial_a\big|\tr[UV^\dag]\big|^2 = 2{\rm Re}\big(i\tr[UV^\dag]^*\tr[H_aUV^\dag]\big)$, we also have $\big|\partial_a\big|\tr[UV^\dag]\big|^2\big| \leq 2\big|\tr[UV^\dag]\big|\big|\tr[H_aUV^\dag]\big|$, and thus
\begin{align}
    &\sum_a\Big|\partial_a\big|\tr[UV^\dag]\big|^2\Big|^2 \notag\\
    &\leq 4\big|\tr[UV^\dag]\big|^2\sum_a\big|\tr[H_aUV^\dag]\big|^2 \\
    &= 4\big|\tr[UV^\dag]\big|^2\Big(d - \f{\big|\tr[UV^\dag]\big|^2}{d}\Big).
\end{align}
Dividing these by $d^2$ and $d^4$, respectively, we obtain
\begin{align}
\label{inteq:lower Laplacian}
    &\sum_a\partial_a^2f_V(U) = -2d\Big(f_V(U) - \f{1}{d^2}\Big), \\
\label{inteq:lower gradient}
    &\sum_a|\partial_af_V(U)|^2 \leq \f{4}{d}f_V(U)\big(1-f_V(U)\big) \leq \f{4}{d}\big(1 - f_V(U)\big).
\end{align}

We next bound the derivatives of $\Pr(V|U)$.
For each sample, we have $\ket{\partial_a\rho_{U,p}}^{\sA'\sB'} = i\sqrt{p}H_a^{\sB'}\ket{\Phi_U}^{\sA'\sB'}$.
Since the reduced state of $\ket{\Phi_U}^{\sA'\sB'}$ on $\sB'$ is the completely mixed state on the $d$-dimensional subspace and $H_a$ is traceless with $\tr[H_a^2] = 1$,
\begin{align}
    &\braket{\rho_{U,p}}{\partial_a\rho_{U,p}} = \f{ip}{d}\tr[H_a] = 0, \\
    &\big\|\ket{\partial_a\rho_{U,p}}\big\|^2 = \f{p}{d}\tr[H_a^2] = \f{p}{d}.
\end{align}
By the product rule, $\ket{\partial_a\psi_U}$ is the sum of $n$ terms, in each of which one of the $n$ copies of $\ket{\rho_{U,p}}$ is replaced with $\ket{\partial_a\rho_{U,p}}$.
In $\big\|\ket{\partial_a\psi_U}\big\|^2$, the inner product of two different terms is $\big|\braket{\rho_{U,p}}{\partial_a\rho_{U,p}}\big|^2 = 0$, and that of each term with itself is $p/d$.
Summing over $a$, we obtain
\begin{align}
\label{inteq:lower sample derivative}
    \sum_a\big\|\ket{\partial_a\psi_U}\big\|^2 = \f{np(d^2-1)}{d}.
\end{align}
By the Cauchy--Schwarz inequality for $E_V \geq 0$, we have
\begin{align}
    |\partial_a\Pr(V|U)|^2 
    &= 4\big({\rm Re}\bra{\partial_a\psi_U}E_V\ket{\psi_U}\big)^2 \\
    &\leq 4\big|\bra{\partial_a\psi_U}E_V\ket{\psi_U}\big|^2 \\
    \label{eq:partial pro bound}
    &\leq 4\Pr(V|U)\bra{\partial_a\psi_U}E_V\ket{\partial_a\psi_U}.
\end{align}
Since $\ket{\partial_a\psi_U}$ lies in the span of the states $\ket{\psi_U}$, integrating over $V$ and using Eqs.~\eqref{inteq:lower sample derivative} and~\eqref{eq:partial pro bound} yields
\begin{align}
    &\bE_U\Big[\int\mathrm{d}V\sum_a\f{|\partial_a\Pr(V|U)|^2}{\Pr(V|U)}\Big] \notag\\
    &\leq 4\bE_U\Big[\int\mathrm{d}V\sum_a\bra{\partial_a\psi_U}E_V\ket{\partial_a\psi_U}\Big] \\
    &= 4\bE_U\Big[\sum_a\big\|\ket{\partial_a\psi_U}\big\|^2\Big] \\
    \label{inteq:lower measurement derivative}
    &= \f{4np(d^2-1)}{d},
\end{align}
where a term with $\Pr(V|U) = 0$ is regarded as zero.
The left-hand side of Eq.~\eqref{inteq:lower measurement derivative} is the averaged Fisher information of the measurement outcome, and this inequality shows that it is bounded by the quantum Fisher information $4\sum_a\big\|\ket{\partial_a\psi_U}\big\|^2$ of the samples~\cite{Braunstein1994Statisticaldistancegeometry}.

We finally relate $\bar{f}$ to these derivatives through integration by parts on the unitary group.
The invariance of the Haar measure implies 
\begin{align}
    \bE_U\big[\Pr(V|U)\partial_a^2f_V(U)\big] 
    = -\bE_U\big[\partial_a\Pr(V|U)\partial_af_V(U)\big].
\end{align}
Summing over $a$, integrating over $V$, and using Eq.~\eqref{inteq:lower Laplacian} and $\int\mathrm{d}V\Pr(V|U) = 1$, we have
\begin{align}
    2d\Big(\bar{f} - \f{1}{d^2}\Big) = \bE_U\Big[\int\mathrm{d}V\sum_a\partial_a\Pr(V|U)\partial_af_V(U)\Big].
\end{align}
Applying the Cauchy--Schwarz inequality to the right-hand side, together with Eqs.~\eqref{inteq:lower gradient} and~\eqref{inteq:lower measurement derivative}, we obtain
\begin{align}
    &4d^2\Big(\bar{f} - \f{1}{d^2}\Big)^2 \notag\\
    &\leq \bE_U\Big[\int\mathrm{d}V\sum_a\f{|\partial_a\Pr(V|U)|^2}{\Pr(V|U)}\Big] \notag\\
    &\hspace{3pc}\times\bE_U\Big[\int\mathrm{d}V\Pr(V|U)\sum_a|\partial_af_V(U)|^2\Big] \\
    &\leq \f{16np(d^2-1)}{d^2}\bE_U\Big[\int\mathrm{d}V\Pr(V|U)\big(1 - f_V(U)\big)\Big] \\
    &= \f{16np(d^2-1)}{d^2}(1 - \bar{f}),
\end{align}
which is rewritten as
\begin{align}
    \big(d^2 - 1 - d^2(1-\bar{f})\big)^2 \leq 4np(d^2-1)(1-\bar{f}).
\end{align}
Expanding the left-hand side and dropping the non-negative term $d^4(1-\bar{f})^2$, we have $(d^2-1)^2 \leq (4np + 2d^2)(d^2-1)(1-\bar{f})$. Dividing both sides by $(4np + 2d^2)(d^2-1) > 0$ yields
\begin{align}
    \label{eq:average fidelity bound}
    \bar{f} \leq 1 - \f{d^2-1}{4np + 2d^2},
\end{align}
which holds for any POVM $\{E_V\}_{V\in\rU(d)}$.
Since the right-hand side of Eq.~\eqref{inteq:lower reduction} is $\bar{f}$ by definition, combining Eqs.~\eqref{inteq:lower reduction} and~\eqref{eq:average fidelity bound} proves Eq.~\eqref{inteq:lower weighted recovery}.

\end{proof}


\subsubsection{Proof of the lower bounds}
\label{sec:lower bound purified proof}

We now prove the first and second statements of \Cref{thm:lower Uhlmann}.
We construct pairs of states for which the condition in \Cref{thm:lower Uhlmann} forces any algorithm to recover $\ket{\Phi}^{\sA\sB}$ from $\ket{\Phi_U}^{\sA\sB}$, and then apply \Cref{lem:weighted unitary recovery}.
Let $\delta$, $s$, and $r$ be as in \Cref{thm:lower Uhlmann}, and let
\begin{align}
\label{inteq:lower parameters}
    d = \min\{r-1, \lfloor 1/(4s)\rfloor\}, \quad \sqrt{p} = \max\{16\delta, \sqrt{2}ds\},
\end{align}
which satisfy $d \geq 2$ and $16\delta \leq \sqrt{p} \leq 1/2$, where we used $\sqrt{2}ds \leq \sqrt{2}/4$.
We take $d_\sA = d_\sB = d+1$ and use the notation of \Cref{sec:unknown unitary recovery}; for larger dimensions, the following states are embedded into the larger systems, which is possible since $d + 1 \leq r \leq \min\{d_\sA, d_\sB\}$.
We consider the pair of states $\rho_{U,p}^\sA$ and $\sigma^\sA$ with the purified states $\ket{\rho_{U,p}}^{\sA\sB}$ in Eq.~\eqref{inteq:lower weighted state} and
\begin{align}
    \ket{\sigma}^{\sA\sB} = \f{1}{\sqrt{2}}\big(\ket{0}^\sA\ket{0}^\sB + \ket{\Phi}^{\sA\sB}\big).
\end{align}
Writing $\Pi = \sum_{j=1}^d\ketbra{j}{j}$, their reduced states are $\rho_{U,p}^\sA = (1-p)\ketbra{0}{0} + (p/d)\Pi$ and $\sigma^\sA = \f{1}{2}\ketbra{0}{0} + \f{1}{2d}\Pi$.
Both have rank $d + 1 \leq r$ and are independent of $U$, while $U$ appears in the purified state $\ket{\rho_{U,p}}^{\sA\sB}$ only through the component $\ket{\Phi_U}^{\sA\sB}$ with amplitude $\sqrt{p}$.
Their fidelity is 
\begin{align}
\label{eq:fidelity on oneside}
    \rF(\rho_{U,p}^\sA, \sigma^\sA) = \f{1}{2}\big(\sqrt{1-p} + \sqrt{p}\big)^2.
\end{align}
Since $\sqrt{\sigma^\sA}\sqrt{\rho_{U,p}^\sA} = \sqrt{(1-p)/2}\ketbra{0}{0} + \big(\sqrt{p}/(\sqrt{2}d)\big)\Pi$, and $\sqrt{p}/d \leq \sqrt{1-p}$, we also have $s_{\rm min} = \sqrt{p}/(\sqrt{2}d) \geq s$ by Eq.~\eqref{inteq:lower parameters}.
Thus, this pair satisfies the conditions in \Cref{thm:lower Uhlmann}.

Let $\cT_U^\sB$ be the quantum channel realized by an algorithm satisfying the condition in \Cref{thm:lower Uhlmann}, when its queries or samples correspond to $\ket{\rho_{U,p}}^{\sA\sB}$ and $\ket{\sigma}^{\sA\sB}$.
We show that the condition forces this channel to recover $\ket{\Phi}^{\sA\sB}$ from $\ket{\Phi_U}^{\sA\sB}$, i.e., to make
\begin{align}
    g_U = \rF\big(\cT_U^\sB(\ketbra{\Phi_U}{\Phi_U}^{\sA\sB}), \ket{\Phi}^{\sA\sB}\big)
\end{align}
close to one.
Let $V_U^{\sB\rarr\sB\sD}$ be a Stinespring isometry of $\cT_U^\sB$, and define vectors $\ket{\mu_U}^\sD = (\bra{0}^\sB\otimes\bI^\sD)V_U^{\sB\rarr\sB\sD}\ket{0}^\sB$ and $\ket{\nu_U}^\sD = (\bra{\Phi}^{\sA\sB}\otimes\bI^\sD)V_U^{\sB\rarr\sB\sD}\ket{\Phi_U}^{\sA\sB}$, which satisfy $\|\ket{\mu_U}\| \leq 1$ and $\|\ket{\nu_U}\|^2 = g_U$.
Since $\ket{0}^\sA$ is orthogonal to the support of $\ket{\Phi_U}^{\sA\sB}$ and $\ket{\Phi}^{\sA\sB}$ on $\sA$, the component $\ket{0}^\sA\ket{0}^\sB$ of $\ket{\rho_{U,p}}^{\sA\sB}$ overlaps after $\cT_U^\sB$ only with that of $\ket{\sigma}^{\sA\sB}$, and $\ket{\Phi_U}^{\sA\sB}$ only with $\ket{\Phi}^{\sA\sB}$.
We thus have
\begin{align}
    &\rF\big(\cT_U^\sB(\ketbra{\rho_{U,p}}{\rho_{U,p}}^{\sA\sB}), \ket{\sigma}^{\sA\sB}\big) \notag\\
    &\hspace{1pc}= \f{1}{2}\big\|\sqrt{1-p}\ket{\mu_U} + \sqrt{p}\ket{\nu_U}\big\|^2 \\
    &\hspace{1pc}\leq \f{1}{2}\big(\sqrt{1-p} + \sqrt{pg_U}\big)^2.
\end{align}
Combining this with Eq.~\eqref{eq:fidelity on oneside}, the fidelity condition in \Cref{thm:lower Uhlmann} implies
\begin{align}
    \delta 
    &\geq \f{1}{2}\big(\sqrt{1-p} + \sqrt{p}\big)^2 - \f{1}{2}\big(\sqrt{1-p} + \sqrt{pg_U}\big)^2 \\
    &\geq \sqrt{p(1-p)}\big(1 - \sqrt{g_U}\big) \\
    &\geq \f{\sqrt{p}(1-g_U)}{4},
\end{align}
where the second inequality follows from $g_U \leq 1$, and the third from $\sqrt{1-p} \geq 1/2$ and $1 - \sqrt{g_U} \geq (1 - g_U)/2$.
Hence, for every $U$, we obtain 
\begin{align}
\label{inteq:lower forced recovery}
    1 - g_U \leq \f{4\delta}{\sqrt{p}}.
\end{align}

In the purified sample access model, the algorithm uses $n$ samples of $\ket{\rho_{U,p}}^{\sA\sB}$, while the samples of $\ket{\sigma}^{\sA\sB}$ can be prepared without any information about $U$.
The channel $\cT_U^\sB$ thus satisfies the assumption of \Cref{lem:weighted unitary recovery}.
Averaging Eq.~\eqref{inteq:lower forced recovery} over $U$ and comparing it with \Cref{lem:weighted unitary recovery}, we obtain $(d^2-1)/(4np + 2d^2) \leq 4\delta/\sqrt{p}$, i.e., 
\begin{align}
    4np 
    &\geq \f{(d^2-1)\sqrt{p}}{4\delta} - 2d^2 \\
    &\geq \f{3d^2\sqrt{p}}{16\delta} - \f{2d^2\sqrt{p}}{16\delta} \\
    &= \f{d^2\sqrt{p}}{16\delta},
\end{align}
where we used $d^2 - 1 \geq 3d^2/4$ for $d \geq 2$ and $\sqrt{p} \geq 16\delta$ in the second inequality.
We thus obtain
\begin{align}
\label{inteq:lower sample bound}
    n \geq \f{d^2}{64\delta\sqrt{p}}.
\end{align}

In the purified query access model, we simulate the queries with samples and then apply \Cref{lem:weighted unitary recovery}.
To this end, we use a unitary preparing $\ket{\rho_{U,p}}^{\sA\sB}$ whose queries consist of reflections about a state of the same form.
Let $p' = (1 - \sqrt{1-p})/2$, which satisfies $p' \leq p/2$, and define
\begin{align}
    W_U^{\sA\sB} = \big(2\ketbra{\rho_{U,p'}}{\rho_{U,p'}}^{\sA\sB} - \bI^{\sA\sB}\big)\big(2\ketbra{0}{0}^\sA\otimes\ketbra{0}{0}^\sB - \bI^{\sA\sB}\big).
\end{align}
Since $\ket{0}^\sA\ket{0}^\sB$ is orthogonal to $\ket{\Phi_U}^{\sA\sB}$, we have 
\begin{align}
    W_U^{\sA\sB}\ket{0}^\sA\ket{0}^\sB 
    &= (1-2p')\ket{0}^\sA\ket{0}^\sB + 2\sqrt{p'(1-p')}\ket{\Phi_U}^{\sA\sB} \\
    &= \ket{\rho_{U,p}}^{\sA\sB},
\end{align}
where we used $1 - 2p' = \sqrt{1-p}$.
As the condition in \Cref{thm:lower Uhlmann} holds for every unitary preparing $\ket{\rho_{U,p}}^{\sA\sB}$, we may take $U_\rho^{\sA\sB} = W_U^{\sA\sB}$, while $U_\sigma^{\sA\sB}$ can be implemented without any information about $U$.
Each query to $W_U^{\sA\sB}$ or $(W_U^{\sA\sB})^\dag$ uses the reflection $2\ketbra{\rho_{U,p'}}{\rho_{U,p'}} - \bI = -e^{i\pi\ketbra{\rho_{U,p'}}{\rho_{U,p'}}}$ once, together with a known reflection.

Suppose that the algorithm uses $q \geq 1$ queries.
We implement the exponential $e^{i\pi\ketbra{\rho_{U,p'}}{\rho_{U,p'}}}$ in each query by the density matrix exponentiation in \Cref{prevthm:cntrl-DME} with $t = \pi$, within error $\delta/(q\sqrt{p})$ measured by one half of the diamond norm, using $m = \lceil4\pi^2q\sqrt{p}/\delta\rceil$ samples of $\ket{\rho_{U,p'}}$.
The resulting algorithm uses no queries and $n = qm \leq 8\pi^2q^2\sqrt{p}/\delta$ samples of $\ket{\rho_{U,p'}}$, where we used $4\pi^2q\sqrt{p}/\delta \geq 1$.
Since the errors of the $q$ exponentials add up to at most $\delta/\sqrt{p}$, Eq.~\eqref{eq:fidelity and diamond} implies that the channel of the algorithm satisfies $1 - g_U \leq 5\delta/\sqrt{p}$ for every $U$ in place of Eq.~\eqref{inteq:lower forced recovery}, where $g_U$ is now defined with this channel.
\Cref{lem:weighted unitary recovery} with $p'$ in place of $p$ then yields $(d^2-1)/(4np' + 2d^2) \leq 5\delta/\sqrt{p}$, i.e.,
\begin{align}
    4np' 
    &\geq \f{(d^2-1)\sqrt{p}}{5\delta} - 2d^2 \\
    &\geq \f{3d^2\sqrt{p}}{20\delta} - \f{2d^2\sqrt{p}}{16\delta} \\
    &= \f{d^2\sqrt{p}}{40\delta},
\end{align}
where we used $d^2 - 1 \geq 3d^2/4$ and $\sqrt{p} \geq 16\delta$ in the second inequality, as in the purified sample access model.
Since $p' \leq p/2$, we obtain $n \geq d^2/(80\delta\sqrt{p})$.
Combining it with $n \leq 8\pi^2q^2\sqrt{p}/\delta$, we have
\begin{align}
\label{inteq:lower query bound}
    q \geq \f{d}{8\sqrt{10}\pi\sqrt{p}} \geq \f{d}{80\sqrt{p}}.
\end{align}
An algorithm with $q = 0$ is excluded by \Cref{lem:weighted unitary recovery} with $n = 0$, which yields $\bE_U[g_U] \leq 1 - (d^2-1)/(2d^2) \leq 5/8$, contradicting Eq.~\eqref{inteq:lower forced recovery} since $4\delta/\sqrt{p} \leq 1/4$ by $\sqrt{p} \geq 16\delta$.

We finally substitute Eq.~\eqref{inteq:lower parameters} into Eqs.~\eqref{inteq:lower sample bound} and~\eqref{inteq:lower query bound}.
Since $r - 1 \geq 2r/3$ and $\lfloor 1/(4s)\rfloor \geq 1/(8s)$ by $1/(4s) \geq 2$, we have $d \geq \min\{2r/3, 1/(8s)\}$.
Together with $1/\sqrt{p} = \min\{1/(16\delta), 1/(\sqrt{2}ds)\}$ and $\delta \leq 1/32$, we obtain
\begin{align}
    n &\geq \min\Big\{\f{d^2}{1024\delta^2}, \f{d}{64\sqrt{2}\delta s}\Big\} \\
    &\geq \min\Big\{\f{r^2}{2304\delta^2}, \f{r}{96\sqrt{2}\delta s}, \f{1}{2048\delta s^2}\Big\},
\end{align}
and
\begin{align}
    q \geq \min\Big\{\f{d}{1280\delta}, \f{1}{80\sqrt{2}s}\Big\} \geq \min\Big\{\f{r}{1920\delta}, \f{1}{320s}\Big\},
\end{align}
which prove the second and first statements, respectively.

The lower bound $\Omega\big(\min\{r/\delta, 1/s\}\big)$ in the first statement of \Cref{thm:lower Uhlmann} also holds when the algorithm may use, in addition to $U_\rho^{\sA\sB}$, $U_\sigma^{\sA\sB}$, and their inverses, their complex conjugates and transposes, as well as the controlled versions of all these unitaries. A controlled query to $W_U^{\sA\sB}$ or $(W_U^{\sA\sB})^\dag$ uses the controlled reflection
\begin{align}
    &\ketbra{0}{0}^\sC\otimes\bI^{\sA\sB} + \ketbra{1}{1}^\sC\otimes\big(2\ketbra{\rho_{U,p'}}{\rho_{U,p'}}^{\sA\sB} - \bI^{\sA\sB}\big) \notag\\
    &= Z^\sC e^{i\pi\ketbra{1}{1}^\sC\otimes\ketbra{\rho_{U,p'}}{\rho_{U,p'}}^{\sA\sB}},
\end{align}
where $Z^\sC$ is the Pauli-$Z$ operator on a control qubit $\sC$. This exponential is implemented by the density matrix exponentiation in \Cref{prevthm:cntrl-DME} with samples of $\ket{1}^\sC\ket{\rho_{U,p'}}^{\sA\sB}$, and the rest of the proof is unchanged.

For the complex conjugate and the transpose, let $\ket{\rho_{U,p}^*}^{\sA\sB}$ denote the complex conjugate of $\ket{\rho_{U,p}}^{\sA\sB}$ in the computational basis. Since $2\ketbra{0}{0}^\sA\otimes\ketbra{0}{0}^\sB - \bI^{\sA\sB}$ is real, we have
\begin{align}
    &W_U^* = \big(2\ketbra{\rho_{U,p'}^*}{\rho_{U,p'}^*} - \bI\big)\big(2\ketbra{0}{0}\otimes\ketbra{0}{0} - \bI\big), \\
    &W_U^\top = \big(2\ketbra{0}{0}\otimes\ketbra{0}{0} - \bI\big)\big(2\ketbra{\rho_{U,p'}^*}{\rho_{U,p'}^*} - \bI\big),
\end{align}
and these queries and their controlled versions are simulated in the same way with samples of $\ket{\rho_{U,p'}^*}$. It thus suffices that \Cref{lem:weighted unitary recovery} holds when $n_-$ of the $n$ samples are replaced with samples of $\ket{\rho_{U,p}^*}$. Because $\ket{\rho_{\bI,p}}$ is real, these samples are $(U^*)^{\sB'}\ket{\rho_{\bI,p}}$, and the reduced state of the known state commutes with $(U^{\sB'})^{\otimes(n-n_-)}\otimes((U^*)^{\sB'})^{\otimes n_-}$, which consists of different representations of the unitary group. Hence, by the extension to different representations and to the recovery of the inverse~\cite[Secs.~III~B and~IV]{Bisio2010Optimalquantumlearningunitary}, the reduction in Eq.~\eqref{inteq:lower reduction} holds with $\ket{\rho_{U,p}}^{\otimes n}$ replaced by the state of these samples. Moreover, $\ket{\partial_a\rho_{U,p}^*} = -i\sqrt{p}(H_a^*)^{\sB'}(U^*)^{\sB'}\ket{\Phi}^{\sA'\sB'}$ satisfies $\braket{\rho_{U,p}^*}{\partial_a\rho_{U,p}^*} = 0$ and $\|\ket{\partial_a\rho_{U,p}^*}\|^2 = p/d$, as $H_a^*$ is also traceless with $\tr[(H_a^*)^2] = 1$. With these properties, Eq.~\eqref{inteq:lower sample derivative} remains valid, and the rest of the proof is unchanged.

One may wonder why we do not take a smaller $p$, which would make the lower bound on $n$ from \Cref{lem:weighted unitary recovery} larger. For such $p$, the condition in \Cref{thm:lower Uhlmann} no longer forces the recovery: the channel that replaces the state of $\sB$ with $\ket{0}^\sB$ gives $g_U = 0$, while $\rF\big(\rho_{U,p}^\sA\otimes\ketbra{0}{0}^\sB, \ket{\sigma}^{\sA\sB}\big) = (1-p)/2$ is smaller than $\rF(\rho_{U,p}^\sA, \sigma^\sA)$ only by $\sqrt{p(1-p)} + p/2$, which is at most $\delta$ for $p = o(\delta^2)$. Hence, $\sqrt{p}$ of order $\delta$ is the smallest one for which this construction yields a lower bound.

For comparison, setting $p = 1$, i.e., $\ket{\rho_{U,1}}^{\sA\sB} = \ket{\Phi_U}^{\sA\sB}$ and $\ket{\sigma}^{\sA\sB} = \ket{\Phi}^{\sA\sB}$, the same argument gives $\Omega(r)$ queries for any $\delta \in (0, 1/4]$.
A lower bound linear in the rank with the same types of queries was also obtained in Ref.~\cite[Theorem~5.4]{cochran2026blackholeradiationdecoding} for a different family of states, whereas our bound also captures the dependence on $\delta$ and $s$.


\subsection{In the mixed sample access model}
\label{sec:lower bound mixed}

We prove the third statement of \Cref{thm:lower Uhlmann}.
We use a family of pairs of states, each consisting of an unknown pure state and a known component orthogonal to it, and obtain the two terms in the third statement from different choices of its parameters.
Without the known component, the Uhlmann transformation requires preparing the complex conjugate of the unknown pure state from its samples, which leads to the term $\Omega(d_\sA/\delta)$.
With the known component, the transformation must also reproduce a phase in the unknown state relative to the known component, which leads to the term $\Omega\big(1/(\delta s)\big)$, as this phase appears in the samples only relative to a small component of amplitude about $s$.
We first show in \Cref{sec:conjugation and isometry} two lemmas that bound how well these tasks can be achieved from samples of the unknown pure state, and then use them in \Cref{sec:lower bound mixed proof} to prove the third statement.

\subsubsection{Complex conjugation and a phase-dependent isometry from samples}
\label{sec:conjugation and isometry}

Let $D \geq 1$, and let $\ket{v}$ be a unit vector in the $D$-dimensional subspace spanned by $\{\ket{j}\}_{j=1}^D$. Let $\ket{\bot}$ be a computational basis vector orthogonal to $\ket{0}, \ldots, \ket{D}$, and for $t \in (0, 1)$, we define $\ket{\psi_{t,v}} = t\ket{0} + \sqrt{1-t^2}\ket{v}$, whose complex conjugate in the computational basis is denoted by $\ket{\psi_{t,v}^*}$. The following lemma bounds how well $\ket{\psi_{t,v}^*}$ can be prepared from samples of $\ket{\psi_{t,v}}$.

\begin{lemma}
\label{lem:conjugation from samples}
Let $D \geq 2$, $n \geq 0$, and $t = 1/\sqrt{2}$, and let $\sB$ be a $(D+1)$-dimensional system with the computational basis $\{\ket{j}\}_{j=0}^D$.
Suppose that, given $n$ samples of $\ket{\psi_{t,v}}$, a quantum algorithm outputs a state $\omega_v^\sB$, which depends on $\ket{v}$ only through the samples.
Then, we have
\begin{align}
\label{inteq:lower conjugation bound}
    \bE_v\big[\rF\big(\omega_v^\sB, \ket{\psi_{t,v}^*}^\sB\big)\big] \leq 1 - \f{D-1}{2(n+D)}.
\end{align}
The expectation $\bE_v[\cdot]$ is taken over $\ket{v}$ uniformly distributed on the unit sphere of the $D$-dimensional subspace, where vectors that differ only in the global phase are distinguished.
\end{lemma}

To prove this lemma, we first bound the average overlap between $\ket{v^*}$ and the output of a completely positive map given $l$ copies of $\ket{v}$. Since the state of the $n$ samples is a superposition of terms with $l$ copies of $\ket{v}$ for $l = 0, \ldots, n$, we then combine this bound over $l$.

\begin{proof}[Proof of \Cref{lem:conjugation from samples}]
We first show the following bound, which for quantum channels coincides with the optimal fidelity of pure-state transposition~\cite{Buscemi2003OptimalRealizationTranspositionMaps, brzic2026optimalphysicalapproximationspurestate}.
We use the identity $\bE_v\big[\ketbra{v}{v}^{\otimes l}\big] = \Pi_{\rm sym}^{(l)}\big/\binom{l+D-1}{l}$, where $\Pi_{\rm sym}^{(l)}$ is the projector onto the symmetric subspace of $l$ copies of the $D$-dimensional subspace.
Let $\cA$ be a completely positive map from $l$ copies of the $D$-dimensional subspace to the $D$-dimensional subspace, and let $J_\cA$ be its unnormalized Choi operator in the computational basis, which satisfies $\tr[J_\cA(X^\top\otimes Y)] = \tr[\cA(X)Y]$ for any operators $X$ and $Y$.
Since $(\ketbra{v}{v}^{\otimes l})^\top = \ketbra{v^*}{v^*}^{\otimes l}$ and $(\Pi_{\rm sym}^{(l)})^\top = \Pi_{\rm sym}^{(l)}$, we have $\bra{v^*}\cA(\ketbra{v}{v}^{\otimes l})\ket{v^*} = \tr\big[J_\cA\ketbra{v^*}{v^*}^{\otimes (l+1)}\big]$ and $\tr\big[\cA(\Pi_{\rm sym}^{(l)})\big] = \tr\big[J_\cA(\Pi_{\rm sym}^{(l)}\otimes\bI)\big]$.
As $\ket{v^*}$ is also uniformly distributed, the identity with $l+1$ copies yields
\begin{align}
    \bE_v\big[\bra{v^*}\cA(\ketbra{v}{v}^{\otimes l})\ket{v^*}\big] 
    &= \f{\tr\big[J_\cA\Pi_{\rm sym}^{(l+1)}\big]}{\binom{l+D}{l+1}} \\
    &\leq \f{\tr\big[J_\cA(\Pi_{\rm sym}^{(l)}\otimes\bI)\big]}{\binom{l+D}{l+1}} \\
    \label{inteq:lower conjugation CP}
    &= \f{l+1}{l+D}\bE_v\big[\tr\big[\cA(\ketbra{v}{v}^{\otimes l})\big]\big],
\end{align}
where we used $J_\cA \geq 0$ and $\Pi_{\rm sym}^{(l+1)} \leq \Pi_{\rm sym}^{(l)}\otimes\bI$ in the inequality and the identity for $l$ in the last line.
The factor $(l+1)/(l+D)$ shows that $\ket{v^*}$ cannot be produced with high fidelity unless $l$ is comparable to $D$.

We now apply Eq.~\eqref{inteq:lower conjugation CP} to the output $\omega_v$.
Let $\cE$ be the quantum channel independent of $\ket{v}$ such that $\omega_v = \cE(\ketbra{\psi_{t,v}}{\psi_{t,v}}^{\otimes n})$.
Since $\ket{v}$ is orthogonal to $\ket{0}$ and $t = 1/\sqrt{2}$, the state of the samples is expanded according to the number $l$ of copies of $\ket{v}$ as
\begin{align}
    \ket{\psi_{t,v}}^{\otimes n} = \sum_{l=0}^n\sqrt{2^{-n}\tbinom{n}{l}} V_l\ket{v}^{\otimes l},
\end{align}
where $V_l$ is the isometry on $l$ copies of the $D$-dimensional subspace that inserts $n - l$ copies of $\ket{0}$ symmetrically.
Substituting this expansion, we have
\begin{align}
    \omega_v = \sum_{l,l'=0}^n 2^{-n}\sqrt{\tbinom{n}{l}\tbinom{n}{l'}}\cE\big(V_l\ket{v}^{\otimes l}\bra{v}^{\otimes l'}V_{l'}^\dag\big).
\end{align}
Let $\bar{\Pi} = \bI - \ketbra{0}{0}$, and define 
\begin{align}
    w_0 = \bE_v\big[\tr[\bar{\Pi}\omega_v]\big], \ \text{and} \ \ 
    w_1 = \bE_v\big[\bra{v^*}\omega_v\ket{v^*}\big],
\end{align}
which are the average weight of $\omega_v$ outside $\ket{0}$ and its average overlap with $\ket{v^*}$, respectively.
Under $\ket{v} \mapsto e^{i\theta}\ket{v}$, the term with $l$ and $l'$ acquires the phase $e^{i(l-l')\theta}$, while $\bar{\Pi}$ and $\ketbra{v^*}{v^*}$ are invariant.
As the distribution of $\ket{v}$ is invariant under this map, averaging over $\theta$ removes the cross terms with $l \neq l'$, and we obtain
\begin{align}
    &w_0 = \sum_{l=0}^n 2^{-n}\tbinom{n}{l}\bE_v\big[\tr\big[\bar{\Pi}\cE(V_l\ketbra{v}{v}^{\otimes l}V_l^\dag)\big]\big], \\
    &w_1 = \sum_{l=0}^n 2^{-n}\tbinom{n}{l}\bE_v\big[\bra{v^*}\bar{\Pi}\cE(V_l\ketbra{v}{v}^{\otimes l}V_l^\dag)\bar{\Pi}\ket{v^*}\big],
\end{align}
where we used $\bar{\Pi}\ket{v^*} = \ket{v^*}$ in $w_1$.
Thus, Eq.~\eqref{inteq:lower conjugation CP} with $\cA(\cdot) = \bar{\Pi}\cE(V_l\cdot V_l^\dag)\bar{\Pi}$ yields
\begin{align}
    w_1 
    &\leq \sum_{l=0}^n 2^{-n}\tbinom{n}{l}\f{l+1}{l+D}\bE_v\big[\tr\big[\bar{\Pi}\cE(V_l\ketbra{v}{v}^{\otimes l}V_l^\dag)\big]\big] \\
    &\leq \f{n+1}{n+D}w_0,
\end{align}
where the last inequality holds because $(l+1)/(l+D)$ is increasing in $l$ and the sum without this factor equals $w_0$.

Finally, we bound the fidelity with $w_0$ and $w_1$.
Since $\ket{\psi_{t,v}^*} = (\ket{0} + \ket{v^*})/\sqrt{2}$, $\bra{0}\omega_v\ket{0} = 1 - \tr[\bar{\Pi}\omega_v]$, and $\omega_v \geq 0$ implies $|\bra{0}\omega_v\ket{v^*}|^2 \leq \bra{0}\omega_v\ket{0}\bra{v^*}\omega_v\ket{v^*}$, we have 
\begin{align}
    \rF\big(\omega_v, \ket{\psi_{t,v}^*}\big) 
    &= \f{1}{2}\big(\bra{0}\omega_v\ket{0} + \bra{v^*}\omega_v\ket{v^*} \notag\\
    &\hspace{6pc}+ 2\mathrm{Re}(\bra{0}\omega_v\ket{v^*})\big) \\
    &\leq \f{1}{2}\big(\sqrt{1 - \tr[\bar{\Pi}\omega_v]} + \sqrt{\bra{v^*}\omega_v\ket{v^*}}\big)^2.
\end{align}
Averaging over $\ket{v}$ and using the concavity of the map $(x, y) \mapsto (\sqrt{x} + \sqrt{y})^2$, we obtain
\begin{align}
    \bE_v\big[\rF\big(\omega_v, \ket{\psi_{t,v}^*}\big)\big]
    &\leq \f{1}{2}\big(\sqrt{1-w_0} + \sqrt{w_1}\big)^2 \\
    &\leq \f{1}{2}\Big(\sqrt{1-w_0} + \sqrt{\tfrac{n+1}{n+D}w_0}\Big)^2 \\
    &\leq \f{1}{2}\Big(1 + \f{n+1}{n+D}\Big) \\
    &= 1 - \f{D-1}{2(n+D)},
\end{align}
where the third inequality follows from the Cauchy--Schwarz inequality.

\end{proof}

We next consider $D = 1$ and $\ket{v} = e^{-i\theta}\ket{1}$ with $\theta \in [0, 2\pi)$, and write the subscript $\theta$ in place of $v$, e.g., $\ket{\psi_{t,\theta}} = t\ket{0} + \sqrt{1-t^2}e^{-i\theta}\ket{1}$.
Let $\sB$ be a system whose computational basis includes $\ket{0}$, $\ket{1}$, and $\ket{\bot}$, and let $V_\theta^\sB$ be the partial isometry on $\sB$ supported on the span of $\ket{\bot}^\sB$ and $\ket{0}^\sB$ such that
\begin{align}
    &V_\theta^\sB\ket{\bot}^\sB = \ket{\bot}^\sB, \\ 
    &V_\theta^\sB\ket{0}^\sB = \ket{\psi_{t,\theta}^*}^\sB = t\ket{0}^\sB + \sqrt{1-t^2}e^{i\theta}\ket{1}^\sB.
\end{align}
We also let $\ket{\Phi_2}^{\sR\sB} = \big(\ket{\bot}^\sR\ket{\bot}^\sB + \ket{0}^\sR\ket{0}^\sB\big)/\sqrt{2}$ with a reference system $\sR$.
The following lemma bounds the number of samples of $\ket{\psi_{t,\theta}}$ required to realize $V_\theta^\sB$ approximately.

\begin{lemma}
\label{lem:phase isometry from samples}
Let $t \in (0, 1/2]$, $\epsilon \in (0, 2^{-12}]$, and $n \geq 0$.
Suppose that, given $n$ samples of $\ket{\psi_{t,\theta}}$, a quantum algorithm realizes a channel $\cT_\theta^\sB$, which depends on $\theta$ only through the samples, such that
\begin{align}
\label{inteq:mixed Choi fidelity}
    \rF\big(\cT_\theta^\sB(\ketbra{\Phi_2}{\Phi_2}^{\sR\sB}), V_\theta^\sB\ket{\Phi_2}^{\sR\sB}\big) \geq 1 - \epsilon,
\end{align}
for every $\theta \in [0, 2\pi)$.
Then, we have $n \geq \f{3\ln 2}{2^{13}t^2\epsilon}$.
\end{lemma}

To prove this lemma, we relate the isometries realized for different $\theta$ through the samples, as in the no-programming theorem~\cite{Nielsen1997ProgrammableQuantumGateArrays}. We then show that this relation is compatible with Eq.~\eqref{inteq:mixed Choi fidelity} only for large $n$.

\begin{proof}[Proof of \Cref{lem:phase isometry from samples}]
We fix a Stinespring isometry $W^{\sB\sP\rarr\sB\sD}$ of the quantum channel independent of $\theta$ that realizes $\cT_\theta^\sB$ from the samples, where $\sP$ is the system of the $n$ samples.
Then, $W_\theta^{\sB\rarr\sB\sD} = W^{\sB\sP\rarr\sB\sD}\big(\bI^\sB\otimes(\ket{\psi_{t,\theta}}^{\otimes n})^\sP\big)$ is a Stinespring isometry of $\cT_\theta^\sB$.
Since $(W^{\sB\sP\rarr\sB\sD})^\dag W^{\sB\sP\rarr\sB\sD} = \bI^{\sB\sP}$, we have
\begin{align}
\label{inteq:mixed inner product}
    (W_0^{\sB\rarr\sB\sD})^\dag W_\theta^{\sB\rarr\sB\sD} = \gamma_\theta\bI^\sB, 
\end{align}
where $\gamma_\theta = \braket{\psi_{t,0}}{\psi_{t,\theta}}^n$. If $\theta$ is small, then $|\gamma_\theta|$ is close to one, and $W_\theta^{\sB\rarr\sB\sD}$ is close to $\gamma_\theta W_0^{\sB\rarr\sB\sD}$.

On the other hand, Eq.~\eqref{inteq:mixed Choi fidelity} requires $W_\theta^{\sB\rarr\sB\sD}$, restricted to the span of $\ket{\bot}^\sB$ and $\ket{0}^\sB$, to be close to $V_\theta^\sB$ tensored with a vector on $\sD$.
Since $W_\theta^{\sB\rarr\sB\sD}$ is a Stinespring isometry of $\cT_\theta^\sB$, Eq.~\eqref{inteq:mixed Choi fidelity} reads
\begin{align}
    \big\|\big(\bra{\Phi_2}^{\sR\sB}(V_\theta^\sB)^\dag\otimes\bI^\sD\big)W_\theta^{\sB\rarr\sB\sD}\ket{\Phi_2}^{\sR\sB}\big\|^2 \geq 1 - \epsilon.
\end{align}
Let $\ket{\chi_\theta}^\sD$ be the unit vector proportional to the vector inside the norm on the left-hand side. Then, we have
\begin{align}
    &\big\|W_\theta^{\sB\rarr\sB\sD}\ket{\Phi_2}^{\sR\sB} - V_\theta^\sB\ket{\Phi_2}^{\sR\sB}\ket{\chi_\theta}^\sD\big\|^2 \notag\\
    &= 2 - 2\big\|\big(\bra{\Phi_2}^{\sR\sB}(V_\theta^\sB)^\dag\otimes\bI^\sD\big)W_\theta^{\sB\rarr\sB\sD}\ket{\Phi_2}^{\sR\sB}\big\| \\
    &\leq 2 - 2\sqrt{1-\epsilon} \\
    &\leq 2\epsilon.
\end{align}
We denote by $\Pi_2^\sB = \ketbra{\bot}{\bot}^\sB + \ketbra{0}{0}^\sB$ the projector onto the span of $\ket{\bot}^\sB$ and $\ket{0}^\sB$.
Since the left-hand side is half of the squared Hilbert--Schmidt norm of $\big(W_\theta^{\sB\rarr\sB\sD} - V_\theta^\sB\otimes\ket{\chi_\theta}^\sD\big)\Pi_2^\sB$, which bounds the squared operator norm, we obtain
\begin{align}
\label{inteq:mixed isometry error}
    \big\|\big(W_\theta^{\sB\rarr\sB\sD} - V_\theta^\sB\otimes\ket{\chi_\theta}^\sD\big)\Pi_2^\sB\big\|_\infty \leq 2\sqrt{\epsilon}.
\end{align}
Combining Eqs.~\eqref{inteq:mixed inner product} and~\eqref{inteq:mixed isometry error} for $\theta$ and $0$ yields 
\begin{align}
    &\big\|\gamma_\theta\Pi_2^\sB - \braket{\chi_0}{\chi_\theta}\Pi_2^\sB(V_0^\sB)^\dag V_\theta^\sB\Pi_2^\sB\big\|_\infty \notag\\
    &= \big\|\Pi_2^\sB(W_0^{\sB\rarr\sB\sD})^\dag\big(W_\theta^{\sB\rarr\sB\sD} - V_\theta^\sB\otimes\ket{\chi_\theta}^\sD\big)\Pi_2^\sB \notag\\
    \label{inteq:165}
    &\hspace{2pc}+ \Pi_2^\sB\big(W_0^{\sB\rarr\sB\sD} - V_0^\sB\otimes\ket{\chi_0}^\sD\big)^\dag\big(V_\theta^\sB\otimes\ket{\chi_\theta}^\sD\big)\Pi_2^\sB\big\|_\infty \\
    &\leq 4\sqrt{\epsilon},
\end{align}
where we used that the operator norms of $W_0^{\sB\rarr\sB\sD}$ and $V_\theta^\sB\otimes\ket{\chi_\theta}^\sD$ are at most one.
In the basis $\{\ket{\bot}, \ket{0}\}$, the operator $\Pi_2^\sB(V_0^\sB)^\dag V_\theta^\sB\Pi_2^\sB$ is diagonal with entries $1$ and $t^2 + (1-t^2)e^{i\theta}$, whose absolute values are at most one.
Thus, the operator in the norm on the left-hand side of Eq.~\eqref{inteq:165} is also diagonal in this basis, and the absolute value of each of its diagonal entries is at most $4\sqrt{\epsilon}$.
Comparing these two diagonal entries by the triangle inequality, we obtain
\begin{align}
    &|\gamma_\theta|(1-t^2)|1 - e^{i\theta}| \notag\\
    &= \big|\gamma_\theta - \gamma_\theta\big(t^2 + (1-t^2)e^{i\theta}\big)\big| \\ 
    &= \big|\gamma_\theta - \braket{\chi_0}{\chi_\theta}\big(t^2 + (1-t^2)e^{i\theta}\big) \notag\\
    &\hspace{4pc} - \big(t^2 + (1-t^2)e^{i\theta}\big)\big(\gamma_\theta - \braket{\chi_0}{\chi_\theta}\big)\big| \\
    &\leq \big|\gamma_\theta - \braket{\chi_0}{\chi_\theta}\big(t^2 + (1-t^2)e^{i\theta}\big)\big|
     + \big|\gamma_\theta - \braket{\chi_0}{\chi_\theta}\big| \\
\label{inteq:mixed overlap inequality}
    &\leq 8\sqrt{\epsilon}.
\end{align}

We finally choose $\theta = 64\sqrt{\epsilon} \leq 1$.
Since $1 - t^2 \geq 3/4$ and $|1 - e^{i\theta}| = 2\sin(\theta/2) \geq 2\theta/\pi$, Eq.~\eqref{inteq:mixed overlap inequality} implies
\begin{align}
    |\gamma_\theta| \leq \f{16\pi\sqrt{\epsilon}}{3\theta} = \f{\pi}{12} < \f{1}{2}.
\end{align}
On the other hand, $|\braket{\psi_{t,0}}{\psi_{t,\theta}}|^2 = 1 - x$ with
\begin{align}
    x = 4t^2(1-t^2)\sin^2(\theta/2),
\end{align}
which satisfies $x \leq t^2\theta^2 = 2^{12}t^2\epsilon \leq 1/4$.
Since $|\gamma_\theta|^2 = (1-x)^n < 1/4$ and $-\ln{(1-x)} \leq 4x/3$ for $x \in [0, 1/4]$, we obtain
\begin{align}
    n 
    = \f{\ln{|\gamma_\theta|^2}}{\ln{(1-x)}}
    > \f{2\ln 2}{-\ln{(1-x)}}
    \geq \f{3\ln 2}{2x}
    \geq \f{3\ln 2}{2^{13}t^2\epsilon}.
\end{align}

\end{proof}


\subsubsection{Proof of the lower bound}
\label{sec:lower bound mixed proof}

We now prove the third statement of \Cref{thm:lower Uhlmann}, using the notation of \Cref{sec:conjugation and isometry}.
For $p \in [0, 1)$ and $t \in (0, 1)$, we consider
\begin{align}
\label{inteq:mixed family}
    &\rho_p^\sA = p\ketbra{\bot}{\bot}^\sA + (1-p)\ketbra{0}{0}^\sA, \\
    &\sigma_{p,t,v}^\sA = p\ketbra{\bot}{\bot}^\sA + (1-p)\ketbra{\psi_{t,v}}{\psi_{t,v}}^\sA.
\end{align}
For $p = 0$, the terms with $\ket{\bot}$ vanish, and the states are supported on the subspace spanned by $\ket{0}, \ldots, \ket{D}$.
As $\ket{\psi_{t,v}}$ is orthogonal to $\ket{\bot}$, the canonical purified states are given by 
\begin{align}
    &\ket{\rho_{p,{\rm c}}}^{\sA\sB} = \sqrt{p}\ket{\bot}^\sA\ket{\bot}^\sB + \sqrt{1-p}\ket{0}^\sA\ket{0}^\sB, \\
    &\ket{\sigma_{p,t,v,{\rm c}}}^{\sA\sB} = \sqrt{p}\ket{\bot}^\sA\ket{\bot}^\sB + \sqrt{1-p}\ket{\psi_{t,v}}^\sA\ket{\psi_{t,v}^*}^\sB,
\end{align}
respectively. For these states, we have 
\begin{align}
    \sqrt{\sigma_{p,t,v}^\sA}\sqrt{\rho_p^\sA} = p\ketbra{\bot}{\bot}^\sA + (1-p)t\ketbra{\psi_{t,v}}{0}^\sA,
\end{align}
whose non-zero singular values are $p$ and $(1-p)t$; the first is absent for $p = 0$.
Thus, it holds that $\sqrt{\rF}(\rho_p^\sA, \sigma_{p,t,v}^\sA) = p + (1-p)t$.

Let $\cT_v^\sB$ be the quantum channel realized by an algorithm satisfying the condition in \Cref{thm:lower Uhlmann}, when its samples correspond to $\rho_p^\sA$ and $\sigma_{p,t,v}^\sA$.
The algorithm uses samples of $\rho_p^\sA$, which is known, and $n$ samples of $\sigma_{p,t,v}^\sA$, each of which can be prepared from a sample of $\ket{\psi_{t,v}}$ by the channel $X \mapsto p\tr[X]\ketbra{\bot}{\bot} + (1-p)X$. It suffices to bound $n$ for algorithms given $n$ samples of $\ket{\psi_{t,v}}$, for which $\cT_v^\sB$ depends on $\ket{v}$ only through the samples.

We first prove the term $\Omega(d_\sA/\delta)$ with $p = 0$, $t = 1/\sqrt{2}$, and $D = d_\sA - 1 \geq 2$, where we used $d_\sA \geq r \geq 3$.
In this case, $\rho_0^\sA = \ketbra{0}{0}^\sA$ and $\sigma_{0,t,v}^\sA = \ketbra{\psi_{t,v}}{\psi_{t,v}}^\sA$ are pure, their fidelity is $t^2 = 1/2$, and $s_{\rm min} = t = 1/\sqrt{2} \geq s$.
Since $\cT_v^\sB$ acts on $\sB$, its output for $\ket{\rho_{0,{\rm c}}}^{\sA\sB} = \ket{0}^\sA\ket{0}^\sB$ is $\ketbra{0}{0}^\sA\otimes\omega_v^\sB$ with $\omega_v^\sB = \cT_v^\sB(\ketbra{0}{0}^\sB)$.
As $|\braket{\psi_{t,v}}{0}|^2 = 1/2$, the condition in \Cref{thm:lower Uhlmann} implies
\begin{align}
    \f{1}{2}\rF\big(\omega_v^\sB, \ket{\psi_{t,v}^*}^\sB\big) 
    &= \rF\big(\ketbra{0}{0}^\sA\otimes\omega_v^\sB, \ket{\psi_{t,v}}^\sA\ket{\psi_{t,v}^*}^\sB\big) \\
    \label{inteq:lower conjugation requirement}
    &\geq \f{1}{2} - \delta,
\end{align}
for every $\ket{v}$, and the algorithm has to prepare the complex conjugate $\ket{\psi_{t,v}^*}$ of the unknown state $\ket{\psi_{t,v}}$ from its samples.
Since $\omega_v^\sB$ depends on $\ket{v}$ only through the samples, \Cref{lem:conjugation from samples} with $D = d_\sA - 1$ and Eq.~\eqref{inteq:lower conjugation requirement} yield $1 - 2\delta \leq 1 - (d_\sA-2)/\big(2(n+d_\sA-1)\big)$, i.e.,
\begin{align}
\label{inteq:mixed dimension bound}
    n \geq \f{d_\sA-2}{4\delta} - (d_\sA-1) \geq \f{d_\sA}{24\delta},
\end{align}
where we used $d_\sA \geq 3$ and $\delta \leq 1/32$ in the last inequality.

We next prove the term $\Omega\big(1/(\delta s)\big)$ with $p = 1/2$, $t = 2s \in (0, 1/2]$, $D = 1$, and $\ket{v} = e^{-i\theta}\ket{1}$.
In this case, the states are supported on the three-dimensional subspace spanned by $\ket{0}$, $\ket{1}$, and $\ket{\bot}$.
Both $\rho_{1/2}^\sA$ and $\sigma_{1/2,t,\theta}^\sA$ have rank two, which is at most $r$, and the non-zero singular values of $\sqrt{\sigma_{1/2,t,\theta}^\sA}\sqrt{\rho_{1/2}^\sA}$ are $1/2$ and $t/2 = s$, and thus $s_{\rm min} = s$.
Since 
\begin{align}
    \tr_\sA\big[\ketbra{\sigma_{1/2,t,\theta,{\rm c}}}{\rho_{1/2,{\rm c}}}^{\sA\sB}\big] 
    = \f{1}{2}\big(\ketbra{\bot}{\bot}^\sB + t\ketbra{\psi_{t,\theta}^*}{0}^\sB\big),
\end{align}
the Uhlmann partial isometry in Eq.~\eqref{eq:explicit form Uhl partial iso} acts on the span of $\ket{\bot}^\sB$ and $\ket{0}^\sB$ as $V_\theta^\sB$ in \Cref{sec:conjugation and isometry}.

We reduce the condition in \Cref{thm:lower Uhlmann} to the condition in \Cref{lem:phase isometry from samples}.
Let $W_\theta^{\sB\rarr\sB\sD}$ be a Stinespring isometry of $\cT_\theta^\sB$, and define $\ket{\mu_\theta}^\sD = (\bra{\bot}^\sB\otimes\bI^\sD)W_\theta^{\sB\rarr\sB\sD}\ket{\bot}^\sB$ and $\ket{\nu_\theta}^\sD = (\bra{\psi_{t,\theta}^*}^\sB\otimes\bI^\sD)W_\theta^{\sB\rarr\sB\sD}\ket{0}^\sB$, where $\|\ket{\mu_\theta}\| \leq 1$.
Since $\ket{\bot}^\sA$ is orthogonal to $\ket{0}^\sA$ and $\ket{\psi_{t,\theta}}^\sA$, and $\braket{\psi_{t,\theta}}{0} = t$, we have
\begin{align}
    &\rF\big(\cT_\theta^\sB(\ketbra{\rho_{1/2,{\rm c}}}{\rho_{1/2,{\rm c}}}^{\sA\sB}), \ket{\sigma_{1/2,t,\theta,{\rm c}}}^{\sA\sB}\big) \notag\\
    &\hspace{1pc}= \f{1}{4}\big\|\ket{\mu_\theta}^\sD + t\ket{\nu_\theta}^\sD\big\|^2.
\end{align}
Together with $\rF(\rho_{1/2}^\sA, \sigma_{1/2,t,\theta}^\sA) = (1+t)^2/4$, the condition in \Cref{thm:lower Uhlmann} implies $\big\|\ket{\mu_\theta} + t\ket{\nu_\theta}\big\|^2 \geq (1+t)^2 - 4\delta$, and thus $\big\|\ket{\mu_\theta} + t\ket{\nu_\theta}\big\| \geq 1 + t - 4\delta$, where we used $\sqrt{b^2 - c} \geq b - c$ for $b \geq 1$ and $c \in [0, 1]$.
Since $\ket{\mu_\theta} + t\ket{\nu_\theta} = (1-t)\ket{\mu_\theta} + t(\ket{\mu_\theta} + \ket{\nu_\theta})$ and $\|\ket{\mu_\theta}\| \leq 1$, the triangle inequality implies 
\begin{align}
    \big\|\ket{\mu_\theta} + \ket{\nu_\theta}\big\| \geq 2 - 4\delta/t = 2(1 - \delta/s).
\end{align}
As $(\bra{\Phi_2}^{\sR\sB}(V_\theta^\sB)^\dag\otimes\bI^\sD)W_\theta^{\sB\rarr\sB\sD}\ket{\Phi_2}^{\sR\sB} = (\ket{\mu_\theta}^\sD + \ket{\nu_\theta}^\sD)/2$, we obtain
\begin{align}
    \rF\big(\cT_\theta^\sB(\ketbra{\Phi_2}{\Phi_2}^{\sR\sB}), V_\theta^\sB\ket{\Phi_2}^{\sR\sB}\big) 
    &= \f{1}{4}\big\|\ket{\mu_\theta}^\sD + \ket{\nu_\theta}^\sD\big\|^2 \\
    &\geq \Big(1 - \f{\delta}{s}\Big)^2 \\
    &\geq 1 - \f{2\delta}{s}.
\end{align}
Since $2\delta/s \leq 2^{-12}$ by the condition $\delta \leq 2^{-13}s$ in the third statement, \Cref{lem:phase isometry from samples} with $t = 2s$ and $\epsilon = 2\delta/s$ yields $n \geq 3\ln 2/(2^{16}\delta s)$.
Together with Eq.~\eqref{inteq:mixed dimension bound}, this proves the third statement.

The same argument bounds the number of samples when the error is measured by the diamond norm.
Suppose that the channel realized by an algorithm is within $\delta \in (0, 2^{-12}]$, measured by one half of the diamond norm, from the unitary channel of a unitary extension of the Uhlmann partial isometry.
For the above states with $p = 1/2$, $t = 2s$, and $D = 1$, any unitary extension acts as $V_\theta^\sB$ on the span of $\ket{\bot}^\sB$ and $\ket{0}^\sB$, and thus the state $\cT_\theta^\sB(\ketbra{\Phi_2}{\Phi_2}^{\sR\sB})$ is within trace distance $\delta$ from $V_\theta^\sB\ket{\Phi_2}^{\sR\sB}$.
Equation~\eqref{inteq:mixed Choi fidelity} then holds with $\epsilon = \delta$, and \Cref{lem:phase isometry from samples} with $t = 2s$ yields $n \geq 3\ln 2/(2^{15}\delta s^2)$.
Thus, in the diamond norm, $\Omega\big(1/(\delta s_{\rm min}^2)\big)$ samples are required even for states of rank two in a three-dimensional system.

The term $\Omega(d_\sA/\delta)$ in the third statement relies on the canonical purification and on the channel acting on $\sB$.
If the channel acts on $\sA$ instead, as in the variant of the Uhlmann transformation, replacing $\sA$ with one sample of $\sigma^\sA$ already attains the fidelity $\rF(\rho^\sA, \sigma^\sA)$ for pure states, which is consistent with the sample complexity of \Cref{thm:variant Uhl overview} being independent of $d_\sA$.

\section{Application 1: Square root fidelity estimation with the Uhlmann transformation algorithm}
\label{sec:fidelity estimation}

We provide a quantum algorithm for estimating the square root fidelity $\sqrt{\rF}(\rho^\sA, \sigma^\sA)$ and evaluate the query and sample complexities in the three models.
The key idea of our approach is to apply the Uhlmann transformation $V^\sB$ to $\ket{\rho}^{\sA\sB}$ and then estimate the absolute value of the inner product between the two pure states $V^\sB\ket{\rho}^{\sA\sB}$ and $\ket{\sigma}^{\sA\sB}$.
By Uhlmann's theorem (Theorem~\ref{thm:Uhlmann theorem}), $\big|\bra{\sigma}V^\sB\ket{\rho}^{\sA\sB}\big|$ coincides with $\sqrt{\rF}(\rho^\sA, \sigma^\sA)$.
Since our algorithms realize $V^\sB$ only approximately, we evaluate how this approximation error affects the fidelity estimation.
In the mixed sample access model, the sample cost depends on whether collective operations over multiple samples are allowed or not.
When they are allowed, we also estimate $\sqrt{\rF}(\rho^\sA, \sigma^\sA)$ directly from the block-encoding of $\tr_{\sA'}\big[\ketbra{\sigma}{\rho}^{\sA'\sB}\big]$ in \Cref{sec:purif sample}, without implementing $V^\sB$.

A notable distinction from previous approaches is that ours keeps the ``purifying system'' throughout the process.
Previous approaches mainly aimed to directly prepare the square root of the states and discarded the purifying system midway, even when purified access was available.
Keeping it enables us to employ an estimation algorithm for pure states. While fidelity estimation between mixed states is generally difficult~\cite{gilyen2022improvedfidelity, Rethinasamy2023estimatingdist}, it is more feasible between pure states, and optimal algorithms are known~\cite{Wang2024OptimalFidelityPure, wang2024sampleoptimal, Fang2025OptimalFidelityToPure}.
This results in a substantial speedup for our algorithm in the purified access model.

Our fidelity estimation algorithms achieve optimal rank dependence, up to logarithmic factors, in the purified query access model and in the mixed sample access model with collective operations. Let $r_{\rm min}$ and $r_{\rm max}$ denote the smaller and larger ranks of $\rho^\sA$ and $\sigma^\sA$, respectively. For additive error $\delta$, the respective complexities are $\til{\cO}(r_{\rm min}/\delta^2)$ queries and $\til{\cO}(r_{\rm min}^2/\delta^2)$ samples, improving upon the previous bounds of $\til{\cO}(r_{\rm min}^{2.5}/\delta^5)$ and $\til{\cO}(r_{\rm min}^{5.5}/\delta^{12})$~\cite{gilyen2022improvedfidelity}. The corresponding lower bounds are $\til{\Omega}(r_{\rm min}/\delta)$~\cite{wang2026unifiedcomplexityframeworkquantum} and $\til{\Omega}(r_{\rm min}^2/\delta^2)$~\cite{wang2026unifiedcomplexityframeworkquantum, lee2026fidelityestimationknownquantum}, and the number of samples is thus optimal in both $r_{\rm min}$ and $\delta$. Reference~\cite{lowe2026randomdimensionreductionlearning} achieves $\cO(r_{\rm max}^2/\delta^2)$ samples for the square root fidelity, and our result replaces $r_{\rm max}$ with $r_{\rm min}$. When one of the states is classically specified, $\cO(r_{\rm min}^2/\delta^2)$ samples of the other state suffice~\cite{lee2026fidelityestimationknownquantum}.
When one of the states is pure and it is known which one, optimal query and sample complexities are known~\cite{Fang2025OptimalFidelityToPure}; subsequent work removes the need to know which state is pure by specializing our Uhlmann transformation algorithm~\cite{liu2026optimalfidelityestimationstate}.

Although in a different setting from ours, fidelity also plays a central role in certifying quantum states~\cite{Eisert2020QuantumCertification, KlieschTheoryofQuantumSystemCertification}.
Certification checks whether an unknown state is close to a classically specified target state, often using only restricted measurements, e.g., single-qubit ones~\cite{Flammia2011DirectFidelityEstimationFewPauli, daSilva2011PracticalCharacterizationwithoutTomography, Pallister2018OptimalVerificationEntangledStates, Huang2025CertifyingAlmostAll}. From a practical viewpoint, this is important for validating near-term quantum devices. Our setting requires no classical description of either state but instead queries to the unitaries preparing the purified states or general quantum operations on samples.
Our results thus complement certification when the target state is available only as a quantum state, e.g., as the output of another device~\cite{Elben2020CrossPlatformVerification} or of a quantum algorithm without an efficient classical description.

The results are organized as follows:
\begin{itemize}
    \item The purified query access model in \Cref{sec:fidelity estimation PQA}.
    \item The purified sample access model in \Cref{sec:fidelity estimation PSS}.
    \item The mixed sample access model in \Cref{sec:fidelity estimation mixed sample}: collective case in \Cref{sec:fidelity estimation collective} and non-collective case in \Cref{sec:fidelity estimation SS}.
\end{itemize}
In the purified query access model, collective operations are not required. In the purified sample access model, the algorithm in \Cref{sec:fidelity estimation PSS} does not use collective operations, while the algorithm in \Cref{sec:fidelity estimation collective} also applies to this model when they are allowed (\Cref{cor:fidelity est purified collective}).


\subsection{In the purified query access model}
\label{sec:fidelity estimation PQA}

Our goal is to estimate the fidelity $\sqrt{\rF}(\rho^\sA, \sigma^\sA)$ using the Uhlmann transformation algorithm in the purified query access model.
For simplicity, we sometimes abbreviate $\sqrt{\rF}(\rho^\sA, \sigma^\sA)$ as $\sqrt{\rF}$.
We propose a quantum query algorithm $\mathtt{UhlFidelityEstPurifQuery}$, whose performance is
characterized by the following theorem.

\begin{theorem}[Square root fidelity estimation in the purified query access model]
\label{thm:fid est quel}
For $\delta \in (0, 1)$, the quantum query algorithm $\mathtt{UhlFidelityEstPurifQuery}(U_\rho^{\sA\sB}, U_\sigma^{\sA\sB}; \delta)$,
which is given by \Cref{alg:fidelestquery}, outputs $\sqrt{\til{\rF}}$ such that $\big|\sqrt{\til{\rF}} - \sqrt{\rF}\big| \leq \delta$ with probability at least $2/3$, using $q = \cO\big(\f{1}{\delta}{\rm min}\big\{\f{1}{s_{\rm min}}\log{\big(\f{1}{\delta}\big)}, \f{R}{\delta}\big\}\big)$ queries to $U_\rho^{\sA\sB}$, $U_\sigma^{\sA\sB}$, and their inverses.
The quantum circuit of this algorithm consists of $\cO\big(q\log{(d_\sA d_\sB)} + (\log{(1/\delta}))^2\big)$ one- and two-qubit gates, and uses $\cO(\log{(d_\sA d_\sB)} + \log{(1/\delta}))$ qubits simultaneously. The circuit description can be computed classically in $\poly\big(q, \log{(d_\sA d_\sB)}, \log{(1/\delta)}\big)$ time.
\end{theorem}

\begin{algorithm}[ht]
\caption{Square root fidelity estimation algorithm in the purified query access model \\ \parbox{\linewidth}{\centering $\mathtt{UhlFidelityEstPurifQuery}(U_\rho^{\sA\sB}, U_\sigma^{\sA\sB}; \delta)$ \\ (In Theorem~\ref{thm:fid est quel})}}
\label{alg:fidelestquery}
\SetKwInput{KwInput}{Input}
\SetKwInput{KwOutput}{Output}
\SetKwInput{KwParameters}{Parameters}

\SetAlgoNoEnd
\SetAlgoNoLine
\KwInput{Unitary oracles $U_\rho^{\sA\sB}$, $U_\sigma^{\sA\sB}$, and their inverses.}
\KwParameters{$\delta \in (0, 1)$.}
\KwOutput{Real number $\sqrt{\til{\rF}}$.}
\SetAlgoLined

Set $\til{W}_\tF^{\sB\sD} \gets \mathtt{UhlmannPurifiedQuery}(U_\rho^{\sA\sB}$, $U_\sigma^{\sA\sB}; \delta, \tF)$ (Algorithm~\ref{alg:Uhl purif query}). \\
Set $\sqrt{\til{\rF}} \gets \mathtt{SqrtAmpEstQuery}(\til{W}_\tF^{\sB\sD}U_\rho^{\sA\sB}, U_\sigma^{\sA\sB}; \delta/2)$ (\Cref{lem:square root amp est query}). \\
Return $\sqrt{\til{\rF}}$.
\end{algorithm}

In \Cref{alg:fidelestquery}, we use the \emph{square root amplitude estimation} algorithm $\mathtt{SqrtAmpEstQuery}$~\cite{Wang2024OptimalFidelityPure}, which is known to be optimal in terms of the number of queries in the case where both states are pure.
Here, we use a slightly adjusted version of the result from Ref.~\cite{Wang2024OptimalFidelityPure}.

\begin{theorem}[Square root amplitude estimation in the purified query access model~{\cite[Theorem III.4]{Wang2024OptimalFidelityPure}}]
\label{lem:square root amp est query}
Suppose that $W$ and $U$ are unitaries on a $d$-dimensional Hilbert space such that $W\ket{0} = \sqrt{c}U\ket{0} + \sqrt{1- c}\ket{\bot}$, where $\ket{\bot}$ is a state satisfying $\bra{\bot}U\ket{0} = 0$.
Then, for any $\delta \in (0, 1)$, there is a quantum query algorithm $\mathtt{SqrtAmpEstQuery}(W, U; \delta)$ that outputs $\sqrt{\til{c}}$ such that $\big|\sqrt{\til{c}} - \sqrt{c}\big| \leq \delta$ holds with probability at least $2/3$, with $\cO(1/\delta)$ queries to $W$, $U$, and their inverses.
The quantum circuit of this algorithm consists of $\cO\big(\log{d}/\delta + (\log{(1/\delta)})^2\big)$ one- and two-qubit gates in addition to the queries, and uses $\cO(\log{d} + \log{(1/\delta)})$ qubits simultaneously.
\end{theorem}

The algorithm $\mathtt{SqrtAmpEstQuery}$ is based on the \emph{quantum phase estimation}.
To estimate a phase of a target unitary $Q$ with accuracy $\delta$ and success probability $p_{\rm succ} \geq 1 - \eta$, the quantum phase estimation uses a total of $g = \cO(1/(\eta\delta))$ applications of $Q$. The quantum circuit includes $\cO\big((\log(1/\delta))^2\big)$ one- and two-qubit gates, which come from the quantum Fourier transform part, and uses $\cO(\log(1/\delta))$ auxiliary qubits.
Here, we took the overlap between the input state to the quantum phase estimation and the eigenstate of $Q$ corresponding to the estimated phase to be of constant order. In fact, while the success probability depends on the overlap, the probability can be amplified with only a constant number of repetitions of the quantum phase estimation when the overlap is of constant order.

In the fidelity estimation, the target unitary is given by a product of reflection unitaries:
\begin{align}
    Q&= e^{i\pi W\ketbra{0}{0}W^\dag}e^{i\pi U\ketbra{0}{0}U^\dag} \\
    &= (\bI - 2W\ketbra{0}{0}W^\dag)(\bI - 2U\ketbra{0}{0}U^\dag)  \\
    \label{inteq:151}
    &= We^{i\pi\ketbra{0}{0}}W^\dag U e^{i\pi\ketbra{0}{0}}U^\dag.
\end{align}
The input state is taken as $W\ket{0}$, since this state has a constant overlap with eigenstates of $Q$.
Then, it was shown that performing quantum phase estimation on $Q$ and a simple post-processing calculation yields an outcome that sufficiently approximates the square root fidelity between two pure states with high probability~\cite{Wang2024OptimalFidelityPure}.

To perform the quantum phase estimation for $Q$ in Eq.~\eqref{inteq:151}, we need the control version of $e^{i\pi\ketbra{0}{0}} = \bI - 2\ketbra{0}{0}$, in addition to $W$ and $U$.
It can be easily constructed using multi-controlled NOT gates that are implemented with a polynomial number of one- and two-qubit gates in the number of qubits~\cite{nielsen2010quantum, saeedi2013NoAncilla, yoder2014fixed}.

With the above discussion, we now prove Theorem~\ref{thm:fid est quel}.

\begin{proof}[Proof of Theorem~\ref{thm:fid est quel}]

As discussed in \Cref{sec:uhlfidquery} (in particular, from just above Eq.~\eqref{inteq:14} to Eq.~\eqref{inteq:13}), a unitary $\til{W}_\tF^{\sB\sD} = \mathtt{UhlmannPurifiedQuery}(U_\rho^{\sA\sB}, U_\sigma^{\sA\sB}; \delta, \tF)$, which is an exact block-encoding of $P_\sgn^{(\rm SV)}(M^\sB)$, satisfies that
\begin{align}
\label{inteq:53}
    \Big|\sqrt{\rF}\big(\til{W}_\tF^{\sB\sD}\ket{\rho}^{\sA\sB}\ket{0}^\sD, \ket{\sigma}^{\sA\sB}\ket{0}^\sD\big) - \sqrt{\rF}\Big| 
    \leq \delta/2,
\end{align}
where $M^\sB = \tr_{\sA'}\big[\ketbra{\sigma}{\rho}^{\sA'\sB}\big]$.
The number of queries for implementing $\til{W}_\tF^{\sB\sD}$ is given by $u_\tF = \cO\big(\min\big\{\f{1}{s_{\rm min}}\log{\big(\f{1}{\delta}\big)}, \f{R}{\delta}\big\}\big)$ (\Cref{thm:Uhlmann alg purif query model}).
Remember that $s_{\rm min}$ and $R$ are the minimum non-zero singular value and the rank of $\sqrt{\sigma^\sA}\sqrt{\rho^\sA}$, respectively.

Next, we consider estimating the value of $\sqrt{\rF'} = \sqrt{\rF}\big(\til{W}_\tF^{\sB\sD}\ket{\rho}^{\sA\sB}\ket{0}^\sD, \ket{\sigma}^{\sA\sB}\ket{0}^\sD\big)$.
Using the Uhlmann transformation algorithm, we can obtain the states $\til{W}_\tF^{\sB\sD}\ket{\rho}^{\sA\sB}\ket{0}^{\sD}$ such that
\begin{align}
    \label{inteq:170}
    \til{W}_\tF^{\sB\sD}\ket{\rho}^{\sA\sB}\ket{0}^{\sD}
    &= \til{W}_\tF^{\sB\sD}U_\rho^{\sA\sB}\ket{0}^{\sA\sB\sD}    \\
    &= \sqrt{\rF'}e^{i\theta}\ket{\sigma}^{\sA\sB}\ket{0}^{\sD} + \ket{\bot}^{\sA\sB\sD} \\
    &= \sqrt{\rF'}U_\sigma^{\sA\sB}\ket{0}^{\sA\sB\sD} + \ket{\bot}^{\sA\sB\sD},
\end{align}
where $\ket{\bot}^{\sA\sB\sD}$ is a vector satisfying $\bra{\bot}U_\sigma^{\sA\sB}\ket{0}^{\sA\sB\sD} = 0$.
A certain phase $\theta$ does not affect the result and is hence absorbed into $U_\sigma^{\sA\sB}$.
Thus, from \Cref{lem:square root amp est query}, when we perform the quantum phase estimation for a target unitary
\begin{align}
    &Q^{\sA\sB\sD} \notag\\
    \label{inteq:48}
    &= \til{W}_\tF^{\sB\sD}U_\rho^{\sA\sB}e^{i\pi\ketbra{0}{0}^{\sA\sB\sD}}(\til{W}_\tF^{\sB\sD}U_\rho^{\sA\sB})^\dag U_\sigma^{\sA\sB}e^{i\pi\ketbra{0}{0}^{\sA\sB\sD}}(U_\sigma^{\sA\sB})^\dag,
\end{align}
with the input state $\til{W}_\tF^{\sB\sD}U_\rho^{\sA\sB} \ket{0}^{\sA\sB\sD}$, we can estimate $\sqrt{\rF'}$ within the error $\delta$ using $\cO(1/\delta)$ applications of $\til{W}_\tF^{\sB\sD}U_\rho^{\sA\sB}$, $U_\sigma^{\sA\sB}$, and their inverses.
That is, the output $\sqrt{\til{\rF}} = \mathtt{SqrtAmpEstQuery}(\til{W}_\tF^{\sB\sD}U_\rho^{\sA\sB}, U_\sigma^{\sA\sB}; \delta/2)$ satisfies that
\begin{align}
\label{inteq:54}
    \Big|\sqrt{\til{\rF}} - \sqrt{\rF'}\Big| \leq \delta/2,
\end{align}
with probability at least $2/3$.

Therefore, we conclude that $\Big|\sqrt{\til{\rF}} - \sqrt{\rF}\Big| \leq \delta$, where we used the triangle inequality and Eqs.~\eqref{inteq:53} and~\eqref{inteq:54}.
The number of queries $q$ to $U_\rho^{\sA\sB}$, $U_\sigma^{\sA\sB}$, and their inverses is determined as
\begin{align}
    q
    &= u_\tF \cO(1/\delta) \\
    \label{inteq:58}
    &=
    \cO\Big(\f{1}{\delta}{\rm min}\Big\{\f{1}{s_{\rm min}}\log{\Big(\f{1}{\delta}\Big)}, \f{R}{\delta}\Big\}\Big).
\end{align}
In the quantum circuit of this algorithm, the unitary $Q^{\sA\sB\sD}$ in Eq.~\eqref{inteq:48} is repeated $\cO(1/\delta)$ times, where $Q^{\sA\sB\sD}$ includes the Uhlmann transformation consisting of $\cO\big(u_\tF\log{(d_\sA d_\sB)}\big)$ gates (\Cref{thm:Uhlmann alg purif query model}).
Including additional $\cO\big(\log{(d_\sA d_\sB)}/\delta + (\log{(1/\delta)})^2\big)$ gates in the square root amplitude estimation (\Cref{lem:square root amp est query}), the total number of one- and two-qubit gates is determined by
\begin{align}
    &\cO(1/\delta)\cO\big(u_\tF\log{(d_\sA d_\sB)}\big) + \cO\big(\log{(d_\sA d_\sB)}/\delta + (\log{(1/\delta)})^2\big) \notag\\
    &= \cO\big(q\log{(d_\sA d_\sB)} + (\log{(1/\delta)})^2\big).
\end{align}
The number of qubits used at any one time is given by $\cO(\log{(d_\sA d_\sB)} + \log{(1/\delta)})$,
where the first term comes from the Uhlmann transformation and the second term from the square root amplitude estimation.
The circuit description can be computed classically in $\poly\big(q, \log{(d_\sA d_\sB)}, \log{(1/\delta)}\big)$ time, as this algorithm consists of the Uhlmann transformation and quantum phase estimation.

This completes the proof of Theorem~\ref{thm:fid est quel}.

\end{proof}


\subsection{In the purified sample access model}
\label{sec:fidelity estimation PSS}

Given multiple copies of $\ket{\rho}^{\sA\sB}$ and $\ket{\sigma}^{\sA\sB}$, we aim to estimate the fidelity $\sqrt{\rF}=\sqrt{\rF}(\rho^\sA, \sigma^\sA)$ without using collective operations.
We propose a quantum sample algorithm $\mathtt{UhlFidelityEstPurifSample}$, for which the following theorem guarantees its performance.
The algorithm is described in \Cref{alg:fidelestsample}, where $\sF=\sH\hat{\sA}\hat{\sB}$ and $\sH$ is a three-qubit system.

\begin{theorem}[Square root fidelity estimation in the purified sample access model]
\label{thm:sqrtfidelestpufifsample}
For $\delta \in (0, 1)$, the quantum sample algorithm $\mathtt{UhlFidelityEstPurifSample}(\ket{\rho}^{\sA\sB}, \ket{\sigma}^{\sA\sB}; \delta)$,
which is given by \Cref{alg:fidelestsample}, outputs $\sqrt{\til{\rF}}$ such that $\big|\sqrt{\til{\rF}} - \sqrt{\rF}\big| \leq \delta$ with probability at least $2/3$, using $n = \cO\big(\f{1}{\delta^2}{\rm min}\big\{\f{1}{s_{\rm min}^2}\big(\log{\big(\f{1}{\delta}\big)}\big)^2, \f{R^2}{\delta^2}\big\}\big)$ samples of $\ket{\rho}^{\sA\sB}$ and $\ket{\sigma}^{\sA\sB}$.
The quantum circuit of this algorithm consists of $\cO(n\log{(d_\sA d_\sB)})$ one- and two-qubit gates, and uses $\cO(\log{(d_\sA d_\sB)} + \log{(1/\delta)})$ qubits simultaneously. The circuit description can be computed classically in $\poly\big(n, \log{(d_\sA d_\sB)}, \log{(1/\delta)}\big)$ time.
\end{theorem}

\begin{algorithm}[ht]
\caption{Square root fidelity estimation algorithm in the purified sample access model (without collective operations) \\ \parbox{\linewidth}{\centering $\mathtt{UhlFidelityEstPurifSample}(\ket{\rho}^{\sA\sB}, \ket{\sigma}^{\sA\sB}; \delta)$ \\ (In Theorem~\ref{thm:sqrtfidelestpufifsample})}}
\label{alg:fidelestsample}
\SetKwInput{KwInput}{Input}
\SetKwInput{KwOutput}{Output}
\SetKwInput{KwParameters}{Parameters}
\SetKwComment{Comment}{$\triangleright$\ }{}
\SetCommentSty{textnormal}
\SetKwProg{Fn}{Subroutine}{}{end}
\SetKwFunction{UPS}{UhlmannPurifiedSample}
\SetKwFunction{SAES}{SqrtAmpEstSample}

\SetAlgoNoEnd
\SetAlgoNoLine
\KwInput{Two pure quantum states $\ket{\rho}^{\sA\sB}$ and $\ket{\sigma}^{\sA\sB}$.}
\KwParameters{$\delta \in (0, 1)$.}
\KwOutput{Real number $\sqrt{\til{\rF}}$.}
\SetAlgoLined

Set $\delta_2 \gets \delta/2$ and $\delta_1 \gets \delta_2/1920$. \\
Set $\ket{\phi}^{\sA\sB\sF} \gets \ket{\rho}^{\sA\sB}\ket{0}^\sH\ket{\sigma}^{\hat{\sA}\hat{\sB}}$ and $\ket{\psi}^{\sA\sB\sF} \gets \ket{\sigma}^{\sA\sB}\ket{0}^\sH\ket{\rho}^{\hat{\sA}\hat{\sB}}$. \\
Set $\cG_\rho^{\sC_1\hat{\sA}\hat{\sB}} \gets \mathtt{BlockEncPureState}(\ket{\rho}^{\hat{\sA}\hat{\sB}}; \delta_1/(2u))$ and $\cG_\sigma^{\sC_2\hat{\sA}\hat{\sB}} \gets \mathtt{BlockEncPureState}(\ket{\sigma}^{\hat{\sA}\hat{\sB}}; \delta_1/(2u))$ (Lemma~\ref{lem:BE pure state}), where $u$ is the number of uses of $\cG^{\sC\hat{\sA}\hat{\sB}\sB}$ and $\cG_{\rm inv}^{\sC\hat{\sA}\hat{\sB}\sB}$ in $\mathtt{QSVTSIGNFEJER}$ or $\mathtt{QSVTSIGN}$ below. \\
Set $\cG^{\sC\hat{\sA}\hat{\sB}\sB} \gets \cG_\rho^{\sC_1\hat{\sA}\hat{\sB}}\circ\cF^{\sB\hat{\sB}}\circ\cG_\sigma^{\sC_2\hat{\sA}\hat{\sB}}$ and $\cG_{\rm inv}^{\sC\hat{\sA}\hat{\sB}\sB} \gets \cG_\sigma^{\sC_2\hat{\sA}\hat{\sB}}\circ\cF^{\sB\hat{\sB}}\circ\cG_\rho^{\sC_1\hat{\sA}\hat{\sB}}$, where $\sC = \sC_1\sC_2$ and $\cF^{\sB\hat{\sB}}(\cdot) = F^{\sB\hat{\sB}}(\cdot)F^{\sB\hat{\sB}}$ is the swap channel. \\
\If{$\mathtt{QSVTSIGNFEJER}$ below uses $\cG^{\sC\hat{\sA}\hat{\sB}\sB}$ and $\cG_{\rm inv}^{\sC\hat{\sA}\hat{\sB}\sB}$ fewer times than $\mathtt{QSVTSIGN}$ below}{
  Set $\cJ^{\sB\sF} \gets \mathtt{QSVTSIGNFEJER}(\cG^{\sC\hat{\sA}\hat{\sB}\sB}, \cG_{\rm inv}^{\sC\hat{\sA}\hat{\sB}\sB}; \delta/(2R))$ (Proposition~\ref{prop:QSVT sign Fejer}). \\
}
\Else{
  Set $\cJ^{\sB\sF} \gets \mathtt{QSVTSIGN}(\cG^{\sC\hat{\sA}\hat{\sB}\sB}, \cG_{\rm inv}^{\sC\hat{\sA}\hat{\sB}\sB}; (\delta/3)^2, s_{\rm min})$ (Proposition~\ref{prop:FPAA general}). \\
}
Set $\cJ_{\rm inv}^{\sB\sF}$ to be the quantum channel obtained by applying the adjoint circuit of the QSVT above, with $\cG^{\sC\hat{\sA}\hat{\sB}\sB}$ and $\cG_{\rm inv}^{\sC\hat{\sA}\hat{\sB}\sB}$ interchanged. \\
Set $\sqrt{\til{\rF}} \gets \mathtt{SqrtAmpEstSample}(\cJ^{\sB\sF}, \cJ_{\rm inv}^{\sB\sF}, \ket{\phi}^{\sA\sB\sF}, \ket{\psi}^{\sA\sB\sF}; \delta_2)$ (Lemma~\ref{lem:modif square root est sample}). \\
Return $\sqrt{\til{\rF}}$.
\end{algorithm}

The estimation strategy is basically the same as in the last subsection.
The difference is that, unlike in the purified query access model, we cannot exactly implement the target unitary $Q$, which is a product of reflections, for the quantum phase estimation.
In \Cref{alg:fidelestsample}, the Uhlmann transformation and its inverse in $Q$ are approximately implemented by the quantum channels $\cJ^{\sB\sF}$ and $\cJ_{\rm inv}^{\sB\sF}$, which are constructed from samples in the same way as in \Cref{alg:Uhl purif sample}, and the other reflections in $Q$ are implemented by the density matrix exponentiation.
Note that $\cJ^{\sB\sF}$ and $\cJ_{\rm inv}^{\sB\sF}$ are used only $\cO(1/\delta)$ times, as in the purified query access model; if we instead used samples of the output state of the Uhlmann transformation in the density matrix exponentiation, $\cO(1/\delta^2)$ runs of the Uhlmann transformation algorithm would be needed.
Here, we use $\mathtt{QSVTSIGNFEJER}$ rather than $\mathtt{QSVTSIGNSTAGES}$ for small $R$, since $\cJ^{\sB\sF}$ needs to approximate a unitary in the diamond norm.

The algorithm $\mathtt{SqrtAmpEstSample}$ is given in Ref.~\cite{wang2024sampleoptimal} and is a sample version of $\mathtt{SqrtAmpEstQuery}$. Originally, it estimates $|\braket{\phi}{\psi}|$ from samples of two pure states $\ket{\phi}$ and $\ket{\psi}$. As in $\mathtt{SqrtAmpEstQuery}$, we include a unitary $W$ in the estimation target as $|\bra{\psi}W\ket{\phi}|$, where $W$ and its inverse are given as quantum channels $\cJ$ and $\cJ_{\rm inv}$ that approximate them. The following lemma shows that $\mathtt{SqrtAmpEstSample}$ works in this setting.

\begin{lemma}[Square root amplitude estimation from samples, adapted from Ref.~{\cite[Theorem 7.1]{wang2024sampleoptimal}}]
\label{lem:modif square root est sample}
Let $\ket{\phi}$ and $\ket{\psi}$ be pure states in a $d$-dimensional Hilbert space, and let $W$ be a unitary on this space. For $\epsilon' \in [0, 1/1920)$, let $\cJ$ and $\cJ_{\rm inv}$ be quantum channels such that $\f{1}{2}\|\cJ - \cW\|_\diamond \leq \epsilon'$ and $\f{1}{2}\|\cJ_{\rm inv} - \cW^\dag\|_\diamond \leq \epsilon'$. Then, for $\delta \in (960\epsilon', 1/2)$, the quantum sample algorithm $\mathtt{SqrtAmpEstSample}(\cJ, \cJ_{\rm inv}, \ket{\phi}, \ket{\psi}; \delta)$ outputs $\sqrt{\til{c}}$ such that $\big|\sqrt{\til{c}} - |\bra{\psi}W\ket{\phi}|\big| \leq \delta$ with probability at least $2/3$, using $\cJ$ and $\cJ_{\rm inv}$ a total of $\cO(1/\delta)$ times and $m = \cO\big(1/\big(\delta(\delta/960 - \epsilon')\big)\big)$ samples of $\ket{\phi}$ and $\ket{\psi}$. In addition to the uses of $\cJ$ and $\cJ_{\rm inv}$, the quantum circuit of this algorithm consists of $\cO(m\log{d})$ one- and two-qubit gates, and uses $\cO(\log{d} + \log{(1/\delta)})$ qubits simultaneously.
\end{lemma}

The key idea of $\mathtt{SqrtAmpEstSample}$ is as follows. As mentioned in the last subsection, the quantum phase estimation on the product of reflections $Q = e^{i\pi\ketbra{\psi}{\psi}}We^{i\pi\ketbra{\phi}{\phi}}W^\dag$ about $\ket{\psi}$ and $W\ket{\phi}$ with the input $\ket{\psi}$ and a simple post-processing calculation yield a good estimate of $|\bra{\psi}W\ket{\phi}|$. Since we cannot exactly implement the unitaries $e^{i\pi\ketbra{\psi}{\psi}}$ and $e^{i\pi\ketbra{\phi}{\phi}}$ in the sample model, we instead approximate them with error $\epsilon$ via the density matrix exponentiation, using $\cO(1/\epsilon)$ samples of $\ket{\psi}$ and $\ket{\phi}$. Typically, when the input state is an eigenstate of $Q$, the quantum phase estimation applies $Q$ a total of $g = \sum_{j=1}^l 2^{j-1} = 2^l - 1 = \cO(1/(\eta\delta))$ times, and uses $\cO\big((\log{(1/\delta)})^2\big)$ gates and $\cO(\log{(1/\delta)})$ auxiliary qubits, to estimate the corresponding phase of $Q$ with error at most $\delta$ and success probability at least $1 - \eta$, where $l = \big\lceil\log{\big(\f{1}{\delta}(2+\f{1}{2\eta})\big)}\big\rceil$~\cite{kitaev1995quantummeasurementsabelian, cleve1997quantumalgorithmrevisit, nielsen2010quantum}. However, if $Q$ is only approximately prepared, the accumulated error reduces the success probability to $p_{\rm succ} \geq 1 - \eta - 2g\epsilon$.
To ensure that $p_{\rm succ}$ remains a constant order, such as $2/3$, it suffices to set $\eta = 1/6$ and $\epsilon \leq 1/(12g)$, resulting in a total sample complexity of $\ket{\psi}$ and $\ket{\phi}$ as $g\cO(1/\epsilon) = \cO(1/\delta^2)$. Since the quantum phase estimation requires a controlled version of the target unitary $Q$, we can simply use the controlled implementation of the density matrix exponentiation~\cite{kimmel2017HamSimSampleComp}, while $W$ and $W^\dag$ need not be controlled since they cancel each other when the reflections are not applied.

We now prove \Cref{lem:modif square root est sample} by evaluating how the error propagates when $W$ and $W^\dag$ are replaced with $\cJ$ and $\cJ_{\rm inv}$.

\begin{proof}[Proof of \Cref{lem:modif square root est sample}]
Using the density matrix exponentiation with $\cO(1/\epsilon)$ samples of $\ket{\phi}$, we can implement the quantum channel $\cL$ that satisfies $\f{1}{2}\big\|\cL - e^{i\pi\ketbra{\phi}{\phi}}(\cdot)e^{-i\pi\ketbra{\phi}{\phi}}\big\|_\diamond \leq \epsilon$.
In the same way, the density matrix exponentiation with $\cO(1/\epsilon)$ samples of $\ket{\psi}$ approximates $e^{i\pi\ketbra{\psi}{\psi}}$ within the error $\epsilon$.

We recall that in the quantum phase estimation on the product of reflections $Q = e^{i\pi\ketbra{\psi}{\psi}}We^{i\pi\ketbra{\phi}{\phi}}W^\dag$, the target unitary $Q$ is used $g = 2^l -1 $ times, where we take $l = \big\lceil\log{\big(\f{\pi}{\delta}(2+\f{1}{2\eta})\big)}\big\rceil$ to estimate the phase of $Q$ within $\delta/\pi$, which yields the error at most $\delta$ in the estimate of $|\bra{\psi}W\ket{\phi}|$.
When we replace all $g$ uses of $e^{i\pi\ketbra{\phi}{\phi}}$ and $e^{i\pi\ketbra{\psi}{\psi}}$ with these channels, and those of $W$ and $W^\dag$ with $\cJ$ and $\cJ_{\rm inv}$, the resulting error is at most $2g\epsilon + 2g\epsilon' = 2g(\epsilon + \epsilon')$.

Hence, the quantum phase estimation outputs $\sqrt{\til{c}}$ that satisfies $\big|\sqrt{\til{c}} - |\bra{\psi}W\ket{\phi}|\big| \leq \delta$ with probability $p_{\rm succ} \geq 1 - \eta - 2g(\epsilon + \epsilon')$.
Setting $\eta = 1/6$ and $\epsilon = 1/(24g) - \epsilon'$, the success probability is at least $3/4 \geq 2/3$.
To ensure that $\epsilon > 0$, it suffices to take $\delta > 960\epsilon'$, since $24g < 960/\delta$ for $g = 2^{\lceil \log{(5\pi/\delta)}\rceil}-1 < 10\pi/\delta$.
The algorithm uses $\cJ$ and $\cJ_{\rm inv}$ a total of $2g = \cO(1/\delta)$ times, and uses $g\cO(1/\epsilon) = \cO\Big(1/\big(\delta(\delta/960 - \epsilon')\big)\Big)$ samples of $\ket{\phi}$ and $\ket{\psi}$.

\end{proof}

With the above discussion, we now prove Theorem~\ref{thm:sqrtfidelestpufifsample}.

\begin{proof}[Proof of Theorem~\ref{thm:sqrtfidelestpufifsample}]

We first evaluate the value to be estimated.
Let $U_1^{\sB\sF}$ be the unitary obtained from $\cJ^{\sB\sF}$ by replacing $\cG^{\sC\hat{\sA}\hat{\sB}\sB}$ and $\cG_{\rm inv}^{\sC\hat{\sA}\hat{\sB}\sB}$ with $U_2^{\sC\hat{\sA}\hat{\sB}\sB}$ in Eq.~\eqref{eq:def of U2} and its inverse, which is a block-encoding of $\ketbra{\rho}{\sigma}^{\hat{\sA}\hat{\sB}}\otimes P_\sgn^{(\rm SV)}(M^\sB)$ as in Eq.~\eqref{inteq:61}, where $M^\sB = \tr_{\sA'}\big[\ketbra{\sigma}{\rho}^{\sA'\sB}\big]$ and $P_\sgn$ is the polynomial in \Cref{prop:QSVT sign Fejer} or \Cref{prop:FPAA general}.
For $\ket{\phi}^{\sA\sB\sF}$ and $\ket{\psi}^{\sA\sB\sF}$ in \Cref{alg:fidelestsample}, we have
\begin{align}
    &\bra{\psi}U_1^{\sB\sF}\ket{\phi}^{\sA\sB\sF} \notag\\
    &= \bra{\sigma}^{\sA\sB}\bra{\rho}^{\hat{\sA}\hat{\sB}}\big(\ketbra{\rho}{\sigma}^{\hat{\sA}\hat{\sB}}\otimes P_\sgn^{(\rm SV)}(M^\sB)\big)\ket{\rho}^{\sA\sB}\ket{\sigma}^{\hat{\sA}\hat{\sB}} \\
    &= \bra{\sigma}^{\sA\sB}P_\sgn^{(\rm SV)}(M^\sB)\ket{\rho}^{\sA\sB}.
\end{align}
Since $P_\sgn$ is chosen in the same way as in \Cref{alg:Uhl purif query} with $\chi = \tF$, the discussion in \Cref{sec:uhlfidquery} (in particular, from Eq.~\eqref{inteq:14} to Eq.~\eqref{inteq:13}) implies that
\begin{align}
\label{inteq:fid est purif poly}
    \Big|\big|\bra{\psi}U_1^{\sB\sF}\ket{\phi}^{\sA\sB\sF}\big| - \sqrt{\rF}\Big| \leq \delta/2.
\end{align}
Recall that $s_{\rm min}$ and $R$ are the minimum non-zero singular value and the rank of $\sqrt{\sigma^\sA}\sqrt{\rho^\sA}$, respectively.

Next, we consider estimating $\big|\bra{\psi}U_1^{\sB\sF}\ket{\phi}^{\sA\sB\sF}\big|$.
As in Eqs.~\eqref{eq:matrix exponentiation of M} and~\eqref{inteq:92}, $\cG^{\sC\hat{\sA}\hat{\sB}\sB}$ and $\cG_{\rm inv}^{\sC\hat{\sA}\hat{\sB}\sB}$ approximate $\cU_2^{\sC\hat{\sA}\hat{\sB}\sB}$ and $(\cU_2^{\sC\hat{\sA}\hat{\sB}\sB})^\dag$, respectively, within error $\delta_1/u$.
Since they are used a total of $u$ times in $\cJ^{\sB\sF}$, the subadditivity of the diamond norm implies that $\f{1}{2}\|\cJ^{\sB\sF} - \cU_1^{\sB\sF}\|_\diamond \leq \delta_1$, and similarly, $\f{1}{2}\|\cJ_{\rm inv}^{\sB\sF} - (\cU_1^{\sB\sF})^\dag\|_\diamond \leq \delta_1$.
Let 
\begin{align}
    \sqrt{\til{\rF}} = \mathtt{SqrtAmpEstSample}(\cJ^{\sB\sF}, \cJ_{\rm inv}^{\sB\sF}, \ket{\phi}^{\sA\sB\sF}, \ket{\psi}^{\sA\sB\sF}; \delta_2).
\end{align}
From \Cref{lem:modif square root est sample} with $W = U_1^{\sB\sF}$ and $\epsilon' = \delta_1$, we observe that $\big|\sqrt{\til{\rF}} - |\bra{\psi}U_1^{\sB\sF}\ket{\phi}^{\sA\sB\sF}|\big| \leq \delta_2$ with probability at least $2/3$, using $\cJ^{\sB\sF}$ and $\cJ_{\rm inv}^{\sB\sF}$ a total of $\cO(1/\delta_2)$ times and $\cO\big(1/\big(\delta_2(\delta_2/960 - \delta_1)\big)\big) = \cO(1/\delta^2)$ samples of $\ket{\phi}^{\sA\sB\sF}$ and $\ket{\psi}^{\sA\sB\sF}$, where we used $\delta_1 = \delta_2/1920$ and $\delta_2 = \delta/2$.
Using the triangle inequality and Eq.~\eqref{inteq:fid est purif poly}, we conclude that $\big|\sqrt{\til{\rF}} - \sqrt{\rF}\big| \leq \delta_2 + \delta/2 = \delta$ with probability at least $2/3$.

Each use of $\cJ^{\sB\sF}$ or $\cJ_{\rm inv}^{\sB\sF}$ uses $\cG^{\sC\hat{\sA}\hat{\sB}\sB}$ and $\cG_{\rm inv}^{\sC\hat{\sA}\hat{\sB}\sB}$ a total of $u$ times, each of which uses $\lceil 8\pi^2u/\delta_1\rceil$ samples of each of $\ket{\rho}^{\hat{\sA}\hat{\sB}}$ and $\ket{\sigma}^{\hat{\sA}\hat{\sB}}$ by \Cref{lem:BE pure state}.
In addition, each sample of $\ket{\phi}^{\sA\sB\sF}$ and $\ket{\psi}^{\sA\sB\sF}$ uses one sample of each of $\ket{\rho}^{\sA\sB}$ and $\ket{\sigma}^{\sA\sB}$.
Thus, the number of samples $n$ of $\ket{\rho}^{\sA\sB}$ and $\ket{\sigma}^{\sA\sB}$ is given by
\begin{align}
    n = \cO(u/\delta)\cO(u/\delta) + \cO(1/\delta^2) = \cO(u^2/\delta^2).
\end{align}
Since \Cref{alg:fidelestsample} uses whichever of $\mathtt{QSVTSIGN}$ and $\mathtt{QSVTSIGNFEJER}$ uses $\cG^{\sC\hat{\sA}\hat{\sB}\sB}$ and $\cG_{\rm inv}^{\sC\hat{\sA}\hat{\sB}\sB}$ fewer times, we have $u = \cO\big(\min\{\log{(1/\delta)}/s_{\rm min}, R/\delta\}\big)$ as in \Cref{sec:uhlfidquery}, and hence
\begin{align}
    n =\cO\bigg(\f{1}{\delta^2}\min\bigg\{\f{1}{s_{\rm min}^2}\Big(\log{\Big(\f{1}{\delta}\Big)}\Big)^2, \f{R^2}{\delta^2}\bigg\}\bigg).
\end{align}
In the quantum circuit of this algorithm, each sample of $\ket{\rho}^{\sA\sB}$ and $\ket{\sigma}^{\sA\sB}$ is used in a step of the density matrix exponentiation, which consists of $\cO(\log{(d_\sA d_\sB)})$ one- and two-qubit gates.
In addition, the QSVT in $\cJ^{\sB\sF}$ and $\cJ_{\rm inv}^{\sB\sF}$ uses $\cO(u)$ gates for each use (Theorem~\ref{prop:general QSVT}), and the quantum phase estimation uses $\cO\big((\log{(1/\delta)})^2\big)$ gates.
Thus, the total number of one- and two-qubit gates is given by $\cO(n\log{(d_\sA d_\sB)})$.
Since the samples are used sequentially, the circuit uses $\cO(\log{(d_\sA d_\sB)} + \log{(1/\delta)})$ qubits simultaneously, where the second term comes from the quantum phase estimation, and the circuit description can be computed classically in $\poly\big(n, \log{(d_\sA d_\sB)}, \log{(1/\delta)}\big)$ time, considering the QSVT and the quantum phase estimation.

\end{proof}

When collective operations are allowed, the algorithm in the next subsection uses $\til{\cO}(r_{\rm min}^2/\delta^2)$ samples without using the purifying systems of the given states (\Cref{cor:fidelity est purified collective}).


\subsection{In the mixed sample access model}
\label{sec:fidelity estimation mixed sample}

We consider the estimation of the fidelity $\sqrt{\rF} = \sqrt{\rF}(\rho^\sA, \sigma^\sA)$ in the mixed sample access model. In \Cref{sec:fidelity estimation collective}, we discuss algorithms that employ collective operations. In \Cref{sec:fidelity estimation SS}, we construct an algorithm with non-collective operations.

\subsubsection{In the mixed sample access model: collective case}
\label{sec:fidelity estimation collective}

We provide two algorithms for square root fidelity estimation in the mixed sample access model when collective operations over multiple samples of the states are allowed.
We first present an algorithm that attains the optimal number of samples with respect to the rank and the accuracy, and then an algorithm whose number of samples depends on $s_{\rm min}$. Here, the optimality of the number of samples up to logarithmic factors follows from the lower bound $\til{\Omega}(r_{\rm min}^2/\delta^2)$~\cite{wang2026unifiedcomplexityframeworkquantum,lee2026fidelityestimationknownquantum}. As with $R$ and $s_{\rm min}$, we regard $r_\rho$ and $r_\sigma$ as known, where $r_\rho$ and $r_\sigma$ are the ranks of $\rho^\sA$ and $\sigma^\sA$, and $r_{\rm min}$ is the smaller of them. The first algorithm $\mathtt{FidelityEstMixedSampleCollective}$ is provided in \Cref{alg:fidelestsamplemixedcollective}, where $\sB = \sB_1\sB_2$ with $d_{\sB_1} = r_{\rm min}$ and $d_{\sB_2} = d_\sA$.

\begin{theorem}[Square root fidelity estimation in the mixed sample access model with collective operations]
\label{thm:fidelity est collective optimal}
For $\delta \in (0, 1)$, the quantum sample algorithm $\mathtt{FidelityEstMixedSampleCollective}(\rho^\sA, \sigma^\sA; \delta)$,
which is given by \Cref{alg:fidelestsamplemixedcollective}, outputs $\sqrt{\til{\rF}}$ such that $\big|\sqrt{\til{\rF}} - \sqrt{\rF}\big| \leq \delta$ with probability at least $2/3$, using $n = \cO\big(\f{r_{\rm min}^2}{\delta^2}\big(\log{\big(\f{r_{\rm min}}{\delta}\big)}\big)^2\big)$ samples of $\rho^\sA$ and $\sigma^\sA$.
The quantum circuit of this algorithm consists of $\til{\cO}(n^4)$ one- and two-qubit gates, and uses $\til{\cO}(n^2)$ qubits simultaneously. The circuit description can be computed classically in $\poly\big(n, \log{d_\sA}, \log{(1/\delta)}\big)$ time.
\end{theorem}

\begin{algorithm}[ht]
\caption{Square root fidelity estimation algorithm in the mixed sample access model with collective operations \\ \parbox{\linewidth}{\centering $\mathtt{FidelityEstMixedSampleCollective}(\rho^{\sA}, \sigma^{\sA}; \delta)$ \\ (In Theorem~\ref{thm:fidelity est collective optimal})}}
\label{alg:fidelestsamplemixedcollective}
\SetKwInput{KwInput}{Input}
\SetKwInput{KwOutput}{Output}
\SetKwInput{KwParameters}{Parameters}

\SetAlgoNoEnd
\SetAlgoNoLine
\KwInput{Two quantum states $\rho^{\sA}$ and $\sigma^{\sA}$.}
\KwParameters{$\delta \in (0, 1)$.}
\KwOutput{Real number $\sqrt{\til{\rF}}$.}
\SetAlgoLined

\If{$r_\rho > r_\sigma$}{
  Interchange $\rho^\sA$ and $\sigma^\sA$. \\
}
Set $L \gets \lceil 2r_\rho/\delta\rceil$ and $\epsilon_1 \gets \delta/(16er_\rho)$. \\
Set $p_l \gets |a_l|/(lS_L)$ for $l = 1, \ldots, L$, where $a_l = \f{4(-1)^{l+1}}{\pi(4l^2 - 1)}$ and $S_L = \sum_{l=1}^{L}|a_l|/l$. \\
Set $N \gets \Big\lceil\f{48e^2r_\rho^2}{\delta^2}S_L\sum_{l=1}^{L}l|a_l|\Big\rceil$ and $m_l \gets 2\lceil 8\pi^2l\rceil$ for $l = 1, \ldots, L$. \\
Draw $l_1, \ldots, l_N \in \{1, \ldots, L\}$ independently, where each $l_j$ equals $l$ with probability $p_l$. \\
\If{$\sum_{j=1}^{N}n_{l_j} > 12N\sum_{l=1}^{L}p_ln_l$, where $n_l = 4lm_l + 1$ is the number of samples of $\rho^\sA$ and $\sigma^\sA$ used below for $l_j = l$}{
  Abort. \\
}
\For{$j = 1, \ldots, N$}{
  Set $l \gets l_j$. \\
  Apply the quantum channels given by $\mathtt{RandomPurification}(\rho^\sA; \epsilon_1, 2lm_l)$ with purifying system $\sB_1$ and $\mathtt{RandomPurification}(\sigma^\sA; \epsilon_1, 2lm_l + 1)$ with purifying system $\sB$ (\Cref{prevthm:random purification}) to fresh samples of $\rho^\sA$ and $\sigma^\sA$, respectively, to approximately obtain samples of $\ket{\rho}^{\sA\sB} = \ket{\rho}^{\sA\sB_1}\ket{0}^{\sB_2}$ and $\ket{\sigma}^{\sA\sB}$. \\
  Set $\til{U}_l^{\sH\hat{\sA}\hat{\sB}\sB}$ to be the block-encoding unitary obtained by \Cref{prop:general QSVT} with the Chebyshev polynomial $T_{2l}$ from $U_2^{\sC\hat{\sA}\hat{\sB}\sB}$ in Eq.~\eqref{eq:def of U2}, where $\sH$ is a three-qubit system. \\
  Run the Hadamard test of $\til{U}_l^{\sH\hat{\sA}\hat{\sB}\sB}$ with the input state $\ketbra{0}{0}^\sH\otimes\ketbra{\sigma}{\sigma}^{\hat{\sA}\hat{\sB}}\otimes\pi^{\sB_1}\otimes\ketbra{0}{0}^{\sB_2}$, where each of the $4l$ unitaries in the controlled $\til{U}_l^{\sH\hat{\sA}\hat{\sB}\sB}$ that depend on $\ket{\rho}^{\hat{\sA}\hat{\sB}}$ or $\ket{\sigma}^{\hat{\sA}\hat{\sB}}$ is replaced by the channel given by $\mathtt{VirtualDME}$ with $t = \pi$ and $m = m_l$ (Theorem~\ref{thm:virtual DME}) applied to $\ket{+}^\sK$, followed by the measurement of $\sK$ in the basis $\{\ket{+}, \ket{-}\}$, using the samples obtained above. \\
  Set $Y_l$ to be the product of $\alpha^{4l}$, the outcome $\pm 1$ of the Hadamard test, and the $4l$ outcomes $\pm 1$ of the measurements of $\sK$, where $\alpha$ is the positive number in \Cref{thm:virtual DME} for $t = \pi$ and $m = m_l$. \\
  Set $c_j \gets a_lY_l/p_l$. \\
}
Set $\sqrt{\til{\rF}} \gets r_\rho\big(\f{2}{\pi} + \f{1}{N}\sum_{j=1}^{N}c_j\big)$. \\
Return $\sqrt{\til{\rF}}$.
\end{algorithm}

Since samples of $\rho^\sA$ and $\sigma^\sA$ are obtained from those of $\ket{\rho}^{\sA\sB}$ and $\ket{\sigma}^{\sA\sB}$ by discarding $\sB$, the same algorithm applies in the purified sample access model.

\begin{corollary}[Square root fidelity estimation in the purified sample access model with collective operations]
\label{cor:fidelity est purified collective}
For $\delta \in (0, 1)$, there is a quantum sample algorithm that outputs $\sqrt{\til{\rF}}$ such that $\big|\sqrt{\til{\rF}} - \sqrt{\rF}\big| \leq \delta$ with probability at least $2/3$, using $n$ samples of $\ket{\rho}^{\sA\sB}$ and $\ket{\sigma}^{\sA\sB}$, where $n$ and the costs of the quantum circuit are the same as in \Cref{thm:fidelity est collective optimal}.
\end{corollary}

The idea of the algorithm in \Cref{thm:fidelity est collective optimal} is to estimate $\sqrt{\rF}$ without actually implementing the Uhlmann partial isometry $V^\sB$.
Since $\sqrt{\rF}$ does not depend on the choice of purified states, we work with purified states $\ket{\rho}^{\sA\sB}$ and $\ket{\sigma}^{\sA\sB}$ prepared from $\rho^\sA$ and $\sigma^\sA$ by the random purification in \Cref{prevthm:random purification}, assuming $r_{\rm min} = r_\rho \leq r_\sigma$ without loss of generality.
For $M^\sB = \tr_{\sA'}\big[\ketbra{\sigma}{\rho}^{\sA'\sB}\big]$, let $\{s_k\}_k$ be its singular values.
As in \Cref{sec:uhlfidquery}, we have $\sqrt{\rF} = \sum_ks_k = \tr\big[V^\sB(M^\sB)^\dag\big]$, which is attained by the Uhlmann transformation.
Since this value coincides with the trace of $|M^\sB| = \sqrt{(M^\sB)^\dag M^\sB}$, we estimate it by expanding the absolute value in the Chebyshev polynomials $T_k$ of degree $k$ as $|x| = \f{2}{\pi} + \sum_{l \geq 1}a_lT_{2l}(x)$ and estimating the trace of $T_{2l}(|M^\sB|)$ on the support of $\rho^\sB$ for each degree. By Theorem~\ref{prop:general QSVT} with the even polynomial $T_{2l}$, the block-encoding unitary $U_2^{\sC\hat{\sA}\hat{\sB}\sB}$ in Eq.~\eqref{eq:def of U2} is converted into a block-encoding unitary of $T_{2l}^{(\rm SV)}\big(\ketbra{\rho}{\sigma}^{\hat{\sA}\hat{\sB}}\otimes M^\sB\big)$. The trace is then estimated by the Hadamard test for block-encoding unitaries~\cite{gilyen2022improvedfidelity}, where the input state on $\sB$ is the completely mixed state on $\supp[\rho^\sB]$.

Three points lead to the number of samples in \Cref{thm:fidelity est collective optimal}.
First, the Hadamard test estimates the trace divided by the dimension of the support of $\rho^\sB$, and the random purification allows us to take this support to be a known subspace, whose dimension is $r_{\rm min}$.
Second, unlike in \Cref{prop:QSVT sign Fejer}, where the polynomial is applied coherently and must be bounded by one, only its trace is estimated here, and the truncated series can be used directly. Since the Chebyshev coefficients of $|x|$ satisfy $|a_l| = \cO(1/l^2)$, the truncation at degree $2L$ with $L = \Theta(r_{\rm min}/\delta)$ suffices.
Third, instead of estimating this trace for each degree separately, we draw the degree $2l$ at random with probability proportional to $|a_l|/l$ in each run, as in randomized methods for estimating matrix functions~\cite{wang2024QubitEfficientRandomized}, which makes the expected number of samples in a run $\cO(\log{L})$. For this estimate to be unbiased, the reflections in $U_2^{\sC\hat{\sA}\hat{\sB}\sB}$ are implemented by the virtual density matrix exponentiation in \Cref{thm:virtual DME} rather than by \Cref{lem:BE pure state}.

\begin{proof}[Proof of \Cref{thm:fidelity est collective optimal}]
Without loss of generality, we assume $r_\rho \leq r_\sigma$; that is, $r_{\rm min} = r_\rho$. We first fix the purified states and the value to be estimated. Let $\sB = \sB_1\sB_2$ with $d_{\sB_1} = r_\rho$ and $d_{\sB_2} = d_\sA$, and take a purified state $\ket{\rho}^{\sA\sB_1}$ of $\rho^\sA$, which exists since $d_{\sB_1} = r_\rho$. On $\sA\sB$, we use the purified states $\ket{\rho}^{\sA\sB} = \ket{\rho}^{\sA\sB_1}\ket{0}^{\sB_2}$ of $\rho^\sA$ and $\ket{\sigma}^{\sA\sB}$ of $\sigma^\sA$. We then define $M^\sB = \tr_{\sA'}\big[\ketbra{\sigma}{\rho}^{\sA'\sB}\big]$ and $\Pi^\sB = \bI^{\sB_1}\otimes\ketbra{0}{0}^{\sB_2}$.
Since $\supp[\rho^\sB]$ is contained in the image of $\Pi^\sB$, we have $M^\sB = M^\sB \Pi^\sB$, and thus $|M^\sB| = \Pi^\sB|M^\sB|\Pi^\sB$.
Hence, for any $l \geq 0$,
\begin{align}
\label{eq:fid collective trace}
    \tr\big[\Pi^\sB T_{2l}(|M^\sB|)\big] = \sum_{k=1}^{r_\rho}T_{2l}(s_k),
\end{align}
where $s_1, \ldots, s_{r_\rho}$ are the eigenvalues of $|M^\sB|$ on $\im[\Pi^\sB]$, that is, the singular values of $M^\sB$ padded with zeros, and $\sum_{k=1}^{r_\rho}s_k = \sqrt{\rF}$ (see \Cref{sec:uhlfidquery}).

We next approximate $\sqrt{\rF}$ by a finite sum.
From the cosine series of $|\cos\theta|$, the Chebyshev series of $|x|$ on $[-1, 1]$ is
\begin{align}
\label{eq:chebyshev abs}
    |x| = \f{2}{\pi} + \sum_{l=1}^{\infty}a_lT_{2l}(x), \ \text{and} \ \  
    a_l = \f{4(-1)^{l+1}}{\pi(4l^2 - 1)},
\end{align}
which converges uniformly because $\sum_{l \geq 1}|a_l| = 2/\pi$~\cite{trefethen2019approximationTheoryPractice}.
Since $|T_{2l}(x)| \leq 1$ on $[-1, 1]$ and $\sum_{l > L}|a_l| = \f{2}{\pi(2L+1)}$, the truncation of the series at degree $2L$ has uniform error at most $\f{2}{\pi(2L+1)}$.
Let $L = \lceil 2r_\rho/\delta\rceil$ and
\begin{align}
\label{eq:fid collective bl}
    b_l = \f{1}{r_\rho}\tr\big[\Pi^\sB T_{2l}(|M^\sB|)\big].
\end{align}
Since $\sqrt{\rF} = \sum_{k=1}^{r_\rho}s_k$ with $s_k \in [0, 1]$, Eq.~\eqref{eq:fid collective trace} and the uniform error of the truncation imply
\begin{align}
    &\Big|\sqrt{\rF} - r_\rho\Big(\f{2}{\pi} + \sum_{l=1}^{L}a_lb_l\Big)\Big| \notag\\
    &= \Big|\sum_{k=1}^{r_\rho}\Big(s_k - \f{2}{\pi} - \sum_{l=1}^{L}a_lT_{2l}(s_k)\Big)\Big| \\
    &\leq \sum_{k=1}^{r_\rho}\Big|s_k - \f{2}{\pi} - \sum_{l=1}^{L}a_lT_{2l}(s_k)\Big| \\
    &\leq \f{2r_\rho}{\pi(2L+1)} \\
    \label{eq:fid collective truncation}
    &\leq \f{\delta}{4}.
\end{align}

To estimate $b_l$ for a given $l \geq 1$, we now construct a quantum circuit whose measurement outcomes yield a real random variable $Y_l$ with expectation $b_l$.
For the above purified states, let $U_2^{\sC\hat{\sA}\hat{\sB}\sB} = W_\rho^{\sC_1\hat{\sA}\hat{\sB}}(\bI^{\hat{\sA}}\otimes F^{\sB\hat{\sB}})W_\sigma^{\sC_2\hat{\sA}\hat{\sB}}$ be the block-encoding unitary of $\ketbra{\rho}{\sigma}^{\hat{\sA}\hat{\sB}}\otimes M^\sB$ as in Eq.~\eqref{eq:def of U2}, where $\sC = \sC_1\sC_2$, $W_\rho^{\sC_1\hat{\sA}\hat{\sB}}$ and $W_\sigma^{\sC_2\hat{\sA}\hat{\sB}}$ are the block-encoding unitaries of $\ketbra{\rho}{\rho}^{\hat{\sA}\hat{\sB}}$ and $\ketbra{\sigma}{\sigma}^{\hat{\sA}\hat{\sB}}$ in \Cref{lem:BE pure state}, and $F^{\sB\hat{\sB}}$ is the swap operator. By \Cref{prop:general QSVT} with the even polynomial $T_{2l}$, we obtain a $(1, 3, 0)$-block-encoding unitary $\til{U}_l^{\sH\hat{\sA}\hat{\sB}\sB}$ of $T_{2l}^{(\rm SV)}\big(\ketbra{\rho}{\sigma}^{\hat{\sA}\hat{\sB}}\otimes M^\sB\big)$ with a three-qubit system $\sH$, using $U_2^{\sC\hat{\sA}\hat{\sB}\sB}$ and its inverse $l$ times each.
Since $T_{2l}$ is even and $\big|\ketbra{\rho}{\sigma}^{\hat{\sA}\hat{\sB}}\otimes M^\sB\big| = \ketbra{\sigma}{\sigma}^{\hat{\sA}\hat{\sB}}\otimes|M^\sB|$, the block is 
\begin{align}
    &T_{2l}^{(\rm SV)}\big(\ketbra{\rho}{\sigma}^{\hat{\sA}\hat{\sB}}\otimes M^\sB\big) \notag\\
    &= \ketbra{\sigma}{\sigma}^{\hat{\sA}\hat{\sB}}\otimes T_{2l}(|M^\sB|) + (-1)^l\big(\bI^{\hat{\sA}\hat{\sB}} - \ketbra{\sigma}{\sigma}^{\hat{\sA}\hat{\sB}}\big)\otimes\bI^\sB,
\end{align}
where the second term comes from $T_{2l}(0) = (-1)^l$. We perform the Hadamard test of $\til{U}_l^{\sH\hat{\sA}\hat{\sB}\sB}$ with a one-qubit system $\sJ$: we prepare $\ket{+}^\sJ$, apply $\til{U}_l^{\sH\hat{\sA}\hat{\sB}\sB}$ controlled by $\sJ$ to the input state $\ketbra{0}{0}^\sH\otimes\ketbra{\sigma}{\sigma}^{\hat{\sA}\hat{\sB}}\otimes\pi^{\sB_1}\otimes\ketbra{0}{0}^{\sB_2}$, and measure $\sJ$ in the basis $\{\ket{+}, \ket{-}\}$. The outcome $\pm 1$ has expectation
\begin{align}
    &{\rm Re}\Big(\tr\big[\big(\ketbra{0}{0}^\sH\otimes\ketbra{\sigma}{\sigma}^{\hat{\sA}\hat{\sB}}\otimes\pi^{\sB_1}\otimes\ketbra{0}{0}^{\sB_2}\big)\til{U}_l^{\sH\hat{\sA}\hat{\sB}\sB}\big]\Big) \notag\\
    &= \tr\big[\big(\ketbra{\sigma}{\sigma}^{\hat{\sA}\hat{\sB}}\otimes\pi^{\sB_1}\otimes\ketbra{0}{0}^{\sB_2}\big)T_{2l}^{(\rm SV)}\big(\ketbra{\rho}{\sigma}^{\hat{\sA}\hat{\sB}}\otimes M^\sB\big)\big] \\
    &= b_l,
\end{align}
where we used $\pi^{\sB_1}\otimes\ketbra{0}{0}^{\sB_2} = \Pi^\sB/r_\rho$. The trace is real since the block is Hermitian, and the second term of the block does not contribute.

The controlled $\til{U}_l$ consists of one- and two-qubit gates and the controlled versions of $W_\rho^{\sC_1\hat{\sA}\hat{\sB}}$ and $W_\sigma^{\sC_2\hat{\sA}\hat{\sB}}$, each used $2l$ times. By Eq.~\eqref{eq:W pure state}, the only part of the controlled $W_\rho^{\sC_1\hat{\sA}\hat{\sB}}$ that depends on $\ket{\rho}^{\hat{\sA}\hat{\sB}}$ is $e^{i\pi\ketbra{1}{1}^\sJ\otimes\ketbra{1}{1}^{\sC_1}\otimes\ketbra{\rho}{\rho}^{\hat{\sA}\hat{\sB}}}$, the exponential of the projector onto the pure state $\ket{1}^\sJ\ket{1}^{\sC_1}\ket{\rho}^{\hat{\sA}\hat{\sB}}$, and the same holds for $W_\sigma^{\sC_2\hat{\sA}\hat{\sB}}$. We replace each of these $4l$ unitaries by the channel given by $\mathtt{VirtualDME}$ in \Cref{thm:virtual DME} with $t = \pi$ and $m = m_l$, where $m_l = 2\lceil 8\pi^2l\rceil$ is even. This channel uses $m_l$ samples of $\ket{\rho}^{\sA\sB}$ or $\ket{\sigma}^{\sA\sB}$. As described below \Cref{thm:virtual DME} with $q = 4l$, we then obtain a real random variable $Y_l$ such that $\bE[Y_l] = b_l$ and $|Y_l| \leq e$, since $m_l \geq 4t^2q$. Since the whole circuit has $4l$ of these unitaries, a run with degree $2l$ uses $4lm_l = \cO(l^2)$ samples in total, in addition to one sample of $\ket{\sigma}^{\sA\sB}$ in the input state.

The samples of $\ket{\rho}^{\sA\sB_1}$ and $\ket{\sigma}^{\sA\sB}$ are prepared from those of $\rho^\sA$ and $\sigma^\sA$ by $\mathtt{RandomPurification}$ in \Cref{prevthm:random purification} with error $\epsilon_1 = \delta/(16er_\rho)$ and with purifying systems $\sB_1$ and $\sB$, respectively. The condition on the rank is satisfied since $d_{\sB_1} = r_\rho$ and $d_\sB \geq d_\sA \geq r_\sigma$. In each run of the Hadamard test, $\mathtt{RandomPurification}$ is applied at once to all the fresh samples of $\rho^\sA$, and likewise for $\sigma^\sA$. The value $b_l$ does not depend on the choice of the purified states, because a unitary on $\sB_1$ applied to $\ket{\rho}^{\sA\sB_1}$ conjugates $|M^\sB|$ by a unitary commuting with $\Pi^\sB$, and a unitary on $\sB$ applied to $\ket{\sigma}^{\sA\sB}$ leaves $|M^\sB|$ unchanged. As the circuit is independent of the unitaries on $\sB_1$ and $\sB$ averaged in Eq.~\eqref{eq:random purification ideal}, the expectation of $Y_l$ is $b_l$ when $\mathtt{RandomPurification}$ is replaced by the ideal channel. Since $|Y_l| \leq e$ and the expectation is linear in the input state, using the actual channel instead changes the expectation by at most $2e\epsilon_1$ for each of $\rho^\sA$ and $\sigma^\sA$, and hence 
\begin{align}
\label{eq:fid collective bias}
    \big|\bE[Y_l] - b_l\big| \leq 4e\epsilon_1 = \f{\delta}{4r_\rho}.
\end{align}

An estimate of $\sqrt{\rF}$ within error $\delta$ is constructed by drawing the degree at random in each run. Let $p_l = |a_l|/(lS_L)$ for $l = 1, \ldots, L$, where $S_L = \sum_{l=1}^{L}|a_l|/l$ satisfies $\f{4}{3\pi} \leq S_L \leq \f{2}{\pi}$.
In each run, we draw $l$ with probability $p_l$, run the above circuit, and output $a_lY_l/p_l$. Let $c_1, \ldots, c_N$ be the outputs of $N$ independent runs.
Each $c_j$ satisfies $\bE[c_j] = \sum_{l=1}^{L}a_l\bE[Y_l]$, and thus Eq.~\eqref{eq:fid collective bias} and $\sum_{l=1}^{L}|a_l| \leq 1$ imply that
\begin{align}
\label{eq:fid collective mean}
    \Big|\bE[c_j] - \sum_{l=1}^{L}a_lb_l\Big| \leq \f{\delta}{4r_\rho}.
\end{align}
Its second moment is bounded as
\begin{align}
\label{eq:fid collective second moment}
    \bE[c_j^2] = \sum_{l=1}^{L}\f{a_l^2\bE[Y_l^2]}{p_l} \leq e^2S_L\sum_{l=1}^{L}l|a_l| = \cO(\log{L}),
\end{align}
where we used $|Y_l| \leq e$, $S_L \leq \f{2}{\pi}$, and $l|a_l| \leq \f{4}{3\pi l}$.
We output 
\begin{align}
    \sqrt{\til{\rF}} = r_\rho\Big(\f{2}{\pi} + \f{1}{N}\sum_{j=1}^{N}c_j\Big),
\end{align}
with $N = \Big\lceil\f{48e^2r_\rho^2}{\delta^2}S_L\sum_{l=1}^{L}l|a_l|\Big\rceil$.
Since $N \geq 48r_\rho^2\bE[c_j^2]/\delta^2$ by Eq.~\eqref{eq:fid collective second moment}, Chebyshev's inequality implies that $\big|\f{1}{N}\sum_{j=1}^{N}c_j - \bE[c_j]\big| \leq \f{\delta}{2r_\rho}$ holds with probability at least $11/12$.
From this bound and Eqs.~\eqref{eq:fid collective truncation} and~\eqref{eq:fid collective mean}, the triangle inequality yields 
\begin{align}
    \big|\sqrt{\til{\rF}} - \sqrt{\rF}\big| 
    \leq \f{\delta}{4} + \f{\delta}{4} + \f{\delta}{2} 
    = \delta.
\end{align}

Since the degree is drawn at random, the number of samples of $\rho^\sA$ and $\sigma^\sA$ in a run is random, where a run with degree $2l$ uses $n_l = 4lm_l + 1 = \cO(l^2)$ samples.
Its expectation is $\sum_{l=1}^{L}p_ln_l = \cO\big(\sum_{l=1}^{L}p_ll^2\big) = \cO(\log{L})$, since $\sum_{l=1}^{L}p_ll^2 = \f{1}{S_L}\sum_{l=1}^{L}l|a_l|$.
To bound the number of samples in the worst case, we draw $l_1, \ldots, l_N$ before the runs and abort if the total number of samples $\sum_{j=1}^{N}n_{l_j}$ exceeds twelve times its expectation, which happens with probability at most $1/12$ by Markov's inequality.
The algorithm thus outputs $\sqrt{\til{\rF}}$ with $\big|\sqrt{\til{\rF}} - \sqrt{\rF}\big| \leq \delta$ with probability at least $1 - 1/12 - 1/12 \geq 2/3$, using
\begin{align}
    n = \cO(N\log{L}) = \cO\bigg(\f{r_\rho^2}{\delta^2}\Big(\log{\Big(\f{r_\rho}{\delta}\Big)}\Big)^2\bigg)
\end{align}
samples of $\rho^\sA$ and $\sigma^\sA$.

In the quantum circuit of this algorithm, a run of the Hadamard test with degree $2l$ applies $\mathtt{RandomPurification}$ to the $\cO(l^2)$ samples of $\rho^\sA$ and to those of $\sigma^\sA$. The circuit of $\mathtt{RandomPurification}$ consists of $\til{\cO}(l^8)$ one- and two-qubit gates and uses $\til{\cO}(l^4)$ qubits at any one time by \Cref{prevthm:random purification}. The rest of the run, i.e., the virtual density matrix exponentiation and the QSVT, uses $\cO(l^2\log{d_\sA})$ gates, where $\log{(d_\sA d_\sB)} = \cO(\log{d_\sA})$.
Since $l \leq L$ and $\sum_{j=1}^{N}l_j^2 = \cO(n)$ unless the algorithm aborts, the total number of gates is $\til{\cO}(L^6n) = \til{\cO}(n^4)$, and the number of qubits used at any one time is $\til{\cO}(L^4) = \til{\cO}(n^2)$, where we used $L^2 = \cO(n)$.
The circuit description can be computed classically in $\poly\big(n, \log{d_\sA}, \log{(1/\delta)}\big)$ time, including that of $\mathtt{RandomPurification}$ with $\log{(1/\epsilon_1)} = \cO\big(\log{(r_\rho/\delta)}\big)$.

\end{proof}

We next provide another algorithm whose number of samples depends on $s_{\rm min}$, obtained by simulating the query algorithm in \Cref{alg:fidelestquery} with samples.

\begin{theorem}[Square root fidelity estimation in the mixed sample access model by simulating queries]
\label{thm:fidelity est local sample collective}
For $\delta \in (0, 1)$, there is a quantum sample algorithm, obtained by applying \Cref{thm:sample to query lift} to \Cref{alg:fidelestquery}, that outputs $\sqrt{\til{\rF}}$ such that $\big|\sqrt{\til{\rF}} - \sqrt{\rF}\big| \leq \delta$ with probability at least $2/3$. The algorithm uses $n = \cO\big(\f{1}{\delta^2}{\rm min}\big\{\f{1}{s_{\rm min}^2}\big(\log{\big(\f{1}{\delta}\big)}\big)^2, \f{R^2}{\delta^2}\big\}\big)$ samples of $\rho^\sA$ and $\sigma^\sA$.
The quantum circuit of this algorithm consists of $\til{\cO}(n^4)$ one- and two-qubit gates, and uses $\til{\cO}(n^2)$ qubits simultaneously. The circuit description can be computed classically in $\poly\big(n, \log d_\sA, \log(1/\delta)\big)$ time.

\end{theorem}

The key idea in the construction of the algorithm in \Cref{thm:fidelity est local sample collective} is to use \Cref{thm:sample to query lift}, which simulates a query algorithm by a sample algorithm.
With this technique, all $q = \cO\big(\f{1}{\delta}{\rm min}\big\{\f{1}{s_{\rm min}}\log{\big(\f{1}{\delta}\big)}, \f{R}{\delta}\big\}\big)$ queries to the purified state preparation unitaries $U_\rho^{\sA\sB}$, $U_\sigma^{\sA\sB}$, and their inverses in the query algorithm $\mathtt{UhlFidelityEstPurifQuery}$ (\Cref{thm:fid est quel}) are replaced by approximate operations constructed from multiple copies of $\rho^\sA$ and $\sigma^\sA$.

As mentioned in Ref.~\cite{tang2025conjugatequerieshelp}, we should note a caveat: \Cref{thm:sample to query lift} assumes that the query algorithm works for any choice of purified state preparation unitary. Since \Cref{thm:sample to query lift} involves the averages $\bE_U[\cdot]$ over Haar random unitaries on the purifying system and $\bE_\theta[\cdot]$ over the phase, the result cannot depend on a specific purification.
While this condition is typically satisfied in common query algorithms, depending on tasks, a particular purification may be more advantageous than others. For instance, this is the case for the Uhlmann transformation itself. Specifically, the Uhlmann transformation works only for the specific purifications determined by the given initial and final states, $\ket{\rho}^{\sA\sB}$ and $\ket{\sigma}^{\sA\sB}$, and not for arbitrary purifications. On the other hand, when the goal is to estimate the fidelity, this issue can be circumvented by applying the purified state preparation unitaries $U_\rho^{\sA\sB}$ and $U_\sigma^{\sA\sB}$ for the initial and final states, along with the Uhlmann transformation (see also Eqs.~\eqref{inteq:170} to~\eqref{inteq:48}). For this reason, the sample cost of fidelity estimation in the mixed sample access model can be improved by this approach, whereas that of the Uhlmann transformation itself cannot.

Given the above discussion, we can straightforwardly prove \Cref{thm:fidelity est local sample collective}.

\begin{proof}[Proof of Theorem~\ref{thm:fidelity est local sample collective}]

Recall that for $\delta \in (0, 1)$ and $\eta \in (0, 1)$ of constant order, the algorithm
\begin{align}
    \mathtt{UhlFidelityEstPurifQuery}(U_\rho^{\sA\sB}, U_\sigma^{\sA\sB}; \delta)
\end{align}
in \Cref{alg:fidelestquery} outputs an estimate $\sqrt{\til{\rF}}$ that satisfies $\big|\sqrt{\til{\rF}} - \sqrt{\rF}\big| \leq \delta$ with success probability $p^{\rm qry} \ge 1-\eta$, using $q = \cO\big(\f{1}{\delta}{\rm min}\big\{\f{1}{s_{\rm min}}\log{\big(\f{1}{\delta}\big)}, \f{R}{\delta}\big\}\big)$ queries to $U_\rho^{\sA\sB}$, $U_\sigma^{\sA\sB}$, and their inverses.
Although the success probability was set to at least $2/3$ in \Cref{thm:fid est quel}, it can be amplified to $1-\eta$ with only a constant increase in the query cost as long as $\eta$ is of constant order (see Sec.~\ref{sec:fidelity estimation PQA}).
Note that the query algorithm $\mathtt{UhlFidelityEstPurifQuery}$ works on \emph{any} choice of purification, and hence, outputs the desired estimate with probability at least $1 - \eta$ independent of which purification is selected.

Following Theorem~\ref{thm:sample to query lift} with the accuracy $\delta_1$, we implement the query algorithm with $q$ queries by the corresponding sample algorithm by approximately constructing the purified state preparation unitaries.
Since Theorem~\ref{thm:sample to query lift} concerns a single state, we apply it to $\rho^\sA$ and then to $\sigma^\sA$, where the success probability after the first application is at least $1 - \eta - \delta_1$ for any choice of $U_\sigma^{\sA\sB}$.
Using $n = \cO(q^2/\delta_1)$ samples of the states $\rho^\sA$ and $\sigma^\sA$, this sample algorithm outputs $\sqrt{\til{\rF}}$ satisfying $\big|\sqrt{\til{\rF}} - \sqrt{\rF}\big| \leq \delta$ with probability $p^{\rm samp}$ such that $p^{\rm samp} \geq 1 - \eta - 2\delta_1$.
By setting $\eta = 1/6$ and $\delta_1 = 1/12$, we conclude that, for any $\delta \in (0, 1)$, the sample algorithm outputs $\sqrt{\til{\rF}}$ that satisfies $\big|\sqrt{\til{\rF}} - \sqrt{\rF}\big| \leq \delta$ with probability $p^{\rm samp} \geq 2/3$, using $n$ samples of the states $\rho^\sA$ and $\sigma^\sA$, where
\begin{align}
    n
    &= \cO(q^2/\delta_1) \\
    &= \cO\Big(\f{1}{\delta^2}{\rm min}\Big\{\f{1}{s_{\rm min}^2}\Big(\log{\Big(\f{1}{\delta}\Big)}\Big)^2, \f{R^2}{\delta^2}\Big\}\Big).
\end{align}

We now evaluate the number of gates and qubits in the quantum circuit of this sample algorithm.
The original query algorithm uses $\cO\big(q\log(d_\sA d_\sB) + (\log(1/\delta))^2\big)$ gates and $\cO(\log(d_\sA d_\sB) + \log(1/\delta))$ qubits at once (\Cref{thm:fid est quel}).
In addition, $n\log(d_\sA d_\sB)$ qubits are required to store the $n$ samples, so the number of qubits becomes $n\log(d_\sA d_\sB) + \cO(\log(d_\sA d_\sB) + \log(1/\delta)) = \cO(n\log(d_\sA d_\sB))$.
From \Cref{thm:sample to query lift}, converting the query algorithm to the sample one takes $\til{\cO}(n^4)$ gates and $\til{\cO}(n^2)$ qubits simultaneously, in the part of the Schur transform~\cite{burchardt2025highdimensionalschur}.
Therefore, the total number of one- and two-qubit gates in the whole sample algorithm is determined by their sum: $\cO\big(q\log(d_\sA d_\sB) + (\log(1/\delta))^2\big) + \til{\cO}(n^4) = \til{\cO}(n^4)$.
The maximum number of qubits used at one time is given by $\max\{\cO(n\log(d_\sA d_\sB)), \til{\cO}(n^2)\} = \til{\cO}(n^2)$.
As is clear from the construction, the circuit description can be computed classically in $\poly\big(n, \log d_\sA, \log(1/\delta)\big)$ time.

\end{proof}

\subsubsection{In the mixed sample access model: non-collective case}
\label{sec:fidelity estimation SS}

We consider the estimation of the fidelity $\sqrt{\rF} = \sqrt{\rF}(\rho^\sA, \sigma^\sA)$, where we are given many identical copies of $\rho^\sA$ and $\sigma^\sA$. Collective operations are not allowed here.
The algorithm $\mathtt{UhlFidelityEstMixedSample}$ is provided in \Cref{alg:fidelestsamplemixed}, where $\sD$ is the system in Theorem~\ref{thm:Uhlmann mixed state sample}, $\sG'$ is a six-qubit system, and $d_\sA = d_\sB$. Here, $\kappa_\rho = 1/\rho_{\rm min}$ and $\kappa_\sigma = 1/\sigma_{\rm min}$ are the reciprocals of the minimum non-zero eigenvalues of $\rho^\sA$ and $\sigma^\sA$, and $r_\rho$ and $r_\sigma$ are their ranks, respectively.

\begin{theorem}[Square root fidelity estimation in the mixed sample access model without collective operations]
\label{thm:fidelity est local sample}
For $\delta \in (0, 1)$, the quantum sample algorithm $\mathtt{UhlFidelityEstMixedSample}(\rho^\sA, \sigma^\sA; \delta)$,
which is given by \Cref{alg:fidelestsamplemixed}, outputs $\sqrt{\til{\rF}}$ such that $\big|\sqrt{\til{\rF}} - \sqrt{\rF}\big| \leq \delta$ with probability at least $2/3$, using $n$ samples of $\rho^\sA$ and $\sigma^\sA$, where
\begin{align}
    n = \til{\cO}\bigg(\f{d_\sA}{\delta^2}\min\Big\{\f{1}{s_{\rm min}^2}, \f{R^2}{\delta^2}\Big\}\min\Big\{\kappa_\rho^2 + \kappa_\sigma^2, \f{r_\rho^2 + r_\sigma^2}{\delta^4}\Big\}\bigg).
\end{align}
The quantum circuit of this algorithm consists of $\cO(n\log{d_\sA})$ one- and two-qubit gates, and uses $\cO(\log{d_\sA} + \log{(1/\delta)})$ qubits simultaneously. The circuit description can be computed classically in $\poly\big(n, \log{d_\sA}, \log{(1/\delta)}\big)$ time.
\end{theorem}

\begin{algorithm}[ht]
\caption{Square root fidelity estimation algorithm in the mixed sample access model without collective operations \\ \parbox{\linewidth}{\centering $\mathtt{UhlFidelityEstMixedSample}(\rho^{\sA}, \sigma^{\sA}; \delta)$ \\ (In Theorem~\ref{thm:fidelity est local sample})}}
\label{alg:fidelestsamplemixed}
\SetKwInput{KwInput}{Input}
\SetKwInput{KwOutput}{Output}
\SetKwInput{KwParameters}{Parameters}
\SetKwComment{Comment}{$\triangleright$\ }{}
\SetCommentSty{textnormal}
\SetKwProg{Fn}{Subroutine}{}{end}
\SetKwFunction{UPS}{UhlmannPurifiedSample}
\SetKwFunction{SAES}{SqrtAmpEstSample}

\SetAlgoNoEnd
\SetAlgoNoLine
\KwInput{Two quantum states $\rho^{\sA}$ and $\sigma^{\sA}$.}
\KwParameters{$\delta \in (0, 1)$.}
\KwOutput{Real number $\sqrt{\til{\rF}}$.}
\SetAlgoLined

Set $\delta_2 \gets \delta/2$ and $\delta_1 \gets \delta_2/1920$. \\
Set $\cL_\rho^{\hat{\sA}\hat{\sB}\sG}$ and $\cL_{\rho, {\rm inv}}^{\hat{\sA}\hat{\sB}\sG}$ to be the quantum channels given by $\mathtt{CanonicalPurification}(\rho^{\sA}; \delta_1/(4u), \delta/24)$, and $\cL_\sigma^{\hat{\sA}\hat{\sB}\sG}$ and $\cL_{\sigma, {\rm inv}}^{\hat{\sA}\hat{\sB}\sG}$ to be those given by $\mathtt{CanonicalPurification}(\sigma^{\sA}; \delta_1/(4u), \delta/24)$ (Algorithm~\ref{alg:canonical purification}), where $u$ is the number of uses of $\cG^{\sB\hat{\sA}\hat{\sB}\sG}$ and $\cG_{\rm inv}^{\sB\hat{\sA}\hat{\sB}\sG}$ in $\mathtt{QSVTSIGNFEJER}$ or $\mathtt{QSVTSIGN}$ below. \\
Set $\cL_\rho^{\sA\sB\sG'}$, $\cL_{\rho, {\rm inv}}^{\sA\sB\sG'}$, $\cL_\sigma^{\sA\sB\sG'}$, and $\cL_{\sigma, {\rm inv}}^{\sA\sB\sG'}$ to be the same quantum channels acting on $\sA\sB\sG'$. \\
Set $\cG^{\sB\hat{\sA}\hat{\sB}\sG} \gets \cL_{\rho, {\rm inv}}^{\hat{\sA}\hat{\sB}\sG}\circ\cF^{\sB\hat{\sB}}\circ\cL_\sigma^{\hat{\sA}\hat{\sB}\sG}$ and $\cG_{\rm inv}^{\sB\hat{\sA}\hat{\sB}\sG} \gets \cL_{\sigma, {\rm inv}}^{\hat{\sA}\hat{\sB}\sG}\circ\cF^{\sB\hat{\sB}}\circ\cL_\rho^{\hat{\sA}\hat{\sB}\sG}$, where $\cF^{\sB\hat{\sB}}(\cdot) = F^{\sB\hat{\sB}}(\cdot)F^{\sB\hat{\sB}}$ is the swap channel. \\
\If{$\mathtt{QSVTSIGNFEJER}$ below uses $\cG^{\sB\hat{\sA}\hat{\sB}\sG}$ and $\cG_{\rm inv}^{\sB\hat{\sA}\hat{\sB}\sG}$ fewer times than $\mathtt{QSVTSIGN}$ below}{
  Set $\cJ^{\sB\sD} \gets \mathtt{QSVTSIGNFEJER}(\cG^{\sB\hat{\sA}\hat{\sB}\sG}, \cG_{\rm inv}^{\sB\hat{\sA}\hat{\sB}\sG}; \delta/(4R))$ (Proposition~\ref{prop:QSVT sign Fejer}). \\
}
\Else{
  Set $\cJ^{\sB\sD} \gets \mathtt{QSVTSIGN}(\cG^{\sB\hat{\sA}\hat{\sB}\sG}, \cG_{\rm inv}^{\sB\hat{\sA}\hat{\sB}\sG}; \delta/4, s_{\rm min}/2)$ (Proposition~\ref{prop:FPAA general}). \\
}
Set $\cJ_{\rm inv}^{\sB\sD}$ to be the quantum channel obtained by applying the adjoint circuit of the QSVT above, with $\cG^{\sB\hat{\sA}\hat{\sB}\sG}$ and $\cG_{\rm inv}^{\sB\hat{\sA}\hat{\sB}\sG}$ interchanged. \\
Set $\cE^{\sA\sB\sG'\sD} \gets \cL_{\sigma, {\rm inv}}^{\sA\sB\sG'}\circ\cJ^{\sB\sD}\circ\cL_\rho^{\sA\sB\sG'}$ and $\cE_{\rm inv}^{\sA\sB\sG'\sD} \gets \cL_{\rho, {\rm inv}}^{\sA\sB\sG'}\circ\cJ_{\rm inv}^{\sB\sD}\circ\cL_\sigma^{\sA\sB\sG'}$. \\
Set $\sqrt{\til{\rF}} \gets \mathtt{SqrtAmpEstSample}(\cE^{\sA\sB\sG'\sD}, \cE_{\rm inv}^{\sA\sB\sG'\sD}, \ket{0}^{\sA\sB\sG'\sD}, \ket{0}^{\sA\sB\sG'\sD}; \delta_2)$ (Lemma~\ref{lem:modif square root est sample}). \\
Return $\sqrt{\til{\rF}}$.

\end{algorithm}

The approach is a sample version of \Cref{alg:fidelestquery}, in which the unitaries $U_\rho^{\hat{\sA}\hat{\sB}\sG}$ and $U_\sigma^{\hat{\sA}\hat{\sB}\sG}$ approximated by the channels given by $\mathtt{CanonicalPurification}$ are used as the oracles, as in \Cref{alg:Uhl mixed sample}.
Unlike in \Cref{sec:fidelity estimation PSS}, samples of the canonical purified states are not available.
We thus include $U_\rho^{\sA\sB\sG'}$ and $(U_\sigma^{\sA\sB\sG'})^\dag$ in the target unitary, and estimate $\big|\bra{0}(U_\sigma^{\sA\sB\sG'})^\dag\til{W}^{\sB\sD}U_\rho^{\sA\sB\sG'}\ket{0}^{\sA\sB\sG'\sD}\big|$, where $\til{W}^{\sB\sD}$ is the unitary approximated by $\cJ^{\sB\sD}$.
In this case, the reflections in the quantum phase estimation other than those including $\til{W}^{\sB\sD}$ are about $\ket{0}^{\sA\sB\sG'\sD}$, and all the samples of $\rho^\sA$ and $\sigma^\sA$ are used in the channels given by $\mathtt{CanonicalPurification}$.
The same channels are used on $\sA\sB\sG'$ as in $\cG^{\sB\hat{\sA}\hat{\sB}\sG}$, since \Cref{lem:overlap purified error} evaluates the overlap between the states prepared by the same unitaries as those in $\til{W}^{\sB\sD}$.
As in \Cref{alg:Uhl mixed sample}, the error $\delta/24$ of these unitaries does not need to be divided by $u$, unlike the error $\delta_1/(4u)$ of their implementation from samples.

With the above discussion, we now prove Theorem~\ref{thm:fidelity est local sample}.

\begin{proof}[Proof of Theorem~\ref{thm:fidelity est local sample}]

We first evaluate the value to be estimated.
As in \Cref{sec:proof eval sample mixed}, let $U_\rho^{\hat{\sA}\hat{\sB}\sG}$ and $U_\sigma^{\hat{\sA}\hat{\sB}\sG}$ be the unitaries approximated by the channels given by $\mathtt{CanonicalPurification}$ in \Cref{alg:fidelestsamplemixed}, which satisfy Eq.~\eqref{inteq:canonical bias} with $\epsilon_1 = \delta/24$, and let $U_\rho^{\sA\sB\sG'}$ and $U_\sigma^{\sA\sB\sG'}$ be the same unitaries acting on $\sA\sB\sG'$.
Let $\til{W}^{\sB\sD}$ be the unitary obtained from $\cJ^{\sB\sD}$ by replacing $\cG^{\sB\hat{\sA}\hat{\sB}\sG}$ and $\cG_{\rm inv}^{\sB\hat{\sA}\hat{\sB}\sG}$ with $W^{\sB\hat{\sA}\hat{\sB}\sG}$ in \Cref{sec:proof eval sample mixed} and its inverse, which is a block-encoding of $P_\sgn^{(\rm SV)}(\til{M}^\sB)$, where $\til{M}^\sB$ is given by Eq.~\eqref{inteq:tilde M} and $P_\sgn$ is the polynomial in \Cref{prop:QSVT sign Fejer} or \Cref{prop:FPAA general}.
We define a unitary $W_1^{\sA\sB\sG'\sD}$ by 
\begin{align}
    W_1^{\sA\sB\sG'\sD} = (U_\sigma^{\sA\sB\sG'})^\dag\til{W}^{\sB\sD}U_\rho^{\sA\sB\sG'}.
\end{align}
Then, we have
\begin{align}
    \bra{0}^{\sA\sB\sG'\sD}W_1^{\sA\sB\sG'\sD}\ket{0}^{\sA\sB\sG'\sD}
    = \bra{\til{\sigma}_{\rm c}}^{\sA\sB\sG'}\bra{0}^\sD\til{W}^{\sB\sD}\ket{\til{\rho}_{\rm c}}^{\sA\sB\sG'}\ket{0}^\sD,
\end{align}
where $\ket{\til{\rho}_{\rm c}}^{\sA\sB\sG'}$ and $\ket{\til{\sigma}_{\rm c}}^{\sA\sB\sG'}$ are the states in Eq.~\eqref{inteq:tilde states} on $\sA\sB\sG'$.
We apply \Cref{lem:overlap purified error} with $\sE = \sA\sG'$, $\ket{\phi}^{\sE\sB} = \ket{\rho_{\rm c}}^{\sA\sB}\ket{0}^{\sG'}$, $\ket{\psi}^{\sE\sB} = \ket{\sigma_{\rm c}}^{\sA\sB}\ket{0}^{\sG'}$, $\ket{\til{\phi}}^{\sE\sB} = \ket{\til{\rho}_{\rm c}}^{\sA\sB\sG'}$, $\ket{\til{\psi}}^{\sE\sB} = \ket{\til{\sigma}_{\rm c}}^{\sA\sB\sG'}$, and $U^{\sB\sD} = \til{W}^{\sB\sD}$, for which $\|M^\sB\|_1 = \sqrt{\rF}$ and the condition on $\epsilon = \epsilon_1$ is satisfied by Eqs.~\eqref{inteq:tilde states} and~\eqref{inteq:tilde states sigma}.
When we use $\mathtt{QSVTSIGNFEJER}$ with $\beta = \delta/(4R)$ or $\mathtt{QSVTSIGN}$ with $\delta_3 = \delta/4$ and $\beta = s_{\rm min}/2$, the conditions of item~1 or item~2 are satisfied by \Cref{prop:QSVT sign Fejer} or \Cref{prop:FPAA general}, respectively.
Thus, in both cases, we obtain that
\begin{align}
\label{inteq:fid est mixed poly}
    \Big|\big|\bra{0}W_1^{\sA\sB\sG'\sD}\ket{0}^{\sA\sB\sG'\sD}\big| - \sqrt{\rF}\Big| \leq \f{\delta}{4} + 6\epsilon_1 = \f{\delta}{2}.
\end{align}
Recall that $s_{\rm min}$ and $R$ are the minimum non-zero singular value and the rank of $\sqrt{\sigma^\sA}\sqrt{\rho^\sA}$, respectively.

Next, we consider estimating $\big|\bra{0}W_1^{\sA\sB\sG'\sD}\ket{0}^{\sA\sB\sG'\sD}\big|$.
The channels given by $\mathtt{CanonicalPurification}$ in \Cref{alg:fidelestsamplemixed} satisfy Eqs.~\eqref{inteq:41} and~\eqref{inteq:41 sigma} with the right-hand sides replaced by $\delta_1/(4u)$.
Thus, as in Eq.~\eqref{inteq:78}, we have $\f{1}{2}\|\cJ^{\sB\sD} - \til{\cW}^{\sB\sD}\|_\diamond \leq 2u \cdot \delta_1/(4u) = \delta_1/2$, and similarly, $\f{1}{2}\|\cJ_{\rm inv}^{\sB\sD} - (\til{\cW}^{\sB\sD})^\dag\|_\diamond \leq \delta_1/2$.
From the subadditivity of the diamond norm, the quantum channels $\cE^{\sA\sB\sG'\sD}$ and $\cE_{\rm inv}^{\sA\sB\sG'\sD}$ in \Cref{alg:fidelestsamplemixed} satisfy that
\begin{align}
    &\f{1}{2}\big\|\cE^{\sA\sB\sG'\sD} - \cW_1^{\sA\sB\sG'\sD}\big\|_\diamond \leq \f{\delta_1}{2} + \f{\delta_1}{2u} \leq \delta_1, \\
    &\f{1}{2}\big\|\cE_{\rm inv}^{\sA\sB\sG'\sD} - (\cW_1^{\sA\sB\sG'\sD})^\dag\big\|_\diamond \leq \f{\delta_1}{2} + \f{\delta_1}{2u} \leq \delta_1.
\end{align}
Let $\sqrt{\til{\rF}} = \mathtt{SqrtAmpEstSample}(\cE^{\sA\sB\sG'\sD}, \cE_{\rm inv}^{\sA\sB\sG'\sD}, \ket{0}^{\sA\sB\sG'\sD},$ $\ket{0}^{\sA\sB\sG'\sD}; \delta_2)$.
From \Cref{lem:modif square root est sample} with $W = W_1^{\sA\sB\sG'\sD}$, $\ket{\phi} = \ket{\psi} = \ket{0}^{\sA\sB\sG'\sD}$, and $\epsilon' = \delta_1$, we observe that $\big|\sqrt{\til{\rF}} - |\bra{0}W_1^{\sA\sB\sG'\sD}\ket{0}^{\sA\sB\sG'\sD}|\big| \leq \delta_2$ with probability at least $2/3$, using $\cE^{\sA\sB\sG'\sD}$ and $\cE_{\rm inv}^{\sA\sB\sG'\sD}$ a total of $\cO(1/\delta_2)$ times and $\cO\big(1/\big(\delta_2(\delta_2/960 - \delta_1)\big)\big) = \cO(1/\delta^2)$ samples of $\ket{0}^{\sA\sB\sG'\sD}$, where we used $\delta_1 = \delta_2/1920$ and $\delta_2 = \delta/2$.
Using the triangle inequality and Eq.~\eqref{inteq:fid est mixed poly}, we conclude that $\big|\sqrt{\til{\rF}} - \sqrt{\rF}\big| \leq \delta_2 + \delta/2 = \delta$ with probability at least $2/3$.

Each use of $\cE^{\sA\sB\sG'\sD}$ or $\cE_{\rm inv}^{\sA\sB\sG'\sD}$ uses $\cG^{\sB\hat{\sA}\hat{\sB}\sG}$ and $\cG_{\rm inv}^{\sB\hat{\sA}\hat{\sB}\sG}$ a total of $u$ times and a pair of the channels acting on $\sA\sB\sG'$.
Each of them uses $f$ samples of $\rho^\sA$ and $\sigma^\sA$ in Eq.~\eqref{inteq:43}, where the parameters of $\mathtt{CanonicalPurification}$ are $\delta_1/(4u)$ and $\delta/24$, that is,
\begin{align}
    f = \til{\cO}\Big(\f{d_\sA u}{\delta}\min\Big\{\kappa_\rho^2 + \kappa_\sigma^2, \f{r_\rho^2 + r_\sigma^2}{\delta^4}\Big\}\Big).
\end{align}
The samples of $\ket{0}^{\sA\sB\sG'\sD}$ require no samples of $\rho^\sA$ and $\sigma^\sA$.
Thus, the number of samples $n$ of $\rho^\sA$ and $\sigma^\sA$ is given by
\begin{align}
    n 
    &= \cO(1/\delta)\cO(u+1)f \\
    &= \til{\cO}\Big(\f{d_\sA u^2}{\delta^2}\min\Big\{\kappa_\rho^2 + \kappa_\sigma^2, \f{r_\rho^2 + r_\sigma^2}{\delta^4}\Big\}\Big).
\end{align}
Since \Cref{alg:fidelestsamplemixed} uses whichever of $\mathtt{QSVTSIGN}$ and $\mathtt{QSVTSIGNFEJER}$ uses $\cG^{\sB\hat{\sA}\hat{\sB}\sG}$ and $\cG_{\rm inv}^{\sB\hat{\sA}\hat{\sB}\sG}$ fewer times, we have $u = \cO\big(\min\{\log{(1/\delta)}/s_{\rm min}, R/\delta\}\big)$ as in \Cref{sec:uhlfidquery}, and hence
\begin{align}
    n = \til{\cO}\bigg(\f{d_\sA}{\delta^2}\min\Big\{\f{1}{s_{\rm min}^2}, \f{R^2}{\delta^2}\Big\}\min\Big\{\kappa_\rho^2 + \kappa_\sigma^2, \f{r_\rho^2 + r_\sigma^2}{\delta^4}\Big\}\bigg).
\end{align}

In the quantum circuit of this algorithm, the most costly component is the canonical purification, which uses $\cO(\log d_\sA)$ one- and two-qubit gates per sample of $\rho^\sA$ and $\sigma^\sA$ (Theorem~\ref{prop:canonical purification}).
The rest of the algorithm, i.e., $\mathtt{QSVTSIGNFEJER}$ (\Cref{prop:QSVT sign Fejer}) or $\mathtt{QSVTSIGN}$ (\Cref{prop:FPAA general}), and the square root amplitude estimation, use fewer gates. Hence, the total number of one- and two-qubit gates is $\cO(n \log d_\sA)$.
Considering both the QSVT and the square root amplitude estimation, $\cO(\log d_\sA + \log(1/\delta))$ qubits are used at any one time in the circuit.
The circuit description can be computed classically in $\poly\big(n, \log{d_\sA}, \log{(1/\delta)}\big)$ time.

\end{proof}

\section{Application 2: Decoupling approach with the Uhlmann transformation algorithm}
\label{sec:decoupling and uhlmann}

The decoupling approach~\cite{hayden2008decoupling, dupuis2010decoupling, Szehr2013decouple2design, dupuis2014one, preskill2016quantumshannon} and the Uhlmann transformation are combined and applied to various important tasks in quantum information processing.
We here consider two examples: quantum communication~\cite{lloyd1997capacity, shor2002quantum, devetak2005private, barnum1998inftrans, datta2013oneshotentassistQCcom, wakakuwa2023oneshot, khatri2024principlemodern} and quantum state merging~\cite{Horodecki2005Partialquantinfo, berta2008singleshot, dupuis2014one, yamasaki2019statemergesmalldimension}.

We employ our algorithms to perform the Uhlmann transformation in these tasks and evaluate the computational cost.
Since these tasks typically involve tripartite or multipartite states, some adjustments are often required beyond the direct application of our Uhlmann transformation algorithms.
For discussions on the optimality of the Uhlmann transformation in the cases involving tripartite or multipartite states, see \Cref{sec:opt local and Uhlmann}.

In \Cref{sec:general decoup and Uhl}, we overview the decoupling and the Uhlmann transformation. In \Cref{sec:qi transmission}, we discuss one-shot quantum communication and analyze the computational costs when we use our Uhlmann transformation algorithms.
In \Cref{sec:entdis and uhl}, we address one-shot quantum state merging and also evaluate the computational costs of our algorithms.

In this section and subsequent sections, we use the following notation.
We denote by $U_\Phi$ a unitary that prepares a maximally entangled state: $U_\Phi^{\sA\sA'}\ket{0}^{\sA\sA'} = \ket{\Phi}^{\sA\sA'}$.
While we usually denote a pure state as $\ket{\omega}$, the corresponding density matrix is sometimes written as $\omega = \ketbra{\omega}{\omega}$. For instance, the density matrix of a pure state $\ket{\omega}^{\sA\sB}$ is denoted by $\omega^{\sA\sB}$.


\subsection{Decoupling and Uhlmann's theorem}
\label{sec:general decoup and Uhl}

The decoupling is a standard approach, often combined with the Uhlmann transformation, that characterizes information-theoretic limits of information tasks.
The key idea of the decoupling is to estimate how much quantum information is leaked to an ``environment'' of a quantum channel. This is specifically quantified by the degree of decoupling.

For a quantum channel $\cF^{\sA\rarr\sC}$, let $V_\cF^{\sA \rarr \sC\sE}$ be a Stinespring isometry~\cite{stinespring1955} of $\cF^{\sA \rarr \sC}$ by an environment $\sE$.
That is, $V_\cF^{\sA \rarr \sC\sE}$ is such that $\cF^{\sA\rarr \sC}$ is represented as $\cF^{\sA \rarr \sC}(\cdot) = \tr_{\sE}\big[V_\cF^{\sA\rarr \sC\sE}(\cdot)(V_\cF^{\sA \rarr \sC\sE})^\dag\big]$.
A complementary channel $\bar{\cF}^{\sA\rarr\sE}$ is defined by $\bar{\cF}^{\sA\rarr\sE}(\cdot) = \tr_\sC\big[V_\cF^{\sA\rarr\sC\sE}(\cdot)(V_\cF^{\sA\rarr\sC\sE})^\dag\big]$.
Note that the Stinespring isometry $V_\cF^{\sA \rarr \sC\sE}$ and the complementary channel $\bar{\cF}^{\sA\rarr\sE}$ are not uniquely determined by $\cF^{\sA\rarr \sC}$; there is freedom to apply an additional isometry on the environment $\sE$.

The decoupling approach is described as follows: suppose that for a pure state $\ket{\omega}^{\sR\sA\sB}$ with a reference system $\sR$ and $\epsilon \in [0, 1]$, there exists a state $\tau^\sE$ such that 
\begin{align}
\label{eq:decoup cond}
    \rF\big(\bar{\cF}^{\sA\rarr\sE}(\omega^{\sR\sA}), \omega^\sR \otimes \tau^\sE\big) \geq 1 - \epsilon.
\end{align}
Then, there is the Uhlmann unitary $V^\sM$ on $\sM = \sB\sC\sD = \hat{\sA}\hat{\sB}\hat{\sE}$
that satisfies
\begin{align}
    &\rF\big(V^\sM V_\cF^{\sA\rarr\sC\sE}\ket{\omega}^{\sR\sA\sB}\ket{0}^\sD, \ket{\omega}^{\sR\hat{\sA}\hat{\sB}}\ket{\tau}^{\sE\hat{\sE}}\big) \notag\\
    &\hspace{1pc}= \rF\big(\bar{\cF}^{\sA\rarr\sE}(\omega^{\sR\sA}), \omega^\sR \otimes \tau^\sE\big) \\
    \label{inteq:80}
    &\hspace{1pc}\geq 1 - \epsilon,
\end{align}
where $\ket{\tau}^{\sE\hat{\sE}}$ is a purified state of $\tau^\sE$.
This implies that, if $\epsilon$ is sufficiently small, one can approximately obtain the state  $\ket{\omega}^{\sR\hat{\sA}\hat{\sB}}\ket{\tau}^{\sE\hat{\sE}}$, by acting only on $\sB\sC\sD$ of the state 
$V_\cF^{\sA\rarr\sC\sE}\ket{\omega}^{\sR\sA\sB}\ket{0}^\sD$,
without manipulating $\sR\sE$.
The inequality in Eq.~\eqref{eq:decoup cond} is known as the \emph{decoupling condition}, which evaluates how well the reference system $\sR$ and the environment $\sE$ are decoupled, with the parameter $\epsilon$ quantifying the degree of decoupling.

The decoupling approach is applicable to a wide range of information-processing tasks.
In particular, the case where $\cF^{\sA\rarr\sC}$ is given by $\cF^{\sA\rarr\sC} = \cG^{\sA\rarr\sC} \circ \cU^\sA$ with a \emph{Haar random unitary} $U^\sA$ and some quantum channel $\cG^{\sA\rarr\sC}$ is often considered.
In this case, the \emph{decoupling theorem}~\cite{dupuis2014one} states that $\epsilon$ in the decoupling condition can be taken sufficiently small with high probability for $\tau^\sE = \bar{\cG}^{\sA\rarr\sE}(\pi^\sA)$, depending on a certain entropy condition specified by the initial state $\ket{\omega}^{\sR\sA\sB}$ and the channel $\cG^{\sA\rarr\sC}$.
Here, $\bar{\cG}^{\sA\rarr\sE}(\pi^\sA)$ is a reduced state on $\sE$ of the \emph{Choi--Jamio\l{}kowski state} $\bar{\cG}^{\sA\rarr\sE}(\Phi^{\sR\sA})$ of $\bar{\cG}^{\sA\rarr\sE}$.
The dependence of $\epsilon$ on $\ket{\omega}^{\sR\sA\sB}$ and $\cG^{\sA\rarr\sC}$ via the entropy condition plays a key role in deriving the theoretical limits of some tasks, such as the quantum capacity~\cite{hayden2008decoupling}.
Yet we do not address this dependence explicitly here, as it is not crucial in our discussion (for details, see, e.g., Refs.~\cite{dupuis2010decoupling, dupuis2014one, khatri2024principlemodern}).

Henceforth, following this convention, we use a Haar random unitary $U^\sA$ and fix $\tau^\sE$ as $\tau^\sE = \bar{\cG}^{\sA\rarr\sE}(\pi^\sA)$.
We remark that even if the Haar random unitary in these tasks is replaced by a unitary 2-design for practical purposes, a similar discussion holds~\cite{Szehr2013decouple2design}.
Unlike the Haar random unitary, the unitary 2-design can be realized efficiently as quantum circuits, e.g., using the Clifford group~\cite{DiVincenzoQdatahiding}, near-linear construction~\cite{cleve2016nearlinearexact2design}, or random unitaries diagonal in the Pauli-$Z$ and -$X$ bases~\cite{Nakata2017randomdiagonal2design}.


\subsection{One-shot quantum communication}
\label{sec:qi transmission}

An important application of the decoupling and Uhlmann's theorem is quantum communication.
We describe the setting in \Cref{seq:setting of ent transmission}, and then, in \Cref{sec:analyze of uhl alg in ent trans}, we investigate the implementation of the Uhlmann transformation for quantum communication, based on our algorithm.

\subsubsection{Setting}
\label{seq:setting of ent transmission}

Alice aims to send a system $\sA$ of a pure state $\ket{\rho}^{\sR\sA}$ to Bob through a noisy quantum channel, where $\sR$ is a reference system.
They may share $(\log{d_\sG})$-ebit entanglement $\ket{\Phi}^{\sG\hat{\sG}}$ in advance, which can be used in the encoding and decoding processes, where $\sG$ is held by Alice and $\hat{\sG}$ by Bob.
When $d_\sG=1$, it is the entanglement-non-assisted setting; otherwise, the entanglement-assisted setting.
Alice encodes the system $\sA$ together with $\sG$ using a Haar random unitary $U^{\sA\sG}$, then sends $\sA\sG$ through a noisy quantum channel $\cN^{\sA\sG\rarr\sB}$ to Bob.
After receiving the system $\sB$, Bob applies a decoding map $\cD^{\sB\hat{\sG}\rarr\hat{\sA}}$.
See also \Cref{fig:diagram of QI transmitt}.

We denote by $\bar{\cN}^{\sA\sG\rarr\sE}$ a complementary channel of $\cN^{\sA\sG\rarr\sB}$, where $\sE$ is an environment.
The decoupling theorem~\cite{dupuis2014one} guarantees that an encoding Haar random unitary $U^{\sA\sG}$ satisfies with high probability that 
\begin{align}
\label{inteq:102}
    \rF\big(\bar{\cN}^{\sA\sG\rarr\sE}\circ\cU^{\sA\sG}(\rho^{\sR\sA} \otimes \pi^\sG), \rho^\sR \otimes \tau^\sE\big) \geq 1 - \epsilon,
\end{align}
where $\tau^\sE = \bar{\cN}^{\sA\sG\rarr\sE}(\pi^{\sA\sG})$.
The parameter $\epsilon$ depends on the state $\ket{\rho}^{\sR\sA}$, the channel $\cN^{\sA\sG\rarr\sB}$, and the number of pre-shared ebits.
Note that Eq.~\eqref{inteq:102} holds independently of the choice of the complementary channel of $\cN^{\sA\sG\rarr\sB}$, as the complementary channel has the freedom to apply an additional isometry on $\sE$ and the fidelity is invariant under such an isometry.

Due to Eq.~\eqref{inteq:80}, there is a quantum channel $\cD^{\sB\hat{\sG} \rarr \hat{\sA}} = \tr_{\hat{\sE}} \circ \cV^\sM \circ \cP_{\ket{0}}^{\bC \rarr \sD}$, where $\sM = \sB\hat{\sG}\sD = \hat{\sA}\hat{\sE}$ and $V^\sM$ is the Uhlmann unitary, which satisfies
\begin{align}
\label{inteq:101}
    \rF\big(\cD^{\sB\hat{\sG} \rarr \hat{\sA}}\circ \cN^{\sA\sG \rarr \sB} \circ \cU^{\sA\sG} (\rho^{\sR\sA} \otimes \Phi^{\sG\hat{\sG}}), \ket{\rho}^{\sR\hat{\sA}}\big) \geq 1 - \epsilon,
\end{align}
where we used Eq.~\eqref{inteq:102}.
This implies that, under the condition that $\sR$ and $\sE$ are sufficiently decoupled in the sense of Eq.~\eqref{inteq:102} with small $\epsilon$, Bob can recover the original state $\ket{\rho}^{\sR\sA}$ from the output of the noisy channel $\cN^{\sA\sG\rarr\sB}$ by applying $\cD^{\sB\hat{\sG}\rarr\hat{\sA}}$, succeeding in transmitting quantum information from Alice to Bob.
We call this quantum channel $\cD^{\sB\hat{\sG}\rarr\hat{\sA}}$ the \emph{Uhlmann decoder}.

\begin{figure}
    \centering
    \includegraphics[width=85mm]{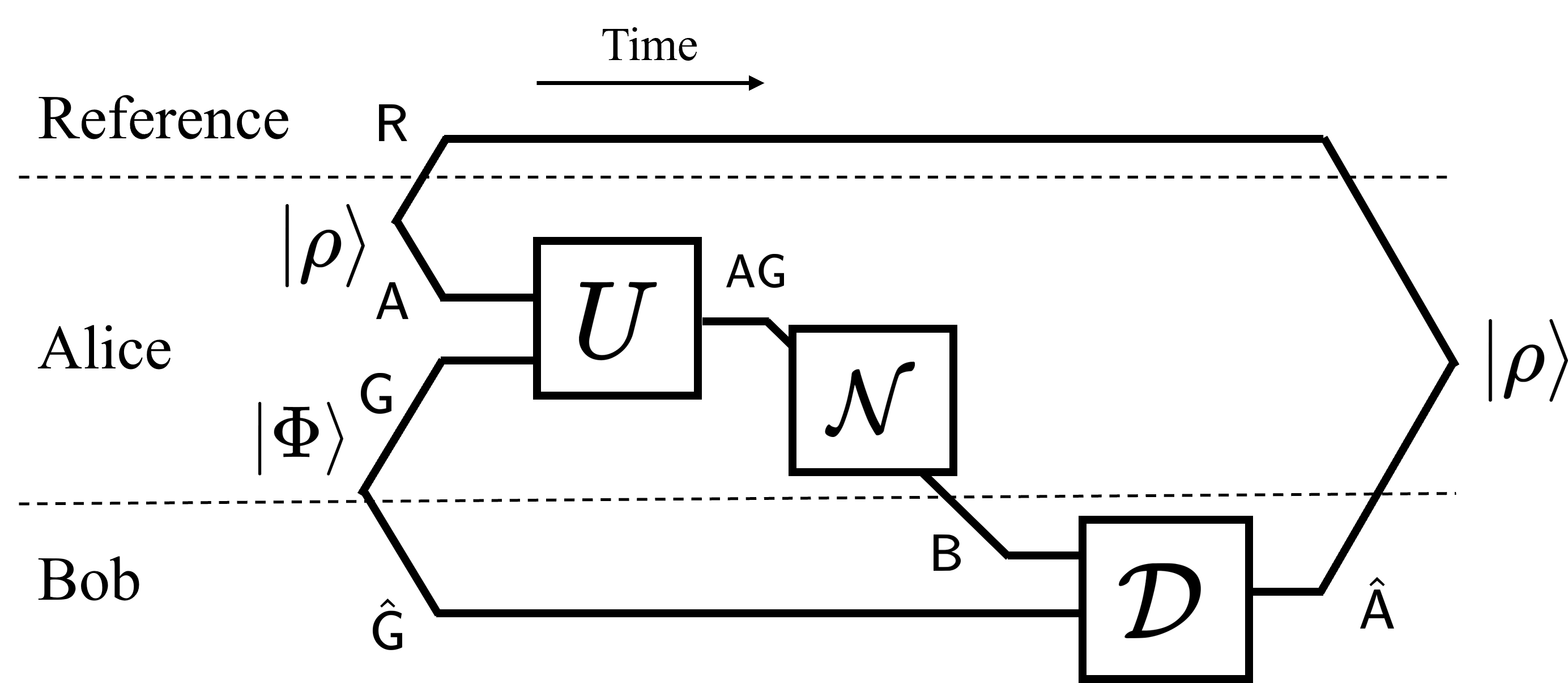}
    \caption{A diagram of quantum communication. 
    Alice applies a Haar random unitary $U^{\sA\sG}$ to encode $\sA$ of a pure state $\ket{\rho}^{\sR\sA}$ with $\sG$, which is part of a possibly pre-shared entangled state $\ket{\Phi}^{\sG\hat{\sG}}$. Alice then transmits $\sA\sG$ to Bob through a noisy quantum channel $\cN^{\sA\sG\rarr\sB}$.
    After receiving $\sB$, Bob applies a decoding map $\cD^{\sB\hat{\sG}\rarr\hat{\sA}}$ on $\sB\hat{\sG}$ to obtain a final state that approximates the initial state $\ket{\rho}^{\sR\hat{\sA}}$.}
    \label{fig:diagram of QI transmitt}
\end{figure}


\subsubsection{Algorithmic implementation of the Uhlmann decoder}
\label{sec:analyze of uhl alg in ent trans}

We consider that Bob implements the Uhlmann decoder using our algorithm.
We assume that Bob knows Alice's operation $U^{\sA\sG}$ and is capable of applying it.
It is standard in decoding to assume that Bob knows how Alice encoded the system.
Even if $U^{\sA\sG}$ is chosen uniformly at random, Bob can determine which unitary was applied, for example, using shared randomness with Alice.
We also assume that Bob knows the state $\ket{\rho}^{\sR\sA}$ and can prepare it.
Let $U_\cN^{\sS}$ be a Stinespring unitary on $\sS = \sA\sG\sC = \sB\sE$, i.e., $U_\cN^{\sS}\ket{0}^{\sC} = V_\cN^{\sA\sG\rarr\sB\sE}$, where $V_\cN^{\sA\sG \to \sB\sE}$ is a Stinespring isometry of $\cN^{\sA\sG \to \sB}$, and $d_\sE$ can be taken as $d_\sA d_\sG d_\sB$ in general.
For algorithmic construction, we consider two different models: 
\begin{enumerate}
    \item A model in which multiple uses of a Stinespring unitary $U_\cN^{\sS'}$ are available.
    \item A model in which multiple uses of the quantum channel $\cN^{\sA'\sG' \to \sB'}$ are available.
\end{enumerate}
We evaluate the number of uses of $U_\cN^{\sS'}$ or $\cN^{\sA'\sG'\rarr\sB'}$ by Bob to implement the Uhlmann decoder in each model.
Note that we should distinguish between the existing systems, such as $\sR$ and $\sS$, and the systems subsequently prepared and simulated by Bob, e.g., $\sS'$.
In Model~2, Bob first applies $\mathtt{RandomPurification}$ in \Cref{prevthm:random purification} to the states prepared by $\cN^{\sA'\sG'\rarr\sB'}$, and then runs the Uhlmann transformation algorithm in the purified sample access model.

We define a pure state $\ket{\Psi}^{\sR\sB\hat{\sG}\sE}$ as 
\begin{align}
    \ket{\Psi}^{\sR\sB\hat{\sG}\sE} = V_\cN^{\sA\sG\rarr\sB\sE} U^{\sA\sG} \ket{\rho}^{\sR\sA}\ket{\Phi}^{\sG\hat{\sG}}.
\end{align}
For a matrix $A$, we use $\lambda_{\rm min}(A)$ and $r(A)$ to denote the minimum non-zero eigenvalue and the rank of $A$, respectively.
Our statement is as follows, where $\mu_{\rm min}$ is given by $\mu_{\rm min} = \sqrt{\lambda_{\rm min}(\rho^\sR)\lambda_{\rm min}(\Psi^{\sB\hat{\sG}})\lambda_{\rm min}(\Psi^{\sR\sB\hat{\sG}})}$.

\begin{theorem}[Algorithmic implementation of the Uhlmann decoder]
\label{thm:Uhlmann decoder ent trans}
In the setting of quantum communication, suppose that an encoding unitary $U^{\sA\sG}$ satisfies
\begin{align}
\label{inteq:136}
    \rF\big(\bar{\cN}^{\sA\sG\rarr\sE}\circ\cU^{\sA\sG}(\rho^{\sR\sA} \otimes \pi^\sG), \rho^\sR \otimes \tau^\sE\big) \geq 1 - \epsilon,
\end{align}
for some $\epsilon \in [0, 1]$ and $\tau^\sE = \bar{\cN}^{\sA\sG\rarr\sE}(\pi^{\sA\sG})$.
Then, for $\delta \in (0, 1)$, there exists a quantum algorithm that implements a quantum channel $\til{\cD}^{\sB\hat{\sG}\rarr\hat{\sA}}$, which satisfies
\begin{align}
\label{inteq:157}
    \rF\big(\til{\cD}^{\sB\hat{\sG} \rarr \hat{\sA}}(\Psi^{\sR\sB\hat{\sG}}), \ket{\rho}^{\sR\hat{\sA}}\big) \geq 1 - 4\epsilon -\delta,
\end{align}
using, in Model~1: 
\begin{align}
\label{inteq:121}
    q
    = \cO\Big(\min\Big\{\f{1}{\mu_{\rm min}}\log{\Big(\f{1}{\delta}\Big)}, \f{r(\Psi^{\sB\hat{\sG}})}{\delta}\Big\}\Big)
\end{align}
uses of a Stinespring unitary $U_\cN^{\sS'}$ and its inverse, or in Model~2:
\begin{align}
\label{inteq:137}
    q 
    = \cO\Big(\min\Big\{\f{r(\Psi^{\sB\hat{\sG}})}{\delta\mu_{\rm min}}\Big(\log{\Big(\f{1}{\delta}\Big)}\Big)^2, \f{r(\Psi^{\sB\hat{\sG}})^2}{\delta^2}\Big\}\Big)
\end{align}
uses of the quantum channel $\cN^{\sA'\sG'\rarr\sB'}$.
The quantum circuit uses $\cO\big(\log{(d_\sR d_\sA d_\sB d_\sG)}\big)$ qubits at any one time in Model~1, and $\til{\cO}(q^2)$ qubits in Model~2.
In both models, the circuit description can be computed classically in $\poly\big(q, \log{(d_\sR d_\sA d_\sB d_\sG)}, \log{(1/\delta)}\big)$ time.
\end{theorem}

In both models, the decoder is implemented efficiently when $U_\cN^{\sS'}$ or $\cN^{\sA'\sG'\rarr\sB'}$, the encoding unitary $U^{\sA\sG}$, and the preparation of $\ket{\rho}^{\sR\sA}$ are realized with polynomially many gates, and the rank $r(\Psi^{\sB\hat{\sG}})$ of Bob's marginal and the inverse $1/\delta$ of the algorithmic error are polynomial in $\log{(d_\sR d_\sA d_\sB d_\sG)}$. Indeed, to the leading order, the circuit complexity is the number of one- and two-qubit gates of $U_\cN^{\sS'}$ or $\cN^{\sA'\sG'\rarr\sB'}$, together with $U^{\sA\sG}$ and the preparation of $\ket{\rho}^{\sR\sA}$, multiplied by the number of its uses in Eq.~\eqref{inteq:121} or Eq.~\eqref{inteq:137}, which is bounded by the terms $r(\Psi^{\sB\hat{\sG}})/\delta$ and $r(\Psi^{\sB\hat{\sG}})^2/\delta^2$ in the minima.
In Model~2, the random purification, which handles the $q$ copies of $\Psi^{\sR'\sB'\hat{\sG}'}$ collectively, adds $\til{\cO}(q^4)$ gates, and the number of qubits and the classical computation time in \Cref{thm:Uhlmann decoder ent trans} are also polynomial.

Before proceeding to the proof of Theorem~\ref{thm:Uhlmann decoder ent trans}, we briefly compare the number of uses of $U_\cN^{\sS'}$ in Eq.~\eqref{inteq:121} with those of existing decoders, the generalized Yoshida--Kitaev decoder, and the Petz-like decoder~\cite{utsumi2024explicitdecoderQSVTt}.
Here, we consider the case where $\ket{\rho}^{\sR\sA} = \ket{\Phi}^{\sR\sA}$ with $d_\sR = d_\sA$, i.e., the entanglement transmission task~\cite{Schumacher1996sendingentanglement, schumacher2001approximateerrorcorrection, hayden2007black, datta2013oneshotentassistQCcom, beigi2016decoding, khatri2024principlemodern}.
In this case, $\mu_{\rm min}$ reduces to $\sqrt{\lambda_{\rm min}(\Psi^{\sB\hat{\sG}})\lambda_{\rm min}(\Psi^{\sR\sB\hat{\sG}})/d_\sA}$, and the number of uses of $U_\cN^{\sS'}$ in these decoders is given by
\begin{align}
\label{inteq:82}
    \cO\Big(\sqrt{d_\sB/\big(d_\sG\lambda_{\rm min}(\Psi^{\sB\hat{\sG}})\big)}\log{(1/\delta)}\Big),
\end{align}
for the generalized Yoshida--Kitaev decoder, and
\begin{align}
\label{inteq:83}
    \cO\Big(\sqrt{d_\sE/\big(d_\sA\lambda_{\rm min}(\Psi^{\sB\hat{\sG}})\big)} \log{(1/\delta)}\Big),
\end{align}
for the Petz-like decoder.
From the comparison of Eqs.~\eqref{inteq:121} and~\eqref{inteq:82}, we see that if an inequality 
\begin{align}
    \f{d_\sG}{d_\sB}\min\bigg\{\f{d_\sA}{\lambda_{\rm min}(\Psi^{\sR\sB\hat{\sG}})}, \f{1}{\delta^2\lambda_{\rm min}(\Psi^{\sB\hat{\sG}})}\bigg\} \leq 1,
\end{align}
is satisfied, the Uhlmann decoder can be implemented more efficiently than the generalized Yoshida--Kitaev one, where we used the relation $r(\Psi^{\sB\hat{\sG}}) \leq 1/\lambda_{\rm min}(\Psi^{\sB\hat{\sG}})$.
Similarly, comparing Eqs.~\eqref{inteq:121} and~\eqref{inteq:83}, we observe that when $\f{d_\sA}{d_\sE}\min\Big\{\f{d_\sA}{\lambda_{\rm min}(\Psi^{\sR\sB\hat{\sG}})}, \f{1}{\delta^2\lambda_{\rm min}(\Psi^{\sB\hat{\sG}})}\Big\} \leq 1$, the Uhlmann decoder can be implemented more efficiently than the Petz-like decoder.

\begin{proof}[Proof of Theorem~\ref{thm:Uhlmann decoder ent trans}]

In both models, we construct the Uhlmann transformation on $\sM = \sB\hat{\sG}\hat{\sR}\hat{\sA}$ that approximately maps $\ket{\Psi}^{\sR\sB\hat{\sG}\sE}\ket{0}^{\hat{\sR}\hat{\sA}}$ to $\ket{\rho}^{\sR\hat{\sA}}\ket{\Psi}^{\hat{\sR}\sB\hat{\sG}\sE}$, where $\hat{\sR}$ is isomorphic to $\sR$.
Here, the purified state $\ket{\tau}^{\sE\hat{\sE}}$ of $\tau^\sE$ is replaced by $\ket{\Psi}^{\hat{\sR}\sB\hat{\sG}\sE}$, whose reduced state on $\sE$ is $\Psi^\sE$ rather than $\tau^\sE$.
This allows Bob to prepare both the input and the output states of the transformation from copies of the single state $\Psi^{\sR'\sB'\hat{\sG}'}$ in Model~2.
The Uhlmann unitary $V^\sM$ for these states satisfies
\begin{align}
   \label{inteq:159} \rF\big(V^\sM\ket{\Psi}^{\sR\sB\hat{\sG}\sE}\ket{0}^{\hat{\sR}\hat{\sA}}, \ket{\rho}^{\sR\hat{\sA}}\ket{\Psi}^{\hat{\sR}\sB\hat{\sG}\sE}\big) = \rF(\Psi^{\sR\sE}, \rho^\sR\otimes\Psi^\sE),
\end{align}
where $\Psi^{\sR\sE} = \bar{\cN}^{\sA\sG\rarr\sE}\circ\cU^{\sA\sG}(\rho^{\sR\sA} \otimes \pi^\sG)$.
We first show that the right-hand side is at least $1 - 4\epsilon$ under the decoupling condition Eq.~\eqref{inteq:136}.

Let $P(\cdot, \cdot) = \sqrt{1 - \rF(\cdot, \cdot)}$ be the purified distance (also known as the sine distance), which satisfies the triangle inequality and monotonicity under quantum channels~\cite{tomamichel2010duality}.
From Eq.~\eqref{inteq:136}, we have $P(\Psi^{\sR\sE}, \rho^\sR\otimes\tau^\sE) \leq \sqrt{\epsilon}$. The monotonicity under the partial trace over $\sR$ implies $P(\Psi^\sE, \tau^\sE) \leq \sqrt{\epsilon}$, and hence,
\begin{align}
    &P(\Psi^{\sR\sE}, \rho^\sR\otimes\Psi^\sE) \notag\\
    &\leq P(\Psi^{\sR\sE}, \rho^\sR\otimes\tau^\sE) + P(\rho^\sR\otimes\tau^\sE, \rho^\sR\otimes\Psi^\sE) \\
    &\leq 2\sqrt{\epsilon},
\end{align}
where we used $P(\rho^\sR\otimes\tau^\sE, \rho^\sR\otimes\Psi^\sE) = P(\tau^\sE, \Psi^\sE)$.
Consequently, we obtain 
\begin{align}
\label{inteq:decoup Psi}
    \rF(\Psi^{\sR\sE}, \rho^\sR\otimes\Psi^\sE) \geq 1 - 4\epsilon.
\end{align}

We now consider Model~1.
Let $U_\rho^{\sR\sA}$ be a unitary preparing $\ket{\rho}^{\sR\sA}$: $U_\rho^{\sR\sA}\ket{0}^{\sR\sA} = \ket{\rho}^{\sR\sA}$.
The states $\ket{\Psi}^{\sR\sB\hat{\sG}\sE}\ket{0}^{\hat{\sR}\hat{\sA}}$ and $\ket{\rho}^{\sR\hat{\sA}}\ket{\Psi}^{\hat{\sR}\sB\hat{\sG}\sE}$ are rephrased as 
$\ket{\Psi}^{\sR\sB\hat{\sG}\sE}\ket{0}^{\hat{\sR}\hat{\sA}} = W_\cN^{\sR\sS\hat{\sG}\hat{\sR}\hat{\sA}}\ket{0}^{\sR\sS\hat{\sG}\hat{\sR}\hat{\sA}}$ and $\ket{\rho}^{\sR\hat{\sA}}\ket{\Psi}^{\hat{\sR}\sB\hat{\sG}\sE} = \til{W}_\cN^{\sR\sS\hat{\sG}\hat{\sR}\hat{\sA}}\ket{0}^{\sR\sS\hat{\sG}\hat{\sR}\hat{\sA}}$, where
\begin{align}
    &W_\cN^{\sR\sS\hat{\sG}\hat{\sR}\hat{\sA}} = U_\cN^{\sS}U^{\sA\sG}(U_\rho^{\sR\sA}\otimes U_\Phi^{\sG\hat{\sG}})\otimes \bI^{\hat{\sR}\hat{\sA}}, \\ 
    &\til{W}_\cN^{\sR\sS\hat{\sG}\hat{\sR}\hat{\sA}} = U_\rho^{\sR\hat{\sA}} \otimes U_\cN^{\sS}U^{\sA\sG}(U_\rho^{\hat{\sR}\sA}\otimes U_\Phi^{\sG\hat{\sG}}),
\end{align}
and $\sS = \sA\sG\sC = \sB\sE$.
Using
\begin{align}
    \mathtt{UhlmannPurifiedQuery}(W_\cN^{\sR'\sS'\hat{\sG}'\hat{\sR}'\hat{\sA}'}, \til{W}_\cN^{\sR'\sS'\hat{\sG}'\hat{\sR}'\hat{\sA}'}; \delta, \tF),
\end{align}
described in \Cref{alg:Uhl purif query}, we can implement a quantum channel $\cT^\sM$ on $\sM$, which approximates the Uhlmann transformation and satisfies
\begin{align}
    &\rF\big(\cT^\sM(\ketbra{\Psi}{\Psi}^{\sR\sB\hat{\sG}\sE}\otimes\ketbra{0}{0}^{\hat{\sR}\hat{\sA}}), \ket{\rho}^{\sR\hat{\sA}}\ket{\Psi}^{\hat{\sR}\sB\hat{\sG}\sE}\big) \notag\\
    &\geq \rF(\Psi^{\sR\sE}, \rho^\sR\otimes\Psi^\sE) -\delta \\
    &\geq 1-4\epsilon - \delta,
\end{align}
where we used Theorem~\ref{thm:Uhlmann alg purif query model} and Eqs.~\eqref{inteq:159} and~\eqref{inteq:decoup Psi}.
From the monotonicity of the fidelity, a quantum channel $\til{\cD}^{\sB\hat{\sG}\rarr\hat{\sA}} = \tr_{\sB\hat{\sG}\hat{\sR}}\circ\cT^\sM \circ \cP_{\ket{0}}^{\bC\rarr\hat{\sR}\hat{\sA}}$ serves as a decoder such that
\begin{align}
    \rF(\til{\cD}^{\sB\hat{\sG}\rarr\hat{\sA}}(\Psi^{\sR\sB\hat{\sG}}), \ket{\rho}^{\sR\hat{\sA}}) 
    \geq 1 - 4\epsilon - \delta.
\end{align}

The number of uses of $U_\cN^{\sS'}$ for applying the channel $\til{\cD}^{\sB\hat{\sG}\rarr\hat{\sA}}$, especially $\cT^\sM$, can be evaluated by Theorem~\ref{thm:Uhlmann alg purif query model} as
\begin{align}
\label{inteq:81}
    q = \cO\big(\min\{\log{(1/\delta)}/s_{\rm min}, R/\delta\}\big), 
\end{align}
where $s_{\rm min}$ and $R$ are the minimum non-zero singular value and the rank of $\sqrt{\Psi^{\sR\sE}}\big(\sqrt{\rho^\sR}\otimes\sqrt{\Psi^\sE}\big)$, respectively.

By applying the following two inequalities:
\begin{align}
    \label{inteq:138}
    &1/s_{\rm min} \leq 1/\sqrt{\lambda_{\rm min}(\rho^\sR)\lambda_{\rm min}(\Psi^{\sB\hat{\sG}})\lambda_{\rm min}(\Psi^{\sR\sB\hat{\sG}})} = 1/\mu_{\rm min},  \\
    \label{inteq:139}
    &R \leq r(\Psi^{\sB\hat{\sG}}),
\end{align}
to Eq.~\eqref{inteq:81}, we finally obtain Eq.~\eqref{inteq:121}.
These inequalities follow from the so-called support lemma~\cite{Renner05securityQKD, wilde2013QItheory}: for any positive semidefinite matrix $S^{\sA\sB}$, $\supp[S^{\sA\sB}] \subset \supp[S^\sA\otimes S^\sB]$.
By combining this fact with \Cref{prop:min singular inequality} in \Cref{sec:appendix product of singvalue}, Eq.~\eqref{inteq:138} is obtained.
Here, we note that $\tr_\sE[\Psi^{\sR\sE}] = \rho^\sR$ and $\tr_\sR[\Psi^{\sR\sE}] = \Psi^\sE$, and that $\ket{\Psi}^{\sR\sB\hat{\sG}\sE}$ is a pure state, which implies $\lambda_{\rm min}(\Psi^{\sR\sE}) = \lambda_{\rm min}(\Psi^{\sB\hat{\sG}})$ and $\lambda_{\rm min}(\Psi^\sE) = \lambda_{\rm min}(\Psi^{\sR\sB\hat{\sG}})$.
The inequality in Eq.~\eqref{inteq:139} results from the rank of a product of matrices being at most the rank of each matrix.
Since $\supp[\Psi^{\sR\sE}] \subset \supp[\rho^\sR \otimes \Psi^\sE]$, we have $\min\{r(\Psi^{\sR\sE}), r(\rho^\sR\otimes \Psi^\sE)\} = r(\Psi^{\sR\sE}) = r(\Psi^{\sB\hat{\sG}})$.
The algorithm runs sequentially, so its quantum circuit uses $\cO\big(\log{(d_\sR d_\sA d_\sB d_\sG)}\big)$ qubits at any one time.

Next, we consider Model~2.
Bob prepares copies of 
\begin{align}
    \Psi^{\sR'\sB'\hat{\sG}'} = \cN^{\sA'\sG'\rarr\sB'}\circ\cU^{\sA'\sG'}(\rho^{\sR'\sA'}\otimes\Phi^{\sG'\hat{\sG}'}),    
\end{align}
each by one use of $\cN^{\sA'\sG'\rarr\sB'}$.
Applying 
\begin{align}
    \mathtt{RandomPurification}(\Psi^{\sR'\sB'\hat{\sG}'}; \delta/2, q)
\end{align}
in \Cref{prevthm:random purification}, Bob obtains a channel $\cQ$ with $\f{1}{2}\|\cQ - \cQ_{\rm ideal}\|_\diamond \leq \delta/2$, where $\cQ_{\rm ideal}$ maps $(\Psi^{\sR'\sB'\hat{\sG}'})^{\otimes q}$ to
\begin{align}
\label{inteq:qc rp output}
    \bE_{U_{\rm RP}}\big[(\ketbra{\Psi_{U_{\rm RP}}}{\Psi_{U_{\rm RP}}}^{\sR'\sB'\hat{\sG}'\sQ'})^{\otimes q}\big].
\end{align}
Here, $\sQ'$ is a purifying system with $d_{\sQ'} = d_\sR d_\sB d_\sG$, $\ket{\Psi_{U_{\rm RP}}}^{\sR'\sB'\hat{\sG}'\sQ'} = U_{\rm RP}^{\sQ'}\ket{\Psi_0}^{\sR'\sB'\hat{\sG}'\sQ'}$, $\ket{\Psi_0}^{\sR'\sB'\hat{\sG}'\sQ'}$ is a fixed purified state of $\Psi^{\sR'\sB'\hat{\sG}'}$, and $U_{\rm RP}^{\sQ'}$ is a Haar random unitary on $\sQ'$.
Since the decoder acts only on $\sB\hat{\sG}$, its performance does not depend on the choice of the environment $\sE$. We may choose $\sE = \sQ'$ and $\ket{\Psi_{U_{\rm RP}}}^{\sR\sB\hat{\sG}\sQ'}$ as the purified state of $\Psi^{\sR\sB\hat{\sG}}$ in the analysis.
Equation~\eqref{inteq:decoup Psi} then holds with $\Psi^{\sR\sE}$ and $\Psi^\sE$ replaced by $\Psi_{U_{\rm RP}}^{\sR\sQ'}$ and $\Psi_{U_{\rm RP}}^{\sQ'}$, as the two purified states are related by a partial isometry on the purifying system.

From the $q$ purified states, Bob prepares the input state $\ket{\Psi_{U_{\rm RP}}}^{\sR'\sB'\hat{\sG}'\sQ'}\ket{0}^{\hat{\sR}'\hat{\sA}'}$ and the output state $\ket{\rho}^{\sR'\hat{\sA}'}\ket{\Psi_{U_{\rm RP}}}^{\hat{\sR}'\sB'\hat{\sG}'\sQ'}$ of the Uhlmann transformation, where the latter is obtained by relabeling $\sR'$ as $\hat{\sR}'$ and attaching $\ket{\rho}^{\sR'\hat{\sA}'}$.
Both states carry the same $U_{\rm RP}^{\sQ'}$, since they are prepared from the common state $\ket{\Psi_{U_{\rm RP}}}^{\sR'\sB'\hat{\sG}'\sQ'}$, and the operations in their preparation do not act on $\sQ'$.
Bob then runs $\mathtt{UhlmannPurifiedSample}\big(\ket{\Psi_{U_{\rm RP}}}\ket{0}, \ket{\rho}\ket{\Psi_{U_{\rm RP}}}; \delta/2, \tF\big)$ in \Cref{alg:Uhl purif sample}, which implements the Uhlmann transformation on $\sM$.
Since $U_{\rm RP}^{\sQ'}$ acts only on $\sQ'$, we obtain
\begin{align}
    &\tr_{\sR\sQ'}\big[\ket{\rho}^{\sR\hat{\sA}}\ket{\Psi_{U_{\rm RP}}}^{\hat{\sR}\sB\hat{\sG}\sQ'}\bra{\Psi_{U_{\rm RP}}}^{\sR\sB\hat{\sG}\sQ'}\bra{0}^{\hat{\sR}\hat{\sA}}\big] \notag\\
    &= \tr_{\sR\sQ'}\big[U_{\rm RP}^{\sQ'}\ket{\rho}^{\sR\hat{\sA}}\ket{\Psi_0}^{\hat{\sR}\sB\hat{\sG}\sQ'}\bra{\Psi_0}^{\sR\sB\hat{\sG}\sQ'}\bra{0}^{\hat{\sR}\hat{\sA}}(U_{\rm RP}^{\sQ'})^\dag\big] \\
    \label{inteq:qc M invariance}
    &= \tr_{\sR\sQ'}\big[\ket{\rho}^{\sR\hat{\sA}}\ket{\Psi_0}^{\hat{\sR}\sB\hat{\sG}\sQ'}\bra{\Psi_0}^{\sR\sB\hat{\sG}\sQ'}\bra{0}^{\hat{\sR}\hat{\sA}}\big],
\end{align}
where we used the cyclic property of the partial trace.
Thus, the operator on $\sM$ in Eq.~\eqref{inteq:qc M invariance}, which corresponds to $M^\sB$ in \Cref{alg:Uhl purif sample}, takes the same value for every $U_{\rm RP}^{\sQ'}$, and so do the Uhlmann partial isometry and the parameters $s_{\rm min}$ and $R$ in Theorem~\ref{thm:algorithm Uhlmann}.
The dependence on $U_{\rm RP}^{\sQ'}$ remains only in the auxiliary systems of \Cref{alg:Uhl purif sample}, which are traced out at the end.

Let $\til{\cD}_{U_{\rm RP}}^{\sB\hat{\sG}\rarr\hat{\sA}}$ be the decoder obtained for a fixed $U_{\rm RP}^{\sQ'}$, defined from the output channel $\cT_{U_{\rm RP}}^\sM$ of \Cref{alg:Uhl purif sample} in the same way as in Model~1.
From Theorem~\ref{thm:algorithm Uhlmann} and Eqs.~\eqref{inteq:159} and~\eqref{inteq:decoup Psi}, we have
\begin{align}
    &\rF\big(\cT_{U_{\rm RP}}^\sM(\ketbra{\Psi_{U_{\rm RP}}}{\Psi_{U_{\rm RP}}}^{\sR\sB\hat{\sG}\sQ'}\otimes\ketbra{0}{0}^{\hat{\sR}\hat{\sA}}), \ket{\rho}^{\sR\hat{\sA}}\ket{\Psi_{U_{\rm RP}}}^{\hat{\sR}\sB\hat{\sG}\sQ'}\big) \notag\\
    &\geq \rF(\Psi_{U_{\rm RP}}^{\sR\sQ'}, \rho^\sR\otimes\Psi_{U_{\rm RP}}^{\sQ'}) - \delta/2 \\
    &\geq 1 - 4\epsilon - \delta/2.
\end{align}
Then, the monotonicity of the fidelity under the partial trace implies that, for every $U_{\rm RP}^{\sQ'}$,
\begin{align}
\label{inteq:qc fixed V}
    \rF\big(\til{\cD}_{U_{\rm RP}}^{\sB\hat{\sG}\rarr\hat{\sA}}(\Psi^{\sR\sB\hat{\sG}}), \ket{\rho}^{\sR\hat{\sA}}\big) \geq 1 - 4\epsilon - \delta/2.
\end{align}
From Eq.~\eqref{inteq:qc rp output}, the channel $\til{\cD}^{\sB\hat{\sG}\rarr\hat{\sA}}$ actually implemented by Bob satisfies $\f{1}{2}\big\|\til{\cD}^{\sB\hat{\sG}\rarr\hat{\sA}} - \bE_{U_{\rm RP}}\big[\til{\cD}_{U_{\rm RP}}^{\sB\hat{\sG}\rarr\hat{\sA}}\big]\big\|_\diamond \leq \delta/2$.
As the target state is pure, $\rF(\cdot, \ket{\rho}^{\sR\hat{\sA}})$ is linear in its first argument, and Eq.~\eqref{inteq:qc fixed V} implies
\begin{align}
    &\rF\big(\bE_{U_{\rm RP}}\big[\til{\cD}_{U_{\rm RP}}^{\sB\hat{\sG}\rarr\hat{\sA}}\big](\Psi^{\sR\sB\hat{\sG}}), \ket{\rho}^{\sR\hat{\sA}}\big) \notag\\
    &= \bE_{U_{\rm RP}}\Big[\rF\big(\til{\cD}_{U_{\rm RP}}^{\sB\hat{\sG}\rarr\hat{\sA}}(\Psi^{\sR\sB\hat{\sG}}), \ket{\rho}^{\sR\hat{\sA}}\big)\Big] \\
    &\geq 1 - 4\epsilon - \delta/2.
\end{align}
Together with the error $\delta/2$ of the random purification, we obtain Eq.~\eqref{inteq:157}, where we used the relation between fidelity and the diamond norm in Eq.~\eqref{eq:fidelity and diamond}.

From Theorem~\ref{thm:algorithm Uhlmann}, the number of samples of each of the input and the output states for \Cref{alg:Uhl purif sample} is $\cO\big(\min\{(R/(\delta s_{\rm min}))(\log{(1/\delta)})^2, R^2/\delta^2\}\big)$, where $s_{\rm min}$ and $R$ are those in Eq.~\eqref{inteq:81}, since $\sqrt{\Psi_{U_{\rm RP}}^{\sR\sQ'}}\big(\sqrt{\rho^\sR}\otimes\sqrt{\Psi_{U_{\rm RP}}^{\sQ'}}\big)$ and $\sqrt{\Psi^{\sR\sE}}\big(\sqrt{\rho^\sR}\otimes\sqrt{\Psi^\sE}\big)$ have the same non-zero singular values.
Each sample is prepared from one copy of $\ket{\Psi_{U_{\rm RP}}}^{\sR'\sB'\hat{\sG}'\sQ'}$, and each copy of $\Psi^{\sR'\sB'\hat{\sG}'}$ requires one use of $\cN^{\sA'\sG'\rarr\sB'}$.
Applying Eqs.~\eqref{inteq:138} and~\eqref{inteq:139}, we obtain Eq.~\eqref{inteq:137}.
From \Cref{prevthm:random purification}, $\mathtt{RandomPurification}$ consists of $\til{\cO}(q^4)$ one- and two-qubit gates and uses $\til{\cO}(q^2)$ qubits at any one time, and the subsequent \Cref{alg:Uhl purif sample} consists of $\cO\big(q\log{(d_\sR d_\sA d_\sB d_\sG)}\big)$ gates and uses $\cO\big(\log{(d_\sR d_\sA d_\sB d_\sG)} + \log{(1/\delta)}\big)$ qubits at once, where we used $R \leq d_\sB d_\sG$ for the system $\sH$ in Theorem~\ref{thm:algorithm Uhlmann}.
In both models, the circuit description can be computed classically in $\poly\big(q, \log{(d_\sR d_\sA d_\sB d_\sG)}, \log{(1/\delta)}\big)$ time.

\end{proof}

In Model~2, the canonical purification in \Cref{alg:canonical purification} can be used instead of the random purification, where the state $\Psi^{\sR'\sB'\hat{\sG}'}$ is purified.
In this case, the algorithm runs sequentially and uses $\cO\big(\log{(d_\sR d_\sA d_\sB d_\sG)}\big)$ qubits at any one time.
Although we omit the details, a step-by-step evaluation along the proof of Theorem~\ref{thm:Uhlmann mixed state sample} shows that the number of uses of $\cN^{\sA'\sG'\rarr\sB'}$ is
\begin{align}
    &\til{\cO}\bigg(\f{d_\sR d_\sB d_\sG}{\delta}\min\Big\{\f{1}{\mu_{\rm min}^2}, \f{r(\Psi^{\sB\hat{\sG}})^2}{\delta^2}\Big\} \notag\\
    &\hspace{6pc}\min\Big\{\f{1}{\lambda_{\rm min}^2(\Psi^{\sR\sB\hat{\sG}})}, \f{r(\Psi^{\sR\sB\hat{\sG}})^2}{\delta^4}\Big\}\bigg),
\end{align}
which includes a factor of the dimension $d_\sR d_\sB d_\sG$ of $\Psi^{\sR'\sB'\hat{\sG}'}$ due to the canonical purification.


\subsection{One-shot quantum state merging}
\label{sec:entdis and uhl}

As another application of the decoupling and the Uhlmann transformation, we consider quantum state merging, which is often regarded as a generalization of entanglement distillation.
In \Cref{sec:setting of state merging}, we introduce the setting, and in \Cref{sec:implement uhl in state merging}, we discuss the implementation of the Uhlmann transformation for quantum state merging using our algorithm.

\subsubsection{Setting}
\label{sec:setting of state merging}

Suppose that a state $\ket{\omega}^{\sR\sA\sB}$ is shared between Alice, Bob, and the reference $\sR$, where $\sA$ and $\sB$ are held by Alice and Bob, respectively.
Their goal is to transfer Alice's system $\sA$ of $\ket{\omega}^{\sR\sA\sB}$ to Bob using only local operations and noiseless classical communication, while simultaneously distilling as much entanglement as possible between Alice and Bob.
That is, the task is to implement the following transformation, with the size of systems $\sS$ and $\hat{\sS}$ taken as large as possible:
\begin{align}
    \cB^{\sB\sX\rarr\hat{\sA}\sB\hat{\sS}}\circ\cA^{\sA\rarr\sS\sX}(\omega^{\sR\sA\sB}) \approx \omega^{\sR\hat{\sA}\sB} \otimes \Phi^{\sS\hat{\sS}},    
\end{align}
where $\cA^{\sA\rarr\sS\sX}$ and $\cB^{\sB\sX\rarr\hat{\sA}\sB\hat{\sS}}$ are performed by Alice and Bob, respectively, and $\sX$ is a classical register transmitted from Alice to Bob.
In the final state, $\sS$ and $\hat{\sA}\sB\hat{\sS}$ are held by Alice and Bob, respectively.

We here follow the protocol in Refs.~\cite{horodecki2007quantum, berta2008singleshot, dupuis2014one, khatri2024principlemodern}, which we describe below.
See also \Cref{fig:diagram of ent distill}.
Alice applies $\cA^{\sA\rarr\sS\sX} = \cE^{\sA \rarr \sS\sX} \circ \cU^\sA$, where $U^\sA$ is a Haar random unitary, and 
\begin{align}
\label{eq:distill E channel}
    \cE^{\sA\rarr\sS\sX}(\cdot) = \sum_x U_x^{\sA\rarr\sS}\Pi_x^\sA(\cdot)\Pi_x^\sA(U_x^{\sA\rarr\sS})^\dag \otimes \ketbra{x}{x}^\sX.
\end{align}
The projector $\Pi_x^\sA$ is onto a $d_\sS$-dimensional subspace $\cH_x^\sA$ of $\cH^\sA$ and $U_x^{\sA\rarr \sS}$ is a partial isometry which maps the subspace $\cH_x^\sA$ to $\cH^\sS$ spanned by a fixed orthonormal basis. For simplicity, we assume $d_\sX = d_\sA/d_\sS$.

\begin{figure}
    \centering
    \includegraphics[width=85mm]{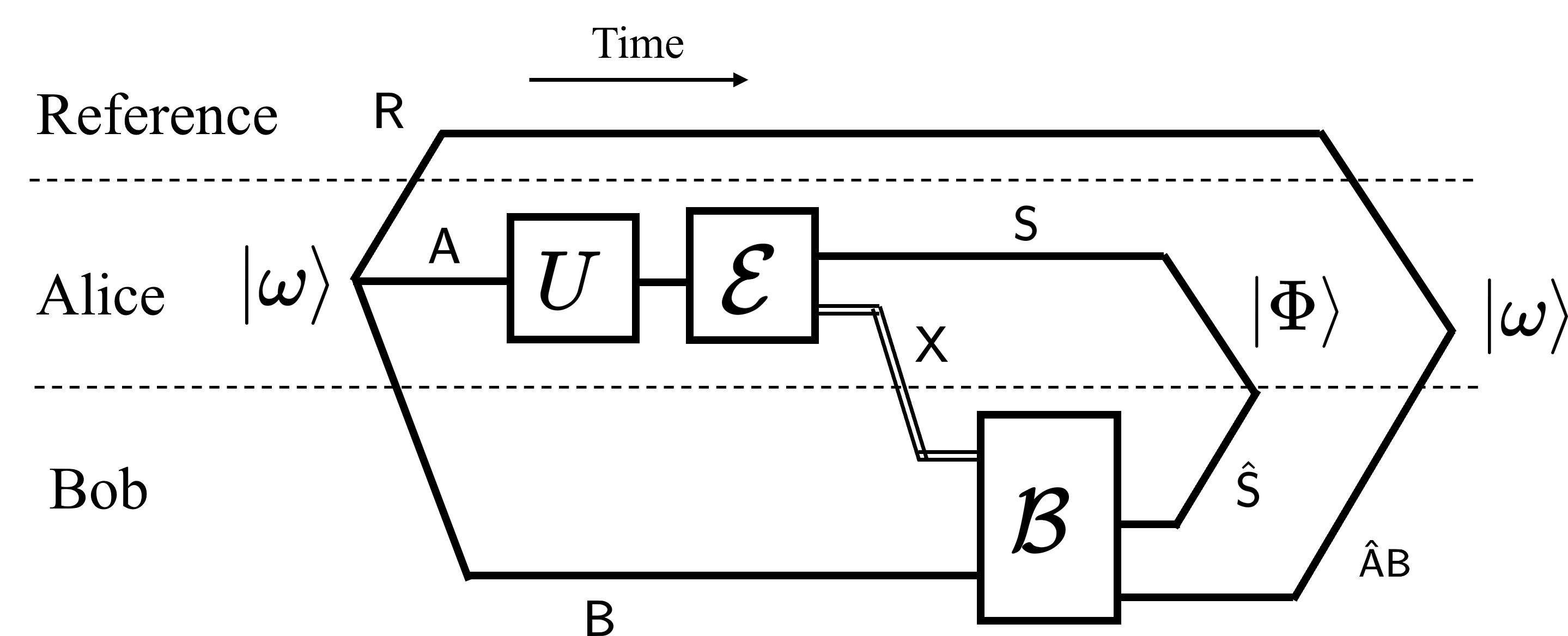}
    \caption{A diagram of quantum state merging.
    The goal is to transfer Alice's subsystem, which is part of a joint quantum state $\ket{\omega}^{\sR\sA\sB}$, to Bob, while simultaneously distilling as much entanglement as possible between $\sS$ and $\hat{\sS}$, using only local operations and classical communication.
    The double line connecting Alice and Bob represents a noiseless classical communication channel.
    }
    \label{fig:diagram of ent distill}
\end{figure}

After receiving the register $\sX$ by classical communication, which stores the measurement outcome $x$, Bob performs the operation $\cB^{\sB\sX \to \hat{\sA}\sB\hat{\sS}}$ specified as follows.
Let $\sqrt{p_x}\ket{\psi_x}^{\sR\sS\sB} = U_x^{\sA \to \sS} \Pi_x^\sA U^\sA\ket{\omega}^{\sR\sA\sB}$, where $p_x$ is the probability that the outcome $x$ occurs, and $\ket{\psi_x}^{\sR\sS\sB}$ is the resulting state.
Then, the operation $\cB^{\sB\sX \to \hat{\sA}\sB\hat{\sS}}$ is constructed as $\cB^{\sB\sX\rarr\hat{\sA}\sB\hat{\sS}} = \tr_\sX \circ \cW^{\hat{\sA}\sB\hat{\sS}\sX} \circ \cP_{\ket{0}}^{\bC\rarr\hat{\sA}\hat{\sS}}$, where $W^{\hat{\sA}\sB\hat{\sS}\sX}$ is a unitary such that $W^{\hat{\sA}\sB\hat{\sS}\sX} = \sum_x V_x^{\hat{\sA}\sB\hat{\sS}} \otimes \ketbra{x}{x}^\sX$, and $ V_x^{\hat{\sA}\sB\hat{\sS}}$ is the Uhlmann unitary satisfying 
\begin{align}
\label{inteq:111}
    \rF\big(V_x^{\hat{\sA}\sB\hat{\sS}}\ket{\psi_x}^{\sR\sS\sB}\ket{0}^{\hat{\sA}\hat{\sS}}, \ket{\omega}^{\sR\hat{\sA}\sB}\ket{\Phi}^{\sS\hat{\sS}}\big) = \rF(\psi_x^{\sR\sS}, \omega^\sR\otimes\pi^\sS).
\end{align}

We remark on the performance of the above protocol.
Let $\bar{\cE}^{\sA\rarr\sS\sE}$ be a complementary channel of $\cE^{\sA\rarr\sX} = \tr_\sS\circ\cE^{\sA\rarr\sS\sX}$, where $\sE$ is an environment of the classical channel from Alice to Bob.
We assume $d_\sE = d_\sX$ without loss of generality.
The decoupling theorem~\cite{dupuis2014one} ensures that the Haar random unitary $U^\sA$ satisfies
\begin{align}
\label{eq:OneShotDecouplingExpectDistill}
    \rF\big(\bar{\cE}^{\sA\rarr\sS\sE}\circ\cU^\sA(\omega^{\sR\sA}), \omega^\sR \otimes \pi^\sS \otimes \pi^\sE \big) \geq 1 - \epsilon,
\end{align}
with high probability.
Here, $\epsilon$ can be taken sufficiently small, depending on the properties of the initial state $\ket{\omega}^{\sR\sA\sB}$ and the channel $\cE^{\sA\rarr\sS\sX}$, such as the dimension of the system $\sS$ and the amount of initial entanglement between the systems $\sA$ and $\sB$~\cite{dupuis2014one}.
In Eq.~\eqref{eq:OneShotDecouplingExpectDistill}, the reduced Choi--Jamio\l{}kowski state on $\sS\sE$ is given by $\bar{\cE}^{\sA\rarr\sS\sE}(\pi^\sA) = \pi^\sS \otimes \pi^\sE$.
This is because, according to the outcome $x$, each $U_x^{\sA\rarr\sS}$ maps from $\cH_x^\sA$ to $\cH^\sS$, which is spanned by a fixed basis, and thus, $U_x^{\sA\rarr \sS}\Pi_x^\sA(U_x^{\sA\rarr\sS})^\dag$ is independent of $x$ and equals $\bI^\sS$.
Then, under the decoupling condition Eq.~\eqref{eq:OneShotDecouplingExpectDistill}, it is known that the operation $\cB^{\sB\sX\rarr\hat{\sA}\sB\hat{\sS}}$ achieves the task of quantum state merging with fidelity at least $1-\epsilon$. For details, see, e.g., Ref.~\cite{khatri2024principlemodern}.


\subsubsection{Algorithmic implementation of Bob's operation based on the Uhlmann transformation}
\label{sec:implement uhl in state merging}

We consider explicitly implementing the unitary $W^{\hat{\sA}\sB\hat{\sS}\sX} = \sum_x V_x^{\hat{\sA}\sB\hat{\sS}} \otimes \ketbra{x}{x}^\sX$, using our Uhlmann transformation algorithm.
In this case, the Uhlmann unitary $V_x^{\hat{\sA}\sB\hat{\sS}}$ is conditioned by $x$ and should be realized for all $x$.
For algorithmic construction, we assume that Bob has full knowledge of Alice's operation $\cA^{\sA\rarr\sS\sX} = \cE^{\sA\rarr\sS\sX}\circ\cU^\sA$, and thus can implement its Stinespring isometry and its inverse. While $U^\sA$ is chosen uniformly at random by Alice, Bob can identify the unitary Alice chose, since Alice can send its description through classical communication.

We address the following two models: 
\begin{enumerate}
    \item A model in which multiple independent and identical copies of the initial state, $\ket{\omega}^{\sR'\sA'\sB'}$, are available.
    \item A model in which multiple independent and identical copies of the reduced state $\omega^{\sA'\sB'}$ are available.
\end{enumerate}
We evaluate the number of samples of $\ket{\omega}^{\sR'\sA'\sB'}$ or $\omega^{\sA'\sB'}$ required for Bob's operation in each model.
We should distinguish the systems used by Bob, e.g., $\sR'$ and $\sA'$, from those that already exist as the reference $\sR$ and Alice's systems $\sA$. These states are sampled and used \emph{locally} by Bob.
It is noteworthy that, in both models, Bob does not need to know a complete description of the initial state $\ket{\omega}^{\sR\sA\sB}$; rather, Bob is provided with copies of the state, whose description remains unknown to Bob.
Model~2 corresponds to the mixed sample access model, where Bob is supplied only with copies of the state that Alice and Bob share. In this model, Bob first applies $\mathtt{RandomPurification}$ in \Cref{prevthm:random purification} to his copies, and then follows the same procedure as in Model~1.

We define a state $\ket{\Psi}^{\sR\sS\sB\sX\sE} = V_\cA^{\sA\rarr\sS\sX\sE}\ket{\omega}^{\sR\sA\sB}$, where $V_\cA^{\sA\rarr\sS\sX\sE}$ is a Stinespring isometry of $\cA^{\sA\rarr\sS\sX} = \cE^{\sA\rarr\sS\sX}\circ \cU^\sA$. 
We introduce $m_{\rm min}$ and $\Lambda$ as  
\begin{align}
    \label{inteq:142}
    &m_{\rm min} = \min_x\Big[\sqrt{p_x}s_{\rm min}\Big(\sqrt{\psi_x^{\sR\sS}}\big(\sqrt{\omega^\sR}\otimes\sqrt{\pi^\sS}\big)\Big)\Big], \\ 
    \label{inteq:143}
    &\Lambda = \sum_x\sqrt{p_x}r\Big(\sqrt{\psi_x^{\sR\sS}}\big(\sqrt{\omega^\sR}\otimes\sqrt{\pi^\sS}\big)\Big),
\end{align}
where $s_{\rm min}(\cdot)$ and $r(\cdot)$ are the minimum non-zero singular value and the rank of an input matrix, respectively.
The minimization is taken over all outcomes $x$ with $p_x > 0$ for which the matrix in the argument of $s_{\rm min}(\cdot)$ is non-zero, and the summation is taken over all outcomes $x$ with $p_x > 0$.
The factor $\sqrt{p_x}$ in both quantities comes from the operator block-encoded in the algorithm; see the proof.
Our statement is then as follows.

\begin{theorem}[Uhlmann transformation algorithm in quantum state merging]
\label{thm:Uhlmann state merging}
In the setting of quantum state merging, suppose that for some $\epsilon \in [0, 1]$,  
\begin{align}
\label{inteq:141}
        \rF\big(\Psi^{\sR\sS\sE}, \omega^\sR \otimes \pi^\sS \otimes \pi^\sE \big) \geq 1 - \epsilon.
\end{align}
Then, for $\delta \in (0, 1)$, there exists a quantum algorithm that implements Bob's operation $\til{\cB}^{\sB\sX\rarr\hat{\sA}\sB\hat{\sS}}$, which satisfies
\begin{align}
    \label{inteq:155}
    \rF\big(\til{\cB}^{\sB\sX\rarr\hat{\sA}\sB\hat{\sS}}(\Psi^{\sR\sS\sB\sX}), \ket{\omega}^{\sR\hat{\sA}\sB}\ket{\Phi}^{\sS\hat{\sS}}\big) \geq 1 - \epsilon - \delta,
\end{align}
using $\zeta$ samples of $\ket{\omega}^{\sR'\sA'\sB'}$ in Model~1, or $\zeta$ samples of $\omega^{\sA'\sB'}$ in Model~2, where 
\begin{align}
    \zeta = \cO\bigg(\min\bigg\{\f{\Lambda}{\delta m_{\rm min}}\Big(\log{\Big(\f{1}{\delta}\Big)}\Big)^2, \f{\Lambda^2}{\delta^2}\bigg\}\bigg),
\end{align}
and $m_{\rm min}$ and $\Lambda$ are given by Eqs.~\eqref{inteq:142} and~\eqref{inteq:143}.
In both models, Bob's operation uses $V_\cA^{\sA\rarr\sS\sX\sE}$ and its inverse, or their controlled versions, $\cO\big(\min\{(\log{(1/\delta)})^2/m_{\rm min}, \Lambda/\delta\}\big)$ times.
In addition to these uses, in Model~1, the quantum circuit for Bob's operation consists of $\cO\big(\zeta\log{(d_\sR d_\sA d_\sB)}\big)$ one- and two-qubit gates and uses $\cO\big(\log{(d_\sR d_\sA d_\sB)} + \log{(1/\delta)}\big)$ qubits at once, and its description can be computed classically in $\poly\big(\zeta, \log{(d_\sR d_\sA d_\sB)}, \log{(1/\delta)}\big)$ time.
In Model~2, these are $\til{\cO}(\zeta^4)$ gates, $\til{\cO}(\zeta^2)$ qubits, and $\poly\big(\zeta, \log{(d_\sA d_\sB)}, \log{(1/\delta)}\big)$ time.
\end{theorem}

This theorem states that when $\Psi^{\sR\sS\sE}$ is sufficiently decoupled as $\Psi^{\sR\sS\sE} \approx \omega^\sR \otimes \pi^\sS \otimes \pi^\sE$ by Alice's operation, by taking small $\delta$ and applying $\til{\cB}^{\sB\sX\rarr\hat{\sA}\sB\hat{\sS}}$---which is explicitly constructed using our algorithm---Bob can accomplish the task of quantum state merging.
Since $\Lambda m_{\rm min} \leq 1$, as shown in the proof, the number of samples is at most $\cO\big((\log{(1/\delta)})^2/(\delta m_{\rm min}^2)\big)$ for all initial states, and this is optimal up to a logarithmic factor in the worst case: any operation of Bob achieving Eq.~\eqref{inteq:155} for all initial states requires $\Omega\big(1/(\delta m_{\rm min}^2)\big)$ samples for some initial states, even when Eq.~\eqref{inteq:141} holds with $\epsilon = 0$ (see \Cref{sec:lower bound merging} for details).
The order of samples is common to the two models, while the numbers of gates and qubits differ.
In Model~2, unlike in Model~1, the algorithm handles the $\zeta$ samples of $\omega^{\sA'\sB'}$ collectively, since the random purification produces purified states with a common unitary on the purifying system. The number of gates and qubits in Model~2 does not depend on $d_\sR$, as the purifying system produced by the random purification is chosen by Bob, independently of $\sR$.

\begin{proof}[Proof of Theorem~\ref{thm:Uhlmann state merging}]

We first consider Model~1.
Let $\sL = \sR\sS$ and $\sM = \hat{\sA}\sB\hat{\sS}$. 
We define $\ket{\sigma}^{\sL\sM}$ and $\Psi^{\sL\sM\sX}$ by $\ket{\sigma}^{\sL\sM} = \ket{\omega}^{\sR\hat{\sA}\sB}\ket{\Phi}^{\sS\hat{\sS}}$ and 
\begin{align}
    \Psi^{\sL\sM\sX} 
    &= \cP_{\ket{0}}^{\bC\rarr\hat{\sA}\hat{\sS}} \circ \cE^{\sA\rarr\sS\sX}\circ\cU^\sA(\ketbra{\omega}{\omega}^{\sR\sA\sB}) \\
    &= \sum_x  p_x\ketbra{\psi_x}{\psi_x}^{\sR\sS\sB} \otimes \ketbra{x}{x}^{\sX} \otimes \ketbra{0}{0}^{\hat{\sA}\hat{\sS}} \\
    &= \sum_x  p_x\ketbra{\rho_x}{\rho_x}^{\sL\sM} \otimes \ketbra{x}{x}^{\sX},
\end{align}
where $\ket{\rho_x}^{\sL\sM} = \ket{\psi_x}^{\sR\sS\sB}\ket{0}^{\hat{\sA}\hat{\sS}}$.
We also define an operator $K_x^{\sL\sM} = \ket{0}^{\hat{\sA}\hat{\sS}}U_x^{\hat{\sA}\rarr\sS}\Pi_x^{\hat{\sA}}U^{\hat{\sA}}\bra{\Phi}^{\sS\hat{\sS}}$, where Alice's operators act on $\hat{\sA}$ in place of $\sA$, so that
\begin{align}
\label{inteq:merging K}
    K_x^{\sL\sM}\ket{\sigma}^{\sL\sM} = \sqrt{p_x}\ket{\rho_x}^{\sL\sM}.
\end{align}
Since $K_x^{\sL\sM}(K_x^{\sL\sM})^\dag = \bI^{\sR\sS\sB}\otimes\ketbra{0}{0}^{\hat{\sA}\hat{\sS}}$ for all $x$ (see the discussion below Eq.~\eqref{eq:OneShotDecouplingExpectDistill}), the operator $\sum_x (K_x^{\sL\sM})^\dag\otimes\ketbra{x}{x}^\sX$ is a partial isometry, and Bob can implement its block-encoding unitary with one use of $(V_\cA^{\sA\rarr\sS\sX\sE})^\dag$.

We consider a quantum channel $\cJ^{\sM\sH\hat{\sL}\hat{\sM}\sX}$ obtained by $\mathtt{QSVTSIGNSTAGES}$ in \Cref{prop:QSVT sign stages} with $\beta \in (0, 1]$ and $s = m_{\rm min}$ specified below, where $\hat{\sL}$ and $\hat{\sM}$ are auxiliary systems of the same number of qubits as $\sL$ and $\sM$, respectively, and $\sH$ consists of the auxiliary qubits used in the algorithm and the block-encoding.
Here, both $\ket{\varphi}^\sE$ and $\ket{\phi}^\sE$ are $\ket{\sigma}^{\hat{\sL}\hat{\sM}}$ tensored with $\ket{0}$ on the auxiliary qubits of the above block-encoding unitary, and $U^{\sB\sE}$ is the block-encoding unitary following $F^{\sM\hat{\sM}}$, with $\sB$ replaced by $\sM\sX$.
Since Bob knows Alice's operation, he can also implement the controlled versions of the block-encoding unitary and its inverse, each with one use of the controlled version of $(V_\cA^{\sA\rarr\sS\sX\sE})^\dag$ or $V_\cA^{\sA\rarr\sS\sX\sE}$.
This aims to transform $\ket{\rho_x}^{\sL\sM}$ to $\ket{\sigma}^{\sL\sM}$ for all $x$ by acting on $\sM\sX$.
From Eqs.~\eqref{eq:projector product} and~\eqref{inteq:merging K}, $M^\sB$ in \Cref{prop:QSVT sign stages} is given by $\sum_x M_x^\sM \otimes \ketbra{x}{x}^\sX$, where $M_x^\sM = \sqrt{p_x}\tr_\sL[\ket{\sigma}\bra{\rho_x}^{\sL\sM}]$.

Since $M_x^\sM$ and $\sqrt{p_x}\sqrt{\psi_x^{\sR\sS}}\big(\sqrt{\omega^\sR}\otimes\sqrt{\pi^\sS}\big)$ have the same non-zero singular values, $m_{\rm min}$ and $\Lambda$ in Eqs.~\eqref{inteq:142} and~\eqref{inteq:143} are written as
\begin{align}
\label{inteq:119}
    m_{\rm min} = \min_x \big[s_{\rm min}(M_x^\sM)\big], \ \text{and} \ \ 
    \Lambda = \sum_x\sqrt{p_x}r(M_x^\sM).
\end{align}
The minimization and the summation over $x$ are crucial to ensure that the algorithm covers all outcomes $x$ with $p_x > 0$.
Here, the minimization is taken over $x$ with $p_x > 0$ and $M_x^\sM \neq 0$, since $M_x^\sM$ has no non-zero singular value when $M_x^\sM = 0$, and such $x$ does not affect the following evaluation.
Since the non-zero singular values of $\sum_x M_x^\sM \otimes \ketbra{x}{x}^\sX$ are those of $M_x^\sM$ over all $x$, they are at least $m_{\rm min}$, as required in \Cref{prop:QSVT sign stages} with $s = m_{\rm min}$.

Let Bob's operation $\til{\cB}^{\sB\sX\rarr\hat{\sA}\sB\hat{\sS}}$ be given by  
\begin{align}
    \til{\cB}^{\sB\sX\rarr\hat{\sA}\sB\hat{\sS}} = \tr_{\sH\hat{\sL}\hat{\sM}\sX}\circ\cJ^{\sM\sH\hat{\sL}\hat{\sM}\sX}\circ\cP_{\ket{\sigma}}^{\bC\rarr\hat{\sL}\hat{\sM}}\circ\cP_{\ket{0}}^{\bC\rarr\sH\hat{\sA}\hat{\sS}}.
\end{align}
We then evaluate the fidelity between the final state $\Psi_{\rm fin}^{\sR\hat{\sA}\sB\sS\hat{\sS}} = \til{\cB}^{\sB\sX\rarr\hat{\sA}\sB\hat{\sS}}\circ\cA^{\sA\rarr\sS\sX}(\omega^{\sR\sA\sB})$ and the target state $\ket{\sigma}^{\sL\sM} = \ket{\omega}^{\sR\hat{\sA}\sB}\ket{\Phi}^{\sS\hat{\sS}}$.
Since $\cA^{\sA\rarr\sS\sX}(\omega^{\sR\sA\sB}) = \sum_x p_x\ketbra{\psi_x}{\psi_x}^{\sR\sS\sB}\otimes\ketbra{x}{x}^\sX$ and the target state is pure, the linearity of the fidelity in its first argument gives
\begin{align}
\label{inteq:114}
    &\rF\big(\Psi_{\rm fin}^{\sR\hat{\sA}\sB\sS\hat{\sS}}, \ket{\sigma}^{\sL\sM}\big) \notag\\
    &= \sum_x p_x\rF\big(\til{\cB}^{\sB\sX\rarr\hat{\sA}\sB\hat{\sS}}(\ketbra{\psi_x}{\psi_x}^{\sR\sS\sB}\otimes\ketbra{x}{x}^\sX), \ket{\sigma}^{\sL\sM}\big).
\end{align}
On the other hand, the condition Eq.~\eqref{inteq:141} implies that
\begin{align}
    1 - \epsilon 
    &\leq \rF(\Psi^{\sR\sS\sE}, \omega^\sR \otimes \pi^\sS \otimes \pi^\sE) \\
    &= \Big(\sum_x \sqrt{\f{p_x}{d_\sX}}\sqrt{\rF}(\psi_x^{\sR\sS}, \omega^\sR \otimes \pi^\sS)\Big)^2 \\
    \label{inteq:116}
    &\leq \sum_x p_x\rF(\psi_x^{\sR\sS}, \omega^\sR \otimes \pi^\sS),
\end{align}
where we used $\Psi^{\sR\sS\sE} = \sum_x p_x\psi_x^{\sR\sS}\otimes\ketbra{x}{x}^\sE$ in the equality, and the Cauchy--Schwarz inequality with $\sum_x 1/d_\sX = 1$ in the last inequality.
Since $\sqrt{\rF}(\psi_x^{\sR\sS}, \omega^\sR \otimes \pi^\sS)$ is the sum of the singular values of $M_x^\sM/\sqrt{p_x}$, if $M_x^\sM = 0$ for all $x$, Eq.~\eqref{inteq:116} implies $\epsilon = 1$, in which case the statement is trivial.

We now evaluate the fidelity for each $x$.
Let $M_x^\sM = \sum_k s_{x, k}\ketbra{\eta_{x, k}}{\xi_{x, k}}^\sM$ be a singular value decomposition for each $x$, which gives a singular value decomposition of $\sum_x M_x^\sM\otimes\ketbra{x}{x}^\sX$, and let $K^{\sM\sH\hat{\sL}\hat{\sM}\sX}$ be the Kraus operator of $\cJ^{\sM\sH\hat{\sL}\hat{\sM}\sX}$ given in \Cref{prop:QSVT sign stages}.
By the same calculation as that leading to Eq.~\eqref{inteq:stages root} in \Cref{sec:uhlfidpurifsamp}, we obtain that, for all $x$ with $p_x > 0$,
\begin{align}
\label{inteq:merging stages}
    &\big|\bra{\sigma}^{\sL\sM}\bra{0}^\sH\bra{\sigma}^{\hat{\sL}\hat{\sM}}\bra{x}^\sX K^{\sM\sH\hat{\sL}\hat{\sM}\sX}\ket{\rho_x}^{\sL\sM}\ket{0}^\sH\ket{\sigma}^{\hat{\sL}\hat{\sM}}\ket{x}^\sX\big| \notag\\
    &\geq \sqrt{\rF}(\psi_x^{\sR\sS}, \omega^\sR \otimes \pi^\sS) - 4\beta r(M_x^\sM)/\sqrt{p_x}.
\end{align}
Here, we used $\tr_\sL\big[\ketbra{\rho_x}{\sigma}^{\sL\sM}\big] = (M_x^\sM)^\dag/\sqrt{p_x} = \sum_k (s_{x, k}/\sqrt{p_x})\ketbra{\xi_{x, k}}{\eta_{x, k}}^\sM$, the fact that $\sum_k s_{x, k}/\sqrt{p_x} = \sqrt{\rF}(\psi_x^{\sR\sS}, \omega^\sR \otimes \pi^\sS)$ by Eq.~\eqref{inteq:111}, and Eq.~\eqref{eq:QSVT sign stages}.
Unlike in \Cref{sec:uhlfidpurifsamp}, the singular values weighting the fidelity are not those of $M_x^\sM$, which appear in Eq.~\eqref{eq:QSVT sign stages}, but are larger by the factor $1/\sqrt{p_x}$, which yields the last term in Eq.~\eqref{inteq:merging stages}.
Since the fidelity is the sum of the non-negative contributions from all the Kraus operators of $\cJ^{\sM\sH\hat{\sL}\hat{\sM}\sX}$, Eq.~\eqref{inteq:merging stages} and the monotonicity of the fidelity under the partial trace imply that
\begin{align}
\label{inteq:merging stages branch}
    &\rF\big(\til{\cB}^{\sB\sX\rarr\hat{\sA}\sB\hat{\sS}}(\ketbra{\psi_x}{\psi_x}^{\sR\sS\sB}\otimes\ketbra{x}{x}^\sX), \ket{\sigma}^{\sL\sM}\big) \notag\\
    &\geq \rF(\psi_x^{\sR\sS}, \omega^\sR \otimes \pi^\sS) - 8\beta r(M_x^\sM)/\sqrt{p_x},
\end{align}
where we used $(a - b)^2 \geq a^2 - 2b$ for $a \in [0, 1]$ and $b \in [0, a]$, and that the right-hand side is negative when $b > a$.
Combining Eqs.~\eqref{inteq:114}, \eqref{inteq:116}, and~\eqref{inteq:merging stages branch}, we obtain
\begin{align}
    \rF\big(\Psi_{\rm fin}^{\sR\hat{\sA}\sB\sS\hat{\sS}}, \ket{\sigma}^{\sL\sM}\big) \geq 1 - \epsilon - 8\beta\Lambda.
\end{align}
Setting $\beta = \delta/(8\Lambda)$, the last term is at most $\delta$.
Here, $\beta \leq 1$ holds when $1 - \epsilon - \delta > 0$, since $\rF(\psi_x^{\sR\sS}, \omega^\sR \otimes \pi^\sS) \leq r(M_x^\sM)$, $p_x \leq \sqrt{p_x}$, and Eq.~\eqref{inteq:116} imply $\Lambda \geq \sum_x p_x\rF(\psi_x^{\sR\sS}, \omega^\sR \otimes \pi^\sS) \geq 1 - \epsilon > \delta$, while the statement is trivial otherwise.

Hence, since $\Psi_{\rm fin}^{\sR\hat{\sA}\sB\sS\hat{\sS}} = \til{\cB}^{\sB\sX\rarr\hat{\sA}\sB\hat{\sS}}\circ\cA^{\sA\rarr\sS\sX}(\omega^{\sR\sA\sB})$, we finally obtain
\begin{align}
    \rF\big(\til{\cB}^{\sB\sX\rarr\hat{\sA}\sB\hat{\sS}}\circ\cA^{\sA\rarr\sS\sX}(\omega^{\sR\sA\sB}), \ket{\omega}^{\sR\hat{\sA}\sB} \ket{\Phi}^{\sS\hat{\sS}}\big) \geq 1 - \epsilon - \delta.
\end{align}

From \Cref{prop:QSVT sign stages} with $\beta = \delta/(8\Lambda)$ and $s = m_{\rm min}$, the number of samples of $\ket{\sigma}^{\hat{\sL}\hat{\sM}}$ for implementing $\cJ^{\sM\sH\hat{\sL}\hat{\sM}\sX}$ is
\begin{align}
    n = \cO\bigg(\f{1}{\beta\max\{\beta, m_{\rm min}\}}\Big(\log{\Big(\f{2\max\{\beta, m_{\rm min}\}}{\beta}\Big)}\Big)^2\bigg),
\end{align}
and the controlled versions of the block-encoding unitary and its inverse are used $\cO(\beta n)$ times in total.
As $(\log{(2x)})^2/x \leq 3$ for $x \geq 1$, taking $x = \max\{\beta, m_{\rm min}\}/\beta$ yields $n = \cO(1/\beta^2) = \cO(\Lambda^2/\delta^2)$ and $\beta n = \cO(1/\beta) = \cO(\Lambda/\delta)$.
Moreover, $\Lambda m_{\rm min} \leq 1$ holds, since $r(M_x^\sM)m_{\rm min} \leq \sum_k s_{x, k} = \sqrt{p_x}\sqrt{\rF}(\psi_x^{\sR\sS}, \omega^\sR \otimes \pi^\sS) \leq \sqrt{p_x}$ for all $x$ with $M_x^\sM \neq 0$ and $\sum_x p_x = 1$.
Thus, 
\begin{align}
    \max\{\beta, m_{\rm min}\}/\beta = \max\{1, 8\Lambda m_{\rm min}/\delta\} \leq 8/\delta,
\end{align}
which, together with $\max\{\beta, m_{\rm min}\} \geq m_{\rm min}$, implies that $n = \cO\big((\Lambda/(\delta m_{\rm min}))(\log{(1/\delta)})^2\big)$ and $\beta n = \cO\big((\log{(1/\delta)})^2/m_{\rm min}\big)$.
One more sample of $\ket{\sigma}^{\hat{\sL}\hat{\sM}}$ is used for $\cP_{\ket{\sigma}}^{\bC\rarr\hat{\sL}\hat{\sM}}$, and each sample of $\ket{\sigma}^{\hat{\sL}\hat{\sM}}$ is prepared from a single copy of $\ket{\omega}^{\sR'\sA'\sB'}$.
Hence, the number of samples of $\ket{\omega}^{\sR'\sA'\sB'}$ for Bob's operation $\til{\cB}^{\sB\sX \rarr\hat{\sA}\sB\hat{\sS}}$ is given by
\begin{align}
    \zeta 
    = \cO\Big(\min\Big\{\f{\Lambda}{\delta m_{\rm min}}\Big(\log{\Big(\f{1}{\delta}\Big)}\Big)^2, \f{\Lambda^2}{\delta^2}\Big\}\Big),
\end{align}
and $V_\cA^{\sA\rarr\sS\sX\sE}$ and its inverse, or their controlled versions, are used
\begin{align}
    \cO\Big(\min\Big\{\f{1}{m_{\rm min}}\Big(\log{\Big(\f{1}{\delta}\Big)}\Big)^2, \f{\Lambda}{\delta}\Big\}\Big)
\end{align}
times in total.

Apart from these uses, the overall quantum circuit of this algorithm includes $\cO\big(\zeta\log{(d_\sR d_\sA d_\sB)}\big)$ one- and two-qubit gates, and $\cO\big(\log{(d_\sR d_\sA d_\sB)} + \log{(1/\delta)}\big)$ qubits at any one time, where the second term comes from the system $\sH$ in \Cref{prop:QSVT sign stages} and we used $\Lambda \leq d_\sX d_\sB \leq d_\sA d_\sB$, which follows from Eq.~\eqref{inteq:143} since $r(\psi_x^{\sR\sS}) = r(\psi_x^\sB) \leq d_\sB$.
From the construction, the circuit description can be computed classically in $\poly\big(\zeta, \log{(d_\sR d_\sA d_\sB)}, \log{(1/\delta)}\big)$ time.

Next, we consider Model~2, where Bob is given multiple samples of $\omega^{\sA'\sB'}$ rather than $\ket{\omega}^{\sR'\sA'\sB'}$.
Applying $\mathtt{RandomPurification}(\omega^{\sA'\sB'}; \delta/2, \zeta)$ in \Cref{prevthm:random purification}, Bob obtains a channel $\cQ$ with $\f{1}{2}\|\cQ - \cQ_{\rm ideal}\|_\diamond \leq \delta/2$, where $\cQ_{\rm ideal}$ maps $(\omega^{\sA'\sB'})^{\otimes\zeta}$ to
\begin{align}
\label{inteq:merging rp output}
    \bE_{U_{\rm RP}}\big[(\ketbra{\omega_{U_{\rm RP}}}{\omega_{U_{\rm RP}}}^{\sQ'\sA'\sB'})^{\otimes \zeta}\big].
\end{align}
Here, $\sQ'$ is a purifying system with $d_{\sQ'} = d_\sA d_\sB$, $\ket{\omega_{U_{\rm RP}}}^{\sQ'\sA'\sB'} = U_{\rm RP}^{\sQ'}\ket{\omega_0}^{\sQ'\sA'\sB'}$, $\ket{\omega_0}^{\sQ'\sA'\sB'}$ is a fixed purified state of $\omega^{\sA'\sB'}$, and $U_{\rm RP}^{\sQ'}$ is a Haar random unitary on $\sQ'$.
Since all purified states of $\omega^{\sA'\sB'}$ are related by a partial isometry on the purifying system, we take $\ket{\omega_0}^{\sQ'\sA'\sB'} = V_{\rm pur}^{\sR'\rarr\sQ'}\ket{\omega}^{\sR'\sA'\sB'}$, where $(V_{\rm pur}^{\sR'\rarr\sQ'})^\dag V_{\rm pur}^{\sR'\rarr\sQ'}$ is the projector onto $\supp[\omega^{\sR'}]$. This is possible even when $\sQ'$ and $\sR'$ have different dimensions.
Bob runs the procedure of Model~1 with these $\zeta$ purified states, with $\delta$ rescaled as $\delta \gets \delta/2$. In this procedure, the system $\hat{\sL}$ of Model~1, which consists of $\sR'$ and a copy of $\sS$, is replaced by the system $\sL'$ consisting of $\sQ'$ and this copy of $\sS$.

We show that the procedure of Model~1 is unaffected by the choice of $U_{\rm RP}^{\sQ'}$.
Let $\ket{\sigma_{U_{\rm RP}}}^{\sL'\hat{\sM}}$ be the state that Bob constructs from $\ket{\omega_{U_{\rm RP}}}^{\sQ'\sA'\sB'}$ in place of $\ket{\sigma}^{\hat{\sL}\hat{\sM}}$ in Model~1, and let $\ket{\rho_{x, U_{\rm RP}}}^{\sL'\hat{\sM}}$ be given by $K_x^{\sL'\hat{\sM}}\ket{\sigma_{U_{\rm RP}}}^{\sL'\hat{\sM}} = \sqrt{p_x}\ket{\rho_{x, U_{\rm RP}}}^{\sL'\hat{\sM}}$. Both states are obtained by operations that do not act on $\sQ'$, and hence $U_{\rm RP}^{\sQ'}$ commutes with these operations.
Thus, $\ket{\sigma_{U_{\rm RP}}}^{\sL'\hat{\sM}} = U_{\rm RP}^{\sQ'}V_{\rm pur}^{\sR'\rarr\sQ'}\ket{\sigma}^{\hat{\sL}\hat{\sM}}$ and $\ket{\rho_{x, U_{\rm RP}}}^{\sL'\hat{\sM}} = U_{\rm RP}^{\sQ'}V_{\rm pur}^{\sR'\rarr\sQ'}\ket{\rho_x}^{\hat{\sL}\hat{\sM}}$, where $\ket{\sigma}^{\hat{\sL}\hat{\sM}}$ and $\ket{\rho_x}^{\hat{\sL}\hat{\sM}}$ are the states in Model~1.
Note that the same $U_{\rm RP}^{\sQ'}$ appears in both states, since they are constructed from the common state $\ket{\omega_{U_{\rm RP}}}^{\sQ'\sA'\sB'}$.
The probability $p_x$ is unchanged, as it is determined by $\omega^{\sA'}$ alone.
Since $U_{\rm RP}^{\sQ'}$ and $V_{\rm pur}^{\sR'\rarr\sQ'}$ do not act on $\hat{\sM}$, we obtain
\begin{align}
\label{inteq:merging M invariance}
    \tr_{\sL'}\big[\ket{\sigma_{U_{\rm RP}}}\bra{\rho_{x, U_{\rm RP}}}^{\sL'\hat{\sM}}\big] 
    = \tr_{\hat{\sL}}\big[\ket{\sigma}\bra{\rho_x}^{\hat{\sL}\hat{\sM}}\big],
\end{align}
where we used the cyclic property of the partial trace and the fact that $(V_{\rm pur}^{\sR'\rarr\sQ'})^\dag (U_{\rm RP}^{\sQ'})^\dag U_{\rm RP}^{\sQ'} V_{\rm pur}^{\sR'\rarr\sQ'}$ acts as the identity on $\ket{\sigma}^{\hat{\sL}\hat{\sM}}$ and $\ket{\rho_x}^{\hat{\sL}\hat{\sM}}$, whose components on $\sR'$ are supported on $\supp[\omega^{\sR'}]$.
Thus $M_x^\sM = \sqrt{p_x}\tr_\sL[\ket{\sigma}\bra{\rho_x}^{\sL\sM}]$ takes the same value as in Model~1 for every $U_{\rm RP}^{\sQ'}$, and so do $V_x^\sM$, $m_{\rm min}$, and $\Lambda$.
The dependence on $U_{\rm RP}^{\sQ'}$ remains only in the states on $\sL'\hat{\sM}$, such as the samples used in $\mathtt{QSVTSIGNSTAGES}$ and the state preparation channel $\cP_{\ket{\sigma_{U_{\rm RP}}}}^{\bC\rarr\sL'\hat{\sM}}$, all of which are traced out at the end.

Let $\til{\cB}_{U_{\rm RP}}^{\sB\sX\rarr\hat{\sA}\sB\hat{\sS}}$ be the channel obtained for a fixed $U_{\rm RP}^{\sQ'}$.
The replacement above occurs only on $\sL'\hat{\sM}$, and the actual systems $\sL\sM$ are untouched.
Equation~\eqref{inteq:merging stages} holds as it stands, since it depends on $U_{\rm RP}^{\sQ'}$ only through $M_x^\sM$ and the states on $\sL\sM$.
Every step of the analysis of Model~1 therefore goes through, and we obtain, for every $U_{\rm RP}^{\sQ'}$,
\begin{align}
\label{inteq:merging fixed V}
    \rF\big(\til{\cB}_{U_{\rm RP}}^{\sB\sX\rarr\hat{\sA}\sB\hat{\sS}}(\Psi^{\sR\sS\sB\sX}), \ket{\omega}^{\sR\hat{\sA}\sB}\ket{\Phi}^{\sS\hat{\sS}}\big) \geq 1 - \epsilon - \delta/2.
\end{align}

From Eq.~\eqref{inteq:merging rp output}, the channel $\til{\cB}^{\sB\sX\rarr\hat{\sA}\sB\hat{\sS}}$ implemented by Bob satisfies $\f{1}{2}\big\|\til{\cB}^{\sB\sX\rarr\hat{\sA}\sB\hat{\sS}} - \bE_{U_{\rm RP}}\big[\til{\cB}_{U_{\rm RP}}^{\sB\sX\rarr\hat{\sA}\sB\hat{\sS}}\big]\big\|_\diamond \leq \delta/2$.
As the target state is pure, $\rF(\cdot, \ket{\omega}^{\sR\hat{\sA}\sB}\ket{\Phi}^{\sS\hat{\sS}})$ is linear in its first argument, and Eq.~\eqref{inteq:merging fixed V} implies
\begin{align}
    &\rF\big(\bE_{U_{\rm RP}}\big[\til{\cB}_{U_{\rm RP}}^{\sB\sX\rarr\hat{\sA}\sB\hat{\sS}}\big](\Psi^{\sR\sS\sB\sX}), \ket{\omega}^{\sR\hat{\sA}\sB}\ket{\Phi}^{\sS\hat{\sS}}\big) \notag\\
    &= \bE_{U_{\rm RP}}\Big[\rF\big(\til{\cB}_{U_{\rm RP}}^{\sB\sX\rarr\hat{\sA}\sB\hat{\sS}}(\Psi^{\sR\sS\sB\sX}), \ket{\omega}^{\sR\hat{\sA}\sB}\ket{\Phi}^{\sS\hat{\sS}}\big)\Big] \\
    &\geq 1 - \epsilon - \delta/2.
\end{align}
Together with the error $\delta/2$ of the random purification, we obtain Eq.~\eqref{inteq:155}, where we used the relation between fidelity and the diamond norm in Eq.~\eqref{eq:fidelity and diamond}.

Since $\delta$ is rescaled only by a constant factor, the number of samples of $\omega^{\sA'\sB'}$ that Bob uses is of the same order as in Model~1, and thus $\zeta$ samples suffice.
From \Cref{prevthm:random purification}, the random purification consists of $\til{\cO}(\zeta^4)$ one- and two-qubit gates and uses $\til{\cO}(\zeta^2)$ qubits at any one time.
As the system $\sR'$ is replaced by $\sQ'$ with $d_{\sQ'} = d_\sA d_\sB$, the subsequent procedure of Model~1 consists of $\cO\big(\zeta\log{(d_\sA d_\sB)}\big)$ gates and uses $\cO\big(\log{(d_\sA d_\sB)} + \log{(1/\delta)}\big)$ qubits at once.
The total number of one- and two-qubit gates is hence $\til{\cO}(\zeta^4)$, and the circuit uses $\til{\cO}(\zeta^2)$ qubits at any one time.
The circuit description can be computed classically in $\poly\big(\zeta, \log{(d_\sA d_\sB)}, \log{(1/\delta)}\big)$ time.

\end{proof}

As the final remark of this section, we discuss the role of the purifying system in Model~2.
It may seem surprising that Bob can perform the merging in Model~2 without access to the reference system $\sR$, since the Uhlmann partial isometry in Eq.~\eqref{eq:explicit form Uhl partial iso} is determined by the states on the entire system.
When the input and output states are both obtained from a single purified state $\ket{\omega}^{\sR\sA\sB}$ by operations that do not act on $\sR$, the partial isometry is determined by $\omega^{\sA\sB}$ without any information about $\sR$.
This is the case not only in quantum state merging but also in other tasks based on the decoupling approach, such as quantum communication in \Cref{sec:qi transmission} and the fully quantum Slepian--Wolf protocol~\cite{Abeyesinghe2009Motherfamilytree, Datta2011apexfamilytree}.

Whether the random purification is applicable is a separate question.
It supplies purified states of a single mixed state sharing a common unitary on the purifying system, whereas no channel purifies distinct mixed states in this manner~\cite{lowe2026randomdimensionreductionlearning}.
The random purification thus appears to be applicable when both the input and the output states of the Uhlmann transformation are prepared from copies of a single mixed state by operations that do not act on the purifying system, as in quantum state merging, where both are prepared from copies of $\omega^{\sA'\sB'}$, and in quantum communication, where both are prepared from copies of $\Psi^{\sR'\sB'\hat{\sG}'}$.

\section{Application 3: Petz recovery map with the Uhlmann transformation algorithm}
\label{sec:sample petz}

The Petz recovery map~\cite{petz1986sufficient, petz1988sufficiency} is a powerful tool with broad applications in quantum science~\cite{barnum2002reversing, ng2010simpleapproach, beigi2016decoding, chen2020entanglement, lauten2022approx, nakayama2023petz}.
In particular, it is known that the map demonstrates nearly optimal performance when it reverses the effect of noise~\cite{barnum2002reversing, beigi2016decoding}.
The Petz recovery map $\cR_{\sigma, \cF}^{\sB \rarr \sA}$ is defined for a reference state $\sigma^\sA$ and a quantum channel $\cF^{\sA\rarr\sB}$ as 
\begin{align}
\label{eq:definition of petz recovery map}
    \cR_{\sigma, \cF}^{\sB \rarr \sA} = \cM_{\sigma^{1/2}}^\sA\circ(\cF^{\sA\rarr\sB})^\dag\circ\cM_{\cF(\sigma)^{-1/2}}^\sB,     
\end{align}
where $(\cF^{\sA\rarr\sB})^\dag$ is an adjoint map of $\cF^{\sA\rarr\sB}$ and $\cM_{\mu}^\sC = \mu^\sC(\cdot)\mu^\sC$ for an operator $\mu^\sC$.
The adjoint map $(\cF^{\sA\rarr\sB})^\dag$ of $\cF^{\sA\rarr\sB}$ is defined as the map satisfying $\tr[L^\sB\cF^{\sA\rarr\sB}(M^\sA)]=\tr[(\cF^{\sA\rarr\sB})^\dag(L^\sB)M^\sA]$ for any operators $L^\sB$ and $M^\sA$.
The inverse $\cF(\sigma)^{-1/2}$ is taken on the support of $\cF^{\sA\rarr\sB}(\sigma^\sA)$, and we denote by $\Pi_{\cF(\sigma)}^\sB$ the projector onto this support.
Since $\tr\big[\cR_{\sigma, \cF}^{\sB\rarr\sA}(Y^\sB)\big] = \tr\big[\Pi_{\cF(\sigma)}^\sB Y^\sB\big]$ for any operator $Y^\sB$, the Petz recovery map is not trace-preserving when $\cF^{\sA\rarr\sB}(\sigma^\sA)$ is not full rank.
To evaluate the approximation error in the diamond norm, we consider the quantum channel
\begin{align}
\label{eq:petz channel}
    \hat{\cR}_{\sigma, \cF}^{\sB\rarr\sA}(\cdot) = \cR_{\sigma, \cF}^{\sB\rarr\sA}(\cdot) + \tr\big[(\bI^\sB - \Pi_{\cF(\sigma)}^\sB)(\cdot)\big]\ketbra{0}{0}^\sA,
\end{align}
which coincides with $\cR_{\sigma, \cF}^{\sB\rarr\sA}$ when $\cF^{\sA\rarr\sB}(\sigma^\sA)$ is full rank.

Quantum algorithms for the Petz recovery map have been studied in Refs.~\cite{gilyen2022petzmap, biswas2023noiseadapted, zhao2026petzrecoverysinglepolar}.
These works represent significant progress toward implementation of the map, but their approaches require, in addition to the use of $\sigma^\sA$, access to a Stinespring unitary of the channel $\cF^{\sA\rarr\sB}$ and its inverse, which is challenging in practice.

We here introduce a method to realize the Petz recovery map using only $\sigma^{\sA}$ and $\cF^{\sA\rarr\sB}$; that is, access to a Stinespring unitary is not required.
Our statement is as follows.

\begin{theorem}[Algorithmic implementation of the Petz recovery map from a channel]
\label{thm:petz implementation}
Let $\epsilon \in (0, 1)$.
For any quantum channel $\cF^{\sA\rarr\sB}$ and any state $\sigma^\sA$, there exists a quantum algorithm that implements a quantum channel $\til{\cR}^{\sB\rarr\sA}$ satisfying $\f{1}{2}\big\|\til{\cR}^{\sB\rarr\sA} - \hat{\cR}_{\sigma, \cF}^{\sB\rarr\sA}\big\|_\diamond \leq \epsilon$, using the channel $\cF^{\sA\rarr\sB}$ a total of $\til{\cO}\big(d_\sA^2 d_\sE/(\epsilon\lambda_{\rm min})\big)$ times and
\begin{align}
\label{eq:sampcomp Petz sigma}
    \til{\cO}\bigg(\f{d_\sE}{\epsilon\lambda_{\rm min}}\min\Big\{\f{1}{\sigma_{\rm min}^2}, \f{d_\sE^2}{\epsilon^4\lambda_{\rm min}^2}\Big\}\bigg)
\end{align}
samples of $\sigma^\sA$. Here, $\lambda_{\rm min}$ and $\sigma_{\rm min}$ are the minimum non-zero eigenvalues of $\cF^{\sA\rarr\sB}(\sigma^\sA)$ and $\sigma^\sA$, respectively, and $d_\sE$ is an upper bound on the rank of the Choi--Jamio{\l}kowski state $\tau^{\sR\sB} = \cF^{\sA\rarr\sB}(\Phi^{\sR\sA})$ with $d_\sR = d_\sA$, which can be taken as $d_\sE = d_\sA d_\sB$ in general.
\end{theorem}

Our idea is as follows. Any purified state $\ket{\tau}^{\sR\sB\sE}$ of $\tau^{\sR\sB}$ is written as $\ket{\tau}^{\sR\sB\sE} = V_\cF^{\sA\rarr\sB\sE}\ket{\Phi}^{\sR\sA}$ with a Stinespring isometry $V_\cF^{\sA\rarr\sB\sE}$ of $\cF^{\sA\rarr\sB}$.
Since $\tau^\sR$ is the completely mixed state $\pi^\sR$, the Uhlmann partial isometry from $\ket{\Phi}^{\sR\sA}\ket{0}^\sG$ to $\ket{\tau}^{\sR\sB\sE}$ coincides with $V_\cF^{\sA\rarr\sB\sE}$ when the input on $\sG$ is $\ket{0}^\sG$, and it can be implemented by our algorithm in the purified sample access model. We prepare the copies of $\ket{\tau}^{\sR\sB\sE}$ from those of $\tau^{\sR\sB}$, each of which is obtained by one use of $\cF^{\sA\rarr\sB}$, using the random purification in \Cref{prevthm:random purification}.
Although the Stinespring isometry depends on the choice of the purified state, the Petz recovery map does not, and it suffices that all the copies share the same purified state, which is ensured by the random purification at the cost of collective operations on the copies of $\tau^{\sR\sB}$.
Following Ref.~\cite{zhao2026petzrecoverysinglepolar}, we obtain the Petz recovery map by applying the QSVT with the sign function to an operator whose polar decomposition gives the map, without implementing the inverse square root of $\cF^{\sA\rarr\sB}(\sigma^\sA)$ as in Ref.~\cite{gilyen2022petzmap}.
Since this QSVT uses a block-encoding of a Stinespring isometry and its inverse sequentially, we show in \Cref{prop:random dilation} that our construction simulates such sequential uses for a randomly chosen Stinespring isometry.

To state \Cref{prop:random dilation}, we parametrize the purified states of $\tau^{\sR\sB}$ by unitaries on $\sE$ and define the corresponding unitaries in the Uhlmann transformation algorithm. Let $V_{\cF, 0}^{\sA\rarr\sB\sE}$ be a Stinespring isometry of $\cF^{\sA\rarr\sB}$, which exists since $d_\sE$ is at least the rank of $\tau^{\sR\sB}$. For a unitary $U^\sE$, let $V_{\cF, U}^{\sA\rarr\sB\sE} = U^\sE V_{\cF, 0}^{\sA\rarr\sB\sE}$ and $\ket{\tau_U}^{\sR\sB\sE} = V_{\cF, U}^{\sA\rarr\sB\sE}\ket{\Phi}^{\sR\sA} = U^\sE\ket{\tau_0}^{\sR\sB\sE}$, where $\ket{\tau_0}^{\sR\sB\sE} = V_{\cF, 0}^{\sA\rarr\sB\sE}\ket{\Phi}^{\sR\sA}$. We take a system $\sG$ such that $\sS = \sA\sG = \sB\sE$, padding the systems if necessary, and set the known state $\ket{\chi}^{\sR\sS} = \ket{\Phi}^{\sR\sA}\ket{0}^\sG$. Let $U_{2, U}^{\sC\hat{\sR}\sS\hat{\sS}}$ be the block-encoding unitary of $\ketbra{\chi}{\tau_U}^{\hat{\sR}\hat{\sS}}\otimes M_U^\sS$ with 
\begin{align}
\label{eq:def of m purified block}
    M_U^\sS = \tr_{\sR}\big[\ketbra{\tau_U}{\chi}^{\sR\sS}\big], 
\end{align}
given by Eq.~\eqref{eq:def of U2} for the Uhlmann transformation from $\ket{\chi}^{\sR\sS}$ to $\ket{\tau_U}^{\sR\sS}$. With $\sN = \sS\sH\hat{\sR}\hat{\sS}$ for a three-qubit system $\sH$, we define $U_{1, U}^\sN = \mathtt{QSVTSIGN}(U_{2, U}, U_{2, U}^\dag; 1/4, 1/d_\sA)$, which uses $U_{2, U}$ and $U_{2, U}^\dag$ a total of $u = \cO(d_\sA)$ times by \Cref{prop:FPAA general}.

\begin{proposition}[Sequential uses of a block-encoding of a random Stinespring isometry from a channel]
\label{prop:random dilation}
Let $\epsilon \in (0, 1)$, and let $\cA_U^{\sX\rarr\sY}$ be the quantum channel implemented by a quantum algorithm that uses $U_{1, U}^\sN$ and $(U_{1, U}^\sN)^\dag$ a total of $m$ times, where the other operations do not depend on $U^\sE$.
Then, there exists a quantum algorithm that implements a quantum channel $\til{\cA}^{\sX\rarr\sY}$ satisfying $\f{1}{2}\big\|\til{\cA}^{\sX\rarr\sY} - \bE_U\big[\cA_U^{\sX\rarr\sY}\big]\big\|_\diamond \leq \epsilon$, using the channel $\cF^{\sA\rarr\sB}$ a total of $n = \cO(m^2d_\sA^2/\epsilon)$ times, where $\bE_U[\cdot]$ denotes the expectation over a Haar random unitary $U^\sE$.
In addition to the operations of the original algorithm, the quantum circuit consists of $\til{\cO}(n^4)$ one- and two-qubit gates and uses $\til{\cO}(n^2)$ qubits at any one time.
\end{proposition}

Conversions related to \Cref{prop:random dilation} were recently studied in Refs.~\cite{girardi2026randomstinespringsuperchannelconverting, yoshida2026randomdilationsuperchannel}, where $n$ parallel uses of a quantum channel are exactly converted into $n$ parallel uses of a randomly chosen Stinespring isometry, and Ref.~\cite{yoshida2026randomdilationsuperchannel} also approximately converts them into sequential uses of a randomly chosen Stinespring dilation unitary and its inverse.
Using the parallel conversion, Ref.~\cite{yoshida2026randomdilationsuperchannel} implements the Petz recovery map exactly and probabilistically, when $\sigma^\sA$ is the completely mixed state and the Kraus rank of $\cF^{\sA\rarr\sB}$ is $d_\sA/d_\sB$, whereas \Cref{thm:petz implementation} applies to any state $\sigma^\sA$ and any channel $\cF^{\sA\rarr\sB}$.

We first prove \Cref{prop:random dilation} and then use it to prove \Cref{thm:petz implementation}.

\begin{proof}[Proof of \Cref{prop:random dilation}]

We first implement $U_{1, U}^\sN$ and its inverse for a fixed $U^\sE$ from copies of $\ket{\tau_U}^{\sR\sS}$, and then account for the preparation of the copies by the random purification.
For $U_{1, U}^\sN$, we replace each use of $U_{2, U}$ and $U_{2, U}^\dag$ with the quantum channels $\til{\cG}$ and $\til{\cG}_{\rm inv}$, respectively, as in \Cref{alg:Uhl purif sample} with inputs $\ket{\chi}^{\sR\sS}$ and $\ket{\tau_U}^{\sR\sS}$, where the block-encoding unitary of $\ketbra{\chi}{\chi}^{\sR\sS}$ is implemented exactly, and that of $\ketbra{\tau_U}{\tau_U}^{\sR\sS}$ is approximated within error $\epsilon/(2mu)$ by the channel given by $\mathtt{BlockEncPureState}(\ket{\tau_U}^{\sR\sS}; \epsilon/(2mu))$ in \Cref{lem:BE pure state}, which uses $\lceil 8\pi^2mu/\epsilon\rceil$ samples of $\ket{\tau_U}^{\sR\sS}$.
For the inverse, we implement the adjoint circuit of $U_{1, U}^\sN$, which is constructed by reversing the gate order and inverting the phase factors~\cite{gilyen2019qsvt, Gilyn2019thesis, martyn2021grand}, with the same replacement. By the subadditivity of the diamond norm, both $U_{1, U}^\sN$ and $(U_{1, U}^\sN)^\dag$ are approximated within error $\epsilon/(2m)$, using
\begin{align}
    w = u\big\lceil 8\pi^2mu/\epsilon\big\rceil = \cO\big(md_\sA^2/\epsilon\big)
\end{align}
copies of $\ket{\tau_U}^{\sR\sS}$ for each use.
Replacing all the $m$ uses in the algorithm, we obtain a channel $\cA_U'^{\sX\rarr\sY}$ satisfying 
\begin{align}
    \f{1}{2}\big\|\cA_U'^{\sX\rarr\sY} - \cA_U^{\sX\rarr\sY}\big\|_\diamond 
    \leq \epsilon/2,
\end{align}
using $n = mw = \cO(m^2d_\sA^2/\epsilon)$ copies of $\ket{\tau_U}^{\sR\sS}$.
The factor $d_\sA^2$ arises from the parameter $1/d_\sA$ of $\mathtt{QSVTSIGN}$ in the definition of $U_{1, U}^\sN$, which originates from the maximally entangled state $\ket{\Phi}^{\sR\sA}$ in $\ket{\chi}^{\sR\sS}$ rather than from the purification.

We next prepare the copies of $\ket{\tau_U}^{\sR\sB\sE}$ by the random purification.
Each copy of $\tau^{\sR\sB}$ is obtained by applying $\cF^{\sA\rarr\sB}$ to $\ket{\Phi}^{\sR\sA}$.
Applying $\mathtt{RandomPurification}(\tau^{\sR\sB}; \epsilon/2, n)$ in \Cref{prevthm:random purification} with $\ket{\tau_0}^{\sR\sB\sE}$ as the fixed purified state, where the condition on the rank is satisfied since $d_\sE$ is at least the rank of $\tau^{\sR\sB}$, we obtain a quantum channel $\cQ$ such that $\f{1}{2}\|\cQ - \cQ_{\rm ideal}\|_\diamond \leq \epsilon/2$, where $\cQ_{\rm ideal}$ maps $(\tau^{\sR\sB})^{\otimes n}$ to $\bE_U\big[(\ketbra{\tau_U}{\tau_U}^{\sR\sB\sE})^{\otimes n}\big]$.
Since the above algorithm processes the copies by a circuit independent of $U^\sE$, it implements $\bE_U\big[\cA_U'^{\sX\rarr\sY}\big]$ when the copies are given by $\cQ_{\rm ideal}$.
Let $\til{\cA}^{\sX\rarr\sY}$ be the quantum channel implemented when the copies are given by $\cQ$ instead, which uses $\cF^{\sA\rarr\sB}$ a total of $n$ times.
Then, we have
\begin{align}
    &\f{1}{2}\big\|\til{\cA}^{\sX\rarr\sY} - \bE_U\big[\cA_U^{\sX\rarr\sY}\big]\big\|_\diamond \notag\\
    &\leq \f{1}{2}\big\|\til{\cA}^{\sX\rarr\sY} - \bE_U\big[\cA_U'^{\sX\rarr\sY}\big]\big\|_\diamond 
    + \f{1}{2}\big\|\bE_U\big[\cA_U'^{\sX\rarr\sY} - \cA_U^{\sX\rarr\sY}\big]\big\|_\diamond \\
    &\leq \f{1}{2}\|\cQ - \cQ_{\rm ideal}\|_\diamond + \bE_U\Big[\f{1}{2}\big\|\cA_U'^{\sX\rarr\sY} - \cA_U^{\sX\rarr\sY}\big\|_\diamond\Big] \\
    &\leq \epsilon,
\end{align}
where we used the subadditivity of the diamond norm for the first term and its convexity for the second term.
From \Cref{prevthm:random purification}, the random purification consists of $\til{\cO}(n^4)$ one- and two-qubit gates and uses $\til{\cO}(n^2)$ qubits at any one time, while the rest of the algorithm is conducted sequentially.

\end{proof}

We now prove \Cref{thm:petz implementation}.

\begin{proof}[Proof of \Cref{thm:petz implementation}]

We first show that a block of $(U_{1, U}^\sN)^\dag$ contains $(V_{\cF, U}^{\sA\rarr\sB\sE})^\dag$.
Since $\tr_{\sR}\big[\ketbra{\Phi}{\Phi}^{\sR\sA}\big] = \pi^\sA$, $M_U^\sS$ in Eq.~\eqref{eq:def of m purified block} can be rephrased as 
\begin{align}
\label{inteq:petz flat}
    M_U^\sS = \f{1}{d_\sA}V_{\cF, U}^{\sA\rarr\sB\sE}\big(\bI^\sA\otimes\bra{0}^\sG\big),
\end{align}
whose non-zero singular values are all equal to $1/d_\sA$.
Thus, the Uhlmann partial isometry $W_{\cF, U}^\sS = \sgn^{({\rm SV})}(M_U^\sS)$ from $\ket{\chi}^{\sR\sS}$ to $\ket{\tau_U}^{\sR\sS}$ satisfies $W_{\cF, U}^\sS\ket{0}^\sG = V_{\cF, U}^{\sA\rarr\sB\sE}$.
From \Cref{prop:FPAA general}, $U_{1, U}^\sN$ is a block-encoding of $\ketbra{\chi}{\tau_U}^{\hat{\sR}\hat{\sS}}\otimes P_\sgn^{({\rm SV})}(M_U^\sS)$, where $P_\sgn$ is a real odd polynomial satisfying $|P_\sgn(x) - \sgn(x)| \leq 1/4$ for $|x| \in [1/d_\sA, 1]$ and $|P_\sgn(x)| \leq 1$ for $|x| \in [0, 1]$ (see also Eq.~\eqref{inteq:61}).
In particular, it holds that $P_\sgn^{({\rm SV})}(M_U^\sS) = cW_{\cF, U}^\sS$ with $c = P_\sgn(1/d_\sA) \in [3/4, 1]$, which is determined by $d_\sA$ and depends on neither $U^\sE$ nor $\cF^{\sA\rarr\sB}$.
Hence, from $\bra{0^3}^\sH(U_{1, U}^\sN)^\dag\ket{0^3}^\sH = c\ketbra{\tau_U}{\chi}^{\hat{\sR}\hat{\sS}}\otimes(W_{\cF, U}^\sS)^\dag$ and $\bra{0}^\sG(W_{\cF, U}^\sS)^\dag = (V_{\cF, U}^{\sA\rarr\sB\sE})^\dag$, we obtain
\begin{align}
\label{eq:oaa block identity}
    \bra{0}^\sG\bra{0^3}^\sH(U_{1, U}^\sN)^\dag\ket{0^3}^\sH\ket{\chi}^{\hat{\sR}\hat{\sS}}
    = c\ket{\tau_U}^{\hat{\sR}\hat{\sS}}\otimes(V_{\cF, U}^{\sA\rarr\sB\sE})^\dag.
\end{align}
The polynomial $P_\sgn$ changes the block only by the known factor $c$, so a constant accuracy of $\mathtt{QSVTSIGN}$ suffices here.

We next express the Petz recovery map with $V_{\cF, U}^{\sA\rarr\sB\sE}$ as the Stinespring isometry for a fixed $U^\sE$.
Let $\omega^\sB = \cF^{\sA\rarr\sB}(\sigma^\sA)$, and let $\sE'$ be a system isomorphic to $\sE$.
We define linear operators
\begin{align}
\label{inteq:petz N}
    &N_U^{\sB\rarr\sA\sE'} = \sqrt{\sigma^\sA}(V_{\cF, U}^{\sA\rarr\sB\sE})^\dag\big(\bI^\sB\otimes\ket{\Gamma}^{\sE\sE'}\big), \\
\label{inteq:petz T}
    &T_U^{\sB\rarr\sA\sE'} = N_U^{\sB\rarr\sA\sE'}(\omega^\sB)^{-1/2},
\end{align}
where $(\omega^\sB)^{-1/2}$ is taken on the support of $\omega^\sB$.
Since $(\cF^{\sA\rarr\sB})^\dag(\cdot) = (V_{\cF, U}^{\sA\rarr\sB\sE})^\dag(\cdot\otimes\bI^\sE)V_{\cF, U}^{\sA\rarr\sB\sE}$ and $\tr_{\sE'}\big[\ketbra{\Gamma}{\Gamma}^{\sE\sE'}\big] = \bI^\sE$, we see that 
\begin{align}
    \tr_{\sE'}\big[T_U^{\sB\rarr\sA\sE'}(\cdot)(T_U^{\sB\rarr\sA\sE'})^\dag\big] = \cR_{\sigma, \cF}^{\sB\rarr\sA}.
\end{align}
Also, since $\tr_\sE\big[V_{\cF, U}^{\sA\rarr\sB\sE}\sigma^\sA(V_{\cF, U}^{\sA\rarr\sB\sE})^\dag\big] = \omega^\sB$, we have $(N_U^{\sB\rarr\sA\sE'})^\dag N_U^{\sB\rarr\sA\sE'} = \omega^\sB$.
The non-zero singular values of $N_U^{\sB\rarr\sA\sE'}$ are the square roots of the non-zero eigenvalues of $\omega^\sB$, and $T_U^{\sB\rarr\sA\sE'} = \sgn^{({\rm SV})}(N_U^{\sB\rarr\sA\sE'})$ is the partial isometry in the polar decomposition of $N_U^{\sB\rarr\sA\sE'}$, which satisfies $(T_U^{\sB\rarr\sA\sE'})^\dag T_U^{\sB\rarr\sA\sE'} = \Pi_{\cF(\sigma)}^\sB$.
Thus, for every unitary $U^\sE$, 
\begin{align}
\label{inteq:petz T channel}
    &\tr_{\sE'}\big[T_U^{\sB\rarr\sA\sE'}(\cdot)(T_U^{\sB\rarr\sA\sE'})^\dag\big] \notag\\
    &\hspace{3pc}+ \tr\big[\big(\bI^\sB - (T_U^{\sB\rarr\sA\sE'})^\dag T_U^{\sB\rarr\sA\sE'}\big)(\cdot)\big]\ketbra{0}{0}^\sA \notag\\
    &= \hat{\cR}_{\sigma, \cF}^{\sB\rarr\sA}.
\end{align}

We implement $T_U^{\sB\rarr\sA\sE'}$ as a block of a unitary by applying the QSVT to a block-encoding of $N_U^{\sB\rarr\sA\sE'}$.
Let $W_{\sqrt{\sigma}}^{\sA\sC_1}$ be a $(1, 5, 0)$-block-encoding unitary of $\sqrt{\sigma^\sA}/(2\sqrt{2})$, which is approximately implemented from samples of $\sigma^\sA$ by \Cref{prop:block enc sqrt state}.
Let $U_{3, U}$ be the unitary that prepares $\ket{\Phi}^{\sE\sE'}$ and $\ket{\chi}^{\hat{\sR}\hat{\sS}}$ from the zero states and then applies $(U_{1, U}^\sN)^\dag$ and $W_{\sqrt{\sigma}}^{\sA\sC_1}$.
Consider its block
\begin{align}
    A_U^{\sB\rarr\sA\sE'\hat{\sR}\hat{\sS}} 
    &= \bra{0}^{\sC_1}\bra{0^3}^\sH\bra{0}^\sG U_{3, U}\ket{0}^{\sC_1}\ket{0^3}^\sH\ket{0}^{\sE\sE'\hat{\sR}\hat{\sS}} \\
    &= \bra{0}^{\sC_1}\bra{0^3}^\sH\bra{0}^\sG W_{\sqrt{\sigma}}^{\sA\sC_1} \notag\\
    &\hspace{3pc}(U_{1, U}^\sN)^\dag\ket{0}^{\sC_1}\ket{0^3}^\sH\ket{\chi}^{\hat{\sR}\hat{\sS}}\ket{\Phi}^{\sE\sE'}.
\end{align}
The projectors specifying this block are in the computational basis, and thus the reflections about them used below are implemented without $\ket{\tau_U}^{\sR\sS}$.
From Eq.~\eqref{eq:oaa block identity} and $\ket{\Phi}^{\sE\sE'} = \ket{\Gamma}^{\sE\sE'}/\sqrt{d_\sE}$, we have
\begin{align}
\label{inteq:petz block}
    A_U^{\sB\rarr\sA\sE'\hat{\sR}\hat{\sS}} = \f{c}{2\sqrt{2d_\sE}}\ket{\tau_U}^{\hat{\sR}\hat{\sS}}\otimes N_U^{\sB\rarr\sA\sE'},
\end{align}
where the factor $1/(2\sqrt{2})$ comes from the normalization of $W_{\sqrt{\sigma}}^{\sA\sC_1}$.
Letting $B_U^{\sB\rarr\sA\sE'\hat{\sR}\hat{\sS}} = \ket{\tau_U}^{\hat{\sR}\hat{\sS}}\otimes T_U^{\sB\rarr\sA\sE'}$ and $a = c\sqrt{\lambda_{\rm min}}/\big(2\sqrt{2d_\sE}\big)$, the non-zero singular values of $A_U$ are at least $a$, and $\sgn^{({\rm SV})}(A_U) = B_U$.

We obtain $B_U$ from $A_U$ by the QSVT with the sign function.
Applying $\mathtt{QSVTSIGN}(U_{3, U}, U_{3, U}^\dag; \epsilon/32, a)$ in \Cref{prop:FPAA general} with the above projectors~\cite{gilyen2019qsvt}, we obtain a unitary whose block $K_U^{\sB\rarr\sA\sE'\hat{\sR}\hat{\sS}}$ specified by the same projectors satisfies
\begin{align}
\label{inteq:petz amplified}
    \|K_U - B_U\|_\infty \leq \f{\epsilon}{32}.
\end{align}
This uses $U_{3, U}$ and $U_{3, U}^\dag$ a total of
\begin{align}
\label{inteq:petz rounds}
    v = \cO\Big(\f{1}{a}\log{\Big(\f{1}{\epsilon}\Big)}\Big) = \til{\cO}\bigg(\sqrt{\f{d_\sE}{\lambda_{\rm min}}}\bigg)
\end{align}
times.
After applying this unitary, we measure the auxiliary systems specifying the block. When all of them are in the zero states, we discard $\sE'\hat{\sR}\hat{\sS}$; otherwise, we output $\ketbra{0}{0}^\sA$.
The resulting quantum channel is
\begin{align}
    \cR_U^{\sB\rarr\sA}(\cdot) 
    &= \tr_{\sE'\hat{\sR}\hat{\sS}}\big[K_U^{\sB\rarr\sA\sE'\hat{\sR}\hat{\sS}}(\cdot)(K_U^{\sB\rarr\sA\sE'\hat{\sR}\hat{\sS}})^\dag\big] \notag\\
    &\hspace{1pc}+ \tr\big[\big(\bI^\sB - (K_U^{\sB\rarr\sA\sE'\hat{\sR}\hat{\sS}})^\dag K_U^{\sB\rarr\sA\sE'\hat{\sR}\hat{\sS}}\big)(\cdot)\big]\ketbra{0}{0}^\sA,
\end{align}
and replacing $K_U$ with $B_U$ in this expression results in $\hat{\cR}_{\sigma, \cF}^{\sB\rarr\sA}$ by Eq.~\eqref{inteq:petz T channel}.
Since $K_U(\cdot)K_U^\dag - B_U(\cdot)B_U^\dag = (K_U - B_U)(\cdot)K_U^\dag + B_U(\cdot)(K_U - B_U)^\dag$ and $K_U^\dag K_U - B_U^\dag B_U = (K_U - B_U)^\dag K_U + B_U^\dag(K_U - B_U)$, where $\|K_U\|_\infty, \|B_U\|_\infty \leq 1$, H\"{o}lder's inequality in Eq.~\eqref{eq:Holder ineq} and Eq.~\eqref{inteq:petz amplified} imply that
\begin{align}
\label{inteq:petz fixed U}
    \f{1}{2}\big\|\cR_U^{\sB\rarr\sA} - \hat{\cR}_{\sigma, \cF}^{\sB\rarr\sA}\big\|_\diamond
    \leq 2\|K_U - B_U\|_\infty
    \leq \f{\epsilon}{16}.
\end{align}

We now implement $W_{\sqrt{\sigma}}^{\sA\sC_1}$ and its inverse from samples of $\sigma^\sA$, with $\delta = \epsilon/(4v)$.
They are approximated within error $\delta$ by \Cref{prop:block enc sqrt state} with $\delta/2$ in place of both $\delta$ and $\epsilon$, using $\til{\cO}\big((1/\delta)\min\{1/\sigma_{\rm min}^2, 1/\delta^4\}\big)$ samples of $\sigma^\sA$ for each.
Since $U_{3, U}$ and $U_{3, U}^\dag$ are used $v$ times in total, the resulting quantum channel $\cR_U'^{\sB\rarr\sA}$, in which $U_{1, U}^\sN$ and its inverse are still used exactly, satisfies $\f{1}{2}\big\|\cR_U'^{\sB\rarr\sA} - \cR_U^{\sB\rarr\sA}\big\|_\diamond \leq v\delta = \epsilon/4$.
The number of samples of $\sigma^\sA$ is
\begin{align}
\label{inteq:petz sigma copies}
    \til{\cO}\Big(\f{v}{\delta}\min\Big\{\f{1}{\sigma_{\rm min}^2}, \f{1}{\delta^4}\Big\}\Big) 
    &= \til{\cO}\Big(\f{v^2}{\epsilon}\min\Big\{\f{1}{\sigma_{\rm min}^2}, \f{v^4}{\epsilon^4}\Big\}\Big) \\
    &= \til{\cO}\Big(\f{d_\sE}{\epsilon\lambda_{\rm min}}\min\Big\{\f{1}{\sigma_{\rm min}^2}, \f{d_\sE^2}{\epsilon^4\lambda_{\rm min}^2}\Big\}\Big),
\end{align}
which is Eq.~\eqref{eq:sampcomp Petz sigma}.

Finally, we replace $U_{1, U}^\sN$ and its inverse.
The quantum channel $\cR_U'^{\sB\rarr\sA}$ is implemented by an algorithm that uses $U_{1, U}^\sN$ and $(U_{1, U}^\sN)^\dag$ a total of $v$ times, where the other operations do not depend on $U^\sE$.
Applying \Cref{prop:random dilation} with $m = v$ and the error $\epsilon/4$, we obtain a quantum channel $\til{\cR}^{\sB\rarr\sA}$ such that $\f{1}{2}\big\|\til{\cR}^{\sB\rarr\sA} - \bE_U\big[\cR_U'^{\sB\rarr\sA}\big]\big\|_\diamond \leq \epsilon/4$, using $\cF^{\sA\rarr\sB}$ a total of
\begin{align}
\label{inteq:petz purified copies}
    \cO\Big(\f{v^2d_\sA^2}{\epsilon}\Big) = \til{\cO}\Big(\f{d_\sA^2 d_\sE}{\epsilon\lambda_{\rm min}}\Big)
\end{align}
times.
Since $\hat{\cR}_{\sigma, \cF}^{\sB\rarr\sA}$ does not depend on $U^\sE$, the triangle inequality, the convexity of the diamond norm, and Eq.~\eqref{inteq:petz fixed U} imply that
\begin{align}
    &\f{1}{2}\big\|\til{\cR}^{\sB\rarr\sA} - \hat{\cR}_{\sigma, \cF}^{\sB\rarr\sA}\big\|_\diamond \notag\\
    &\leq \f{1}{2}\big\|\til{\cR}^{\sB\rarr\sA} - \bE_U\big[\cR_U'^{\sB\rarr\sA}\big]\big\|_\diamond \notag\\
    &\hspace{0.5pc}+ \bE_U\Big[\f{1}{2}\big\|\cR_U'^{\sB\rarr\sA} - \cR_U^{\sB\rarr\sA}\big\|_\diamond + \f{1}{2}\big\|\cR_U^{\sB\rarr\sA} - \hat{\cR}_{\sigma, \cF}^{\sB\rarr\sA}\big\|_\diamond\Big] \\
    &\leq \f{\epsilon}{4} + \f{\epsilon}{4} + \f{\epsilon}{16} \\
    &\leq \epsilon.
\end{align}

\end{proof}

The canonical purification in \Cref{alg:canonical purification} can be used instead of the random purification, where the purifying system of $\tau^{\sR\sB}$ is of dimension $d_\sE = d_\sA d_\sB$.
In this case, the unitary preparing the canonical purified state of $\tau^{\sR\sB}$ and its inverse are approximated by the quantum channels in \Cref{prop:canonical purification}, and are used as the oracles in the Uhlmann transformation algorithm, as in \Cref{alg:Uhl mixed sample}.
Since the canonical purified state is fixed, \Cref{prop:random dilation} is not needed, and the algorithm does not require collective operations; it runs sequentially and uses $\cO\big(\log{(d_\sA d_\sB)}\big)$ qubits at any one time.
Although we omit the details, a step-by-step evaluation along the proof of \Cref{thm:petz implementation}, where both errors in \Cref{prop:canonical purification} are taken as $\Theta\big(\epsilon/(d_\sA v)\big)$ with $v$ in Eq.~\eqref{inteq:petz rounds} since the error is measured in the diamond norm, shows that the number of uses of $\cF^{\sA\rarr\sB}$ is
\begin{align}
    \til{\cO}\bigg(\f{d_\sA^4 d_\sB^2}{\epsilon\lambda_{\rm min}}\min\bigg\{\f{1}{\tau_{\rm min}^2}, \f{d_\sA^6 d_\sB^2 r_\tau^2}{\epsilon^4\lambda_{\rm min}^2}\bigg\}\bigg),
\end{align}
where $r_\tau$ and $\tau_{\rm min}$ are the rank and the minimum non-zero eigenvalue of $\tau^{\sR\sB}$, respectively.
This includes the dimensional factor and the dependence on $\tau_{\rm min}$ due to the canonical purification.

\acknowledgments
We thank Keisuke Fujii, Zhaoyi Li, Seth Lloyd, Takaya Matsuura, and Kosuke Mitarai for helpful discussions. We are also grateful to Hayato Arai for facilitating our collaboration.

T.U.\ acknowledges the support of JST CREST Grant Number JPMJCR23I3 and JST SPRING Grant Number JPMJSP2108.
Y.N.\ is supported by MEXT KAKENHI Grant-in-Aid for Transformative Research Areas (A) ``Extreme Universe'' Grant Numbers JP21H05182 and JP21H05183, JST CREST Grant Number JPMJCR23I3, and JST PRESTO Grant Number JPMJPR2456.
Q.W.\ acknowledges the support of the Startup Funding from Shanghai Jiao Tong University and the Engineering and Physical Sciences Research Council under Grant \mbox{EP/X026167/1}.
Part of the work of Q.W.\ was done when the author was with the School of Informatics, University of Edinburgh, Edinburgh, U.K.
R.T.\ acknowledges the support of JST CREST Grant Number JPMJCR23I3, JSPS KAKENHI Grant Numbers JP24K16975 and JP25K00924, and MEXT KAKENHI Grant-in-Aid for Transformative Research Areas (A) ``Extreme Universe'' Grant Number JP24H00943.

We used Claude Opus 5.5 (Anthropic) and GPT-6 Astra (OpenAI) to assist with checking derivations and proofs, suggesting arguments for some of the proofs, revising the text for clarity, and searching the literature. All AI-assisted content was directed and verified by the authors, who take full responsibility for the content of this paper.

\bibliographystyle{bibsty}
\bibliography{ref}

\clearpage

\appendix

\crefalias{section}{appendix}
\crefalias{subsection}{appendix}
\crefalias{subsubsection}{appendix}

\section{Algorithm for a variant of the Uhlmann transformation in the mixed sample access model}
\label{sec:altanative mixed sample uhlmann}

We provide a quantum algorithm for implementing a variant of the Uhlmann transformation, given mixed sample access to $\rho^\sA$ and $\sigma^\sA$.
The variant transformation corresponds to a partial isometry $V^\sA$ such that
\begin{align}
\label{inteq:144}
    \rF(\rho^\sA, \sigma^\sA) = \rF(V^\sA\ket{\rho_{\rm c}}^{\sA\sB}, \ket{\sigma_{\rm c}}^{\sA\sB}),
\end{align}
where $\ket{\rho_{\rm c}}^{\sA\sB}$ and $\ket{\sigma_{\rm c}}^{\sA\sB}$ are the canonical purified states~\cite{Winter2004ExtrinsicIntrinsic} of $\rho^\sA$ and $\sigma^\sA$, respectively.
Recall that this partial isometry $V^\sA$ is equivalent to the Uhlmann partial isometry constructed from states $\rho^\top$ and $\sigma^\top$ on the system $\sB$.
As shown by the inequality:
\begin{align}
    \rF(\rho^\sA, \sigma^\sA) 
    &= \rF(V^\sA\ket{\rho_{\rm c}}^{\sA\sB}, \ket{\sigma_{\rm c}}^{\sA\sB}) \\
    &\leq \rF\big(V^\sA\rho^\sA(V^\sA)^\dag, \sigma^\sA\big),
\end{align}
the partial isometry $V^\sA$ has the property that it always brings $\rho^\sA$ closer to $\sigma^\sA$.

Our statement is as follows, where we denote by $\rho_{\rm min}$ and $\sigma_{\rm min}$ the minimum non-zero eigenvalues of $\rho^\sA$ and $\sigma^{\sA}$, respectively, and by $s_{\rm min}$ and $R$ the minimum non-zero singular value and the rank of $\sqrt{\sigma^\sA}\sqrt{\rho^\sA}$, respectively.
It should be noted that $d_\sA = d_\sB$.

\begin{theorem}[Algorithm for a variant of the Uhlmann transformation in the mixed sample access model]
\label{thm:variant of uhlmann mixed sample}
    
Let $\delta \in (0, 1)$ and $\chi \in \{\diamond, \tF\}$. 
Then, the quantum sample algorithm $\mathtt{VariantUhlmannMixedSample}$ given by \Cref{alg:Uthm:variant of uhlmann mixed sample} satisfies the following.

A quantum channel $\cJ_\diamond^{\sA\sH}$ given by 
\begin{align}
    \cJ_\diamond^{\sA\sH}  
    =\mathtt{VariantUhlmannMixedSample}(\rho^{\hat{\sA}}, \sigma^{\hat{\sA}}; \delta, \diamond), 
\end{align}
satisfies that
\begin{equation}
    \f{1}{2}\big\|\cJ_\diamond^{\sA\sH}\circ\cP_{\ket{0}}^{\bC\rarr\sH} - \cU_{\rm ideal}^{\sA\sH}\circ\cP_{\ket{0}}^{\bC\rarr\sH}\big\|_\diamond \leq \delta,
\end{equation}
where $U_{\rm ideal}^{\sA\sH}$ is an exact block-encoding unitary of $V^\sA$, and $V^\sA$ is a partial isometry such that
\begin{align}
    \rF(V^\sA\ket{\rho_{\rm c}}^{\sA\sB}, \ket{\sigma_{\rm c}}^{\sA\sB}) = \rF(\rho^\sA, \sigma^\sA).
\end{align}
The algorithm uses $\zeta_\diamond$ samples of $\rho^{\hat{\sA}}$ and $\sigma^{\hat{\sA}}$, where 
\begin{align}
\label{inteq:146}
    \zeta_\diamond = \til{\cO}\Big(\f{1}{\delta s_{\rm min}^2} \min\Big\{\f{1}{\rho_{\rm min}^2} + \f{1}{\sigma_{\rm min}^2}, \f{1}{\delta^4 s_{\rm min}^4}\Big\}\Big).
\end{align}

A quantum channel $\cJ_\tF^{\sA\sH}$ given by 
\begin{align}
    \cJ_\tF^{\sA\sH} 
    = \mathtt{VariantUhlmannMixedSample}(\rho^{\hat{\sA}}, \sigma^{\hat{\sA}}; \delta, \tF),
\end{align}
satisfies that
\begin{align}
    \rF\big(\cT^{\sA}(\ketbra{\rho_{\rm c}}{\rho_{\rm c}}^{\sA\sB}), \ket{\sigma_{\rm c}}^{\sA\sB}\big) \geq \rF(\rho^\sA, \sigma^\sA) - \delta,
\end{align}
where $\cT^{\sA}= \tr_\sH\circ\cJ_\tF^{\sA\sH}\circ\cP_{\ket{0}}^{\bC\rarr\sH}$.
The algorithm uses $\zeta_{\tF}$ samples of $\rho^{\hat{\sA}}$ and $\sigma^{\hat{\sA}}$, where 
\begin{align}
\label{inteq:145}
    \zeta_{\tF} = \til{\cO}\Big(\f{1}{\delta}\min\Big\{\f{1}{s_{\rm min}^2}, \f{R^2}{\delta^2}\Big\}\min\Big\{\f{1}{\rho_{\rm min}^2} + \f{1}{\sigma_{\rm min}^2}, \f{R^2}{\delta^4}\Big\}\Big).
\end{align}

In both cases, the quantum circuit for implementing the algorithm consists of $\cO\big(\zeta_\chi \log{d_\sA}\big)$ one- and two-qubit gates, and uses $\cO\big(\log{d_\sA}\big)$ qubits simultaneously. The circuit description can be computed classically in $\poly\big(\zeta_\chi, \log{d_\sA}, \log{(1/\delta)}\big)$ time.

\end{theorem}

Unlike \Cref{alg:Uhl mixed sample} in Sec.~\ref{sec:mixed sample}, \Cref{alg:Uthm:variant of uhlmann mixed sample} can be implemented with substantially fewer samples, since Eqs.~\eqref{inteq:146} and~\eqref{inteq:145} do not include the dimensional factor $d_\sA$, which appears in the sample complexity in Theorem~\ref{thm:Uhlmann mixed state sample}.
This is because \Cref{alg:Uthm:variant of uhlmann mixed sample} avoids the use of the canonical purification algorithm $\mathtt{CanonicalPurification}$ in \Cref{alg:canonical purification}, which takes a high sample complexity.
Instead, it employs $\mathtt{BlockEncSqrtState}$ in \Cref{prop:block enc sqrt state} to directly block-encode the square root of given mixed states.

\begin{algorithm}[ht]
\caption{Algorithm for a variant of the Uhlmann transformation in the mixed sample access model \\ \parbox{\linewidth}{\centering $\mathtt{VariantUhlmannMixedSample}(\rho^{\hat{\sA}}, \sigma^{\hat{\sA}}; \delta, \chi)$ \\ (In Theorem~\ref{thm:variant of uhlmann mixed sample})}}
\label{alg:Uthm:variant of uhlmann mixed sample}
\SetKwInput{KwInput}{Input}
\SetKwInput{KwOutput}{Output}
\SetKwInput{KwParameters}{Parameters}
\SetKwComment{Comment}{$\triangleright$\ }{}
\SetCommentSty{textnormal}
\SetKwProg{Fn}{Subroutine}{}{end}
\SetKwFunction{UPS}{UhlmannPurifiedSample}

\SetAlgoNoEnd
\SetAlgoNoLine
\KwInput{Two quantum states $\rho^{\hat{\sA}}$ and $\sigma^{\hat{\sA}}$.}
\KwParameters{$\delta \in (0, 1)$ and $\chi \in \{\diamond, \tF\}$.}
\KwOutput{Quantum channel $\cJ^{\sA\sH}$.}
\SetAlgoLined

\If{$\chi = \diamond$}{
  Set $\delta_1 \gets \delta/(8u)$ and $\epsilon_1 \gets \delta/(8u)$, where $u$ is the number of uses of $\cL^{\sA\sC_1\sC_2}$ and $\cL_{\rm inv}^{\sA\sC_1\sC_2}$ in $\mathtt{QSVTSIGN}$ below. \\
}
\Else{
  Set $\delta_1 \gets \delta/(4u)$ and $\epsilon_1 \gets \delta/(128\sqrt{R})$, where $u$ is the number of uses of $\cL^{\sA\sC_1\sC_2}$ and $\cL_{\rm inv}^{\sA\sC_1\sC_2}$ in $\mathtt{QSVTSIGNFEJER}$ or $\mathtt{QSVTSIGN}$ below. \\
}
Set $\cF^{\sA\sC_1} \gets \mathtt{BlockEncSqrtState}(\rho^{\hat{\sA}}; \delta_1, \epsilon_1)$ 
and $\cF_{\rm inv}^{\sA\sC_1} \gets \mathtt{BlockEncSqrtState}^\dag(\rho^{\hat{\sA}}; \delta_1, \epsilon_1)$. \\
Set $\cG^{\sA\sC_2} \gets \mathtt{BlockEncSqrtState}(\sigma^{\hat{\sA}}; \delta_1, \epsilon_1)$ 
and $\cG_{\rm inv}^{\sA\sC_2} \gets \mathtt{BlockEncSqrtState}^\dag(\sigma^{\hat{\sA}}; \delta_1, \epsilon_1)$ (Proposition~\ref{prop:block enc sqrt state}). \\
Set $\cL^{\sA\sC_1\sC_2} \gets \cG^{\sA\sC_2} \circ \cF^{\sA\sC_1}$
and $\cL_{\rm inv}^{\sA\sC_1\sC_2} \gets \cF_{\rm inv}^{\sA\sC_1} \circ \cG_{\rm inv}^{\sA\sC_2}$. \\
\If{$\chi = \tF$ and $\mathtt{QSVTSIGNFEJER}$ below uses $\cL^{\sA\sC_1\sC_2}$ and $\cL_{\rm inv}^{\sA\sC_1\sC_2}$ fewer times than $\mathtt{QSVTSIGN}$ below}{
  Set $\cJ^{\sA\sH} \gets \mathtt{QSVTSIGNFEJER}(\cL^{\sA\sC_1\sC_2}, \cL_{\rm inv}^{\sA\sC_1\sC_2}; \delta/(64R))$ (Proposition~\ref{prop:QSVT sign Fejer}). \\
}
\Else{
  Set $\beta \gets s_{\rm min}/8$ if $\chi = \diamond$, and $\beta \gets s_{\rm min}/16$ if $\chi = \tF$. \\
  Set $\cJ^{\sA\sH} \gets \mathtt{QSVTSIGN}(\cL^{\sA\sC_1\sC_2}, \cL_{\rm inv}^{\sA\sC_1\sC_2}; (\delta/6)^2, \beta)$ (Proposition~\ref{prop:FPAA general}). \\
}
Return $\cJ^{\sA\sH}$.
\end{algorithm}

\begin{proof}[Proof of Theorem~\ref{thm:variant of uhlmann mixed sample}]
    
From Uhlmann's theorem,  there exists a partial isometry $V^\sA$ which satisfies
\begin{align}
\label{inteq:49}
    \rF(V^\sA\ket{\rho_{\rm c}}^{\sA\sB}, \ket{\sigma_{\rm c}}^{\sA\sB}) 
    &= \rF(\rho_{\rm c}^{\sB}, \sigma_{\rm c}^{\sB}) \\
    &= \rF\big((\rho^\sA)^\top, (\sigma^\sA)^\top\big) \\
    \label{inteq:129}
    &= \rF(\rho^\sA, \sigma^\sA).
\end{align}
An explicit form of this partial isometry $V^\sA$ is given by 
\begin{align}
    V^\sA 
    &= \sgn^{(\rm SV)}\big(\tr_\sB\big[\ketbra{\sigma_{\rm c}}{\rho_{\rm c}}^{\sA\sB}\big]\big) \\
    &= \sgn^{(\rm SV)}\big(\tr_\sB\big[\sqrt{\sigma^\sA}\ketbra{\Gamma}{\Gamma}^{\sA\sB}\sqrt{\rho^\sA}\big]\big) \\
    &= \sgn^{(\rm SV)}\big(\sqrt{\sigma^\sA}\sqrt{\rho^\sA}\big),
\end{align}
where $\ket{\Gamma}^{\sA\sB} = \sum_i\ket{i}^\sA\ket{i}^\sB$.
Thus, it suffices to directly construct a unitary that block-encodes $\sqrt{\sigma^\sA}\sqrt{\rho^\sA}$, and then lift up all singular values to unity using the QSVT for sign function $\mathtt{QSVTSIGN}$ (Proposition~\ref{prop:FPAA general}).

From a block-encoding of the square root of a quantum state (Proposition~\ref{prop:block enc sqrt state}), we obtain quantum channels $\cF^{\sA\sC_1} = \mathtt{BlockEncSqrtState}(\rho^{\hat{\sA}}; \delta_1, \epsilon_1)$ and $\cG^{\sA\sC_2} = \mathtt{BlockEncSqrtState}(\sigma^{\hat{\sA}}; \delta_1, \epsilon_1)$ with $\delta_1, \epsilon_1 \in (0, 1)$, such that
\begin{align}
\label{inteq:127}
    \f{1}{2}\big\|\cF^{\sA\sC_1} - \cU_\rho^{\sA\sC_1}\big\|_\diamond \leq \delta_1, \ \ \text{and} \ \ \
    \f{1}{2}\big\|\cG^{\sA\sC_2} - \cU_\sigma^{\sA\sC_2}\big\|_\diamond \leq \delta_1,
\end{align}
where $U_\rho^{\sA\sC_1}$ and $U_\sigma^{\sA\sC_2}$ are unitaries satisfying
\begin{align}
\label{inteq:variant sqrt bias}
    \big\|U_\rho^{\sA\sC_1} - W_{\sqrt{\rho}}^{\sA\sC_1}\big\|_\infty \leq \epsilon_1, \ \ \text{and} \ \ \
    \big\|U_\sigma^{\sA\sC_2} - W_{\sqrt{\sigma}}^{\sA\sC_2}\big\|_\infty \leq \epsilon_1
\end{align}
for $(1, 5, 0)$-block-encoding unitaries $W_{\sqrt{\rho}}^{\sA\sC_1}$ and $W_{\sqrt{\sigma}}^{\sA\sC_2}$ of $\sqrt{\rho^\sA}/(2\sqrt{2})$ and $\sqrt{\sigma^\sA}/(2\sqrt{2})$, respectively.
These channels use $h = \til{\cO}\Big(\f{1}{\delta_1}\min\Big\{\f{1}{\rho_{\rm min}^2} + \f{1}{\sigma_{\rm min}^2}, \f{1}{\epsilon_1^4}\Big\}\Big)$ samples of $\rho^{\hat{\sA}}$ and $\sigma^{\hat{\sA}}$.

We here use the following lemma regarding the product of block-encoded matrices.
\begin{lemma}[Product of block-encoded matrices~{\cite[Lemma 30]{gilyen2019qsvt}}; see also Ref.~{\cite[Lemma 3.3.10]{Gilyn2019thesis}}]
\label{lem:product of BE}
Suppose that $U_1^{\sA\sB}$ is an $(\alpha_1, a_1, \epsilon_1)$-block-encoding unitary of a matrix $M_1^\sA$, and $U_2^{\sA\sC}$ is an $(\alpha_2, a_2, \epsilon_2)$-block-encoding unitary of a matrix $M_2^\sA$. Then, $(\bI^\sC\otimes U_1^{\sA\sB})(\bI^\sB \otimes U_2^{\sA\sC})$ is an $(\alpha_1\alpha_2, a_1+a_2, \alpha_1\epsilon_2 + \alpha_2\epsilon_1)$-block-encoding of $M_1^\sA M_2^\sA$.
\end{lemma}

Let $W^{\sA\sC_1\sC_2} = W_{\sqrt{\sigma}}^{\sA\sC_2}W_{\sqrt{\rho}}^{\sA\sC_1}$.
From \Cref{lem:product of BE}, we see that
\begin{align}
    W^{\sA\sC_1\sC_2} 
    = 
\begin{blockarray}{ccc}
 & \bra{0}^{\sC_1\sC_2} &  & \vspace{1mm}\\
\begin{block}{c(cc)}
  \ket{0}^{\sC_1\sC_2} \hspace*{1mm} & \f{1}{8}\sqrt{\sigma^\sA}\sqrt{\rho^\sA} & \hspace{0mm} * \hspace{3mm}\\
   \hspace*{1mm} & * & \hspace{0mm} * \hspace{3mm}\\
\end{block}
\end{blockarray}\hspace{1mm}.
\end{align}
We also define $U^{\sA\sC_1\sC_2} = U_\sigma^{\sA\sC_2}U_\rho^{\sA\sC_1}$.
When we set $\cL^{\sA\sC_1\sC_2} = \cG^{\sA\sC_2}\circ\cF^{\sA\sC_1}$, we have 
\begin{align}
\label{inteq:130}
    &\f{1}{2}\big\|\cL^{\sA\sC_1\sC_2} - \cU^{\sA\sC_1\sC_2}\big\|_\diamond \leq 2\delta_1, \\
\label{inteq:variant W bias}
    &\big\|U^{\sA\sC_1\sC_2} - W^{\sA\sC_1\sC_2}\big\|_\infty \leq 2\epsilon_1,
\end{align}
where we used Eqs.~\eqref{inteq:127} and~\eqref{inteq:variant sqrt bias}, and that $\|U_1U_2 - V_1V_2\|_\infty \leq \|U_1 - V_1\|_\infty + \|U_2 - V_2\|_\infty$ for unitaries $U_1$, $U_2$, $V_1$, and $V_2$.
Thus, the $(1,10,0)$-block-encoding unitary $W^{\sA\sC_1\sC_2}$ of $\sqrt{\sigma^\sA}\sqrt{\rho^\sA}/8$ is approximately obtained with error $2\delta_1 + 2\epsilon_1$ in the diamond norm, where we used Eq.~\eqref{eq:rel of operator and diamond norm}.
Similarly, a quantum channel $\cL_{\rm inv}^{\sA\sC_1\sC_2}$ that satisfies $\f{1}{2}\|\cL_{\rm inv}^{\sA\sC_1\sC_2} - (\cU^{\sA\sC_1\sC_2})^\dag\|_\diamond \leq 2\delta_1$ is obtained from the same number of samples of $\rho^{\hat{\sA}}$ and $\sigma^{\hat{\sA}}$.

The rest is similar to the discussion in Sec.~\ref{sec:proofUhlmann}.
Let $\cJ^{\sA\sH}$ be the quantum channel given by $\mathtt{QSVTSIGNFEJER}$ or $\mathtt{QSVTSIGN}$ in \Cref{alg:Uthm:variant of uhlmann mixed sample}, and let $u$ be the number of uses of $\cL^{\sA\sC_1\sC_2}$ and $\cL_{\rm inv}^{\sA\sC_1\sC_2}$ in it.
We denote by $\til{U}^{\sA\sH}$ and $\til{W}^{\sA\sH}$ the unitaries obtained from $\cJ^{\sA\sH}$ by replacing $\cL^{\sA\sC_1\sC_2}$ and $\cL_{\rm inv}^{\sA\sC_1\sC_2}$ with $U^{\sA\sC_1\sC_2}$ and $(U^{\sA\sC_1\sC_2})^\dag$, and with $W^{\sA\sC_1\sC_2}$ and $(W^{\sA\sC_1\sC_2})^\dag$, respectively.
From Eqs.~\eqref{inteq:130} and~\eqref{inteq:variant W bias}, we see that
\begin{align}
\label{inteq:variant J error}
    &\f{1}{2}\big\|\cJ^{\sA\sH} - \til{\cU}^{\sA\sH}\big\|_\diamond \leq 2u\delta_1, \\
\label{inteq:variant tilde U bias}
    &\big\|\til{U}^{\sA\sH} - \til{W}^{\sA\sH}\big\|_\infty \leq 2u\epsilon_1,
\end{align}
where we used the subadditivity of the diamond norm.

First, regarding the evaluation in the diamond norm, $\mathtt{QSVTSIGN}$ is used with $\delta_2 \in (0, 1/2)$ and $\beta = \beta_\diamond = s_{\rm min}/8$, which gives $u = \cO\big(\log(1/\delta_2)/\beta_\diamond\big)$, and then, it holds that $\big\|P_\sgn^{(\rm SV)}(\til{M}^\sA) - V^\sA\big\|_\infty \leq \delta_2$, where $\til{M}^\sA = \sqrt{\sigma^\sA}\sqrt{\rho^\sA}/8$ and $V^\sA = \sgn^{(\rm SV)}(\til{M}^\sA)$.
Since $\til{W}^{\sA\sH}$ is a block-encoding of $P_\sgn^{(\rm SV)}(\til{M}^\sA)$, \Cref{lem:FPAA diamond} with $d_\sE = 1$ implies that
\begin{align}
\label{inteq:variant ideal}
    \f{1}{2}\big\|\til{\cW}^{\sA\sH}\circ\cP_{\ket{0}}^{\bC\rarr\sH} - \cU_{\rm ideal}^{\sA\sH} \circ \cP_{\ket{0}}^{\bC\rarr\sH}\big\|_\diamond \leq 3\sqrt{\delta_2}.
\end{align}
Then, we have that
\begin{align}
    &\f{1}{2}\big\|\cJ^{\sA\sH}\circ\cP_{\ket{0}}^{\bC\rarr\sH} - \cU_{\rm ideal}^{\sA\sH} \circ \cP_{\ket{0}}^{\bC\rarr\sH}\big\|_\diamond \notag\\
    &\leq \f{1}{2}\big\|\cJ^{\sA\sH} - \til{\cU}^{\sA\sH}\big\|_\diamond + \f{1}{2}\big\|\til{\cU}^{\sA\sH} - \til{\cW}^{\sA\sH}\big\|_\diamond \notag\\
    &\hspace{2pc}+ \f{1}{2}\big\|\til{\cW}^{\sA\sH}\circ\cP_{\ket{0}}^{\bC\rarr\sH} - \cU_{\rm ideal}^{\sA\sH} \circ \cP_{\ket{0}}^{\bC\rarr\sH}\big\|_\diamond \\
\label{inteq:128}
    &\leq 2u\delta_1 + 2u\epsilon_1 + 3\sqrt{\delta_2},
\end{align}
where we used Eqs.~\eqref{inteq:variant J error} and~\eqref{inteq:variant tilde U bias}, and Eq.~\eqref{inteq:variant ideal} together with the relation between the diamond and operator norms in Eq.~\eqref{eq:rel of operator and diamond norm}.
Rescaling the parameters as $\delta_1 = \epsilon_1 = \delta/(8u)$ and $\delta_2 = (\delta/6)^2$, Eq.~\eqref{inteq:128} is bounded by $\delta$. In this case, the number of samples of $\rho^{\hat{\sA}}$ and $\sigma^{\hat{\sA}}$ is evaluated as $\zeta_\diamond = hu$, that is, 
\begin{align}
    \zeta_\diamond 
    &= \cO\Big(\f{1}{\beta_\diamond}\log{\Big(\f{1}{\delta_2}\Big)}\Big)\til{\cO}\Big(\f{1}{\delta_1}\min\Big\{\f{1}{\rho_{\rm min}^2} + \f{1}{\sigma_{\rm min}^2}, \f{1}{\epsilon_1^4}\Big\}\Big) \\
    &= \til{\cO}\Big(\f{1}{\delta s_{\rm min}^2}\min\Big\{\f{1}{\rho_{\rm min}^2}+\f{1}{\sigma_{\rm min}^2}, \f{1}{\delta^4s_{\rm min}^4}\Big\}\Big).
\end{align}

Next, we evaluate the error in the fidelity difference.
In this case, we directly evaluate $\til{U}^{\sA\sH}$ without using $\til{W}^{\sA\sH}$, which allows us to choose $\epsilon_1$ independently of $u$.
From \Cref{prop:block enc sqrt state}, the unitaries $U_\rho^{\sA\sC_1}$ and $U_\sigma^{\sA\sC_2}$ are block-encodings of $P_\rho(\rho^\sA/2)$ and $P_\sigma(\sigma^\sA/2)$, respectively, for real odd polynomials $P_\rho$ and $P_\sigma$ satisfying
\begin{align}
\label{inteq:variant sqrt poly}
    &\big\|P_\rho(\rho^\sA/2) - \sqrt{\rho^\sA}/\big(2\sqrt{2}\big)\big\|_\infty \leq \epsilon_1/2, \\
\label{inteq:variant sqrt poly sigma}
    &\big\|P_\sigma(\sigma^\sA/2) - \sqrt{\sigma^\sA}/\big(2\sqrt{2}\big)\big\|_\infty \leq \epsilon_1/2.
\end{align}
From \Cref{lem:product of BE}, $U^{\sA\sC_1\sC_2}$ is a block-encoding of 
\begin{align}
\label{inteq:variant N}
    N^\sA = P_\sigma(\sigma^\sA/2)P_\rho(\rho^\sA/2),
\end{align}
and $\til{U}^{\sA\sH}$ is a block-encoding of $P_\sgn^{(\rm SV)}(N^\sA)$, where $P_\sgn$ is the real odd polynomial in \Cref{prop:QSVT sign Fejer} or \Cref{prop:FPAA general}, which satisfies $|P_\sgn(x)| \leq 1$ for $|x| \in [0, 1]$.

We first evaluate the difference between $N^\sA$ and $\til{M}^\sA$ in the trace norm.
Let $D_\rho^\sA = P_\rho(\rho^\sA/2) - \sqrt{\rho^\sA}/(2\sqrt{2})$ and $D_\sigma^\sA = P_\sigma(\sigma^\sA/2) - \sqrt{\sigma^\sA}/(2\sqrt{2})$.
Since $P_\rho(0) = 0$, $D_\rho^\sA$ is a function of $\rho^\sA$ that vanishes on $\ker[\rho^\sA]$, and it can be written as $D_\rho^\sA = \sqrt{\rho^\sA}X_\rho^\sA$ with an operator $X_\rho^\sA$.
Similarly, $D_\sigma^\sA = X_\sigma^\sA\sqrt{\sigma^\sA}$ with an operator $X_\sigma^\sA$.
Then, we obtain that
\begin{align}
\label{inteq:variant N difference}
    &N^\sA - \til{M}^\sA \notag\\
    &= \f{1}{2\sqrt{2}}\big(D_\sigma^\sA\sqrt{\rho^\sA} + \sqrt{\sigma^\sA}D_\rho^\sA\big) + D_\sigma^\sA D_\rho^\sA \\
    &= \f{1}{2\sqrt{2}}\big(X_\sigma^\sA\sqrt{\sigma^\sA}\sqrt{\rho^\sA} + \sqrt{\sigma^\sA}\sqrt{\rho^\sA}X_\rho^\sA\big) + X_\sigma^\sA\sqrt{\sigma^\sA}\sqrt{\rho^\sA}X_\rho^\sA.
\end{align}
Each term on the right-hand side has rank at most $R$, and so does $N^\sA = \big(\bI^\sA/(2\sqrt{2}) + X_\sigma^\sA\big)\sqrt{\sigma^\sA}\sqrt{\rho^\sA}\big(\bI^\sA/(2\sqrt{2}) + X_\rho^\sA\big)$.
From H\"{o}lder's inequality in Eq.~\eqref{eq:Holder ineq}, $\|\sqrt{\rho^\sA}\|_2 = \|\sqrt{\sigma^\sA}\|_2 = 1$, and Eqs.~\eqref{inteq:variant sqrt poly} and~\eqref{inteq:variant sqrt poly sigma}, we have $\|D_\sigma^\sA\sqrt{\rho^\sA}\|_2 \leq \epsilon_1/2$, $\|\sqrt{\sigma^\sA}D_\rho^\sA\|_2 \leq \epsilon_1/2$, and $\|D_\sigma^\sA D_\rho^\sA\|_\infty \leq \epsilon_1^2/4$.
Since $\|Y\|_1 \leq \sqrt{R}\|Y\|_2 \leq R\|Y\|_\infty$ for any matrix $Y$ of rank at most $R$, Eq.~\eqref{inteq:variant N difference} implies that
\begin{align}
    \big\|N^\sA - \til{M}^\sA\big\|_1 
    &\leq \f{\sqrt{R}\epsilon_1}{2\sqrt{2}} + \f{R\epsilon_1^2}{4} \\
\label{inteq:variant trace error}
    &\leq \f{\delta}{256},
\end{align}
where we used $\epsilon_1 = \delta/(128\sqrt{R})$ in the last inequality.

We next evaluate the overlap $z = \bra{\sigma_{\rm c}}^{\sA\sB}\bra{0}^\sH\til{U}^{\sA\sH}\ket{\rho_{\rm c}}^{\sA\sB}\ket{0}^\sH$.
Let $\{s_k\}_k$ and $\{t_k\}_k$ be the singular values of $\sqrt{\sigma^\sA}\sqrt{\rho^\sA}$ and $N^\sA$, respectively, both arranged in non-increasing order and including zeros, where $\sum_k s_k = \sqrt{\rF}(\rho^\sA, \sigma^\sA) \leq 1$ and $t_k \leq 1$.
Since $\ket{\rho_{\rm c}}^{\sA\sB} = (\sqrt{\rho^\sA}\otimes\bI^\sB)\ket{\Gamma}^{\sA\sB}$ and $\ket{\sigma_{\rm c}}^{\sA\sB} = (\sqrt{\sigma^\sA}\otimes\bI^\sB)\ket{\Gamma}^{\sA\sB}$, we have
\begin{align}
    z = \tr\big[P_\sgn^{(\rm SV)}(N^\sA)\big(\sqrt{\sigma^\sA}\sqrt{\rho^\sA}\big)^\dag\big].
\end{align}
Since $\|P_\sgn^{(\rm SV)}(N^\sA)\|_\infty \leq 1$, H\"{o}lder's inequality implies that $|z| \leq \sqrt{\rF}(\rho^\sA, \sigma^\sA)$ and $\big|\tr\big[P_\sgn^{(\rm SV)}(N^\sA)Y^\dag\big]\big| \leq \|Y\|_1$ for any matrix $Y$.
Using $\tr\big[P_\sgn^{(\rm SV)}(N^\sA)(N^\sA)^\dag\big] = \sum_k P_\sgn(t_k)t_k$ and $\sqrt{\rF}(\rho^\sA, \sigma^\sA) = 8\|\til{M}^\sA\|_1 \leq 8\sum_k t_k + 8\|N^\sA - \til{M}^\sA\|_1$, we obtain
\begin{align}
    &\sqrt{\rF}(\rho^\sA, \sigma^\sA) - |z| \notag\\
    &\leq \sqrt{\rF}(\rho^\sA, \sigma^\sA) - 8\sum_k P_\sgn(t_k)t_k + 8\big\|N^\sA - \til{M}^\sA\big\|_1 \\
\label{inteq:variant fejer}
    &\leq 8\sum_k \big(1 - P_\sgn(t_k)\big)t_k + \f{\delta}{16},
\end{align}
where we used Eq.~\eqref{inteq:variant trace error} in the last inequality.
When $\mathtt{QSVTSIGNFEJER}$ is used with $\beta = \delta/(64R)$, \Cref{prop:QSVT sign Fejer} implies that $\big(1 - P_\sgn(t_k)\big)t_k \leq \beta$ for each $k$.
Since the rank of $N^\sA$ is at most $R$, at most $R$ of $t_k$ are non-zero, and thus the first term on the right-hand side of Eq.~\eqref{inteq:variant fejer} is at most $8R\beta = \delta/8$.
When $\mathtt{QSVTSIGN}$ is used with $\delta_2 = (\delta/6)^2$ and $\beta = s_{\rm min}/16$, \Cref{prop:FPAA general} implies that $1 - P_\sgn(t_k) \leq \delta_2$ for $k$ such that $t_k \geq s_{\rm min}/16$.
For $k$ such that $t_k < s_{\rm min}/16$, we have $t_k \leq |t_k - s_k/8|$, since either $s_k \geq s_{\rm min}$ or $s_k = 0$.
Mirsky's inequality~\cite{mirsky1960symmetric,bhatia2013matrix} for the trace norm gives $\sum_k |t_k - s_k/8| \leq \|N^\sA - \til{M}^\sA\|_1$, since $\{s_k/8\}_k$ are the singular values of $\til{M}^\sA$.
Using $1 - P_\sgn(t_k) \leq 2$ and $\sum_k t_k \leq \|\til{M}^\sA\|_1 + \|N^\sA - \til{M}^\sA\|_1 \leq 1/8 + \delta/256$, we obtain that
\begin{align}
    &8\sum_k \big(1 - P_\sgn(t_k)\big)t_k \notag\\
    &\leq 8\delta_2\sum_{k; \, t_k \geq s_{\rm min}/16}t_k + 16\sum_{k; \, t_k < s_{\rm min}/16}\big|t_k - s_k/8\big| \\
    &\leq \delta_2\Big(1 + \f{\delta}{32}\Big) + \f{\delta}{16} \\
    &\leq \f{\delta}{8},
\end{align}
where we used Eq.~\eqref{inteq:variant trace error} in the second inequality.
Thus, in both cases, Eq.~\eqref{inteq:variant fejer} implies that $\sqrt{\rF}(\rho^\sA, \sigma^\sA) - |z| \leq 3\delta/16$.
Since $0 \leq |z| \leq \sqrt{\rF}(\rho^\sA, \sigma^\sA) \leq 1$, it follows that
\begin{align}
    \rF(\rho^\sA, \sigma^\sA) - |z|^2 
    \leq 2\big(\sqrt{\rF}(\rho^\sA, \sigma^\sA) - |z|\big) 
    \leq \f{3\delta}{8}.
\end{align}
Let $\cT^\sA = \tr_\sH\circ \cJ^{\sA\sH} \circ \cP_{\ket{0}}^{\bC\rarr\sH}$. We obtain that
\begin{align}
    &\rF\big(\cT^\sA(\ketbra{\rho_{\rm c}}{\rho_{\rm c}}^{\sA\sB}), \ket{\sigma_{\rm c}}^{\sA\sB}\big) \notag\\
    &\geq \rF\big(\til{U}^{\sA\sH}\ket{\rho_{\rm c}}^{\sA\sB}\ket{0}^\sH, \ket{\sigma_{\rm c}}^{\sA\sB}\ket{0}^\sH\big) - \f{1}{2}\big\|\cJ^{\sA\sH} - \til{\cU}^{\sA\sH}\big\|_\diamond \\
    &\geq |z|^2 - 2u\delta_1 \\
    &\geq \rF(\rho^\sA, \sigma^\sA) - \f{3\delta}{8} - 2u\delta_1,
\end{align}
where we used the monotonicity of the fidelity under the partial trace, Eqs.~\eqref{eq:fidelity and diamond} and~\eqref{inteq:variant J error}.
Hence, rescaling the parameters as $\delta_1 = \delta/(4u)$ and $\epsilon_1 = \delta/(128\sqrt{R})$ as in \Cref{alg:Uthm:variant of uhlmann mixed sample} yields the result that $\rF\big(\cT^\sA(\ketbra{\rho_{\rm c}}{\rho_{\rm c}}^{\sA\sB}), \ket{\sigma_{\rm c}}^{\sA\sB}\big) \geq \rF(\rho^\sA, \sigma^\sA) -\delta$.
Since \Cref{alg:Uthm:variant of uhlmann mixed sample} uses the one of $\mathtt{QSVTSIGN}$ and $\mathtt{QSVTSIGNFEJER}$ that uses $\cL^{\sA\sC_1\sC_2}$ and $\cL_{\rm inv}^{\sA\sC_1\sC_2}$ fewer times, we have $u = \cO\big(\min\{\log{(1/\delta)}/s_{\rm min}, R/\delta\}\big)$.
The number of samples $\zeta_\tF = hu$ of $\rho^{\hat{\sA}}$ and $\sigma^{\hat{\sA}}$ is thus given by
\begin{align}
    \zeta_\tF 
    &= \cO(u)\til{\cO}\Big(\f{1}{\delta_1}\min\Big\{\f{1}{\rho_{\rm min}^2} + \f{1}{\sigma_{\rm min}^2}, \f{1}{\epsilon_1^4}\Big\}\Big) \\
    &= \til{\cO}\Big(\f{u^2}{\delta}\min\Big\{\f{1}{\rho_{\rm min}^2} + \f{1}{\sigma_{\rm min}^2}, \f{R^2}{\delta^4}\Big\}\Big) \\
    &= \til{\cO}\Big(\f{1}{\delta}\min\Big\{\f{1}{s_{\rm min}^2}, \f{R^2}{\delta^2}\Big\}\min\Big\{\f{1}{\rho_{\rm min}^2} + \f{1}{\sigma_{\rm min}^2}, \f{R^2}{\delta^4}\Big\}\Big).
\end{align}

The algorithm $\mathtt{BlockEncSqrtState}$ uses $\cO(h\log{d_\sA})$ one- and two-qubit gates and is repeated $\cO(u)$ times in $\mathtt{QSVTSIGN}$ or $\mathtt{QSVTSIGNFEJER}$.
Thus, the total number of gates in the quantum circuit for the entire algorithm is $\cO(hu\log{d_\sA}) = \cO(\zeta_\chi\log{d_\sA})$, where $\chi \in \{\diamond, \tF\}$.
Due to the sequential property of the algorithm, its quantum circuit uses $\cO(\log{d_\sA})$ qubits at any one time.
Based on its construction, the circuit description can be computed classically in $\poly\big(\zeta_\chi, \log d_\sA, \log(1/\delta)\big)$ time.

\end{proof}


\section{Naive state tomography-based approach for the Uhlmann transformation}
\label{sec:Comparing with a naive tomography-based}

The most straightforward strategy for implementing the Uhlmann transformation is to use quantum state tomography~\cite{ODonnell2016EffTomo, Haah2017samploptTomo, Guta2020FastTomo, Apeldoorn2023TomoPurifQuery, Chen2023WhenAdapTomo, hu2024sampleoptimalmemoryefficient}.
We first obtain approximate classical descriptions of $\ket{\rho}^{\sA\sB}$ and $\ket{\sigma}^{\sA\sB}$ via state tomography, and then directly compute the Uhlmann partial isometry $V^\sB = \sgn^{(\rm SV)}\big(\tr_{\sA}\big[\ketbra{\sigma}{\rho}^{\sA\sB}\big]\big)$ on a classical computer.
This partial isometry coincides with that determined by the Schmidt bases of $\ket{\rho}^{\sA\sB}$ and $\ket{\sigma}^{\sA\sB}$.
To implement the partial isometry by quantum circuits, we need to extend it to a unitary and, for instance, decompose it into quantum circuits by a brute-force method. This clearly results in an exponential number of gates in general~\cite{nielsen2010quantum}.
Moreover, in a classical computation step, this approach generally requires exponential time and space in terms of classical bits.

Below, we provide a statement about the query and sample complexities of the Uhlmann transformation based on the state tomography approach, where the approximation error is evaluated by the fidelity difference as in our algorithms.
Here, the results in the purified and mixed sample access models can be obtained by allowing collective measurements over multiple copies, potentially even an exponential number of them.
The derivations are discussed in Appendices~\ref{sec:tomography based purif query},~\ref{sec:tomography based purif sample}, and~\ref{sec:tomography based mixed sample}, corresponding to each of the three computational models.

\begin{proposition}[Uhlmann transformation by a state tomography-based approach]
\label{prop:tomography Uhlmann}
In each of the three query/sample models, for $\delta \in (0, 1)$, there exists a quantum query/sample algorithm based on state tomography that realizes a quantum channel $\cT^\sB$ satisfying
\begin{align}
\rF\big(\cT^\sB(\ketbra{\rho}{\rho}^{\sA\sB}), \ket{\sigma}^{\sA\sB}\big) \geq \rF(\rho^\sA, \sigma^\sA) - \delta,
\end{align}
where, in the mixed sample access model, $\ket{\rho}^{\sA\sB}$ and $\ket{\sigma}^{\sA\sB}$ are taken to be the canonical purified states $\ket{\rho_{\rm c}}^{\sA\sB}$ and $\ket{\sigma_{\rm c}}^{\sA\sB}$, respectively.
The query/sample complexity of the algorithm in each model is given as follows:
\begin{itemize}
\item Purified query access model: 
\begin{align}
    \til{\cO}\Big(d_\sA d_\sB\min\Big\{\f{1}{\delta}, \f{1}{\sqrt{\delta s_{\rm min}}}\Big\}\Big)
\end{align}
queries to $U_\rho^{\sA\sB}$, $U_\sigma^{\sA\sB}$, and their inverses.
\item Purified sample access model: 
\begin{align}
    \til{\cO}\Big(d_\sA d_\sB\min\Big\{\f{1}{\delta^2}, \f{1}{\delta s_{\rm min}}\Big\}\Big)
\end{align}
samples of $\ket{\rho}^{\sA\sB}$ and $\ket{\sigma}^{\sA\sB}$.
\item Mixed sample access model: 
\begin{align}
    \til{\cO}\Big(d_\sA (r_\rho + r_\sigma)\min\Big\{\f{1}{\delta^2}, \f{1}{\delta s_{\rm min}}\Big\}\Big)
\end{align}
samples of $\rho^\sA$ and $\sigma^\sA$.
\end{itemize}
The algorithms generally require exponential time and space in classical computation, as well as an exponential number of one- and two-qubit gates.
\end{proposition}

In each of the three models, the algorithm obtains classical descriptions of approximations of the purified states, and implements a unitary that extends the Uhlmann partial isometry for these approximations.
Here, we say that a unitary $U$ extends a partial isometry $V$ if $UV^\dag V = V$, i.e., $U$ coincides with $V$ on the support of $V$.
The following lemma bounds the resulting loss of the fidelity by the errors of the approximations, and is used in all three models.

\begin{lemma}[Fidelity loss of an approximate Uhlmann unitary]
\label{lem:tomography fidelity loss}
Let $\ket{\rho}^{\sA\sB}$, $\ket{\sigma}^{\sA\sB}$, $\ket{\til{\rho}}^{\sA\sB}$, and $\ket{\til{\sigma}}^{\sA\sB}$ be pure states, and let $\til{U}^\sB$ be a unitary that extends the Uhlmann partial isometry $\sgn^{(\rm SV)}\big(\tr_\sA\big[\ketbra{\til{\sigma}}{\til{\rho}}^{\sA\sB}\big]\big)$ from $\ket{\til{\rho}}^{\sA\sB}$ to $\ket{\til{\sigma}}^{\sA\sB}$.
Then, we have
\begin{align}
\label{inteq:tomography fidelity loss}
    \rF(\rho^\sA, \sigma^\sA) - \rF\big(\til{U}^\sB\ket{\rho}^{\sA\sB}, \ket{\sigma}^{\sA\sB}\big) \leq \min\bigg\{4\eta, \f{12\eta^2}{s_{\rm min}}\bigg\},
\end{align}
where $\eta = \big\|\ket{\til{\rho}}^{\sA\sB} - \ket{\rho}^{\sA\sB}\big\| + \big\|\ket{\til{\sigma}}^{\sA\sB} - \ket{\sigma}^{\sA\sB}\big\|$, and $s_{\rm min}$ is the minimum non-zero singular value of $\sqrt{\sigma^\sA}\sqrt{\rho^\sA}$.
\end{lemma}

The first bound in Eq.~\eqref{inteq:tomography fidelity loss} follows from the optimality of $\til{U}^\sB$ for the approximations, and the second additionally uses that the non-zero singular values of $\tr_\sA\big[\ketbra{\sigma}{\rho}^{\sA\sB}\big]$ are at least $s_{\rm min}$.

\begin{proof}[Proof of \Cref{lem:tomography fidelity loss}]

Let $M^\sB = \tr_\sA\big[\ketbra{\sigma}{\rho}^{\sA\sB}\big]$ and $\til{M}^\sB = \tr_\sA\big[\ketbra{\til{\sigma}}{\til{\rho}}^{\sA\sB}\big]$.
By Uhlmann's theorem in \Cref{thm:Uhlmann theorem}, $\|M\|_1 = \sqrt{\rF}(\rho^\sA, \sigma^\sA)$, and as noted after \Cref{prop:explicit form of Uhl}, the non-zero singular values of $M$ coincide with those of $\sqrt{\sigma^\sA}\sqrt{\rho^\sA}$.
If $M = 0$, the left-hand side of Eq.~\eqref{inteq:tomography fidelity loss} is non-positive, and we thus assume $M \neq 0$.
Since $\bra{\sigma}^{\sA\sB}\til{U}^\sB\ket{\rho}^{\sA\sB} = \tr\big[\til{U}^\dag M\big]^*$ and $\big|\tr\big[\til{U}^\dag M\big]\big| \leq \|M\|_1 \leq 1$ by H\"{o}lder's inequality in Eq.~\eqref{eq:Holder ineq}, we obtain
\begin{align}
    &\rF(\rho^\sA, \sigma^\sA) - \rF\big(\til{U}^\sB\ket{\rho}^{\sA\sB}, \ket{\sigma}^{\sA\sB}\big) \notag\\
    &= \|M\|_1^2 - \big|\tr\big[\til{U}^\dag M\big]\big|^2 \\
    &= \Big(\|M\|_1 + \big|\tr\big[\til{U}^\dag M\big]\big|\Big)\Big(\|M\|_1 - \big|\tr\big[\til{U}^\dag M\big]\big|\Big) \\
    \label{inteq:tomography loss real}
    &\leq 2\ell,
\end{align}
where $\ell = \|M\|_1 - {\rm Re}\big(\tr\big[\til{U}^\dag M\big]\big)$ and ${\rm Re}(x)$ denotes the real part of $x$.

We first bound $\ell$ by the difference between $M$ and $\til{M}$.
Let $U$ be any unitary that extends the Uhlmann partial isometry $\sgn^{(\rm SV)}(M)$, which satisfies $\tr[U^\dag M] = \|M\|_1$.
Since $\til{U}$ extends $\sgn^{(\rm SV)}(\til{M})$, we similarly have $\tr[\til{U}^\dag\til{M}] = \|\til{M}\|_1 \geq {\rm Re}\big(\tr[U^\dag\til{M}]\big)$, where the inequality follows from H\"{o}lder's inequality.
Thus, it holds that
\begin{align}
\label{inteq:tomography loss difference}
    \ell 
    = {\rm Re}\big(\tr\big[(U - \til{U})^\dag M\big]\big)
    \leq {\rm Re}\big(\tr\big[(U - \til{U})^\dag(M - \til{M})\big]\big).
\end{align}
For any vectors $\ket{x}^{\sA\sB}$ and $\ket{y}^{\sA\sB}$, we have $\big\|\tr_\sA\big[\ketbra{x}{y}^{\sA\sB}\big]\big\|_1 \leq \|\ket{x}^{\sA\sB}\|\|\ket{y}^{\sA\sB}\|$, since the trace norm is the maximum of $\big|\tr\big[W^\sB\tr_\sA\big[\ketbra{x}{y}^{\sA\sB}\big]\big]\big| = \big|\bra{y}^{\sA\sB}W^\sB\ket{x}^{\sA\sB}\big|$ over unitaries $W^\sB$.
Applying it to $\ketbra{\sigma}{\rho} - \ketbra{\til{\sigma}}{\til{\rho}} = (\ket{\sigma} - \ket{\til{\sigma}})\bra{\rho} + \ket{\til{\sigma}}(\bra{\rho} - \bra{\til{\rho}})$, we obtain
\begin{align}
\label{inteq:tomography operator error}
    \|M - \til{M}\|_2 \leq \|M - \til{M}\|_1 \leq \eta.
\end{align}
Since $\|U - \til{U}\|_\infty \leq 2$, H\"{o}lder's inequality and Eqs.~\eqref{inteq:tomography loss difference} and~\eqref{inteq:tomography operator error} yield $\ell \leq 2\eta$, and thus Eq.~\eqref{inteq:tomography loss real} gives the first term in Eq.~\eqref{inteq:tomography fidelity loss}.

We next show the second term by choosing $U$ close to $\til{U}$.
Let $M = U_\rL(\Sigma\oplus 0)U_\rR^\dag$ be a singular value decomposition, where $U_\rL$ and $U_\rR$ are unitaries, and $\Sigma$ is the $R \times R$ diagonal matrix of the non-zero singular values $s_1, \ldots, s_R$ of $M$, which satisfy $s_i \geq s_{\rm min}$.
We write
\begin{align}
    U_\rL^\dag\til{U}U_\rR = 
    \begin{pmatrix}
        X_{11} & X_{12} \\
        X_{21} & X_{22}
    \end{pmatrix},
\end{align}
where $X_{11}$ and $X_{22}$ are $R \times R$ and $k \times k$ matrices, respectively, with $k = d_\sB - R$; when $k = 0$, the blocks other than $X_{11}$ are absent in the following.
Let $\bI_R$ be the $R \times R$ identity matrix.
We take a unitary $Q$ that extends $\sgn^{(\rm SV)}(X_{22})$, and set $U = U_\rL(\bI_R\oplus Q)U_\rR^\dag$, which extends $\sgn^{(\rm SV)}(M) = U_\rL(\bI_R\oplus 0)U_\rR^\dag$.
Since $\big|(X_{11})_{ii}\big| \leq 1$ for each diagonal entry, we have
\begin{align}
\label{inteq:tomography loss support}
    \ell = \sum_{i=1}^R s_i\big(1 - {\rm Re}\big((X_{11})_{ii}\big)\big) \geq s_{\rm min}a,
\end{align}
where $    a = R - {\rm Re}\big(\tr[X_{11}]\big)$.
As $U_\rL^\dag\til{U}U_\rR$ is a unitary, it holds that $\|X_{11}\|_2^2 + \|X_{12}\|_2^2 = R$, $\|X_{21}\|_2^2 + \|X_{22}\|_2^2 = k$, and $\|X_{12}\|_2^2 + \|X_{22}\|_2^2 = k$.
Together with $\tr[Q^\dag X_{22}] = \|X_{22}\|_1 \geq \|X_{22}\|_2^2$, which holds since the singular values of $X_{22}$ are at most one, we obtain
\begin{align}
    \|\til{U} - U\|_2^2
    &= \|X_{11} - \bI_R\|_2^2 + \|X_{12}\|_2^2 + \|X_{21}\|_2^2 + \|X_{22} - Q\|_2^2 \\
    &= 2a + 2\big(k - \|X_{22}\|_1\big) \\
    &\leq 2a + 2\big(R - \|X_{11}\|_2^2\big) \\
    \label{inteq:tomography unitary distance}
    &\leq 6a,
\end{align}
where the last inequality follows from $\|X_{11} - \bI_R\|_2^2 = \|X_{11}\|_2^2 - 2{\rm Re}\big(\tr[X_{11}]\big) + R \geq 0$.
From Eqs.~\eqref{inteq:tomography loss difference},~\eqref{inteq:tomography operator error},~\eqref{inteq:tomography loss support}, and~\eqref{inteq:tomography unitary distance}, the Cauchy--Schwarz inequality yields
\begin{align}
    \ell \leq \|\til{U} - U\|_2\|M - \til{M}\|_2 \leq \sqrt{\f{6\ell}{s_{\rm min}}}\,\eta,
\end{align}
and thus $\ell \leq 6\eta^2/s_{\rm min}$.
Equation~\eqref{inteq:tomography loss real} then gives the second term in Eq.~\eqref{inteq:tomography fidelity loss}.

\end{proof}

\subsection{In the purified query access model}
\label{sec:tomography based purif query}

In the purified query access model, the best-known query complexity for state tomography of a pure state in a $d$-dimensional Hilbert space, up to trace distance error $\epsilon$ with a constant success probability, is given by $\til{\cO}(d/\epsilon)$~\cite{Apeldoorn2023TomoPurifQuery}.
Thus, to obtain classical descriptions of states $\til{\rho}^{\sA\sB}$ and $\til{\sigma}^{\sA\sB}$ that satisfy
\begin{align}
\label{eq:tomography descp. error}
    \f{1}{2}\|\til{\rho}^{\sA\sB} - \ketbra{\rho}{\rho}^{\sA\sB}\|_1 \leq \epsilon, 
    \ \ \text{and} \ \ \  
    \f{1}{2}\|\til{\sigma}^{\sA\sB} - \ketbra{\sigma}{\sigma}^{\sA\sB}\|_1 \leq \epsilon,
\end{align}
it suffices to make $\til{\cO}\big(d_\sA d_\sB/\epsilon\big)$ queries to the unitaries $U_\rho^{\sA\sB}$ and $U_\sigma^{\sA\sB}$, which prepare $\ket{\rho}^{\sA\sB}$ and $\ket{\sigma}^{\sA\sB}$, respectively, and their inverses.
We then compute an approximation $\til{V}^\sB$ of the Uhlmann partial isometry $V^\sB$ from the classical description of $\til{\rho}^{\sA\sB}$ and $\til{\sigma}^{\sA\sB}$.

While there is some choice in how to construct $\til{V}^\sB$ from $\til{\rho}^{\sA\sB}$ and $\til{\sigma}^{\sA\sB}$, we specifically consider 
\begin{align}
    \til{V}^\sB = \sgn^{(\rm SV)}\big(\tr_\sA[\ketbra{\til{\sigma}_1}{\til{\rho}_1}^{\sA\sB}]\big), 
\end{align}
where $\ket{\til{\rho}_1}^{\sA\sB}$ and $\ket{\til{\sigma}_1}^{\sA\sB}$ are eigenstates corresponding to the largest eigenvalues of $\til{\rho}^{\sA\sB}$ and $\til{\sigma}^{\sA\sB}$, respectively.
We implement a unitary $\til{U}^\sB$ that extends $\til{V}^\sB$, which is an Uhlmann unitary for $\ket{\til{\rho}_1}^{\sA\sB}$ and $\ket{\til{\sigma}_1}^{\sA\sB}$ by \Cref{prop:explicit form of Uhl}.
As we will see below, this unitary achieves a fidelity close to $\rF(\rho^\sA, \sigma^\sA)$.

To evaluate the fidelity achieved by $\til{U}^\sB$, we first evaluate the distance between $\ket{\til{\rho}_1}^{\sA\sB}$ and $\ket{\rho}^{\sA\sB}$. Let the eigenvalue decomposition of $\til{\rho}^{\sA\sB}$ be given by $\til{\rho}^{\sA\sB} = \sum_{j=1}^{r_{\til{\rho}}} \til{\rho}_j \ketbra{\til{\rho}_j}{\til{\rho}_j}^{\sA\sB}$, where $\til{\rho}_1 \geq \til{\rho}_2 \geq \ldots \geq \til{\rho}_{r_{\til{\rho}}} > 0$.
Since $\til{\rho}_1$ is the largest eigenvalue, we have $\til{\rho}_1 \geq \bra{\rho}^{\sA\sB}\til{\rho}^{\sA\sB}\ket{\rho}^{\sA\sB} \geq 1 - \epsilon$, where the second inequality follows from Eq.~\eqref{eq:tomography descp. error} and the variational characterization of the trace norm in Eq.~\eqref{eq:trace distance variational} with $M = \ketbra{\rho}{\rho}^{\sA\sB}$.
Since $\sum_{j=1}^{r_{\til{\rho}}}\til{\rho}_j = 1$, we see that
\begin{align}
\label{inteq:85}
    (1 - \til{\rho}_1) + \sum_{j=2}^{r_{\til{\rho}}} \til{\rho}_j = 2(1 - \til{\rho}_1) \leq 2\epsilon.
\end{align}
Thus, we have 
\begin{align}
    &\big\|\ketbra{\til{\rho}_1}{\til{\rho}_1}^{\sA\sB} - \ketbra{\rho}{\rho}^{\sA\sB}\big\|_1  \notag\\
    &\leq \big\|\ketbra{\til{\rho}_1}{\til{\rho}_1}^{\sA\sB} - \til{\rho}^{\sA\sB}\big\|_1 + \big\|\til{\rho}^{\sA\sB} - \ketbra{\rho}{\rho}^{\sA\sB}\big\|_1 \\
    &= (1 - \til{\rho}_1) + \sum_{j=2}^{r_{\til{\rho}}} \til{\rho}_j + \big\|\til{\rho}^{\sA\sB} - \ketbra{\rho}{\rho}^{\sA\sB}\big\|_1 \\
    \label{inteq:86}
    &\leq 4 \epsilon,
\end{align}
where we used Eqs.~\eqref{eq:tomography descp. error} and~\eqref{inteq:85}.
Hence, for the state $\til{\rho}^{\sA\sB}$ obtained by state tomography, its eigenvector $\ket{\til{\rho}_1}^{\sA\sB}$ corresponding to the largest eigenvalue sufficiently approximates $\ket{\rho}^{\sA\sB}$.
The same holds for the states $\ket{\til{\sigma}_1}^{\sA\sB}$ and $\ket{\sigma}^{\sA\sB}$:
\begin{align}
    \label{inteq:107}
    \big\|\ketbra{\til{\sigma}_1}{\til{\sigma}_1}^{\sA\sB} - \ketbra{\sigma}{\sigma}^{\sA\sB}\big\|_1
    &\leq 4 \epsilon.
\end{align}
Since the absolute values of the overlaps evaluated below do not depend on the global phases of $\ket{\til{\rho}_1}^{\sA\sB}$ and $\ket{\til{\sigma}_1}^{\sA\sB}$, we choose the phases that achieve the minimum in the relation between the trace norm and the Euclidean norm in Eq.~\eqref{eq:relation of Euclidean and trace norm}.
From Eqs.~\eqref{inteq:86} and~\eqref{inteq:107}, we then have
\begin{align}
\label{inteq:tomography euclid}
    \big\|\ket{\til{\rho}_1}^{\sA\sB} - \ket{\rho}^{\sA\sB}\big\| \leq 2\sqrt{2}\epsilon,
    \ \ \text{and} \ \ \  
    \big\|\ket{\til{\sigma}_1}^{\sA\sB} - \ket{\sigma}^{\sA\sB}\big\| \leq 2\sqrt{2}\epsilon.
\end{align}

We next evaluate the fidelity.
Since $\til{U}^\sB$ extends the Uhlmann partial isometry $\til{V}^\sB$ from $\ket{\til{\rho}_1}^{\sA\sB}$ to $\ket{\til{\sigma}_1}^{\sA\sB}$, \Cref{lem:tomography fidelity loss} with $\eta \leq 4\sqrt{2}\epsilon$ from Eq.~\eqref{inteq:tomography euclid} yields
\begin{align}
\label{inteq:tomography fidelity}
    \rF(\rho^\sA, \sigma^\sA) - \rF\big(\til{U}^\sB\ket{\rho}^{\sA\sB}, \ket{\sigma}^{\sA\sB}\big) \leq \min\bigg\{16\sqrt{2}\epsilon, \f{384\epsilon^2}{s_{\rm min}}\bigg\}.
\end{align}
The second term is smaller than the first when $\epsilon < s_{\rm min}/(12\sqrt{2})$.

The algorithm implements $\til{U}^\sB$ according to the outcomes of state tomography, and thus realizes the quantum channel $\cT^\sB$ given by the average of $\til{U}^\sB(\cdot)(\til{U}^\sB)^\dag$ over the outcomes.
Since the fidelity with the pure state $\ket{\sigma}^{\sA\sB}$ is linear in the other state, $\rF\big(\cT^\sB(\ketbra{\rho}{\rho}^{\sA\sB}), \ket{\sigma}^{\sA\sB}\big)$ is the average of $\rF\big(\til{U}^\sB\ket{\rho}^{\sA\sB}, \ket{\sigma}^{\sA\sB}\big)$ over the outcomes.
We let each state tomography fail to satisfy Eq.~\eqref{eq:tomography descp. error} with probability at most $\delta/4$, which increases the number of queries only by a logarithmic factor. Then, with probability at least $1 - \delta/2$, both of them succeed and Eq.~\eqref{inteq:tomography fidelity} holds, while the fidelity is non-negative otherwise.
We choose
\begin{align}
\label{inteq:tomography epsilon}
    \epsilon = \max\bigg\{\f{\delta}{32\sqrt{2}}, \sqrt{\f{\delta s_{\rm min}}{768}}\bigg\},
\end{align}
for which the right-hand side of Eq.~\eqref{inteq:tomography fidelity} is at most $\delta/2$, and $s_{\rm min}$ can be replaced by any positive lower bound on it.
Equation~\eqref{inteq:tomography fidelity} then implies
\begin{align}
    \rF\big(\cT^\sB(\ketbra{\rho}{\rho}^{\sA\sB}), \ket{\sigma}^{\sA\sB}\big)
    &\geq \Big(1 - \f{\delta}{2}\Big)\Big(\rF(\rho^\sA, \sigma^\sA) - \f{\delta}{2}\Big) \\
    &\geq \rF(\rho^\sA, \sigma^\sA) - \delta.
\end{align}

Consequently, we conclude that the required query complexity is given by
\begin{align}
    \til{\cO}\Big(\f{d_\sA d_\sB}{\epsilon}\Big) = \til{\cO}\Big(d_\sA d_\sB\min\Big\{\f{1}{\delta}, \f{1}{\sqrt{\delta s_{\rm min}}}\Big\}\Big).
\end{align}
This completes the derivation of the query complexity for the Uhlmann transformation using the naive state tomography-based approach.


\subsection{In the purified sample access model}
\label{sec:tomography based purif sample}

The quantum state tomography-based approach in the purified sample access model is mostly the same as the approach in the purified query access model described in \Cref{sec:tomography based purif query}.
The difference from the discussion in \Cref{sec:tomography based purif query} lies only in the number of samples required in state tomography.
The optimal sample complexity for a pure state in a $d$-dimensional Hilbert space, up to trace distance error $\epsilon$, is given by $\til{\cO}(d/\epsilon^2)$~\cite{ODonnell2016EffTomo, Haah2017samploptTomo}, where collective measurements over multiple copies, potentially even exponentially many, are allowed.
Thus, to obtain classical descriptions of states $\til{\rho}^{\sA\sB}$ and $\til{\sigma}^{\sA\sB}$ that satisfy 
\begin{align}
\label{eq:tomography descp. error 2}
    \f{1}{2}\|\til{\rho}^{\sA\sB} - \ketbra{\rho}{\rho}^{\sA\sB}\|_1 \leq \epsilon, 
    \ \ \text{and} \ \ \ 
    \f{1}{2}\|\til{\sigma}^{\sA\sB} - \ketbra{\sigma}{\sigma}^{\sA\sB}\|_1 \leq \epsilon,
\end{align}
we use $\til{\cO}\big(d_\sA d_\sB/\epsilon^2\big)$ samples of $\ket{\rho}^{\sA\sB}$ and~$\ket{\sigma}^{\sA\sB}$.

Following the same strategy in \Cref{sec:tomography based purif query} with $\epsilon$ in Eq.~\eqref{inteq:tomography epsilon}, it follows that for implementing the Uhlmann transformation within the error $\delta$, the total number of samples of $\ket{\rho}^{\sA\sB}$ and $\ket{\sigma}^{\sA\sB}$ is
\begin{align}
    \til{\cO}\Big(\f{d_\sA d_\sB}{\epsilon^2}\Big) 
    = \til{\cO}\Big(d_\sA d_\sB\min\Big\{\f{1}{\delta^2}, \f{1}{\delta s_{\rm min}}\Big\}\Big).
\end{align}


\subsection{In the mixed sample access model}
\label{sec:tomography based mixed sample}

In the mixed sample access model, from multiple copies of $\rho^\sA$ and $\sigma^\sA$, we obtain approximations $\til{\rho}^\sA$ and $\til{\sigma}^\sA$ of the original states, respectively.
The optimal sample complexity of state tomography for a rank-$r$ quantum state $\omega$ in a $d$-dimensional Hilbert space, up to an infidelity of $\epsilon$, i.e., $1 - \rF(\til{\omega}, \omega) \leq \epsilon$ for its estimate $\til{\omega}$, is given by $\til{\cO}(rd/\epsilon)$~\cite{Haah2017samploptTomo}, where one is allowed to perform collective measurements on multiple copies, possibly even an exponential number of them.
This holds regardless of whether the infidelity is defined by the fidelity or by the square root fidelity, since $1 - \rF(\til{\omega}, \omega) \leq 2\big(1 - \sqrt{\rF}(\til{\omega}, \omega)\big)$.

Unlike the other two models, the mixed sample access model requires more careful analysis.
Specifically, even though $\omega^\sA$ can be estimated with a small error, it is necessary to evaluate whether the purified state can also be estimated with a small error.
In the evaluation, we use the following inequality: for any states $\til{\omega}^\sA$ and $\omega^\sA$,
\begin{align}
\label{inteq:26}
    \big\|\ket{\til{\omega}_{\rm c}}^{\sA\sB} - \ket{\omega_{\rm c}}^{\sA\sB}\big\|^2
    \leq 2\big(1 - \rF(\til{\omega}^\sA, \omega^\sA)\big),
\end{align}
where $\ket{\omega_{\rm c}}^{\sA\sB}$ and $\ket{\til{\omega}_{\rm c}}^{\sA\sB}$ are the canonical purified states of $\omega^\sA$ and $\til{\omega}^\sA$, respectively, and $d_\sA = d_\sB$.
We derive this inequality at the end of this section.

The sample complexities to obtain $\til{\rho}^\sA$ and $\til{\sigma}^\sA$ which satisfy
\begin{align}
    1 - \rF(\til{\rho}^\sA, \rho^\sA) \leq \epsilon',
    \ \  \text{and} \ \ \  
    1 - \rF(\til{\sigma}^\sA, \sigma^\sA) \leq \epsilon',
\end{align}
are given by $\til{\cO}\big(d_\sA r_\rho/\epsilon'\big)$ and $\til{\cO}\big(d_\sA r_\sigma /\epsilon'\big)$, respectively, where $r_\rho$ and $r_\sigma$ are the ranks of $\rho^\sA$ and $\sigma^\sA$.

Once classical descriptions of $\til{\rho}^\sA$ and $\til{\sigma}^\sA$ are obtained, one can simply diagonalize them via classical computation, which immediately yields classical descriptions of $\ket{\til{\rho}_{\rm c}}^{\sA\sB}$ and $\ket{\til{\sigma}_{\rm c}}^{\sA\sB}$.
By Eq.~\eqref{inteq:26}, to obtain classical descriptions of $\ket{\til{\rho}_{\rm c}}^{\sA\sB}$ and $\ket{\til{\sigma}_{\rm c}}^{\sA\sB}$ such that 
\begin{align}
\label{inteq:tomography canonical}
    \big\|\ket{\til{\rho}_{\rm c}}^{\sA\sB} - \ket{\rho_{\rm c}}^{\sA\sB}\big\| \leq \epsilon, 
    \ \ \text{and} \ \ \  
    \big\|\ket{\til{\sigma}_{\rm c}}^{\sA\sB} - \ket{\sigma_{\rm c}}^{\sA\sB}\big\| \leq \epsilon,
\end{align}
the number of samples of $\rho^\sA$ and $\sigma^\sA$ is $\til{\cO}(d_\sA (r_\rho + r_\sigma) / \epsilon^2)$, where $\epsilon'$ is chosen as $\epsilon' = \epsilon^2/2$.

The procedure hereafter follows similarly to that in \Cref{sec:tomography based purif query}.
We classically compute $\til{V}^{\sB} = \sgn^{(\rm SV)}\big(\tr_\sA\big[\ketbra{\til{\sigma}_{\rm c}}{\til{\rho}_{\rm c}}^{\sA\sB}\big]\big)$ from the descriptions of $\ket{\til{\rho}_{\rm c}}^{\sA\sB}$ and $\ket{\til{\sigma}_{\rm c}}^{\sA\sB}$, and implement a unitary $\til{U}^\sB$ that extends $\til{V}^{\sB}$.
\Cref{lem:tomography fidelity loss} with $\eta \leq 2\epsilon$ from Eq.~\eqref{inteq:tomography canonical} yields
\begin{align}
    \rF(\rho^\sA, \sigma^\sA) - \rF\big(\til{U}^\sB\ket{\rho_{\rm c}}^{\sA\sB}, \ket{\sigma_{\rm c}}^{\sA\sB}\big) \leq \min\Big\{8\epsilon, \f{48\epsilon^2}{s_{\rm min}}\Big\}.
\end{align}
Choosing $\epsilon = \max\big\{\delta/16, \sqrt{\delta s_{\rm min}/96}\big\}$ and the failure probability of state tomography as in \Cref{sec:tomography based purif query}, the total number of samples of $\rho^\sA$ and $\sigma^\sA$ for the Uhlmann transformation within the error $\delta$ using the state tomography-based approach is given by 
\begin{align}
    \til{\cO}\bigg(\f{d_\sA (r_\rho + r_\sigma)}{\epsilon^2}\bigg) 
    \label{eq:tomobase sample mixed}
    = \til{\cO}\Big(d_\sA (r_\rho + r_\sigma)\min\Big\{\f{1}{\delta^2}, \f{1}{\delta s_{\rm min}}\Big\}\Big).
\end{align}
For fixed $d_\sA$, $r_\rho$, and $r_\sigma$, this is $\til{\cO}\big(1/(\delta s_{\rm min})\big)$ when $\delta \leq s_{\rm min}$, which shows that the term $\Omega\big(1/(\delta s)\big)$ in the third statement of \Cref{thm:lower Uhlmann} is tight with respect to both $\delta$ and $s$ up to logarithmic factors.

Finally, we derive the inequality in Eq.~\eqref{inteq:26} via a technique of vectorization, which is represented by a linear map $\mathrm{Vec}$ such that, for a given orthonormal basis $\{\ket{i}\}_i$, $\mathrm{Vec}(\ket{\psi}\bra{\varphi}) = \ket{\psi}\ket{\varphi^*}$, where the complex conjugate is taken in the basis $\{\ket{i}\}_i$, i.e., $\ket{\varphi^*} = \sum_ic_i^*\ket{i}$ when $\ket{\varphi} = \sum_i c_i\ket{i}$.
It is straightforward to verify that the vectorization has the property that, for any matrix $M$, $\|M\|_2 = \|\mathrm{Vec}(M)\|$.

We utilize this property in our evaluation. This implies that
\begin{align}
    \big\|\ket{\til{\omega}_{\rm c}}^{\sA\sB} - \ket{\omega_{\rm c}}^{\sA\sB}\big\|^2
    &= \big\|\mathrm{Vec}\big(\sqrt{\til{\omega}^\sA}\big) - \mathrm{Vec}\big(\sqrt{\omega^\sA}\big)\big\|^2 \\
    &= \big\|\sqrt{\til{\omega}^\sA} - \sqrt{\omega^\sA}\big\|_2^2 \\
    \label{inteq:27}
    &= 2 - 2\tr\big[\sqrt{\til{\omega}^\sA}\sqrt{\omega^\sA}\big],
\end{align}
where we used that $\ket{\omega_{\rm c}}^{\sA\sB} = \mathrm{Vec}\big(\sqrt{\omega^\sA}\big)$ in the first equation, and used $\tr[\til{\omega}^\sA] = \tr[\omega^\sA] = 1$ in the last equation.
Moreover, omitting the system labels and applying H\"{o}lder's inequality in Eq.~\eqref{eq:Holder ineq} twice, we obtain
\begin{align}
    \sqrt{\rF}(\til{\omega}, \omega)
    &= \big\|\sqrt{\til{\omega}}\sqrt{\omega}\big\|_1 \\
    &= \big\|\til{\omega}^{1/4}\big(\til{\omega}^{1/4}\omega^{1/4}\big)\omega^{1/4}\big\|_1 \\
    &\leq \big\|\til{\omega}^{1/4}\big\|_4\big\|\til{\omega}^{1/4}\omega^{1/4}\big\|_2\big\|\omega^{1/4}\big\|_4 \\
    \label{inteq:28}
    &= \Big(\tr\big[\sqrt{\til{\omega}}\sqrt{\omega}\big]\Big)^{1/2},
\end{align}
where we used $\|\til{\omega}^{1/4}\|_4 = \|\omega^{1/4}\|_4 = 1$ and $\|\til{\omega}^{1/4}\omega^{1/4}\|_2^2 = \tr\big[\sqrt{\til{\omega}}\sqrt{\omega}\big]$ in the last equation.
From Eqs.~\eqref{inteq:27} and~\eqref{inteq:28}, we complete the derivation of Eq.~\eqref{inteq:26}.


\section{Overview of previous approach to the Uhlmann transformation}
\label{sec:metger and yuen's algorithm}

An explicit algorithm for the Uhlmann transformation was proposed by T. Metger and H. Yuen in Ref.~\cite{metger2023stateqipspace}.
This algorithm demonstrates that the Uhlmann transformation can be implemented with polynomial space complexity.
We here provide a brief and high-level overview of the algorithm tailored to the purified query access model; for full details, see Ref.~\cite{metger2023stateqipspace}.

\begin{theorem}[Uhlmann transformation algorithm in the purified query access model based on the approach in~{\cite[Theorem 7.4 in its full version]{metger2023stateqipspace}}]
\label{prevthm:Metger&Yuen's result}
For $\delta \in (0, 1)$, there exists a quantum query algorithm that realizes a quantum channel $\cT^\sB$ satisfying 
\begin{align}
    \rF\big(\cT^\sB(\ketbra{\rho}{\rho}^{\sA\sB}), \ket{\sigma}^{\sA\sB}\big) \geq \rF(\rho^\sA, \sigma^\sA) - \delta,
\end{align}
using $\gamma$ queries to $U_\rho^{\sA\sB}$ and $U_\sigma^{\sA\sB}$, and their inverses, where $\gamma = \til{\cO}\Big(\f{d_\sA^{13} d_\sB^{9}}{\delta^2}\min\Big\{\f{1}{s_{\rm min}^3}, \f{R^3}{\delta^3}\Big\}\Big)$. The quantum circuit of this algorithm consists of $\til{\cO}(\gamma)$ one- and two-qubit gates, and uses $\cO\big(\log{(d_\sA d_\sB)}\big)$ qubits at once.
The circuit description can be computed classically in $\poly\big(\gamma, \log{(d_\sA d_\sB)}, \log{(1/\delta)}\big)$ time.
\end{theorem}

In Ref.~\cite{metger2023stateqipspace}, the authors carefully analyzed the space complexity in the algorithm and showed that, at every stage of this algorithm, even in classical computation, only polynomial space is used; it is designed so that an exponential amount of classical data is not stored.
Meanwhile, as seen from the subsequent derivation, this approach involves operations such as incrementally summing measurement outcomes, which generally leads to exponential time in the number of bits in classical computation.

This algorithm uses the QSVT with the sign function; however, the polynomial approximating the sign function has exponential degree, making it not obvious whether the polynomial can be computed in polynomial classical space, as mentioned in Sec.~\ref{sec:BE and QSVT intro}.
By extending the approach in Ref.~\cite{Low2017HamiltonianSB} using linear combinations of Chebyshev polynomials \cite{chebyshev1853teory, chebyshev1857surles}, Ref.~\cite{metger2023stateqipspace} successfully demonstrates that the algorithm is implementable in polynomial space at every step. In the following analysis, however, we do not consider their sophisticated polynomial construction; instead, we use the conventional polynomial from Refs.~\cite{Low2017HamiltonianSB, Gilyn2019thesis} for simplicity. Full polynomial-space computability can be ensured by following the polynomial construction and its implementation in Ref.~\cite{metger2023stateqipspace}, whose query complexity we do not analyze here.

In the rest of this appendix, we derive \Cref{prevthm:Metger&Yuen's result}.
A main idea behind this algorithm is sketched as follows.
In the first step, we prepare the state $\ket{\rho}^{\sA\sB} = U_\rho^{\sA\sB}\ket{0}^{\sA\sB}$.
Then, for each $i$ and $j$, we apply Pauli measurements to estimate $\tr[\Lambda^{\sA\sB}\rho^{\sA\sB}]$ for the Pauli operators $\Lambda^{\sA\sB} \in \{\bI, X, Y, Z\}^{\otimes \log{d_{\sA\sB}}}$ satisfying $\bra{i}\Lambda\ket{j}^{\sA\sB} \neq 0$, where $\rho^{\sA\sB} = \ketbra{\rho}{\rho}^{\sA\sB}$.
From the estimated value $\til{\alpha}_\Lambda$, we classically compute an element $\til{c}_{ij} = \f{1}{d_{\sA\sB}}\sum_\Lambda \til{\alpha}_\Lambda \bra{i}\Lambda\ket{j}^{\sA\sB}$, which approximates $c_{ij} = \bra{i}\rho\ket{j}^{\sA\sB}$, and subsequently prepare a one-qubit gate
\begin{align}
\til{U}_{ij}^{\sC} = 
\begin{pmatrix} 
\til{c}_{ij}/2 & -\sqrt{1 - |\til{c}_{ij}|^2/4} \\
\sqrt{1 - |\til{c}_{ij}|^2/4} & \til{c}_{ij}^*/2
\end{pmatrix}.
\end{align}
Since each $\til{\alpha}_\Lambda$ is an average of outcomes in $\{-1, 1\}$, we have $|\til{c}_{ij}| \leq 1$ for all measurement outcomes, and thus $\til{U}_{ij}^\sC$ is always unitary.
The factor $1/2$ makes $\til{U}_{ij}^\sC$ robust against the estimation error in $\til{c}_{ij}$, as we will see below.
Then, we construct a unitary $\til{W}_\rho^{\sA\sB\hat{\sA}\hat{\sB}\sC}$ such that
\begin{align}
\label{inteq:MY W}
    \til{W}_\rho^{\sA\sB\hat{\sA}\hat{\sB}\sC} = H^{\hat{\sA}\hat{\sB}} F^{\sA\sB\hat{\sA}\hat{\sB}} \Big( \sum_{i, j} \ketbra{j}{j}^{\sA\sB} \otimes \ketbra{i}{i}^{\hat{\sA}\hat{\sB}} \otimes \til{U}_{ij}^\sC \Big) H^{\hat{\sA}\hat{\sB}},
\end{align}
where $F$ is the swap operator, and $H$ is the Hadamard operator: 
\begin{align}
    H^\sA\ket{i}^\sA = \f{1}{\sqrt{d_\sA}}\sum_j (-1)^{i\cdot j}\ket{j}^\sA.
\end{align}
We can see that this unitary satisfies $\bra{0}^{\hat{\sA}\hat{\sB}\sC}\til{W}_\rho^{\sA\sB\hat{\sA}\hat{\sB}\sC}\ket{0}^{\hat{\sA}\hat{\sB}\sC} = \til{\rho}^{\sA\sB}/(2d_{\sA\sB})$, where $\til{\rho}^{\sA\sB} = \sum_{i,j}\til{c}_{ij}\ketbra{i}{j}^{\sA\sB}$.
We denote by $U_{ij}^\sC$ and $W_\rho^{\sA\sB\hat{\sA}\hat{\sB}\sC}$ the unitaries given in the same way with $\til{c}_{ij}$ replaced by $c_{ij}$, and then $W_\rho^{\sA\sB\hat{\sA}\hat{\sB}\sC}$ is a $(2d_{\sA\sB}, 1 + \log{d_{\sA\sB}}, 0)$-block-encoding unitary of $\rho^{\sA\sB}$.

In Ref.~\cite{metger2023stateqipspace}, the corresponding values are computed deterministically in polynomial space from the descriptions of the circuits preparing the states~{\cite[Lemma 6.2]{metger2023stateqipspace}}.
In the purified query access model, we instead estimate them by the measurements in the quantum algorithm given in the proof of Ref.~{\cite[Lemma 6.2]{metger2023stateqipspace}}, and the estimates differ each time $\til{W}_\rho^{\sA\sB\hat{\sA}\hat{\sB}\sC}$ is constructed.
We take this into account by evaluating the distance between $\til{W}_\rho^{\sA\sB\hat{\sA}\hat{\sB}\sC}$ and $W_\rho^{\sA\sB\hat{\sA}\hat{\sB}\sC}$ in the operator norm.
Let $\sD = \hat{\sA}\hat{\sB}\sC$.
We analyze the number of queries required for $\til{W}_\rho^{\sA\sB\sD}$ to be close to $W_\rho^{\sA\sB\sD}$.
For each $i$ and $j$, only the $d_{\sA\sB}$ Pauli operators $\Lambda^{\sA\sB}$ satisfying $\bra{i}\Lambda\ket{j}^{\sA\sB} \neq 0$ contribute to $\til{c}_{ij}$, and we estimate $\tr[\Lambda^{\sA\sB}\rho^{\sA\sB}]$ only for them. To obtain $\til{\alpha}_\Lambda$ such that $\big|\til{\alpha}_\Lambda - \tr[\Lambda^{\sA\sB}\rho^{\sA\sB}]\big| \leq \delta_1/d_{\sA\sB}$ for all these $d_{\sA\sB}$ Pauli operators with probability at least $1 - \eta$, it suffices to have $d_{\sA\sB}\cO\big(\f{d_{\sA\sB}^2}{\delta_1^2} \log(d_{\sA\sB}/\eta)\big) = \til{\cO}(d_{\sA\sB}^3/\delta_1^2)$ copies of $\ket{\rho}^{\sA\sB} = U_\rho^{\sA\sB}\ket{0}^{\sA\sB}$~\cite{aaronson2007learnability, yuen2022lecturetomography}. This is the result from the Chernoff bound for i.i.d. sampling~\cite{chernoff1952measure, hoeffding1963Probability, motwani1995randomized, mitzenmacher2005probability}.
By sequentially summing over each $\Lambda^{\sA\sB}$, the value $\til{c}_{ij}$ can be computed using only polynomial space.
Since this procedure is performed for all $i$ and $j$, the total number of queries is given by $d_{\sA\sB}^2 \til{\cO}(d_{\sA\sB}^3/\delta_1^2) = \til{\cO}(d_{\sA\sB}^5/\delta_1^2)$.

When all the estimates succeed, we have
\begin{align}
    \big|\til{c}_{ij} - c_{ij}\big|
    &\leq \f{1}{d_{\sA\sB}}\sum_\Lambda \big|\til{\alpha}_\Lambda - \tr[\Lambda^{\sA\sB}\rho^{\sA\sB}]\big| \big|\bra{i}\Lambda\ket{j}^{\sA\sB}\big| \\
    &\leq \f{\delta_1}{d_{\sA\sB}},
\end{align}
for all $i$ and $j$, where we used that, for each $i$ and $j$, $\big|\bra{i}\Lambda\ket{j}^{\sA\sB}\big| = 1$ holds for exactly $d_{\sA\sB}$ Pauli operators $\Lambda^{\sA\sB}$ and $\bra{i}\Lambda\ket{j}^{\sA\sB} = 0$ holds for the others.
Since $|c_{ij}| \leq 1$ and $|\til{c}_{ij}| \leq 1$, we have
\begin{align}
    \big\|\til{U}_{ij}^\sC - U_{ij}^\sC\big\|_\infty
    &\leq \big\|\til{U}_{ij}^\sC - U_{ij}^\sC\big\|_2 \\
    &\leq \big|\til{c}_{ij} - c_{ij}\big|,
\end{align}
where in the last inequality we used $\big|\sqrt{1 - x^2} - \sqrt{1 - y^2}\big| \leq |x - y|$ for $x, y \in [0, 1/2]$.
Since the operator in the parentheses in Eq.~\eqref{inteq:MY W} is block diagonal with respect to $i$ and $j$, we obtain
\begin{align}
\label{inteq:MY W error}
    \big\|\til{W}_\rho^{\sA\sB\sD} - W_\rho^{\sA\sB\sD}\big\|_\infty
    = \max_{i, j}\big\|\til{U}_{ij}^\sC - U_{ij}^\sC\big\|_\infty
    \leq \f{\delta_1}{d_{\sA\sB}}.
\end{align}
Similarly, we perform this procedure for $\ket{\sigma}^{\sA\sB} = U_\sigma^{\sA\sB}\ket{0}^{\sA\sB}$ and obtain a unitary $\til{W}_\sigma^{\sA\sB\sD}$ from the same number of queries, which satisfies $\big\|\til{W}_\sigma^{\sA\sB\sD} - W_\sigma^{\sA\sB\sD}\big\|_\infty \leq \delta_1/d_{\sA\sB}$ for a $(2d_{\sA\sB}, 1 + \log{d_{\sA\sB}}, 0)$-block-encoding unitary $W_\sigma^{\sA\sB\sD}$ of $\ketbra{\sigma}{\sigma}^{\sA\sB}$.

In the next step, we construct a unitary 
\begin{align}
    \til{U}^{\sA'\sB\sD_1\sD_2} = \til{W}_\sigma^{\sA'\sB\sD_1}U_\sigma^{\sA'\sB} (U_\rho^{\sA'\sB})^\dag(\til{W}_\rho^{\sA'\sB\sD_2})^\dag.
\end{align}
Let $U^{\sA'\sB\sD_1\sD_2}$ be the unitary given in the same way with $\til{W}_\rho^{\sA'\sB\sD_2}$ and $\til{W}_\sigma^{\sA'\sB\sD_1}$ replaced by $W_\rho^{\sA'\sB\sD_2}$ and $W_\sigma^{\sA'\sB\sD_1}$, respectively.
Then, $U^{\sA'\sB\sD_1\sD_2}$ is a $(4d_{\sA\sB}^2, 2(1+\log{d_{\sA\sB}}), 0)$-block-encoding unitary of $\ketbra{\sigma}{\rho}^{\sA'\sB}$.
This follows from \Cref{lem:product of BE}, since $W_\sigma^{\sA'\sB\sD_1}U_\sigma^{\sA'\sB}$ and $(U_\rho^{\sA'\sB})^\dag(W_\rho^{\sA'\sB\sD_2})^\dag$ are $(2d_{\sA\sB}, 1 + \log{d_{\sA\sB}}, 0)$-block-encoding unitaries of $\ketbra{\sigma}{0}^{\sA'\sB}$ and $\ketbra{0}{\rho}^{\sA'\sB}$, respectively.
Moreover, Eq.~\eqref{inteq:MY W error} and the same bound for $\til{W}_\sigma^{\sA'\sB\sD_1}$ imply $\big\|\til{U}^{\sA'\sB\sD_1\sD_2} - U^{\sA'\sB\sD_1\sD_2}\big\|_\infty \leq 2\delta_1/d_{\sA\sB}$.

We need to trace over system $\sA'$ from $\ketbra{\sigma}{\rho}^{\sA'\sB}$.
This is achieved by a technique of the linear combination of block-encoded matrices~\cite{gilyen2019qsvt, Gilyn2019thesis} as follows: let $X_i$ be a tensor product of the Pauli-$X$ operators such that $X_i\ket{0} = \ket{i}$.
Then, we construct a unitary
\begin{align}
\label{inteq:153}
    \til{G}^{\sA'\sB\sA''\sE} =  H^{\sA''}\Big(\sum_{i=0}^{d_\sA-1} \ketbra{i}{i}^{\sA''} \otimes X_i^{\sA'} \til{U}_i^{\sA'\sB\sE} X_i^{\sA'}\Big)H^{\sA''}, 
\end{align}
where $\sE = \sD_1\sD_2$, and $\til{U}_i^{\sA'\sB\sE}$ for $i = 0, 1, \ldots, d_\sA-1$ are unitaries constructed in the same way as $\til{U}^{\sA'\sB\sE}$, each from its own estimates.
Let $G^{\sA'\sB\sA''\sE}$ be the unitary given in the same way with $\til{U}_i^{\sA'\sB\sE}$ replaced by $U^{\sA'\sB\sE}$ for all $i$.
Then, $G^{\sA'\sB\sA''\sE}$ is a $(4d_\sA^3d_\sB^2, \cO(\log{d_{\sA\sB}}), 0)$-block-encoding unitary of $\tr_{\sA'}[\ketbra{\sigma}{\rho}^{\sA'\sB}]$, because we have
\begin{align}
    &4d_\sA^3 d_\sB^2 \bra{0}^{\sA'\sA''\sE} G^{\sA'\sB\sA''\sE} \ket{0}^{\sA'\sA''\sE} \notag\\
    &= 4d_{\sA\sB}^2\sum_i\bra{i}^{\sA'}\bra{0}^\sE U^{\sA'\sB\sE}\ket{0}^\sE\ket{i}^{\sA'} \\
    &= \tr_{\sA'}\big[\ketbra{\sigma}{\rho}^{\sA'\sB}\big].
\end{align}
Since the operator in the parentheses in Eq.~\eqref{inteq:153} is block diagonal with respect to $i$, we also have
\begin{align}
\label{inteq:MY G error}
    \big\|\til{G}^{\sA'\sB\sA''\sE} - G^{\sA'\sB\sA''\sE}\big\|_\infty
    = \max_i\big\|\til{U}_i^{\sA'\sB\sE} - U^{\sA'\sB\sE}\big\|_\infty
    \leq \f{2\delta_1}{d_{\sA\sB}}.
\end{align}
Since $\til{G}^{\sA'\sB\sA''\sE}$ in Eq.~\eqref{inteq:153} includes $d_\sA$ unitaries $\til{U}_i^{\sA'\sB\sE}$ for $i = 0, 1, \ldots, d_\sA-1$, the number of queries up to this point is $h = d_\sA \til{\cO}(d_{\sA\sB}^5/\delta_1^2) = \til{\cO}(d_\sA^6d_\sB^5/\delta_1^2)$.

Let $\sF = \sA'\sA''\sE$ and $M^\sB = \tr_{\sA'}\big[\ketbra{\sigma}{\rho}^{\sA'\sB}\big]$.
Now, we see that the unitary $G^{\sB\sF}$ takes the form:
\begin{align}
        G^{\sB\sF}
        = \hspace{1mm}
\begin{blockarray}{ccc}
 & \bra{0}^\sF &  & \vspace{1mm}\\
\begin{block}{c(cc)}
  \ket{0}^\sF \hspace*{1mm} & \f{1}{4d_\sA^3 d_\sB^2}M^\sB & \hspace{0mm} * \hspace{3mm}\\
   \hspace*{1mm} & * & \hspace{0mm} * \hspace{3mm}\\
\end{block}
\end{blockarray}\hspace{1mm}.
\end{align}
By applying the QSVT with the sign function to $G^{\sB\sF}$, we obtain a unitary $W^{\sB\sH}$ that is a $(1, \cO(\log{d_{\sA\sB}}), 0)$-block-encoding of $P_\sgn^{(\rm SV)}\big(M^\sB/(4d_\sA^3d_\sB^2)\big)$ for a degree-$u$ odd polynomial $P_\sgn$, using $G^{\sB\sF}$ and its inverse a total of $u=\cO(\log(1/\delta_2)/\beta)$ times, where $\beta$ will be chosen appropriately later.

The algorithm implements the unitary $\til{W}^{\sB\sH}$ obtained from $W^{\sB\sH}$ by replacing the $k$-th use of $G^{\sB\sF}$ or its inverse with that of $\til{G}_k^{\sB\sF}$ for $k = 1, 2, \ldots, u$, where each $\til{G}_k^{\sB\sF}$ is constructed as in Eq.~\eqref{inteq:153} from its own estimates.
Since the other operations in $W^{\sB\sH}$ and $\til{W}^{\sB\sH}$ are the same unitaries, the triangle inequality and Eq.~\eqref{inteq:MY G error} imply that
\begin{align}
\label{inteq:134}
    \big\|\til{W}^{\sB\sH} - W^{\sB\sH}\big\|_\infty
    \leq \sum_{k=1}^u \big\|\til{G}_k^{\sB\sF} - G^{\sB\sF}\big\|_\infty
    \leq \eta_1,
\end{align}
where $\eta_1 = 2u\delta_1/d_{\sA\sB}$.
Note that this evaluation holds even though the unitaries $\til{G}_k^{\sB\sF}$ differ for each $k$.

We finally evaluate the approximation error using the fidelity difference.
Let
\begin{align}
    &\til{\rF} = \rF\big(\til{W}^{\sB\sH}\ket{\rho}^{\sA\sB}\ket{0}^\sH, \ket{\sigma}^{\sA\sB}\ket{0}^\sH\big), \\
    &\rF' = \big|\bra{\sigma}^{\sA\sB}P_\sgn^{(\rm SV)}\big(M^\sB/(4d_\sA^3d_\sB^2)\big)\ket{\rho}^{\sA\sB}\big|^2, \\
    &\rF = \rF(\rho^\sA, \sigma^\sA) = \big|\bra{\sigma}^{\sA\sB} V^{\sB}\ket{\rho}^{\sA\sB}\big|^2,
\end{align}
where $V^\sB$ is the exact Uhlmann partial isometry.
Note that $\sgn^{(\rm SV)}\big(M^\sB/(4d_\sA^3d_\sB^2)\big)$ corresponds to $V^\sB$.
Since $\rF' = \big|\bra{\sigma}^{\sA\sB}\bra{0}^\sH W^{\sB\sH}\ket{\rho}^{\sA\sB}\ket{0}^\sH\big|^2$, Eq.~\eqref{inteq:134} implies that
\begin{align}
\label{inteq:125}
    \Big|\sqrt{\til{\rF}} - \sqrt{\rF'}\Big|
    \leq \big|\bra{\sigma}^{\sA\sB}\bra{0}^\sH\big(\til{W}^{\sB\sH} - W^{\sB\sH}\big)\ket{\rho}^{\sA\sB}\ket{0}^\sH\big|
    \leq \eta_1.
\end{align}

Also, from a calculation similar to that in Sec.~\ref{sec:uhlfidquery}, we see that
\begin{align}
    &\big|\sqrt{\rF'} - \sqrt{\rF}\big|  \notag\\
    &\leq \Big|\tr\Big[\Big(\sgn^{(\rm SV)}\big(M^\sB/(4d_\sA^3 d_\sB^2)\big)  \notag\\
    &\hspace{6pc}- P_\sgn^{(\rm SV)}\big(M^\sB/(4d_\sA^3 d_\sB^2)\big)\Big)(M^{\sB})^\dag\Big]\Big| \\
    &\leq \sum_k \big|1 - P_\sgn\big(s_k/(4d_\sA^3 d_\sB^2)\big)\big|s_k \\
    &= \sum_{k\in I_\beta} \big|1 - P_\sgn\big(s_k/(4d_\sA^3 d_\sB^2)\big)\big|s_k  \notag\\
    &\hspace{6pc}+ \sum_{k\in \bar{I}_\beta} \big|1 - P_\sgn\big(s_k/(4d_\sA^3 d_\sB^2)\big)\big|s_k \\
    &\leq \delta_2 + 8d_\sA^3 d_\sB^2\beta R,
\end{align}
where $\{s_k\}_k$ and $R$ are the singular values and the rank of $M^\sB$, respectively, and $I_\beta = \{k \in \bN; s_k/(4d_\sA^3 d_\sB^2) \geq \beta\}$ and $\bar{I}_\beta = \{k \in \bN; s_k/(4d_\sA^3 d_\sB^2) < \beta\}$.
Here, we used $\big|\sgn(x) - P_\sgn(x)\big| \leq \delta_2$ for $x \in [\beta, 1]$ when $u=\cO\big(\log(1/\delta_2)/\beta\big)$.
Thus, setting $\beta$ as $1/\beta = 4d_\sA^3 d_\sB^2\min\{1/s_{\rm min}, 2R/\delta_2\}$ ensures at least 
\begin{align}
\label{inteq:126}
    \big|\sqrt{\rF} - \sqrt{\rF'}\big| \leq 2\delta_2.
\end{align}

From Eqs.~\eqref{inteq:125} and~\eqref{inteq:126}, we obtain that $\big|\til{\rF} - \rF\big| \leq 2\eta_1 + 4\delta_2$ when all the estimates succeed, where we used the inequality that $|x-y|\leq 2 |\sqrt{x} - \sqrt{y}|$ for $0 \leq x, y \leq 1$.
The algorithm realizes the quantum channel $\cT^\sB$ given by the average of $\tr_\sH \circ \til{\cW}^{\sB\sH} \circ \cP_{\ket{0}}^{\bC\rarr\sH}$ over the outcomes of the measurements.
Since the fidelity with the pure state $\ket{\sigma}^{\sA\sB}$ is linear in the other state and monotonic under the partial trace, $\rF\big(\cT^\sB(\ketbra{\rho}{\rho}^{\sA\sB}), \ket{\sigma}^{\sA\sB}\big)$ is at least the average of $\til{\rF}$ over the outcomes.
Each $\til{c}_{ij}$ is estimated $2d_\sA$ times in each $\til{G}_k^{\sB\sF}$, and thus the estimation is performed $N = 2ud_\sA d_{\sA\sB}^2$ times in total.
We let each estimation fail with probability at most $\eta = \delta/(2N)$, which increases the number of queries only by a logarithmic factor.
Then, all the estimates succeed with probability at least $1 - \delta/2$, and we have
\begin{align}
    \rF\big(\cT^\sB(\ketbra{\rho}{\rho}^{\sA\sB}), \ket{\sigma}^{\sA\sB}\big)
    &\geq \Big(1 - \f{\delta}{2}\Big)\big(\rF(\rho^\sA, \sigma^\sA) - 2\eta_1 - 4\delta_2\big) \\
    &\geq \rF(\rho^\sA, \sigma^\sA) - 2\eta_1 - 4\delta_2 - \f{\delta}{2}.
\end{align}
By rescaling the parameters as $\eta_1 = 2u\delta_1/d_{\sA\sB} = \delta/8$, i.e., $\delta_1 = d_{\sA\sB}\delta/(16u)$, and $\delta_2 = \delta/16$, we finally obtain
\begin{align}
    \rF\big(\cT^\sB(\ketbra{\rho}{\rho}^{\sA\sB}), \ket{\sigma}^{\sA\sB}\big) \geq \rF(\rho^\sA, \sigma^\sA) - \delta.
\end{align}

Consequently, the total number of queries in the entire algorithm is given by
\begin{align}
    \gamma 
    &= hu \\
    &= \til{\cO}\big(d_\sA^6d_\sB^5/\delta_1^2\big)\cO\big(\log{(1/\delta_2)}/\beta\big) \\
    &= \til{\cO}\Big(\f{d_\sA^{13}d_\sB^{9}}{\delta^2}\min\Big\{\f{1}{s_{\rm min}^3}, \f{R^3}{\delta^3}\Big\}\Big).
\end{align}
The most foundational operation throughout this algorithm is the Pauli measurement, which, up to logarithmic factors, has to be performed as many times as the total number of queries. 
Thus, the quantum circuit of this algorithm consists of $\tilde{\cO}(\gamma)$ one- and two-qubit gates.
Due to its sequential structure, the quantum circuit uses $\cO(\log d_{\sA\sB})$ qubits simultaneously.
Since this algorithm is based on simple measurements and the QSVT, the circuit description can be computed classically in $\poly\big(\gamma, \log{(d_\sA d_\sB)}, \log{(1/\delta)}\big)$ time.


\section{Derivation of an inequality between the minimum non-zero singular values}
\label{sec:appendix product of singvalue}

We here derive the following proposition. We use $\supp[A]$ to denote the support of $A$, i.e., the orthogonal complement of the kernel of $A$, and $\im[A]$ to denote the image of $A$.

\begin{proposition}[Inequality about the minimum non-zero singular value of a product of matrices]
\label{prop:min singular inequality}
For any matrices $A$ and $B$ such that $\im[B] \subset \supp[A]$, it holds that $s_{\rm min}(AB) \geq s_{\rm min}(A) s_{\rm min}(B)$, where $s_{\rm min}(\cdot)$ denotes the minimum non-zero singular value of the input matrix.
\end{proposition}

Note that if a matrix $A$ is Hermitian, then $\im[A] = \supp[A]$ holds.
If both $A$ and $B$ are Hermitian, $s_{\rm min}(AB) = s_{\rm min}(BA)$.

\begin{proof}[Proof of \Cref{prop:min singular inequality}]

From the min-max principle~\cite{Gohberg2003basicclass, Simon2005traceideals, bhatia2013matrix}, the minimum non-zero singular value of $A$ is given by $s_{\rm min}(A) = \min_{x \in \supp[A]; x \neq 0}\f{\|Ax\|}{\|x\|}$. Hence, $s_{\rm min}(AB)$ can be calculated as follows:
\begin{align}
    s_{\rm min}(AB) 
    &= \min_{x \in \supp[AB]; x \neq 0}\f{\|ABx\|}{\|x\|} \\
    &= \min_{x \in \supp[AB]; x \neq 0}\f{\|ABx\|}{\|Bx\|}\f{\|Bx\|}{\|x\|} \\
    &\geq \min_{x \in \supp[AB]; x \neq 0}\f{\|ABx\|}{\|Bx\|}\min_{x \in \supp[AB]; x \neq 0}\f{\|Bx\|}{\|x\|} \\
    &\geq \min_{y \in \supp[A]; y \neq 0}\f{\|Ay\|}{\|y\|}\min_{x \in \supp[B]; x \neq 0}\f{\|Bx\|}{\|x\|} \\
    &= s_{\rm min}(A) s_{\rm min}(B),
\end{align}
where, in the second inequality, we used the fact that $\supp[AB] \subset \supp[B]$ and the assumption $\im[B] \subset \supp[A]$.

\end{proof}


\section{Optimal local transformation between generally mixed states}
\label{sec:opt local and Uhlmann}

We show that Uhlmann's theorem characterizes an optimal local transformation between two states, even when the involved states are generally mixed states.

\begin{proposition}[Optimal local operation via the Uhlmann transformation]
\label{prop:opt local mixed via Uhl}
Let $\ket{\rho}^{\sA\sB\sE}$ and $\ket{\sigma}^{\sA\sC\sE}$ be purified states of $\rho^{\sA\sB}$ and $\sigma^{\sA\sC}$, respectively. Then, any quantum channel $\cT^{\sB\rarr\sC}$ satisfies that
    \begin{align}
    \label{eq:opt ineq local op}
        \rF\big(\cT^{\sB\rarr\sC}(\rho^{\sA\sB}), \sigma^{\sA\sC}\big) \leq \max_{\cL^\sE} \rF\big(\rho^{\sA\sE}, \cL^\sE(\sigma^{\sA\sE})\big),
    \end{align}
where the maximization is taken over all quantum channels acting on the system $\sE$.

Let $\til{\cL}^\sE$ be a quantum channel that attains the maximization in Eq.~\eqref{eq:opt ineq local op}, and $V_{\til{\cL}}^{\sE\rarr\sE\sF}$ be its Stinespring isometry.
A quantum channel $\til{\cT}^{\sB\rarr\sC}$ that achieves equality in Eq.~\eqref{eq:opt ineq local op} is given by 
\begin{align}
    \til{\cT}^{\sB\rarr\sC}(\cdot) = \tr_\sF\big[U^\sG(\cdot \otimes \ketbra{0}{0}^\sD)(U^\sG)^\dag\big],
\end{align}
where $U^\sG$ is an Uhlmann unitary acting on $\sG = \sB\sD = \sC\sF$, which satisfies 
\begin{align}
\rF\big(U^\sG\ket{\rho}^{\sA\sB\sE}\ket{0}^\sD, V_{\til{\cL}}^{\sE\rarr\sE\sF}\ket{\sigma}^{\sA\sC\sE}\big) = \rF\big(\rho^{\sA\sE}, \til{\cL}^\sE(\sigma^{\sA\sE})\big).
\end{align}
\end{proposition}

Note that, since the optimal transformation $\til{\cT}^{\sB\rarr\sC}$ is constructed from the Uhlmann unitary $U^\sG$ connecting $\ket{\rho}^{\sA\sB\sE}\ket{0}^\sD$ and $V_{\til{\cL}}^{\sE\rarr\sE\sF}\ket{\sigma}^{\sA\sC\sE}$, implementing it would generally require knowledge of the purified states $\ket{\rho}^{\sA\sB\sE}$, $\ket{\sigma}^{\sA\sC\sE}$, and also the quantum channel $\til{\cL}^\sE$ that attains the maximization of Eq.~\eqref{eq:opt ineq local op}.

In particular, when we apply this proposition to the decoupling argument discussed in Sec.~\ref{sec:decoupling and uhlmann}, the state $\sigma^{\sA\sE}$ can be taken to be a product state $\rho^\sA \otimes \tau^\sE$. In this case, the right-hand side of Eq.~\eqref{eq:opt ineq local op} reduces to
$\max_{\tau^\sE}\rF(\rho^{\sA\sE}, \rho^\sA\otimes\tau^\sE)$, where the maximization is taken over all states on $\sE$.
By identifying a state $\til{\tau}^\sE$ that achieves this maximization, we can construct the optimal local operation via the Uhlmann transformation.

We now prove the proposition.

\begin{proof}[Proof of \Cref{prop:opt local mixed via Uhl}]

From Uhlmann's theorem, for any quantum channel $\cT^{\sB\rarr\sC}$, there is a unitary $\hat{U}^{\sE\sF}$ such that
\begin{align}
\label{inteq:108}
    \rF\big(\cT^{\sB\rarr\sC}(\rho^{\sA\sB}), \sigma^{\sA\sC}\big) = \rF\big(V_\cT^{\sB\rarr\sC\sF}\ket{\rho}^{\sA\sB\sE}, \hat{U}^{\sE\sF}\ket{\sigma}^{\sA\sC\sE}\ket{0}^\sF\big),
\end{align}
where $V_\cT^{\sB\rarr\sC\sF}$ is a Stinespring isometry of $\cT^{\sB\rarr\sC}$.
Since the fidelity is monotonic under the partial trace, the right-hand side of Eq.~\eqref{inteq:108} is bounded as
\begin{align}
    &\rF\big(V_\cT^{\sB\rarr\sC\sF}\ket{\rho}^{\sA\sB\sE}, \hat{U}^{\sE\sF}\ket{\sigma}^{\sA\sC\sE}\ket{0}^\sF\big)  \notag\\
    &\leq \rF\big(\rho^{\sA\sE}, \tr_\sF\big[\hat{U}^{\sE\sF}(\sigma^{\sA\sE}\otimes\ketbra{0}{0}^\sF)(\hat{U}^{\sE\sF})^\dag\big]\big) \\
    &\leq \max_{\cL^{\sE}}\rF\big(\rho^{\sA\sE}, \cL^\sE(\sigma^{\sA\sE})\big), 
\end{align}
where we take the maximization over all quantum channels acting on system $\sE$, considering that $\tr_\sF\big[\hat{U}^{\sE\sF}(\cdot\otimes\ketbra{0}{0}^\sF)(\hat{U}^{\sE\sF})^\dag\big]$ gives a quantum channel on $\sE$.
Thus, any quantum channel $\cT^{\sB\rarr\sC}$ satisfies
\begin{align}
\label{inteq:109}
    \rF\big(\cT^{\sB\rarr\sC}(\rho^{\sA\sB}), \sigma^{\sA\sC}\big)
    \leq \max_{\cL^{\sE}}\rF\big(\rho^{\sA\sE}, \cL^\sE(\sigma^{\sA\sE})\big).
\end{align}

Regarding the equality condition, let $\til{\cL}^\sE$ be a quantum channel that attains the maximization, and let $V_{\til{\cL}}^{\sE \to \sE\sF}$ be its Stinespring isometry.
Then, from Uhlmann's theorem, there exists a unitary $U^\sG$ on $\sG = \sB\sD = \sC\sF$ satisfying that
\begin{align}
    \rF\big(\rho^{\sA\sE}, \til{\cL}^\sE(\sigma^{\sA\sE})\big)
    &= \rF\big(U^\sG\ket{\rho}^{\sA\sB\sE}\ket{0}^\sD, V_{\til{\cL}}^{\sE\rarr\sE\sF}\ket{\sigma}^{\sA\sC\sE}\big) \\
    \label{inteq:110}
    &\leq \rF\big(\til{\cT}^{\sB\rarr\sC}(\rho^{\sA\sB}), \sigma^{\sA\sC}\big),
\end{align}
where $\til{\cT}^{\sB\rarr\sC}(\cdot) = \tr_\sF\big[U^\sG(\cdot\otimes\ketbra{0}{0}^\sD)(U^\sG)^\dag\big]$, and in the inequality we used the monotonic property of the fidelity under the partial trace again.
From Eqs.~\eqref{inteq:109} and~\eqref{inteq:110}, we see that $\til{\cT}^{\sB\rarr\sC}$ gives the optimal transformation from $\rho^{\sA\sB}$ to $\sigma^{\sA\sC}$ through a local operation and achieves that $\rF\big(\til{\cT}^{\sB\rarr\sC}(\rho^{\sA\sB}), \sigma^{\sA\sC}\big) = \max_{\cL^{\sE}}\rF\big(\rho^{\sA\sE}, \cL^\sE(\sigma^{\sA\sE})\big)$.

\end{proof}

\section{Derivation of the explicit form of the Uhlmann partial isometry}
\label{sec:analyze partial iso}

Following the approach in Ref.~\cite{Jozsa1994FidelityforMixedQuantumStates}, we show that the partial isometry $V^\sB = \sgn^{(\rm SV)}\big(\tr_{\sA'}\big[\ketbra{\sigma}{\rho}^{\sA'\sB}\big]\big)$ satisfies that $\rF(\rho^\sA, \sigma^\sA) = \rF(V^\sB\ket{\rho}^{\sA\sB}, \ket{\sigma}^{\sA\sB})$, where the sign function acts on the singular values of the input matrix, as shown in Eq.~\eqref{inteq:158}.
This implies that $V^\sB$ is the Uhlmann partial isometry.

To this end, we should recall the variational characterization of the trace norm: for any square matrix $M$, it holds that $\|M\|_1 = \max_{U}|\tr[UM]|$, where the maximization is taken over all unitaries. The maximization is attained by the inverse of the unitary polar factor of $M$.
On the block specified by the left and right singular spaces corresponding to the non-zero singular values of $M$, the unitary polar factor of $M$ is uniquely determined and given by $\sgn^{(\rm SV)}(M)$. Outside the block, it can act arbitrarily, but this arbitrariness does not affect the result.

Using the above fact, we can derive the explicit form of the Uhlmann partial isometry in a straightforward manner.
By Uhlmann's theorem, $\rF(\rho^\sA, \sigma^\sA)$ is rephrased as
\begin{align}
    \rF(\rho^\sA, \sigma^\sA) 
    &= \max_{U^\sB}\rF(U^\sB\ket{\rho}^{\sA\sB}, \ket{\sigma}^{\sA\sB}) \\
    &= \max_{U^\sB} \big|\bra{\sigma}^{\sA\sB}U^\sB\ket{\rho}^{\sA\sB}\big|^2 \\
    &= \max_{U^\sB} \big|\tr\big[U^\sB\tr_\sA[\ketbra{\rho}{\sigma}^{\sA\sB}]\big]\big|^2,
\end{align}
where the maximization is taken over all unitaries $U^\sB$.
On the block specified by the left and right singular spaces corresponding to the non-zero singular values of $\tr_\sA[\ketbra{\rho}{\sigma}^{\sA\sB}]$, the unitary that achieves the maximization is uniquely determined up to a global phase as $\big(\sgn^{(\rm SV)}(\tr_\sA[\ketbra{\rho}{\sigma}^{\sA\sB}])\big)^\dag = \sgn^{(\rm SV)}(\tr_\sA[\ketbra{\sigma}{\rho}^{\sA\sB}])$, which is a partial isometry. We denote the partial isometry by $V^\sB$ and refer to it as the Uhlmann partial isometry.

Finally, we see that the singular values of $\tr_\sA[\ketbra{\sigma}{\rho}^{\sA\sB}]$ correspond to those of $\sqrt{\sigma^\sA}\sqrt{\rho^\sA}$.
Without loss of generality, we suppose $d_\sA \leq d_\sB$, which is justified by the fact that we can arbitrarily pad $\sB$ with the state $\ket{0}$.

Suppose that the Schmidt decompositions of the purified states $\ket{\rho}^{\sA\sB}$ and $\ket{\sigma}^{\sA\sB}$ are given by
\begin{align}
    \label{inteq:schmidt decomp rho_sigma}
      \ket{\rho}^{\sA\sB} = \sum_{i=1}^{r_\rho}\sqrt{p_i}\ket{e_i}^\sA\ket{f_i}^\sB, \ \text{and} \ \ \ket{\sigma}^{\sA\sB} = \sum_{j=1}^{r_\sigma}\sqrt{q_j}\ket{g_j}^\sA\ket{h_j}^\sB,
\end{align}
respectively.
Then, we have
\begin{align}
    &\tr_\sA\big[\ketbra{\sigma}{\rho}^{\sA\sB}\big]  \notag\\
    &= \sum_{i,j}\sqrt{p_i}\sqrt{q_j}\braket{e_i}{g_j}\ketbra{h_j}{f_i}^\sB \\
    &= \Big(\sum_{i,j}\sqrt{p_i}\sqrt{q_j}\braket{g_j}{e_i}\ketbra{h_j^*}{f_i^*}^\sB\Big)^* \\
    &= \big(V_2^{\sA\rarr\sB}\sum_{i,j}\sqrt{p_i}\sqrt{q_j}\braket{g_j}{e_i} \ketbra{g_j}{e_i}^\sA(V_1^{\sA\rarr\sB})^\dag\big)^* \\
    &= \big(V_2^{\sA\rarr\sB}\sqrt{\sigma^\sA}\sqrt{\rho^\sA}(V_1^{\sA\rarr\sB})^\dag\big)^*,
\end{align}
where $V_1^{\sA\rarr\sB}$ and $V_2^{\sA\rarr\sB}$ are isometries that map $\ket{e_i}^\sA$ to $\ket{f_i^*}^\sB$ and $\ket{g_j}^\sA$ to $\ket{h_j^*}^\sB$, respectively.

Since singular values are invariant under complex conjugation and isometries, we conclude that $\tr_\sA\big[\ketbra{\sigma}{\rho}^{\sA\sB}\big]$ and $\sqrt{\sigma^\sA}\sqrt{\rho^\sA}$ have identical singular values.
Note that all of these singular values lie within the range $[0,1]$.
The rank of $\tr_{\sA}\big[\ketbra{\sigma}{\rho}^{\sA\sB}\big]$ also coincides with the rank of $\sqrt{\sigma^\sA}\sqrt{\rho^\sA}$, as they have the same number of singular values.


\section{Proof of \texorpdfstring{\Cref{prop:QSVT sign stages}}{Proposition~\ref{prop:QSVT sign stages}}}
\label{sec:proof of prop QSVT sign stages}

We prove \Cref{prop:QSVT sign stages}, which we restate below.

\propsignstages*

The algorithm $\mathtt{QSVTSIGNSTAGES}$ implements the reflections about $\ket{\varphi}^\sE$ and $\ket{\phi}^\sE$ from their samples, using block-encoding unitaries of $\ketbra{\varphi}{\varphi}^\sE$ and $\ketbra{\phi}{\phi}^\sE$ controlled by another qubit.
To evaluate the error in terms of the fidelity, we use one of the Kraus operators of the resulting channel, as given in the following lemma.

\begin{lemma}[Controlled block-encoding of a pure state from samples]
\label{lem:ctrl BE pure state}
Let $m \geq 3$ be an integer, and let $\ket{\psi}^\sA$ be a pure state in a $d_\sA$-dimensional Hilbert space.
Let $W^{\sC\sA}$ be the $(1, 1, 0)$-block-encoding unitary of $\ketbra{\psi}{\psi}^\sA$ given by
\begin{align}
\label{eq:W ctrl pure state}
    W^{\sC\sA} = X^\sC H^\sC e^{i\pi\ketbra{1}{1}^\sC\otimes\ketbra{\psi}{\psi}^\sA} H^\sC,
\end{align}
where $X^\sC$ and $H^\sC$ are the one-qubit Pauli-$X$ and Hadamard gates, respectively.
Then, for a one-qubit system $\sQ$, there exists a quantum algorithm that outputs a quantum channel $\cG_\ctrl^{\sQ\sC\sA}$ having a Kraus operator of the form $\ketbra{0}{0}^\sQ\otimes\bI^{\sC\sA} + \ketbra{1}{1}^\sQ\otimes K^{\sC\sA}$, where $K^{\sC\sA}$ satisfies 
\begin{align}
    \big\|K^{\sC\sA} - W^{\sC\sA}\big\|_\infty \leq \f{\pi^2}{2m}.
\end{align}
The algorithm uses $m$ samples of $\ket{\psi}^\sA$ and $\cO(m\log{d_\sA})$ one- and two-qubit gates.
\end{lemma}

We first provide the proof of \Cref{lem:ctrl BE pure state}, which will be used to prove \Cref{prop:QSVT sign stages}.

\begin{proof}[Proof of \Cref{lem:ctrl BE pure state}]
Let $P^{\sC\sA} = \ketbra{1}{1}^\sC\otimes\ketbra{\psi}{\psi}^\sA$ and $\Delta t = \pi/m$, and let $\sC'\sA'$ be a system isomorphic to $\sC\sA$.
We define a quantum channel $\cD^{\sQ\sC\sA}$ by
\begin{align}
\label{eq:controlled DME step}
    \cD^{\sQ\sC\sA}(\cdot) = \tr_{\sC'\sA'}\big[V\big(\cdot\otimes\ketbra{1}{1}^{\sC'}\otimes\ketbra{\psi}{\psi}^{\sA'}\big)V^\dag\big],
\end{align}
where $V = \ketbra{0}{0}^\sQ\otimes\bI + \ketbra{1}{1}^\sQ\otimes e^{i\Delta t F}$, and $F$ is the swap operator between $\sC\sA$ and $\sC'\sA'$.
The channel $\cD^{\sQ\sC\sA}$ is one step of the density matrix exponentiation of $P^{\sC\sA}$ controlled by $\sQ$, and it uses one sample of $\ket{\psi}^\sA$.
We define $\cG_\ctrl^{\sQ\sC\sA}$ as the channel that applies $H^\sC$, $\cD^{\sQ\sC\sA}$ for $m$ times, $H^\sC$, and $X^\sC$ controlled by $\sQ$ in this order.

Using $e^{i\Delta tF} = \cos(\Delta t)\bI + i\sin(\Delta t)F$, we see that $\cD^{\sQ\sC\sA}$ has a Kraus operator
\begin{align}
    &\big(\bI^{\sQ\sC\sA}\otimes\bra{1}^{\sC'}\bra{\psi}^{\sA'}\big)V\big(\bI^{\sQ\sC\sA}\otimes\ket{1}^{\sC'}\ket{\psi}^{\sA'}\big) \notag\\
    &\hspace{0pc}= \ketbra{0}{0}^\sQ\otimes\bI^{\sC\sA} + \ketbra{1}{1}^\sQ\otimes\big(\cos(\Delta t)\bI^{\sC\sA} + i\sin(\Delta t)P^{\sC\sA}\big).
\end{align}
Since the product of Kraus operators of channels is a Kraus operator of their composition, the channel applying $\cD^{\sQ\sC\sA}$ for $m$ times has a Kraus operator
\begin{align}
    &\ketbra{0}{0}^\sQ\otimes\bI^{\sC\sA} + \ketbra{1}{1}^\sQ\otimes\big(\cos(\Delta t)\bI^{\sC\sA} + i\sin(\Delta t)P^{\sC\sA}\big)^m \notag\\
    &\hspace{0pc}= \ketbra{0}{0}^\sQ\otimes\bI^{\sC\sA} + \ketbra{1}{1}^\sQ\otimes\big(c_m(\bI^{\sC\sA} - P^{\sC\sA}) - P^{\sC\sA}\big),
\end{align}
where $c_m = \cos^m(\pi/m)$, and we used $(\cos(\Delta t) + i\sin(\Delta t))^m = e^{i\pi} = -1$.
Applying $H^\sC$ on both sides and $X^\sC$ controlled by $\sQ$, we see that $\cG_\ctrl^{\sQ\sC\sA}$ has the Kraus operator $\ketbra{0}{0}^\sQ\otimes\bI^{\sC\sA} + \ketbra{1}{1}^\sQ\otimes K^{\sC\sA}$, where
\begin{align}
    K^{\sC\sA} = X^\sC H^\sC\big(c_m(\bI^{\sC\sA} - P^{\sC\sA}) - P^{\sC\sA}\big)H^\sC.
\end{align}
Since $W^{\sC\sA} = X^\sC H^\sC e^{i\pi P^{\sC\sA}}H^\sC$ and $e^{i\pi P^{\sC\sA}} = (\bI^{\sC\sA} - P^{\sC\sA}) - P^{\sC\sA}$, it holds that
\begin{align}
    \big\|K^{\sC\sA} - W^{\sC\sA}\big\|_\infty
    &= (1 - c_m)\big\|X^\sC H^\sC(\bI^{\sC\sA} - P^{\sC\sA})H^\sC\big\|_\infty \\
    &= 1 - c_m \\
    &\leq \f{\pi^2}{2m},
\end{align}
where we used $\cos(\pi/m) \geq 1 - \pi^2/(2m^2) \geq 0$ and Bernoulli's inequality $(1 - y)^m \geq 1 - my$ for $y \in [0, 1]$ in the last line.

Each application of $\cD^{\sQ\sC\sA}$ uses $\cO(\log{d_\sA})$ one- and two-qubit gates as in \Cref{prevthm:cntrl-DME}, and controlling them by $\sQ$ changes the number of gates only by a constant factor.
Hence, $\cG_\ctrl^{\sQ\sC\sA}$ uses $m$ samples of $\ket{\psi}^\sA$ and $\cO(m\log{d_\sA})$ one- and two-qubit gates.

\end{proof}

We now prove \Cref{prop:QSVT sign stages}.
We first construct a unitary that applies the QSVT with the sign function in stages, assuming that the reflections about $\ket{\varphi}^\sE$ and $\ket{\phi}^\sE$ are implemented exactly.
We then evaluate the effect of implementing these reflections from samples of $\ket{\varphi}^\sE$ and $\ket{\phi}^\sE$, based on \Cref{lem:ctrl BE pure state}.
Since the components reaching a later stage have smaller singular values, the reflections in that stage can be implemented with lower accuracy.

\begin{proof}[Proof of \Cref{prop:QSVT sign stages}]
Without loss of generality, we assume that $\beta$ is a power of $1/2$, since otherwise we can replace $\beta$ by $2^{-\lceil\log{(1/\beta)}\rceil} \in (\beta/2, \beta]$, which strengthens Eq.~\eqref{eq:QSVT sign stages} and changes the costs only by constant factors.
Let $J$ be the smallest non-negative integer such that $2^{-J} \leq \max\{\beta, s/2\}$, which determines the number of stages below.
Since $\beta$ is a power of $1/2$ and $\max\{\beta, s/2\} \leq 1$, it holds that
\begin{align}
\label{eq:last threshold}
    \beta \leq 2^{-J} \leq \max\{\beta, s/2\} < 2^{-J+1}.
\end{align}
In particular, $2^{-J} = \beta$ when $s \leq 2\beta$.

Let $\sC$ be a one-qubit system.
We define a projector $\Pi^{\sC\sE}$ and unitaries $S^{\sC\sE\sB}$ and $Q^{\sC\sE}$ by
\begin{align}
    &\Pi^{\sC\sE} = \ketbra{0}{0}^\sC\otimes\ketbra{\varphi}{\varphi}^\sE + \ketbra{1}{1}^\sC\otimes\ketbra{\phi}{\phi}^\sE, \\
    &S^{\sC\sE\sB} = \ketbra{0}{1}^\sC\otimes U^{\sB\sE} + \ketbra{1}{0}^\sC\otimes(U^{\sB\sE})^\dag, \\
    &Q^{\sC\sE} = 2\Pi^{\sC\sE} - \bI^{\sC\sE},
\end{align}
and let $\Omega^{\sC\sE\sB} = Q^{\sC\sE}S^{\sC\sE\sB}$. Note that $S$ and $Q$ are Hermitian unitaries.
From the definition of $M^\sB$, we have
\begin{align}
\label{eq:PiSPi}
    \Pi S\Pi 
    = \ketbra{0}{1}\otimes\ketbra{\varphi}{\phi}\otimes M + \ketbra{1}{0}\otimes\ketbra{\phi}{\varphi}\otimes M^\dag.
\end{align}

We first show that $(\Omega + \Omega^\dag)/2$ acts on invariant subspaces of $\Omega$ as the singular values of $M^\sB$ up to sign. Let $\ket{v}$ be a normalized eigenvector of $\Pi S\Pi$ in $\im[\Pi]$ with eigenvalue $\lambda \in [-1, 1]$. When $|\lambda| < 1$, we define a normalized vector $\ket{v^\bot} = (S\ket{v} - \lambda\ket{v})/\sqrt{1 - \lambda^2}$, which is in $\ker[\Pi]$ as $\Pi S\ket{v} = \lambda\ket{v}$ and $\Pi\ket{v} = \ket{v}$. Since $S^2 = \bI$, we have $S\ket{v^\bot} = \sqrt{1-\lambda^2}\ket{v} - \lambda\ket{v^\bot}$, and thus, on the subspace $\cH_v$ spanned by $\{\ket{v}, \ket{v^\bot}\}$, the operators $S$ and $\Omega$ are represented as
\begin{align}
    S = 
    \begin{pmatrix}
        \lambda & \sqrt{1-\lambda^2} \\
        \sqrt{1-\lambda^2} & -\lambda
    \end{pmatrix}, \ \ \textrm{and} \ \ 
    \Omega = 
    \begin{pmatrix}
        \lambda & \sqrt{1-\lambda^2} \\
        -\sqrt{1-\lambda^2} & \lambda
    \end{pmatrix},
\end{align}
where we used $Q\ket{v} = \ket{v}$ and $Q\ket{v^\bot} = -\ket{v^\bot}$.
When $|\lambda| = 1$, we have $S\ket{v} = \lambda\ket{v}$, and we let $\cH_v$ be spanned by $\ket{v}$.
In both cases, $\Omega$ preserves $\cH_v$, and it holds that
\begin{align}
\label{eq:cosine walk}
    \f{1}{2}(\Omega + \Omega^\dag)\ket{w} = \lambda\ket{w},
\end{align}
for any $\ket{w} \in \cH_v$.
Moreover, for two orthogonal eigenvectors $\ket{v}$ and $\ket{v'}$ of $\Pi S\Pi$ in $\im[\Pi]$, the subspaces $\cH_v$ and $\cH_{v'}$ are orthogonal, since $\bra{v}S\ket{v'} = \bra{v}\Pi S\Pi\ket{v'} = 0$.
From Eq.~\eqref{eq:PiSPi}, for a singular value decomposition $M^\sB = \sum_k s_k\ketbra{\eta_k}{\xi_k}^\sB$, the vectors
\begin{align}
\label{eq:v pm}
    \ket{v_k^\pm}^{\sC\sE\sB} = \f{1}{\sqrt{2}}\big(\ket{0}^\sC\ket{\varphi}^\sE\ket{\eta_k}^\sB \pm \ket{1}^\sC\ket{\phi}^\sE\ket{\xi_k}^\sB\big)
\end{align}
are orthogonal eigenvectors of $\Pi S\Pi$ in $\im[\Pi]$ with eigenvalues $\pm s_k$.

We next construct a block-encoding unitary of $(\Omega + \Omega^\dag)/2$.
Let $\sL$ be a one-qubit system, and let $V^{\sL\sC\sE\sB}$ be a unitary defined by
\begin{align}
    V^{\sL\sC\sE\sB} 
    &= H^\sL\big(\ketbra{0}{0}^\sL\otimes\Omega^{\sC\sE\sB} + \ketbra{1}{1}^\sL\otimes(\Omega^{\sC\sE\sB})^\dag\big)H^\sL \\
    &= H^\sL\big(\ketbra{0}{0}^\sL\otimes\bI^{\sC\sE\sB} + \ketbra{1}{1}^\sL\otimes S^{\sC\sE\sB}\big)Q^{\sC\sE} \notag\\
    &\hspace{3pc} \big(\ketbra{0}{0}^\sL\otimes S^{\sC\sE\sB} + \ketbra{1}{1}^\sL\otimes\bI^{\sC\sE\sB}\big)H^\sL,
\end{align}
where $H^\sL$ is the Hadamard gate on $\sL$, and we used $\Omega^\dag = SQ$.
Then, $V^{\sL\sC\sE\sB}$ is a $(1, 1, 0)$-block-encoding of $(\Omega + \Omega^\dag)/2$, which uses $Q$ once and $S$ twice.
Each use of $S$ requires one use of each of the controlled versions of $U^{\sB\sE}$ and $(U^{\sB\sE})^\dag$.
For a real even or odd polynomial $P$ satisfying $|P(x)| \leq 1$ for $x \in [-1, 1]$, \Cref{prop:general QSVT} applied to $V^{\sL\sC\sE\sB}$ yields a unitary $V_P^{\sG\sC\sE\sB}$ that is a block-encoding of $P^{(\rm SV)}\big((\Omega + \Omega^\dag)/2\big) = P\big((\Omega + \Omega^\dag)/2\big)$, where $\sG$ is a two-qubit system including $\sL$.
Here, we used that, for a Hermitian matrix, the singular value transformation by an even or odd polynomial coincides with the polynomial of the matrix.
Since $V_P^{\sG\sC\sE\sB}$ consists of $V^{\sL\sC\sE\sB}$, its inverse, and gates acting on $\sG$, it preserves $\cH^\sG\otimes\cH_v$, and Eq.~\eqref{eq:cosine walk} implies that
\begin{align}
\label{eq:scalar QSVT}
    \bra{0}^\sG V_P^{\sG\sC\sE\sB}\ket{0}^\sG\ket{w}^{\sC\sE\sB} = P(\lambda)\ket{w}^{\sC\sE\sB},
\end{align}
for any $\ket{w}^{\sC\sE\sB} \in \cH_v$, and the same holds with $\ket{w}^{\sC\sE\sB}$ tensored with any other system.

We now choose the polynomials used in each stage.
For $a \in (0, 1]$ and $\nu \in (0, 1/64]$, let $P_{a, \nu}$ be the polynomial $P_\sgn$ in \Cref{prop:FPAA general} with $(\delta, \beta) = (\nu/8, a)$, and define
\begin{align}
    T_{a, \nu}(x) = \f{1}{4}\Big(P_{a/6, \nu/2}\Big(\f{x + 3a/2}{3}\Big) + P_{a/6, \nu/2}\Big(\f{x - 3a/2}{3}\Big)\Big)^2.
\end{align}
Since $T_{a,\nu}$ is the square of a real odd polynomial bounded by one on $[-1, 1]$, it is a real even polynomial satisfying $0 \leq T_{a,\nu}(x) \leq 1$ for $x \in [-1, 1]$.
For $x \in [-1, 1]$, both arguments $(x \pm 3a/2)/3$ lie in $[-1, 1]$.
When $|x| \leq a$, they have opposite signs and their absolute values are at least $a/6$, and when $|x| \geq 2a$, they have the same sign as $x$ and their absolute values are at least $a/6$.
From \Cref{prop:FPAA general}, it holds that
\begin{align}
\label{eq:stage polynomials}
    &T_{a, \nu}(x) \leq \nu/8 \ \ \textrm{for} \ \ |x| \leq a, \\
\label{eq:stage polynomials large}
    &T_{a, \nu}(x) \geq 1 - \nu/8 \ \ \textrm{for} \ \ |x| \geq 2a, \\ 
\label{eq:stage polynomials sign}
    &|P_{a,\nu}(x) - \sgn(x)| \leq \nu/8 \ \ \textrm{for} \ \ |x| \geq a,
\end{align}
and the degrees of $T_{a,\nu}$ and $P_{a,\nu}$ are $\cO\big(\log{(1/\nu)}/a\big)$.

To describe the stages, for $j = 0, \ldots, J$, we set $a_j = 2^{-j}$ and $\nu_j = (\beta/a_j)^2/64$, which satisfies $\nu_j \leq 1/64$ since $\beta \leq a_j$ by Eq.~\eqref{eq:last threshold}, and write 
\begin{align}
    T_j = T_{a_j, \nu_j}, \ \ \text{and} \ \ \ P_j = P_{a_j, \nu_j},
\end{align}
for simplicity.
The polynomial $T_j$ serves as a filter: the $j$-th stage stops the part of the component with weight $T_j^2$, which is close to one for $|x| \geq 2a_j$ and to zero for $|x| \leq a_j$, applies $P_j$ to it, and passes the rest to the next stage. Let $V_{T_j}^{\sG_j\sC\sE\sB}$ and $V_{P_j}^{\sG_j'\sC\sE\sB}$ be the unitaries obtained above for the polynomials $T_j$ and $P_j$, respectively, and let $\sQ_j$, $\sY_j$, and $\sX_j$ be one-qubit systems. The $j$-th stage $A_j$ consists of the following operations, as shown in \Cref{fig:QSVT sign stages circuit}:
\begin{enumerate}
    \item Controlled by $\sQ_j$, apply $V_{T_j}$.
    \item Flip $\sY_j$ if $\sQ_j$ is in $\ket{1}$ and $\sG_j$ is in $\ket{0}$.
    \item Controlled by $\sQ_j$ and $\sY_j$, apply the Hadamard gate on $\sX_j$, then $V_{P_j}$ further controlled by $\sX_j$, and the Hadamard gate on $\sX_j$ again.
    \item Flip $\sQ_{j+1}$ if $\sQ_j$ is in $\ket{1}$ and $\sY_j$ is in $\ket{0}$.
\end{enumerate}
We refer to $\sQ_j$ as the active flag of the $j$-th stage, and say that a component stops at the $j$-th stage if $\sY_j$ is in $\ket{1}$.
Let $A = A_J\cdots A_0X^{\sQ_0}$, and let $\mathrm{CZ}^{\sY_0\sX_0\cdots\sY_J\sX_J}$ be the product of the controlled-$Z$ gates between $\sY_j$ and $\sX_j$ for $j = 0, \ldots, J$.
The whole circuit realizes the unitary $A^\dag \mathrm{CZ}AX^\sC$. The system $\sH = \sC\sD$ has $\cO(J) = \cO\big(\log{(1/\beta)}\big)$ qubits, where $\sD$ consists of $\sQ_j$, $\sG_j$, $\sY_j$, $\sG_j'$, and $\sX_j$ for all $j$, and $\sQ_{J+1}$.

\begin{figure*}[t]
    \centering
    \tikzset{
        wire/.style={line width=0.5pt},
        gate/.style={draw, fill=white, line width=0.5pt, minimum width=0.6cm, minimum height=0.5cm, inner sep=2pt},
        block/.style={gate, minimum width=1.0cm},
        qsvt/.style={gate, minimum width=1.4cm, minimum height=1.25cm},
        ctrl/.style={circle, fill=black, inner sep=0pt, minimum size=4pt},
        openctrl/.style={circle, draw, fill=white, line width=0.5pt, inner sep=0pt, minimum size=5pt},
        targ/.style={circle, draw, fill=white, line width=0.5pt, inner sep=0pt, minimum size=7pt, path picture={\draw[line width=0.5pt] (path picture bounding box.north) -- (path picture bounding box.south) (path picture bounding box.west) -- (path picture bounding box.east);}},
        brace/.style={decorate, decoration={brace, mirror, amplitude=4pt}, line width=0.4pt}
    }
    \begin{tikzpicture}[x=0.8cm, y=0.8cm, font=\small]
        \node[anchor=west] at (-1.9,0.7) {(a)};
        \node[anchor=east] at (-0.2,0) {$\sH\sE\sB$};
        \draw[wire] (0,0) -- (12.6,0);
        \draw[wire] (0.10,-0.08) -- (0.21,0.08);
        \node[gate] at (0.75,0) {$X^\sC$};
        \node[gate] at (1.95,0) {$X^{\sQ_0}$};
        \node[block] at (3.2,0) {$A_0$};
        \node[fill=white, inner xsep=4pt] at (4.35,0) {$\cdots$};
        \node[block] at (5.5,0) {$A_J$};
        \node[gate, minimum width=0.8cm] at (6.9,0) {$\mathrm{CZ}$};
        \node[block] at (8.3,0) {$A_J^\dag$};
        \node[fill=white, inner xsep=4pt] at (9.45,0) {$\cdots$};
        \node[block] at (10.6,0) {$A_0^\dag$};
        \node[gate] at (11.8,0) {$X^{\sQ_0}$};
        \draw[brace] (1.5,-0.45) -- (6.05,-0.45) node[midway, below=6pt] {$A$};
        \draw[brace] (7.75,-0.45) -- (12.25,-0.45) node[midway, below=6pt] {$A^\dag$};
    \end{tikzpicture}

    \vspace{-2mm}

    \begin{tikzpicture}[x=0.8cm, y=0.8cm, font=\small]
        \node[anchor=west] at (-1.9,0.6) {(b)};
        \foreach \y in {0,-0.8,-1.6,-2.4,-3.2,-4.0,-4.8} {
            \draw[wire] (0,\y) -- (12.6,\y);
        }
        \node[anchor=east] at (-0.7,0) {$\sQ_j$};
        \node[anchor=east] at (-0.7,-0.8) {$\sG_j$};
        \node[anchor=east] at (-0.7,-1.6) {$\sC\sE\sB$};
        \node[anchor=east] at (-0.7,-2.4) {$\sG_j'$};
        \node[anchor=east] at (-0.7,-3.2) {$\sX_j$};
        \node[anchor=east] at (-0.7,-4.0) {$\sY_j$};
        \node[anchor=east] at (-0.7,-4.8) {$\sQ_{j+1}$};
        \foreach \y in {-0.8,-2.4,-3.2,-4.0,-4.8} {
            \node[anchor=east] at (0,\y) {$\ket{0}$};
        }
        \foreach \y in {-0.8,-1.6,-2.4} {
            \draw[wire] (0.10,\y-0.08) -- (0.21,\y+0.08);
        }

        \draw[wire] (1.4,0) -- (1.4,-1.2);
        \node[ctrl] at (1.4,0) {};
        \node[qsvt] at (1.4,-1.2) {$V_{T_j}$};

        \draw[wire] (3.6,0) -- (3.6,-4.0);
        \node[ctrl] at (3.6,0) {};
        \node[openctrl] at (3.6,-0.8) {};
        \node[targ] at (3.6,-4.0) {};

        \draw[wire] (5.4,0) -- (5.4,-4.0);
        \node[ctrl] at (5.4,0) {};
        \node[ctrl] at (5.4,-4.0) {};
        \node[gate] at (5.4,-3.2) {$H$};

        \draw[wire] (7.6,0) -- (7.6,-4.0);
        \node[ctrl] at (7.6,0) {};
        \node[ctrl] at (7.6,-3.2) {};
        \node[ctrl] at (7.6,-4.0) {};
        \node[qsvt] at (7.6,-2.0) {$V_{P_j}$};

        \draw[wire] (9.8,0) -- (9.8,-4.0);
        \node[ctrl] at (9.8,0) {};
        \node[ctrl] at (9.8,-4.0) {};
        \node[gate] at (9.8,-3.2) {$H$};

        \draw[wire] (11.8,0) -- (11.8,-4.8);
        \node[ctrl] at (11.8,0) {};
        \node[openctrl] at (11.8,-4.0) {};
        \node[targ] at (11.8,-4.8) {};
    \end{tikzpicture}
    \caption{A circuit for the staged QSVT in \Cref{prop:QSVT sign stages}.
    (a) The unitary $A^\dag \mathrm{CZ}AX^\sC$, where $A = A_J\cdots A_0X^{\sQ_0}$ and $\mathrm{CZ}^{\sY_0\sX_0\cdots\sY_J\sX_J}$ is the product of the controlled-$Z$ gates between $\sY_j$ and $\sX_j$ for $j = 0, \ldots, J$.
    (b) The $j$-th stage $A_j$, where $V_{T_j}$ acts on $\sG_j\sC\sE\sB$ and $V_{P_j}$ acts on $\sG_j'\sC\sE\sB$. Filled and open circles denote controls on $\ket{1}$ and $\ket{0}$, respectively, and the open circle on $\sG_j$ denotes the control on all qubits of $\sG_j$ being in $\ket{0}$. The symbol $\oplus$ denotes the NOT gate. The qubit $\sY_j$ records whether a component stops at the $j$-th stage, and $\sQ_{j+1}$ is the active flag of the $(j+1)$-th stage. In the implementation from samples, each controlled $Q$ in $V_{T_j}$, $V_{P_j}$, and their inverses is implemented by \Cref{lem:ctrl BE pure state}.}
    \label{fig:QSVT sign stages circuit}
\end{figure*}

We first show that, when the reflections are implemented exactly, $A^\dag \mathrm{CZ}AX^\sC$ satisfies Eq.~\eqref{eq:QSVT sign stages} with $2\beta + \beta/64$ in place of $4\beta$.
To this end, we show that it applies $P_j$ to the components that stop at the $j$-th stage.
For $j = 0, \ldots, J+1$, we define an even polynomial 
\begin{align}
\label{eq:active probability}
    p_j(\lambda) = \prod_{i=0}^{j-1}\big(1 - T_i(\lambda)^2\big),
\end{align}
where $\lambda \in [-1, 1]$ is the eigenvalue of $\Pi S\Pi$ for $\ket{v}$ above, and $p_0(\lambda) = 1$.
We follow $A_{j-1}\cdots A_0X^{\sQ_0}\ket{0}^\sD\ket{v}^{\sC\sE\sB}$ stage by stage. Since all the operations in the stages preserve $\cH^\sD\otimes\cH_v$, Eq.~\eqref{eq:scalar QSVT} allows us to replace $V_{T_j}$ and $V_{P_j}$ by $T_j(\lambda)$ and $P_j(\lambda)$ on the blocks with $\sG_j$ and $\sG_j'$ in $\ket{0}$.
In the $j$-th stage, steps~1 and~2 stop the part of the component with $\ket{1}^{\sQ_j}$ that has $\sG_j$ in $\ket{0}$, whose squared norm is $T_j(\lambda)^2$ times that of the component, and step~4 passes the rest to $\ket{1}^{\sQ_{j+1}}$. Since $p_0(\lambda) = 1$ and $p_j(\lambda)(1 - T_j(\lambda)^2) = p_{j+1}(\lambda)$, we have
\begin{align}
\label{eq:active norm}
    \big\|\ketbra{1}{1}^{\sQ_j}A_{j-1}\cdots A_0X^{\sQ_0}\ket{0}^\sD\ket{v}^{\sC\sE\sB}\big\|^2 = p_j(\lambda),
\end{align}
and the component that stops at the $j$-th stage has the squared norm $p_j(\lambda)T_j(\lambda)^2$.
The stopped components are not changed by the subsequent stages. Step~3 together with $\mathrm{CZ}$ and $A^\dag$ forms the Hadamard test of $V_{P_j}$, which contributes $p_j(\lambda)T_j(\lambda)^2P_j(\lambda)$ to the expectation value of $\mathrm{CZ}$.
Hence, we have 
\begin{align}
\label{eq:expectation Z}
    \bra{0}\bra{v}A^\dag \mathrm{CZ}A\ket{0}\ket{v} = \sum_{j=0}^J p_j(\lambda)T_j(\lambda)^2P_j(\lambda) + p_{J+1}(\lambda),
\end{align}
where the last term comes from the component that does not stop at any stage.

Let $\ket{\psi_\pm} = \ket{0}(\ket{v_k^+} \pm \ket{v_k^-})/\sqrt{2}$, which are $\ket{0}^\sH\ket{\varphi}^\sE\ket{\eta_k}^\sB$ and $X^\sC\ket{0}^\sH\ket{\phi}^\sE\ket{\xi_k}^\sB$, respectively, by Eq.~\eqref{eq:v pm}.
Then, Eq.~\eqref{eq:expectation Z} implies that
\begin{align}
    &\bra{\psi_+}A^\dag \mathrm{CZ}A\ket{\psi_-} \notag\\
    &= \f{1}{2}\big(\bra{0}\bra{v_k^+} + \bra{0}\bra{v_k^-}\big)A^\dag \mathrm{CZ}A\big(\ket{0}\ket{v_k^+} - \ket{0}\ket{v_k^-}\big) \\
    &= \f{1}{2}\big(\bra{0}\bra{v_k^+}A^\dag \mathrm{CZ}A\ket{0}\ket{v_k^+} - \bra{0}\bra{v_k^-}A^\dag \mathrm{CZ}A\ket{0}\ket{v_k^-}\big) \\
    &= \f{1}{2}\sum_{j=0}^J\big(p_j(s_k)T_j(s_k)^2P_j(s_k) \notag\\
    &\hspace{6pc}- p_j(-s_k)T_j(-s_k)^2P_j(-s_k)\big) \\
    \label{eq:exact block stages}
    &= h(s_k),
\end{align}
where $h(x) = \sum_{j=0}^J p_j(x)T_j(x)^2P_j(x)$ is an odd polynomial. The cross terms vanish since $A^\dag \mathrm{CZ}A$ preserves the orthogonal subspaces $\cH_{v_k^+}$ and $\cH_{v_k^-}$, and we used that $p_j$ and $T_j$ are even and $P_j$ is odd.
Since $A^\dag \mathrm{CZ}A$ preserves $\cH_{v_k^\pm}$ for all $k$, including those for the zero singular values when $\{\ket{\eta_k}^\sB\}_k$ and $\{\ket{\xi_k}^\sB\}_k$ are completed to orthonormal bases, $A^\dag \mathrm{CZ}AX^\sC$ applies $h$ to the singular values of $M^\sB$ when the reflections are exact:
\begin{align}
    A^\dag \mathrm{CZ}AX^\sC
     = \!\hspace{1mm}
\begin{blockarray}{ccc}
& \bra{0}^\sH\bra{\phi}^\sE &  & \vspace{1mm}\\
\begin{block}{c(cc)}
  \ket{0}^\sH\ket{\varphi}^\sE \hspace*{1mm} & h^{(\rm SV)}(M^\sB) & \hspace{0mm} * \hspace{3mm}\\
   \hspace*{1mm} & * & \hspace{0mm} * \hspace{3mm}\\
\end{block}
\end{blockarray}\hspace{1mm}.
\end{align}

The components reaching the $j$-th stage have small singular values: for $j = 0, \ldots, J+1$ and $\lambda \in [-1, 1]$,
\begin{align}
\label{eq:weight bound}
    |\lambda|\sqrt{p_j(\lambda)} \leq 4a_j,
\end{align}
where $a_{J+1} = 2^{-(J+1)}$.
This holds for $|\lambda| < 4a_j$ since $p_j(\lambda) \leq 1$.
For $|\lambda| \geq 4a_j$, let $i \in [1, j-1]$ be the integer such that $2a_i \leq |\lambda| < 2a_{i-1}$, which exists since $|\lambda| \leq 1 < 2a_0$.
Since $|\lambda| \geq 2a_l$ for $l = i, \ldots, j-1$, Eq.~\eqref{eq:stage polynomials large} implies $1 - T_l(\lambda)^2 \leq \nu_l/4 \leq 1/256$ for these $l$, and thus
\begin{align}
    |\lambda|\sqrt{p_j(\lambda)} 
    &\leq 2a_{i-1}\cdot 16^{-(j-i)} \\
    &= 4a_j\cdot 8^{-(j-i)} \\
    &\leq 4a_j.
\end{align}

We bound $s_k\big(1 - h(s_k)\big)$ for the singular values $s_k$ of $M^\sB$.
Since $\sum_{j=0}^J p_jT_j^2 + p_{J+1} = 1$ holds by the definition of $p_j$, we have that, for $x \in [0, 1]$,
\begin{align}
    x\big(1 - h(x)\big) = \sum_{j=0}^J x\big(1 - P_j(x)\big)p_j(x)T_j(x)^2 + xp_{J+1}(x),
\end{align}
which is non-negative since $|P_j(x)| \leq 1$.
The sum is less than $\beta/64$, since each of its terms is at most $a_j\nu_j/2$ and $\sum_{j=0}^J a_j\nu_j = (\beta^2/64)\sum_{j=0}^J 2^j < \beta^2/(32a_J) \leq \beta/32$ by Eq.~\eqref{eq:last threshold}.
Indeed, for $x \geq a_j$, the $j$-th term is at most $(\nu_j/8)x\sqrt{p_j(x)} \leq a_j\nu_j/2$ by Eqs.~\eqref{eq:stage polynomials sign} and~\eqref{eq:weight bound}, where we used $p_j \leq \sqrt{p_j}$ and $T_j^2 \leq 1$, and for $x < a_j$, it is at most $2a_jT_j(x)^2 \leq a_j\nu_j/4$ by Eq.~\eqref{eq:stage polynomials}.
The last term is at most $2\beta$ at $x = s_k$, as follows. By Eq.~\eqref{eq:weight bound}, it satisfies $s_kp_{J+1}(s_k) \leq 2a_J\sqrt{p_{J+1}(s_k)}$.
If $s_k \geq 2a_J$, Eq.~\eqref{eq:stage polynomials large} implies $p_{J+1}(s_k) \leq 1 - T_J(s_k)^2 \leq \nu_J/4 = \big(\beta/(16a_J)\big)^2$, and thus the last term is at most $\beta/8$.
Otherwise, either $s_k = 0$, for which the last term vanishes, or $s \leq s_k < 2a_J$ by the assumption on $M^\sB$, for which Eq.~\eqref{eq:last threshold} implies $a_J = \beta$ and the last term is at most $2\beta$.
Combining these bounds, we obtain that, for all $k$,
\begin{align}
\label{eq:exact error stages}
    s_k\big(1 - h(s_k)\big) \leq 2\beta + \f{\beta}{64}.
\end{align}
Since the matrix element in Eq.~\eqref{eq:QSVT sign stages} is $h(s_k)$ by Eq.~\eqref{eq:exact block stages}, Eq.~\eqref{eq:QSVT sign stages} holds with $2\beta + \beta/64$ in place of $4\beta$ when the reflections are exact.

We now implement the reflections in $A$ and $A^\dag$ approximately by using samples.
Since $\ketbra{0}{0}^\sC\otimes\ketbra{\varphi}{\varphi}^\sE$ and $\ketbra{1}{1}^\sC\otimes\ketbra{\phi}{\phi}^\sE$ are orthogonal projectors, we have 
\begin{align}
    Q^{\sC\sE} = -e^{i\pi\ketbra{0}{0}^\sC\otimes\ketbra{\varphi}{\varphi}^\sE}e^{i\pi\ketbra{1}{1}^\sC\otimes\ketbra{\phi}{\phi}^\sE}.
\end{align}
From Eq.~\eqref{eq:W ctrl pure state}, it holds that $e^{i\pi\ketbra{1}{1}^\sC\otimes\ketbra{\phi}{\phi}^\sE} = H^\sC X^\sC W_\phi^{\sC\sE}H^\sC$ and $e^{i\pi\ketbra{0}{0}^\sC\otimes\ketbra{\varphi}{\varphi}^\sE} = X^\sC H^\sC X^\sC W_\varphi^{\sC\sE}H^\sC X^\sC$, where $W_\varphi^{\sC\sE}$ and $W_\phi^{\sC\sE}$ are the unitaries in Eq.~\eqref{eq:W ctrl pure state} for $\ket{\varphi}^\sE$ and $\ket{\phi}^\sE$, respectively.
In the $j$-th stage, each use of $Q$ is controlled by a qubit $\sF$ storing whether $\sQ_j$ is in $\ket{1}$, or whether $\sQ_j$, $\sY_j$, and $\sX_j$ are all in $\ket{1}$.
We implement the controlled versions of $W_\varphi^{\sC\sE}$ and $W_\phi^{\sC\sE}$ by \Cref{lem:ctrl BE pure state}, each with $m_j$ samples, where $m_j$ is chosen below, and the other controlled gates exactly.
\Cref{lem:ctrl BE pure state} implies that the channel for each of them has a Kraus operator of the form $\ketbra{0}{0}^\sF\otimes\bI^{\sC\sE} + \ketbra{1}{1}^\sF\otimes K^{\sC\sE}$, where $K^{\sC\sE}$ satisfies $\|K^{\sC\sE} - W^{\sC\sE}\|_\infty \leq \pi^2/(2m_j)$ for $W^{\sC\sE} = W_\varphi^{\sC\sE}$ or $W_\phi^{\sC\sE}$.
Thus, the resulting channel for the controlled $Q$ has a Kraus operator of the form $\ketbra{0}{0}\otimes\bI + \ketbra{1}{1}\otimes\til{Q}$, where $\til{Q}$ satisfies $\|\til{Q} - Q\|_\infty \leq \pi^2/m_j$.
We here used that Kraus operators of quantum channels have operator norms at most one.

Let $q_j$ be the number of uses of $Q$ in the $j$-th stage, and let $\til{A}_j$ be the operator obtained from $A_j$ by replacing $Q$ with $\til{Q}$ in each controlled $Q$.
Then, $\til{A} = \til{A}_J\cdots\til{A}_0X^{\sQ_0}$ is a Kraus operator of the channel that approximates $A$ by samples. Since $\til{A}_j$ and $A_j$ act as the identity when $\sQ_j$ is in $\ket{0}$, and since $\|\til{A}_j - A_j\|_\infty \leq q_j\pi^2/m_j$, we have
\begin{align}
\label{eq:stage error}
    \big\|(\til{A}_j - A_j)\ket{\psi}\big\| \leq \eta_j\big\|\ketbra{1}{1}^{\sQ_j}\ket{\psi}\big\|,
\end{align}
for any vector $\ket{\psi}$, where $\eta_j = q_j\pi^2/m_j$.
Similarly, we implement $A^\dag$ by reversing the order of the gates and replacing each gate with its inverse, where each controlled $Q$, which is Hermitian, is implemented in the same way.
Let $\til{A}_{\rm inv}$ be the Kraus operator of this channel obtained by replacing $Q$ with $\til{Q}$, which differs from $\til{A}^\dag$ in general since $\til{Q}$ is not unitary.
Then, $(\til{A}_{\rm inv})^\dag$ is the product of the operators corresponding to the stages, each of which satisfies Eq.~\eqref{eq:stage error} in place of $\til{A}_j$.

We evaluate the effect of this replacement on Eq.~\eqref{eq:exact block stages}.
Equation~\eqref{eq:stage error} implies that
\begin{align}
    &\big\|(\til{A} - A)\ket{\psi_\pm}\big\| \notag\\
    &= \Big\|\sum_{j=0}^J\til{A}_J\cdots\til{A}_{j+1}(\til{A}_j - A_j)A_{j-1}\cdots A_0X^{\sQ_0}\ket{\psi_\pm}\Big\| \\
    &\leq \sum_{j=0}^J\eta_j\big\|\ketbra{1}{1}^{\sQ_j}A_{j-1}\cdots A_0X^{\sQ_0}\ket{\psi_\pm}\big\| \\
    &= \sum_{j=0}^J\eta_j\sqrt{p_j(s_k)},
\end{align}
where we used $\|\til{A}_l\|_\infty \leq 1$ in the inequality, and Eq.~\eqref{eq:active norm} for $\ket{v_k^\pm}$ with $p_j(s_k) = p_j(-s_k)$ and the orthogonality of $\cH_{v_k^+}$ and $\cH_{v_k^-}$ in the last equality. The same bound holds for $\big\|((\til{A}_{\rm inv})^\dag - A)\ket{\psi_\pm}\big\|$.
Since $\til{A}_{\rm inv}\mathrm{CZ}\til{A} - A^\dag \mathrm{CZ}A = (\til{A}_{\rm inv} - A^\dag)\mathrm{CZ}\til{A} + A^\dag \mathrm{CZ}(\til{A} - A)$, Eq.~\eqref{eq:exact block stages} implies that
\begin{align}
    &\big|\bra{\psi_+}\til{A}_{\rm inv}\mathrm{CZ}\til{A}\ket{\psi_-} - h(s_k)\big| \notag\\
    &\leq \big\|((\til{A}_{\rm inv})^\dag - A)\ket{\psi_+}\big\| + \big\|(\til{A} - A)\ket{\psi_-}\big\| \\
    \label{eq:matrix element error}
    &\leq 2\sum_{j=0}^J\eta_j\sqrt{p_j(s_k)}.
\end{align}

Let $\cJ^{\sB\sH\sE}$ be the channel realized by the whole circuit, which approximates $A^\dag \mathrm{CZ}AX^\sC$ from samples. Its Kraus operator $K^{\sB\sH\sE} = \til{A}_{\rm inv}\mathrm{CZ}\til{A}X^\sC$ has $\bra{\psi_+}\til{A}_{\rm inv}\mathrm{CZ}\til{A}\ket{\psi_-}$ as the matrix element in Eq.~\eqref{eq:QSVT sign stages}. Combining Eq.~\eqref{eq:matrix element error} with Eqs.~\eqref{eq:weight bound} and~\eqref{eq:exact error stages}, we obtain that
\begin{align}
    &s_k\big|1 - \bra{\eta_k}^\sB\bra{0}^\sH\bra{\varphi}^\sE K^{\sB\sH\sE}\ket{0}^\sH\ket{\phi}^\sE\ket{\xi_k}^\sB\big| \notag\\
    &\leq s_k\big(1 - h(s_k)\big) + s_k\big|\bra{\psi_+}\til{A}_{\rm inv}\mathrm{CZ}\til{A}\ket{\psi_-} - h(s_k)\big| \\
    &\leq s_k\big(1 - h(s_k)\big) + 2\sum_{j=0}^J\eta_js_k\sqrt{p_j(s_k)} \\
    &\leq 2\beta + \f{\beta}{64} + 8\sum_{j=0}^J a_j\eta_j.
\end{align}
We now choose $m_j = \big\lceil 64\pi^2q_ja_j^{3/2}/(\beta a_J^{1/2})\big\rceil$, which is at least $3$ since $a_j \geq a_J \geq \beta$, and ensures $a_j\eta_j = \pi^2q_ja_j/m_j \leq (\beta/64)(a_J/a_j)^{1/2}$.
Then, we have
\begin{align}
    8\sum_{j=0}^Ja_j\eta_j
    \leq \f{\beta}{8}\sum_{j=0}^J 2^{-(J-j)/2}
    < \f{\beta}{2}.
\end{align}
Thus, Eq.~\eqref{eq:QSVT sign stages} follows.

Finally, we evaluate the computational costs.
From \Cref{prop:general QSVT} and the degrees of $T_j$ and $P_j$, the $j$-th stage uses $V^{\sL\sC\sE\sB}$ and its inverse a total of $q_j = \cO\big(\log{(1/\nu_j)}/a_j\big)$ times.
Since $1/\nu_j = 64(a_j/a_J)^2(a_J/\beta)^2$, we have $\log{(1/\nu_j)} = \cO\big(J-j+\log{(2a_J/\beta)}\big)$, and thus $q_j = \cO\big((J-j+\log{(2a_J/\beta)})2^j\big)$.
As each use of $V^{\sL\sC\sE\sB}$ requires one use of $Q$ and two uses of $S$, the number of uses of the controlled versions of $U^{\sB\sE}$ and $(U^{\sB\sE})^\dag$ in $A$ and $A^\dag$ is
\begin{align}
    \cO\Big(\sum_{j=0}^J q_j\Big) = \cO\Big(\f{1}{a_J}\log{\Big(\f{2a_J}{\beta}\Big)}\Big).
\end{align}
Each use of $Q$ in the $j$-th stage consumes $m_j$ samples of each of $\ket{\varphi}^\sE$ and $\ket{\phi}^\sE$, and thus the total number of samples is
\begin{align}
    \cO\Big(\sum_{j=0}^J q_jm_j\Big) 
    &= \cO\bigg(\sum_{j=0}^J \f{q_j^2a_j^{3/2}}{\beta a_J^{1/2}}\bigg) \\
    &= \cO\bigg(\f{1}{\beta a_J}\sum_{j=0}^J \Big(J-j+\log{\Big(\f{2a_J}{\beta}\Big)}\Big)^2 2^{-(J-j)/2}\bigg) \\
    &= \cO\Big(\f{1}{\beta a_J}\Big(\log{\Big(\f{2a_J}{\beta}\Big)}\Big)^2\Big).
\end{align}
Equation~\eqref{eq:last threshold} implies $\max\{\beta, s\}/4 < a_J \leq \max\{\beta, s\}$, and hence the number of samples of $\ket{\varphi}^\sE$ and $\ket{\phi}^\sE$ is
\begin{align}
    n = \cO\bigg(\f{1}{\beta\max\{\beta, s\}}\Big(\log{\Big(\f{2\max\{\beta, s\}}{\beta}\Big)}\Big)^2\bigg),
\end{align}
and the number of uses of the controlled versions of $U^{\sB\sE}$ and $(U^{\sB\sE})^\dag$ is $\cO\big((1/\max\{\beta, s\})\log{(2\max\{\beta, s\}/\beta)}\big)$, which is $\cO(\beta n)$ since $\log{(2\max\{\beta, s\}/\beta)} \geq 1$.

By \Cref{lem:ctrl BE pure state}, the density matrix exponentiation uses $\cO(\log{d_\sE})$ one- and two-qubit gates for each sample, and the other gates in $A$ and $A^\dag$, including those in the QSVT and on the flags, are $\cO\Big(\sum_{j=0}^J q_j\Big)$ in total, which is at most the number of samples.
Thus, in addition to the uses of the controlled versions of $U^{\sB\sE}$ and $(U^{\sB\sE})^\dag$, the circuit consists of $\cO(\log{d_\sE})$ one- and two-qubit gates for each sample.
Since the samples are used sequentially, the circuit uses the systems $\sB$, $\sE$, and $\sH$, a system storing one sample at a time, and a constant number of control qubits, and hence $\cO\big(\log{(d_\sB d_\sE)} + \log{(1/\beta)}\big)$ qubits simultaneously.
The circuit description can be computed classically in $\poly\big(1/\beta, \log{d_\sE}\big)$ time, as it requires the phase factors of the QSVT for the polynomials of degrees $\cO\big(\log{(1/\nu_j)}/a_j\big)$, where $\nu_j \geq \beta^2/64$ and $a_j \geq \beta$.

\end{proof}


\section{Lower bound on the number of samples in quantum state merging}
\label{sec:lower bound merging}

We provide a lower bound on the number of samples for quantum state merging and compare it with the upper bound in \Cref{thm:Uhlmann state merging}.
We use the same setting and notation as in \Cref{sec:entdis and uhl}.
The lower bound applies to any operation of Bob, which is not necessarily based on the Uhlmann transformation and may process the samples collectively.

In the setting of \Cref{sec:setting of state merging}, we consider the case without the reference $\sR$, i.e., $d_\sR = 1$. We also take $d_\sS = d_\sB = r$ with $r \geq 2$ and $d_\sX = N$, so that $d_\sA = rN$.
Let $p_x > 0$ for $x = 1, \ldots, N$ satisfy $\sum_{x=1}^N p_x = 1$.
For unitaries $G_x^\sB$, which are unknown to Bob, we write $G = (G_1^\sB, \ldots, G_N^\sB)$ and consider the family of initial states $\ket{\omega_G}^{\sA\sB}$ given by
\begin{align}
\label{inteq:merging lower initial}
    \ket{\omega_G}^{\sA\sB} = \sum_x\sqrt{p_x}(U^\sA)^\dag(U_x^{\sA\rarr\sS})^\dag G_x^\sB\ket{\Phi}^{\sS\sB},
\end{align}
where $(U_x^{\sA\rarr\sS})^\dag$ maps $\cH^\sS$ to $\cH_x^\sA$, and $U^\sA$ is the Haar random unitary used in Alice's operation $\cA^{\sA\rarr\sS\sX}$.
Since the subspaces $\cH_x^\sA$ are mutually orthogonal, $\cA^{\sA\rarr\sS\sX}$ yields each outcome $x$ with probability $p_x$ and the corresponding state is given by $\ket{\psi_x}^{\sS\sB} = G_x^\sB\ket{\Phi}^{\sS\sB}$.
For each $x$, the Uhlmann unitary $V_x^{\hat{\sA}\sB\hat{\sS}}$ in Bob's operation $\cB^{\sB\sX\rarr\hat{\sA}\sB\hat{\sS}}$ satisfies
\begin{align}
    \rF\big(V_x^{\hat{\sA}\sB\hat{\sS}}\ket{\psi_x}^{\sS\sB}\ket{0}^{\hat{\sA}\hat{\sS}}, \ket{\omega_G}^{\hat{\sA}\sB}\ket{\Phi}^{\sS\hat{\sS}}\big) = \rF(\psi_x^\sS, \pi^\sS) = 1,
\end{align}
where the first equality is Eq.~\eqref{inteq:111} with $d_\sR = 1$, and the second follows from $\psi_x^\sS = \pi^\sS$.
Thus, $\cB^{\sB\sX\rarr\hat{\sA}\sB\hat{\sS}}$ accomplishes the state merging exactly for this family.
The two models in \Cref{sec:implement uhl in state merging} coincide for this family, as $\ket{\omega_G}^{\sA\sB}$ is pure on $\sA\sB$.

Since the Uhlmann unitaries $V_x^{\hat{\sA}\sB\hat{\sS}}$ depend on the unknown unitaries $G_x^\sB$, Bob cannot implement $\cB^{\sB\sX\rarr\hat{\sA}\sB\hat{\sS}}$ itself and has to approximate it using samples of $\ket{\omega_G}^{\sA'\sB'}$.
The following proposition shows that the error of any such operation of Bob, which we denote by $\cC^{\sB\sX\rarr\hat{\sA}\sB\hat{\sS}}$, is bounded from below in terms of the number of samples.

\begin{proposition}
\label{prop:lower merging}
In the setting above, suppose that Bob's operation $\cC^{\sB\sX\rarr\hat{\sA}\sB\hat{\sS}}$ uses $\zeta$ samples of $\ket{\omega_G}^{\sA'\sB'}$, together with a Stinespring isometry $V_\cA^{\sA\rarr\sS\sX\sE}$ of Alice's operation $\cA^{\sA\rarr\sS\sX}$ and its inverse any number of times.
Then,
\begin{align}
\label{inteq:merging lower bound}
    &\bE_G\Big[1 - \rF\big(\cC^{\sB\sX\rarr\hat{\sA}\sB\hat{\sS}}\circ\cA^{\sA\rarr\sS\sX}(\omega_G^{\sA\sB}), \ket{\omega_G}^{\hat{\sA}\sB}\ket{\Phi}^{\sS\hat{\sS}}\big)\Big] \notag\\
    &\geq \sum_x\f{p_x(r^2-1)}{4\zeta p_x + 2r^2},
\end{align}
where $\bE_G[\cdot]$ denotes the expectation over independent Haar random unitaries $G_1^\sB, \ldots, G_N^\sB$.
\end{proposition}

We prove \Cref{prop:lower merging} by reducing the state merging to the recovery of $\ket{\Phi}^{\sS\hat{\sS}}$ and applying \Cref{lem:weighted unitary recovery} to each outcome.

\begin{proof}[Proof of \Cref{prop:lower merging}]

We choose the Stinespring isometry $V_\cA^{\sA\rarr\sS\sX\sE} = \sum_x U_x^{\sA\rarr\sS}\Pi_x^\sA U^\sA \otimes \ket{x}^\sX\ket{x}^\sE$, which gives
\begin{align}
    \ket{\Psi_G}^{\sS\sB\sX\sE}
    &= V_\cA^{\sA\rarr\sS\sX\sE}\ket{\omega_G}^{\sA\sB} \\
    &= \sum_x\sqrt{p_x}G_x^\sB\ket{\Phi}^{\sS\sB}\ket{x}^\sX\ket{x}^\sE.
\end{align}
Each sample $\ket{\omega_G}^{\sA'\sB'}$ is obtained from a copy of $\ket{\Psi_G}^{\sS'\sB'\sX'\sE'}$ by a quantum channel that applies $(V_\cA^{\sA'\rarr\sS'\sX'\sE'})^\dag$ on the range of $V_\cA^{\sA'\rarr\sS'\sX'\sE'}$, in which $\ket{\Psi_G}^{\sS'\sB'\sX'\sE'}$ lies.
Let $\cD_G^{\sB\sX\rarr\hat{\sS}}$ be the quantum channel realized by the following procedure using $\zeta$ copies of $\ket{\Psi_G}^{\sS'\sB'\sX'\sE'}$: it converts each copy into a sample of $\ket{\omega_G}^{\sA'\sB'}$ in this way, applies $\cC^{\sB\sX\rarr\hat{\sA}\sB\hat{\sS}}$ with these samples, and traces out $\hat{\sA}\sB$.
The procedure itself is independent of $G$, which enters only through the copies.
We then have
\begin{align}
    &\rF\big(\cC^{\sB\sX\rarr\hat{\sA}\sB\hat{\sS}}\circ\cA^{\sA\rarr\sS\sX}(\omega_G^{\sA\sB}), \ket{\omega_G}^{\hat{\sA}\sB}\ket{\Phi}^{\sS\hat{\sS}}\big) \notag\\
    &\leq \rF\big(\tr_{\hat{\sA}\sB}\circ\cC^{\sB\sX\rarr\hat{\sA}\sB\hat{\sS}}\circ\cA^{\sA\rarr\sS\sX}(\omega_G^{\sA\sB}), \ket{\Phi}^{\sS\hat{\sS}}\big) \\
    &= \rF\Big(\cD_G^{\sB\sX\rarr\hat{\sS}}\Big(\sum_x p_x\ketbra{\psi_x}{\psi_x}^{\sS\sB}\otimes\ketbra{x}{x}^\sX\Big), \ket{\Phi}^{\sS\hat{\sS}}\Big) \\
    \label{inteq:merging branch decomposition}
    &= \sum_x p_x f_x(G),
\end{align}
where $f_x(G) = \rF\big(\cD_G^{\sB\sX\rarr\hat{\sS}}\big(\ketbra{\psi_x}{\psi_x}^{\sS\sB}\otimes\ketbra{x}{x}^\sX\big), \ket{\Phi}^{\sS\hat{\sS}}\big)$.
We used the monotonicity of the fidelity under the partial trace in the first inequality, the definition of $\cD_G^{\sB\sX\rarr\hat{\sS}}$ and $\cA^{\sA\rarr\sS\sX}(\omega_G^{\sA\sB}) = \sum_x p_x\ketbra{\psi_x}{\psi_x}^{\sS\sB}\otimes\ketbra{x}{x}^\sX$ in the first equality, and the linearity of the fidelity for the pure target state in the last equality.

Fix $x$ and the unitaries $G_y^\sB$ for $y \neq x$.
Since Bob may ignore additional information, an upper bound on $f_x(G)$ derived when the descriptions of these unitaries are supplied to Bob also holds without them.
For $p_x < 1$, we define the state
\begin{align}
    \ket{\xi_x}^{\sS'\sB'\sX'\sE'} = \f{1}{\sqrt{1-p_x}}\sum_{y \neq x}\sqrt{p_y}G_y^{\sB'}\ket{\Phi}^{\sS'\sB'}\ket{y}^{\sX'}\ket{y}^{\sE'},
\end{align}
which is known to Bob and orthogonal to $\ket{i}^{\sS'}\ket{j}^{\sB'}\ket{x}^{\sX'}\ket{x}^{\sE'}$ for all $i$ and $j$.
Each copy of $\ket{\Psi_G}$ is then written as
\begin{align}
    &\ket{\Psi_G}^{\sS'\sB'\sX'\sE'} \notag\\
    &= \sqrt{p_x}G_x^{\sB'}\ket{\Phi}^{\sS'\sB'}\ket{x}^{\sX'}\ket{x}^{\sE'} + \sqrt{1-p_x}\ket{\xi_x}^{\sS'\sB'\sX'\sE'}.
\end{align}

We here recall \Cref{lem:weighted unitary recovery} with $d = r$, written in the present systems.
Let $\ket{0}^\sS$ and $\ket{0}^\sB$ be unit vectors orthogonal to $\cH^\sS$ and $\cH^\sB$, respectively, and for a unitary $U^\sB$ on $\cH^\sB$ and $p \in (0, 1]$, let
\begin{align}
    \ket{\rho_{U,p}}^{\sS\sB} = \sqrt{1-p}\ket{0}^\sS\ket{0}^\sB + \sqrt{p}U^\sB\ket{\Phi}^{\sS\sB}.
\end{align}
If a quantum algorithm that does not depend on $U^\sB$ realizes a quantum channel $\cT_U^{\sB\rarr\hat{\sS}}$ using $n$ samples of $\ket{\rho_{U,p}}^{\sS'\sB'}$, then
\begin{align}
\label{inteq:merging lemma recalled}
    \bE_U\Big[\rF\big(\cT_U^{\sB\rarr\hat{\sS}}\big(U^\sB\ketbra{\Phi}{\Phi}^{\sS\sB}(U^\sB)^\dag\big), \ket{\Phi}^{\sS\hat{\sS}}\big)\Big] \leq 1 - \f{r^2-1}{4np + 2r^2},
\end{align}
where $\bE_U[\cdot]$ denotes the expectation over a Haar random unitary $U^\sB$ on $\cH^\sB$.

The isometry that maps $\ket{0}^{\sS'}\ket{0}^{\sB'}$ to $\ket{\xi_x}^{\sS'\sB'\sX'\sE'}$ and $\ket{i}^{\sS'}\ket{j}^{\sB'}$ to $\ket{i}^{\sS'}\ket{j}^{\sB'}\ket{x}^{\sX'}\ket{x}^{\sE'}$ for $1 \leq i, j \leq r$ is independent of $G_x^\sB$, and maps $\ket{\rho_{G_x,p_x}}^{\sS'\sB'}$ to $\ket{\Psi_G}^{\sS'\sB'\sX'\sE'}$.
For $p_x = 1$, i.e., $N = 1$, the same holds without $\ket{0}^{\sS'}\ket{0}^{\sB'}$.
Applying this isometry to each of $\zeta$ samples of $\ket{\rho_{G_x,p_x}}^{\sS'\sB'}$ and then the procedure of $\cD_G^{\sB\sX\rarr\hat{\sS}}$ with the input label $x$, we obtain such an algorithm with $U^\sB \gets G_x^\sB$, $p \gets p_x$, and $n \gets \zeta$, which realizes $\cD_G^{\sB\sX\rarr\hat{\sS}}(\cdot\otimes\ketbra{x}{x}^\sX)$.
Since $\ket{\psi_x}^{\sS\sB} = G_x^\sB\ket{\Phi}^{\sS\sB}$, Eq.~\eqref{inteq:merging lemma recalled} implies that, for every choice of $G_y^\sB$ with $y \neq x$,
\begin{align}
    \bE_{G_x}\big[1 - f_x(G)\big] \geq \f{r^2-1}{4\zeta p_x + 2r^2}.
\end{align}
Since $G_1^\sB, \ldots, G_N^\sB$ are independent, taking the expectation over $G_y^\sB$ with $y \neq x$ gives $\bE_G\big[1 - f_x(G)\big] \geq (r^2-1)/(4\zeta p_x + 2r^2)$.
Combining this with Eq.~\eqref{inteq:merging branch decomposition} and $\sum_x p_x = 1$, we obtain
\begin{align}
    &\bE_G\Big[1 - \rF\big(\cC^{\sB\sX\rarr\hat{\sA}\sB\hat{\sS}}\circ\cA^{\sA\rarr\sS\sX}(\omega_G^{\sA\sB}), \ket{\omega_G}^{\hat{\sA}\sB}\ket{\Phi}^{\sS\hat{\sS}}\big)\Big] \notag\\
    &\geq \sum_x p_x\bE_G\big[1 - f_x(G)\big] \\
    &\geq \sum_x\f{p_x(r^2-1)}{4\zeta p_x + 2r^2},
\end{align}
which proves Eq.~\eqref{inteq:merging lower bound}.

\end{proof}

When the probabilities are uniform, \Cref{prop:lower merging} gives a lower bound in terms of
\begin{align}
    \Lambda = \sum_x\sqrt{p_x}r\Big(\sqrt{\psi_x^{\sR\sS}}\big(\sqrt{\omega^\sR}\otimes\sqrt{\pi^\sS}\big)\Big),
\end{align}
in Eq.~\eqref{inteq:143}, where $r(\cdot)$ denotes the rank of an input matrix.
We compare this lower bound with \Cref{thm:Uhlmann state merging}.

\begin{corollary}
\label{cor:lower merging uniform}
In the setting of \Cref{prop:lower merging}, let $p_x = 1/N$ for all $x$ and $\delta \in (0, 1/4]$.
If Bob's operation $\cC^{\sB\sX\rarr\hat{\sA}\sB\hat{\sS}}$ satisfies
\begin{align}
    \rF\big(\cC^{\sB\sX\rarr\hat{\sA}\sB\hat{\sS}}\circ\cA^{\sA\rarr\sS\sX}(\omega_G^{\sA\sB}), \ket{\omega_G}^{\hat{\sA}\sB}\ket{\Phi}^{\sS\hat{\sS}}\big) \geq 1 - \delta,
\end{align}
for all unitaries $G_1^\sB, \ldots, G_N^\sB$, then
\begin{align}
    \zeta \geq \f{Nr^2}{16\delta} = \f{\Lambda^2}{16\delta},
\end{align}
where $\ket{\omega_G}^{\sA\sB}$ is given by Eq.~\eqref{inteq:merging lower initial}.
\end{corollary}

\begin{proof}[Proof of \Cref{cor:lower merging uniform}]

Since the average is at most the maximum, \Cref{prop:lower merging} with $p_x = 1/N$ gives $\delta \geq (r^2-1)/(4\zeta/N + 2r^2)$, i.e., $\zeta \geq N(r^2 - 1 - 2\delta r^2)/(4\delta)$.
Using $r \geq 2$ and $\delta \leq 1/4$, we obtain $\zeta \geq Nr^2/(16\delta)$.
Since $d_\sR = 1$ and $\psi_x^\sS = \pi^\sS$, the operator $\sqrt{\psi_x^{\sR\sS}}\big(\sqrt{\omega^\sR}\otimes\sqrt{\pi^\sS}\big)$ equals $\pi^\sS$ for all $x$, whose rank is $r$, and thus $\Lambda = r\sqrt{N}$.

\end{proof}

For uniform probabilities, $\Psi^{\sS\sE} = \pi^\sS\otimes\pi^\sE$, and the condition in Eq.~\eqref{inteq:141} holds with decoupling parameter $\epsilon = 0$.
In addition, the quantity in Eq.~\eqref{inteq:142}, i.e.,
\begin{align}
    m_{\rm min} = \min_x\Big[\sqrt{p_x}s_{\rm min}\Big(\sqrt{\psi_x^{\sR\sS}}\big(\sqrt{\omega^\sR}\otimes\sqrt{\pi^\sS}\big)\Big)\Big]
\end{align}
equals $1/(r\sqrt{N}) = 1/\Lambda$, as implied in the proof of \Cref{cor:lower merging uniform}, where $s_{\rm min}(\cdot)$ denotes the minimum non-zero singular value of an input matrix.
Hence, the number of samples in \Cref{thm:Uhlmann state merging} becomes
\begin{align}
    \zeta 
    &= \cO\bigg(\min\bigg\{\f{\Lambda}{\delta m_{\rm min}}\Big(\log{\Big(\f{1}{\delta}\Big)}\Big)^2, \f{\Lambda^2}{\delta^2}\bigg\}\bigg) \\
    &= \cO\Big(\f{\Lambda^2}{\delta}\Big(\log{\Big(\f{1}{\delta}\Big)}\Big)^2\Big),
\end{align}
which matches the lower bound in \Cref{cor:lower merging uniform} up to the factor $(\log{(1/\delta)})^2$.
This lower bound holds not only for our algorithm but for any operation $\cC^{\sB\sX\rarr\hat{\sA}\sB\hat{\sS}}$ of Bob with the same guarantee as \Cref{thm:Uhlmann state merging}, i.e.,
\begin{align}
    \rF\big(\cC^{\sB\sX\rarr\hat{\sA}\sB\hat{\sS}}\circ\cA^{\sA\rarr\sS\sX}(\omega^{\sR\sA\sB}), \ket{\omega}^{\sR\hat{\sA}\sB}\ket{\Phi}^{\sS\hat{\sS}}\big) \geq 1 - \epsilon - \delta,
\end{align}
for all initial states $\ket{\omega}^{\sR\sA\sB}$ and all $\epsilon$ satisfying the decoupling condition in Eq.~\eqref{inteq:141}, since this guarantee with $\epsilon = 0$ reduces to the condition in \Cref{cor:lower merging uniform} for this family.
As this family attains $m_{\rm min} = 1/(r\sqrt{N}) = 1/\Lambda$ for any $r \geq 2$ and $N \geq 1$, we conclude that the dependence $\Lambda/(\delta m_{\rm min})$ in \Cref{thm:Uhlmann state merging}, which equals $1/(\delta m_{\rm min}^2)$ for this family, is optimal in the worst case up to the logarithmic factor, and the quadratic dependence on $\Lambda$ cannot be improved.


\end{document}